\documentclass{article}

\usepackage{fullpage}
\usepackage{graphicx}
\usepackage{xspace}
\usepackage{url}
\usepackage{amssymb}
\usepackage{latexsym}
\usepackage{amsthm}
\usepackage{amsmath}
\usepackage{amstext}
\usepackage{times}
\usepackage{amstext}
\usepackage{calc}
\usepackage{amsopn}
\usepackage[noend]{algorithmic}
\usepackage[boxed]{algorithm}
\usepackage{eucal}
\usepackage{latexsym}
\usepackage{tabularx}

\usepackage[
normalsections, 
normalmargins, 
normallists, 
normalfloats, 
normalindent, 
normaltitle, 
normalleading, 
normallooseness, 
]{savetrees}

{\bf}{\rm}

\newcommand{\Z}{\ensuremath{\mathbb{Z}}}

\newcommand{\ra}{\ensuremath{\rightarrow}}

\newcommand{\tw}{\ensuremath{\mathbf{tw}}}
\newcommand{\pw}{\ensuremath{\mathbf{pw}}}
\newcommand{\eps}{\varepsilon}

\newcommand{\indeg}{{\rm indeg}}
\newcommand{\outdeg}{{\rm outdeg}}

\newcommand{\insubs}{{\rm insubs}}
\newcommand{\outsubs}{{\rm outsubs}}
\newcommand{\subs}{{\rm subs}}

\newcommand{\cut}{{\rm cut}}

\newcommand{\bag}[1]{B_{#1}}
\newcommand{\subbags}[1]{V_{#1}}
\newcommand{\subedges}[1]{E_{#1}}
\newcommand{\subarcs}[1]{E_{#1}}
\newcommand{\weight}{\omega}
\newcommand{\anchor}{v_1}
\newcommand{\rootv}{r}

\newcommand{\targetW}{W}
\newcommand{\targetk}{k}
\newcommand{\iterW}{w}
\newcommand{\iterk}{i}

\newcommand{\maxw}{N}
\newcommand{\treedecomp}{\mathbb{T}}

\newcommand{\univ}{U}
\newcommand{\sols}{\mathcal{S}}
\newcommand{\cand}{\mathcal{R}}
\newcommand{\objs}{\mathcal{C}}
\newcommand{\conncomp}{\mathtt{cc}}
\newcommand{\countproc}{\mathtt{CountC}}
\newcommand{\compnomark}{\mathtt{cc}}

\newcommand{\mustset}{S}

\floatname{algorithm}{Algorithm}

\renewcommand{\subset}{\subseteq}

\newcommand{\defproblemu}[3]{
  \vspace{1mm}
\noindent\fbox{
  \begin{minipage}{\textwidth}
  #1 \\
  {\bf{Input:}} #2  \\
  {\bf{Question:}} #3
  \end{minipage}
  }
  \vspace{1mm}
}
\newcommand{\defparproblemu}[4]{
  \vspace{1mm}
\noindent\fbox{
  \begin{minipage}{\textwidth}
  \begin{tabular*}{\textwidth}{@{\extracolsep{\fill}}lr} #1 & {\bf{Parameter:}} #3 \\ \end{tabular*}
  {\bf{Input:}} #2  \\
  {\bf{Question:}} #4
  \end{minipage}
  }
  \vspace{1mm}
}

\newcommand{\zero}{{\mathbf{0}}}
\newcommand{\zeron}{{\mathbf{0}_N}}
\newcommand{\zeroy}{{\mathbf{0}_Y}}
\newcommand{\zerol}{{\mathbf{0}_L}}
\newcommand{\zeror}{{\mathbf{0}_R}}
\newcommand{\zerozero}{{\mathbf{0}_0}}
\newcommand{\zeroone}{{\mathbf{0}_1}}
\newcommand{\zerotwo}{{\mathbf{0}_2}}
\newcommand{\zeroj}{{\mathbf{0}_j}}
\newcommand{\one}{{\mathbf{1}}}
\newcommand{\two}{{\mathbf{2}}}

\newcommand{\oneone}{{\mathbf{1}_1}}
\newcommand{\onetwo}{{\mathbf{1}_2}}
\newcommand{\onej}{{\mathbf{1}_j}}
\newcommand{\zz}{{\mathbf{00}}}
\newcommand{\zzone}{{\mathbf{00}_1}}
\newcommand{\zztwo}{{\mathbf{00}_2}}
\newcommand{\zoone}{{\mathbf{01}_1}}
\newcommand{\zotwo}{{\mathbf{01}_2}}
\newcommand{\zoj}{{\mathbf{01}_j}}
\newcommand{\ozone}{{\mathbf{10}_1}}
\newcommand{\oztwo}{{\mathbf{10}_2}}
\newcommand{\ozj}{{\mathbf{10}_j}}
\newcommand{\oo}{{\mathbf{11}}}
\newcommand{\ooone}{{\mathbf{11}_1}}
\newcommand{\ootwo}{{\mathbf{11}_2}}

\newcommand{\kkhittingsetname}{$k \times k$ {\sc{Permutation Hitting Set}}\xspace}
\newcommand{\kknphittingsetname}{$k \times k$ {\sc{Hitting Set}}\xspace}

\newcommand{\vertexcover}{{\sc Vertex Cover}\xspace}
\newcommand{\indset}{{\sc Independent Set}\xspace}
\newcommand{\domset}{{\sc Dominating Set}\xspace}
\newcommand{\cutwidth}{{\sc Cutwidth}\xspace}
\newcommand{\minmaxmatching}{{\sc Minimum Maximal Matching}\xspace}

\newcommand{\hampath}{{\sc Hamiltonian Path}\xspace}
\newcommand{\hamcycle}{{\sc Hamiltonian Cycle}\xspace}
\newcommand{\longestpath}{{\sc Longest Path}\xspace}

\newcommand{\longestcycle}{{\sc Longest Cycle}\xspace}
\newcommand{\gmtsp}{{\sc Graph Metric Travelling Salesman Problem}\xspace}
\newcommand{\tsp}{{\sc Travelling Salesman Problem}\xspace}

\newcommand{\cvertexcover}{{\sc Connected Vertex Cover}\xspace}
\newcommand{\constrainedcvertexcover}{{\sc Constrained Connected Vertex Cover}\xspace}
\newcommand{\cdomset}{{\sc Connected Dominating Set}\xspace}
\newcommand{\constrainedcdomset}{{\sc Constrained Connected Dominating Set}\xspace}
\newcommand{\steinertree}{{\sc Steiner Tree}\xspace}

\newcommand{\fvs}{{\sc Feedback Vertex Set}\xspace}
\newcommand{\constrainedfvs}{{\sc Constrained Feedback Vertex Set}\xspace}
\newcommand{\cfvs}{{\sc Connected Feedback Vertex Set}\xspace}
\newcommand{\constrainedcfvs}{{\sc Constrained Connected Feedback Vertex Set}\xspace}
\newcommand{\oct}{{\sc Odd Cycle Transversal}\xspace}
\newcommand{\coct}{{\sc Connected Odd Cycle Transversal}\xspace}
\newcommand{\constrainedcoct}{{\sc Constrained Connected Odd Cycle Transversal}\xspace}

\newcommand{\maxccdsname}{{\sc{Maximally Disconnected Dominating Set}}\xspace}

\newcommand{\mincyclecovername}{{\sc{Min Cycle Cover}}\xspace}
\newcommand{\maxcyclecovername}{{\sc{Max Cycle Cover}}\xspace}
\newcommand{\cyclepackingname}{{\sc{Cycle Packing}}\xspace}

\newcommand{\partialcycles}{{\sc{Partial Cycle Cover}}\xspace}
\newcommand{\dirpartialcycles}{{\sc{Directed Partial Cycle Cover}}\xspace}

\newcommand{\minleaf}{{\sc{Minimum Leaf Tree}}\xspace}
\newcommand{\exactleaf}{{\sc{Exact $k$-leaf Spanning Tree}}\xspace}
\newcommand{\maxleaf}{{\sc{Maximum Leaf Tree}}\xspace}
\newcommand{\maxoutbranching}{{\sc{Maximum Leaf Outbranching}}\xspace}
\newcommand{\minoutbranching}{{\sc{Minimum Leaf Outbranching}}\xspace}
\newcommand{\exactoutbranching}{{\sc{Exact $k$-Leaf Outbranching}}\xspace}
\newcommand{\maxspantree}{{\sc{Maximum Full Degree Spanning Tree}}\xspace}
\newcommand{\exactspantree}{{\sc{Exact Full Degree Spanning Tree}}\xspace}

\newcommand{\sat}{{\sc{SAT}}\xspace}

\newtheorem{theorem}{Theorem}[section]
\newtheorem{lemma}[theorem]{Lemma}
\newtheorem{definition}[theorem]{Definition}
\newtheorem{corollary}[theorem]{Corollary}
\newtheorem{remark}[theorem]{Remark}
\newtheorem{proposition}[theorem]{Proposition}

\begin{document}

  \date{}

  \author{Marek Cygan\thanks{Institute of Informatics, University of Warsaw, Poland, \texttt{cygan@mimuw.edu.pl}}
					\and
					Jesper Nederlof\thanks{Department of Informatics, University of Bergen, Norway, \texttt{Jesper.Nederlof@ii.uib.no}}
					\and
		      Marcin Pilipczuk\thanks{Institute of Informatics, University of Warsaw, Poland, \texttt{malcin@mimuw.edu.pl}}
		      \and
		      Micha\l{} Pilipczuk\thanks{Faculty of Mathematics, Informatics and Mechanics, University of Warsaw, Poland, \texttt{michal.pilipczuk@students.mimuw.edu.pl}}
		      \and
					Johan van Rooij\thanks{Department of Information and Computing Sciences, Utrecht University, The Netherlands, \texttt{jmmrooij@cs.uu.nl}}
					\and
		      Jakub Onufry Wojtaszczyk\thanks{Google Inc., Cracow, Poland, \texttt{onufry@google.com}}
  }

  \title{Solving connectivity problems parameterized by treewidth in single exponential time}

\maketitle
\begin{abstract}
	For the vast majority of local graph problems standard dynamic programming
	techniques give $c^{\tw} |V|^{O(1)}$ algorithms, where $\tw$ is
	the treewidth of the input graph.
	On the other hand, for problems with a global requirement (usually
	connectivity) the best--known algorithms were naive dynamic programming
	schemes running in $\tw^{O(\tw)} |V|^{O(1)}$ time.
	
	We breach this gap by introducing a technique we dubbed Cut\&Count 
	that allows to produce
	$c^{\tw} |V|^{O(1)}$ Monte Carlo algorithms for most connectivity-type
	problems, including \hampath, \fvs and \cdomset, 
  consequently answering the question raised by Lokshtanov, Marx and Saurabh [SODA'11]
  in a surprising way.
	We also show that (under reasonable complexity assumptions) the gap
	cannot be breached for some problems for which Cut\&Count does not
	work, like \cyclepackingname{}.

	The constant $c$ we obtain is in all cases small (at most $4$ for
	undirected problems and at most $6$ for directed ones), and in several
	cases we are able to show that improving those constants
  would cause the Strong Exponential Time Hypothesis to fail.

	Our results have numerous consequences in various fields, like FPT
	algorithms, exact and approximate algorithms on planar and $H$-minor-free
	graphs and algorithms on graphs of bounded degree.
	In all these fields we are able to improve the best-known results for some
	problems.
\end{abstract}

\section{Introduction and notation}

The notion of {\em treewidth}, introduced in 1984 by Robertson and Seymour
\cite{rs:minors3},
has in many cases proved to be a good measure of the intrinsic difficulty
of various NP-hard problems on graphs, and a useful tool for attacking those
problems. 
Many of them can be efficiently solved through dynamic programming
if we assume the input graph to have bounded treewidth.
For example, an expository algorithm to solve \vertexcover and
\indset running in time $4^{\tw(G)}|V|^{O(1)}$ is described in the
algorithms textbook by Kleinberg and Tardos \cite{kleinberg-tardos}, while
the book of Niedermeier \cite{niedermeier:book} on fixed-parameter
algorithms presents an algorithm with running time $2^{\tw(G)}|V|^{O(1)}$.

The interest in algorithms for graphs of bounded treewidth stems from their
utility: 
such algorithms are used as sub-routines in a variety of settings.
Amongst them prominent are approximation algorithms 
\cite{bidimensionality, eppstein} and parametrized algorithms 
\cite{subexponential-bound-genus} for a vast number of
problems on planar, bounded-genus and $H$-minor-free graphs, 
including \vertexcover, \domset and \indset;
there are applications for parametrized algorithms in general graphs
\cite{enumerate-and-expand, cutwidth-linear} for problems like \cvertexcover
and \cutwidth;
and exact algorithms \cite{Fomin06ontwo, rooij:inclusion/exclusion} such as
\minmaxmatching and \domset.

In many cases, where the problem to be solved is ``local'' (loosely speaking
this means that the property of the object to be found can be verified by
checking separately the neighbourhood of each vertex) matching upper and
lower bounds for the runtime of the optimal solution are known. 
For instance
for the aforementioned $2^{\tw(G)}|V|^{O(1)}$ algorithm for \vertexcover
there is a matching lower bound --- unless the Strong Exponential Time
Hypothesis fails, there is no algorithm for \vertexcover running quicker
than $(2-\eps)^{\tw(G)}$ for any $\eps > 0$.

On the other hand, when the problem involves some sort of a ``global''
constraint --- e.g., connectivity --- the best known algorithms usually 
have a runtime on the order of $\tw(G)^{O(\tw(G))} |V|^{O(1)}$.
In these cases the typical dynamic program has to keep track of all the
ways in which the solution can traverse the corresponding separator of the
tree decomposition, that is $\Omega(l^l)$ on the size $l$ of the separator,
and therefore of treewidth.
This obviously implies weaker results in the applications mentioned above.
This problem was observed, for instance, by F. Dorn, F. Fomin and D.
Thilikos \cite{subexponential-bound-genus, Catalan} and by H. Bodlaender
et al. in \cite{sphere-cut}.
The question whether the known $2^{O(\tw(G) \log \tw(G))} |V|^{O(1)}$
parametrized algorithms for \hampath, \cvertexcover and
\cdomset are optimal was asked by D. Lokshtanov, D. Marx and
S. Saurabh \cite{marx:superexp}.


To explain why the $2^{O(\tw(G) \log \tw(G))}$ dynamic programming 
algorithms for connectivity problems were thought to be optimal we recall
the concept of a Myhill-Nerode style equivalence class.
In this case two partial solutions of a subtree of the tree decomposition
are said to be equivalent if they are consistent with the same set of
partial solutions on the remainder of the tree decomposition 
\cite{hopcroft-ullman}.
The states of the naive dynamic program reflect Myhill-Nerode style
equivalence classes, and it seemed to be necessary for any algorithm to 
memoize the information about each class, as it could be needed during 
further computation (see for example the recent work by Lokshtanov et al. 
\cite{treewidth-lower, marx:superexp}).
From this point of view the results of this paper come as a significant
surprise.

\subsection{Our results} \label{ssec:results}
In this paper we introduce a technique we dubbed ``Cut\&Count'' that allows
us to deal with connectivity-type problems through randomization.
For most problems involving a global constraint our technique gives a
randomized algorithm with runtime $c^{\tw(G)} |V|^{O(1)}$.
In particular we are able to give such algorithms for the three problems
mentioned in \cite{marx:superexp}, as well as for all the other sample 
problems mentioned in \cite{Catalan}: \longestpath, \longestcycle, \fvs, 
\hamcycle and \gmtsp.
Moreover, both the constant $c$ and the exponent in $|V|^{O(1)}$ is in all 
cases well defined and small.

The randomization we mention comes from the usage of the Isolation
Lemma \cite{isolation}.
This gives us Monte Carlo algorithms with a one-sided error. The formal
statement of a typical result is as follows:
\begin{theorem}
There exists a randomized algorithm, which given a graph $G$ with
$n$ vertices, a tree decomposition of $G$ of width $t$ and a number $k$
in $3^t n^{O(1)}$ time either states that there exists a connected
vertex cover of size at most $k$ in $G$, or that it could not verify this
hypothesis. If there indeed exists such a cover, the algorithm will return
``unable to verify'' with probability at most $1\slash 2$.
\end{theorem}
We shall denote such an algorithm, which either confirms the
existence of the object we are asking about or returns 
``unable to verify'', and returns ``unable to verify'' with probability no 
larger than $1\slash 2$ if the answer is positive, an algorithm with 
{\em false negatives}.

We see similar results for a plethora of other global problems.
As the exact value of $c$ in the $c^{\tw(G)}$ expression is often important,
we gather here the results we obtain:
\begin{theorem}\label{main}
There exist Monte-Carlo algorithms that given a tree decomposition of the
(underlying undirected graph of the) input graph of width $t$ solve the 
following problems:
\begin{enumerate}
\item \steinertree in $3^t |V|^{O(1)}$ time.
\item \fvs in $3^t |V|^{O(1)}$ time.
\item \cvertexcover in $3^t |V|^{O(1)}$ time.
\item \cdomset in $4^t |V|^{O(1)}$ time.
\item \cfvs in $4^t |V|^{O(1)}$ time.
\item \coct in $4^t |V|^{O(1)}$ time.
\item \mincyclecovername in $6^t |V|^{O(1)}$ time for directed graphs
and in $4^t |V|^{O(1)}$ time for undirected graphs.
\item Directed \longestpath in $6^t |V|^{O(1)}$ time and undirected
\longestpath in $4^t |V|^{O(1)}$ time. This, in particular, gives the same
times for solving \hampath for the directed and undirected cases,
respectively.
\item Directed \longestcycle in $6^t |V|^{O(1)}$ time and undirected
\longestcycle in $4^t |V|^{O(1)}$ time. This, in particular, gives the same
times for solving \hamcycle for the directed and undirected cases,
respecitvely.
\item \exactleaf and, in particular, \minleaf and \maxleaf in
$4^t |V|^{O(1)}$ time for undirected graphs.
\item \exactoutbranching and, in particular \minoutbranching and 
\maxoutbranching in $6^t |V|^{O(1)}$ time for directed graphs.
\item \maxspantree in $4^t |V|^{O(1)}$ time.
\item \gmtsp in $4^t |V|^{O(1)}$ time for undirected graphs.
\end{enumerate}
The algorithms cannot give false positives and may give false negatives with probability at most $1/2$.
\end{theorem}
For a number of these results we have matching lower bounds, such as the
following one:
\begin{theorem} Unless the Strong Exponential Time Hypothesis is false,
there do not exist a constant $\eps>0$ and an algorithm that given an instance $(G=(V,E),T,k)$
together 
with a path decomposition of the graph $G$ of width $p$ solves the \steinertree problem in $(3-\eps)^p |V|^{O(1)}$ time.
\end{theorem}
We have such matching lower bounds for the following problems: 
\cvertexcover, \cdomset, \cfvs, \coct, \fvs, \steinertree and \exactleaf.
We feel that the first four results are of
particular interest here and should be compared to the algorithms and
lower bounds for the analogous problems without the connectivity
requirement.
For instance in the case of \cvertexcover the results show that the increase
in running time to $3^{\tw(G)} n^{O(1)}$ from the 
$2^{\tw(G)} n^{O(1)}$ algorithm of \cite{niedermeier:book} for 
\vertexcover is not an
artifact of the Cut\&Count technique, but rather an intrinsic characteristic
of the problem.
We see a similar increase of the base constant by one for the other three
mentioned problems.

We have found Cut\&Count to fail for two maximization problems:
\cyclepackingname and \maxcyclecovername. 
We believe this is an example of a more general phenomenon --- problems
that ask to maximize (instead of minimizing) the number of connected
components in the solution seem more difficult to solve than the
problems of minimization (including problems where we demand that the
solution forms a single connected component).
As evidence we present lower bounds for the time complexity of solutions
to such problems, proving that $c^{\tw(G)}$ solutions
of these problems are unlikely:
\begin{theorem}
Unless the Exponential Time Hypothesis is false, there does not exist
a $2^{o(p \log p)}|V|^{O(1)}$ algorithm for solving \cyclepackingname
or \maxcyclecovername.
The parameter $p$ denotes the width of a given path decomposition of the input graph.
\end{theorem}

To further verify this intuition, we investigated an artificial problem
(the \maxccdsname), in which we ask for a dominating set with the largest
possible number of connected components, and indeed we found a similar
phenomenon:
\begin{theorem}
Unless the Exponential Time Hypothesis is false, there does not exist
a $2^{o(p \log p)}|V|^{O(1)}$ algorithm for solving \maxccdsname.
The parameter $p$ denotes the width of a given path decomposition of the input graph.
\end{theorem}

A reader interested in just the basic workings of the Cut\&Count technique 
will likely find the descriptions and the intuition needed in the first 
three sections of this work.
The rest of the rather formidable volume is devoted to refined applications,
in some cases needed to obtained the optimal constants, as well as arguments
showing the aforementioned optimality, and can --- in a sense --- be
considered ``advanced material''.

\subsection{Previous work}\label{ssec:previous}
The Cut\&Count technique has two main ingredients.
The first is an algebraic approach, where we assure that objects we are
not interested in are counted an even number of times, and then do the
calculations in $\Z_2$ or in a field of characteristic $2$, which causes
them to disappear.
This line of reasoning goes back to Tutte~\cite{tutte}, and
was recently used by Bj{\"o}rklund~\cite{bjorklund-focs} and Bj{\"o}rklund et. 
al~\cite{bjorklund-arxiv}.

The second is the idea of defining the connectivity requirement through
cuts, which is frequently used in approximation algorithms via linear
programming relaxations.
In particular cut based constraints were used in the Held and Karp
relaxation for the \tsp problem from 1970~\cite{held-karp-relaxation, held-karp-relaxation2}
and appear up to now in the best known approximation algorithms,
for example in the recent algorithm for the \steinertree problem by Byrka et al.~\cite{byrka}.
To the best of our knowledge the idea of defining problems through
cuts was never used in the exact and parameterized settings.

A number of papers circumvent the problems stemming from the lack of
singly exponential algorithms parametrized by treewidth for
connectivity--type problems.
For instance in the case of parametrized algorithms,
sphere cuts \cite{subexponential-bound-genus, sphere-cut} (for planar and
bounded genus graphs) and Catalan structures \cite{Catalan} (for 
$H$-minor-free graphs) were used to obtain $2^{O(\sqrt{k})}|V|^{O(1)}$ algorithms 
for a number of problems with connectivity requirements.
To the best of our knowledge, however, no attempt to attack the problem
directly was published before;
indeed the non-existence of $2^{o(\tw(G) \log \tw(G))}$ algorithms was 
deemed to be more likely.

\subsection{Consequences of the Cut\&Count technique}
\label{ssec:consequences}
As alredy mentioned, algorithms for graphs with a bounded treewidth have
a number of applications in various branches of algorithmics.
Thus, it is not a surprise that the results obtained by our technique give
a large number of corollaries.
To keep the volume of this paper manageable, we do not explore all possible
applications, but only give sample applications in various directions.

We would like to emphasize that the strength of the Cut\&Count technique
shows not only in the quality of the results obtained in various fields,
which are frequently better than the previously best known ones, achieved
through a plethora of techniques and approaches, but also in the ease in
which new strong results can be obtained.

\subsubsection{Consequences for FPT algorithms}
\label{sssec:fpt}
Let us recall the definition of the \fvs problem:

\defparproblemu{\fvs}{An undirected graph $G$ and an integer $k$}{$k$}{
	Is it possible to remove $k$ vertices from $G$ so that the remaining
	vertices induce a forest?}

This problem is on Karp's original list of 21 NP-complete problems
\cite{np-karp}.
It has also been extensively studied from the parametrized 
complexity point of view.
Let us recall that in the fixed-parameter setting (FPT) the problem comes
with a parameter $k$, and we are looking for a solution with time complexity
$f(k) n^{O(1)}$, where $n$ is the input size and $f$ is some function
(usually exponential in $k$).
Thus, we seek to move the intractability of the problem from the input size
to the parameter.

There is a long sequence of FPT algorithms for \fvs 
\cite{fvs:4krand, fvs1, fvs:5k, fvs7, fvs2, fvs3, guo:fvs, fvs6, fvs4,
fvs5}.
The best --- so far --- result in this series is the 
$3.83^k kn^2$ result of Cao, Chen and Liu
\cite{fvs:3.83k}. Our technique gives an
improvement of their result:
\begin{theorem}\label{fvs_theorem}
There exists a Monte-Carlo algorithm solving the \fvs problem in a graph 
with $n$ vertices in $3^k n^{O(1)}$ time and polynomial space.
The algorithm cannot give false positives and may give false negatives with probability at most $1/2$.
\end{theorem}

We give similar improvements for \cvertexcover (from the $2.4882^k n^{O(1)}$
of \cite{cvc-br} to $2^kn^{O(1)}$) and \cfvs (from the $46.2^k n^{O(1)}$ of
\cite{DBLP:conf/walcom/MisraPRSS10} to $3^kn^{O(1)}$).

\subsubsection{Parametrized algorithms for $H$-minor-free graphs}

A large branch of applications of algorithms parametrized with treewidth
is the bidimensionality theory, used to find subexponential algorithms
for various problems in $H$-minor-free graphs.
In this theory we use the celebrated minor theorem of Robertson and
Seymour \cite{graph-minor-theorem}, which ensures that any 
$H$-minor-free graph either has treewidth bounded by $C\sqrt{k}$,
or a $2\sqrt{k} \times 2\sqrt{k}$ lattice as a minor.
In the latter case we are assumed to be able to answer the problem in
question (for instance a $2\sqrt{k} \times 2\sqrt{k}$ lattice as a minor
guarantees that the graph does not have a \vertexcover or \cvertexcover
smaller than $k$).
Thus, we are left with solving the problem with the assumption of bounded
treewidth.
In the case of, for instance, \vertexcover a standard dynamic algorithm
suffices, thus giving us a $2^{O(\sqrt{k})}$ algorithm to check whether
a graph has a vertex cover no larger than $k$.
In the case of \cvertexcover, however, the standard dynamic algorithm gives
a $2^{O(\sqrt{k}\log k)}$ complexity --- thus, we lose a logarithmic factor
in the exponent.

There were a number of attempts to deal with this problem, taking into
account the structure of the graph, and using it to deduce some
properties of the tree decomposition under consideration.
The latest and most efficient of those approaches is due to Dorn, Fomin
and Thilikos \cite{Catalan}, and exploits the so called Catalan structures.
The approach deals with most of the problems mentioned in our paper, and is
probably applicable to the remaining ones.
Thus, the gain here is not in improving the running times (though
our approach does improve the constants hidden in the big-$O$ notation
these are rarely considered to be important in the bidimensionality theory),
but rather in simplifying the proof --- instead of delving into
the combinatorial structure of each particular problem, we are back to 
a simple framework of applying the Robertson-Seymour theorem and then
following up with a dynamic program on the obtained tree decomposition.

The situation is more complicated in the case of problems on directed
graphs.
A full equivalent of the bidimensionality theory is not developed for such
problems, and only a few problems have subexponential parametrized
algorithms available \cite{fastFAST, beyond-bidimensionality}.
One of the approaches is to mimic the bidimensionality approach, which again
leads to solving a problem on a graph of bounded treewidth --- such an
approach is taken by Dorn et al. in \cite{beyond-bidimensionality} for
\maxoutbranching to obtain a $2^{O(\sqrt{k} \log k)}$ algorithm.
In this case, a straightforward substitution of our $6^{\tw(G)} |V|^{O(1)}$
algorithm for the dynamic algorithm used by Dorn et al. will give the
following improvement:
\begin{theorem}
There exists a Monte-Carlo algorithm solving the $k$-\maxoutbranching 
problem in $2^{O(\sqrt{k})}|V|^{O(1)}$ time for directed graphs
for which the underlying undirected graph excludes a fixed graph $H$ as a 
minor.
The algorithm cannot give false positives and 
may give false negatives with probability at most $1/2$.
\end{theorem}

\subsubsection{Consequences for Exact Algorithms for graphs of bounded
degree}
Another application of our methods can be found in the field of solving
problems with a global constraint in graphs of bounded degree.
The problems that have been studied in this setting are mostly local in
nature (such as \vertexcover, see, e.g., \cite{BottomUp});
however global problems such as the \tsp and \hamcycle have also received
considerable attention \cite{bjorklund-tsp-bound, eppstein-tsp-cubic,
gebauer-tsp-bound, iwama-tsp-cubic}.

Here the starting point is the following theorem by Fomin et al.
\cite{Fomin06ontwo}:
\begin{theorem}[Fomin, Gaspers, Saurabh, Stepanov]
For any $\eps > 0$ there exists an integer $n_\eps$ such that for any
graph $G$ with $n > n_\eps$ vertices,
$$\pw(G) \leq \frac{1}{6} n_3 + \frac{1}{3} n_4 + \frac{13}{30} n_5 +
n_{\geq 6} + \eps n,$$
where $n_i$ is the number of vertices of degree $i$ in $G$ for any 
$i \in \{3,\ldots,5\}$ and $n_{\geq 6}$ is the number of vertices of degree
at least $6$.
\end{theorem}
This theorem is constructive, and the corresponding path decompostion (and,
consequently, tree decomposition) can be found in polynomial time.

Combining this theorem with our results gives algorithms running in faster
than $2^n$ for graphs of maximum degree $3$, $4$ and (in the case of the
$3^{\tw(G)}$ and $4^{\tw(G)}$ algorithms) $5$, as follows:

\begin{corollary}
There exist randomized algorithms that solve the following problems:
\begin{itemize}
\item \steinertree, \fvs and \cvertexcover in $O(1.201^n)$ time for cubic 
graphs, $O(1.443^n)$ time for graphs of maximum degree $4$ and 
$O(1.61^n)$ time for graphs of maximum degree $5$;
\item \cdomset, \coct, \cfvs, \exactleaf, \maxspantree and
\gmtsp as well as undirected versions of \mincyclecovername,
\longestpath and \longestcycle in $O(1.26^n)$ time for cubic graphs,
$O(1.588^n)$ for graphs of maximum degree $4$ and $O(1.824^n)$ for 
graphs of maximum degree $5$;
\item Directed versions of \mincyclecovername, \longestpath and
\longestcycle, as well as for \exactoutbranching, in $O(1.349^n)$ for cubic 
graphs and in $O(1.818^n)$ for graphs of maximum degree $4$.
\end{itemize}
All the aforementioned algorithms are Monte Carlo algorithms with false
negatives.
\end{corollary}

The \tsp in its full generality does not fall under the Cut\&Count
regime;
however for graphs of degree four the $O(1.588^n)$ algorithm obtained for 
\gmtsp (which can easily be extended to the case of the \tsp with 
polynomially bounded weights)
is significantly faster than the best known algorithm for the case of 
unbounded weights of Gebauer \cite{gebauer-tsp-bound}, which runs in 
$1.733^n$ time.
For the case of degree $5$ (where we give an $O(1.824^n)$ algorithm for the
\tsp with polynomially bounded weights) the best known result in the general
weight case is the $(2-\eps)^n$ algorithm of Bj{\"o}rklund et al.
\cite{bjorklund-tsp-bound}.
It is worth noticing that in the case of cubic graphs we automatically obtain an algorithm
for \gmtsp running in $2^{n/3+\eps n} n^{O(1)}$ time, which coincides with the
time complexity of the algorithm for \tsp of Eppstein \cite{eppstein-tsp-cubic}.
The currently fastest algorithm for \tsp in cubic graphs is due to Iwama and Nakashima \cite{iwama-tsp-cubic}. 

\subsubsection{Consequences for exact algorithms on planar graphs}
Here we begin with a consequence of the work of Fomin and Thilikos
\cite{twbound}:
\begin{proposition}
For any planar graph $G$, $\tw(G) + 1\leq \frac{3}{2}\sqrt{4.5 n} \leq
3.183 \sqrt{n}$.
Moreover a tree decomposition of such width can be found in polynomial time.
\end{proposition}

Using this we immediately obtain $c^{\sqrt{n}}$ algorithms for solving
problems with a global constraint on planar graphs with good constants.
For instance for the \hamcycle problem on planar graphs we obtain the
following result:
\begin{corollary}
There exists a randomized algorithm with false negatives
solving \hamcycle on planar graphs in $O(4^{3.183\sqrt{n}}) = 
O(2^{6.366\sqrt{n}})$ time.
\end{corollary}

To the best of our knowledge the best algorithm known so far was the
$O(2^{6.903\sqrt{n}})$ of Bodlaender et al. \cite{sphere-cut}.

Similarly, we obtain an $O(2^{6.366\sqrt{n}})$ algorithm for \longestcycle
on planar graphs (compare to the $O(2^{7.223\sqrt{n}})$ of 
\cite{sphere-cut}), and --- as in the previous subsections --- well-behaved
$c^{\sqrt{n}}$ algorithms for all mentioned problems.

\subsection{Organization of the paper}
In the introduction we present the contents of the paper.
Subsection \ref{ssec:results} states the main results obtained by us.
After analyzing the connections to previous works in Subsection 
\ref{ssec:previous} we turn to giving sample consequences in Subsection 
\ref{ssec:consequences}.
We finish the introduction by giving this outline.

Section \ref{sec:preliminaries} is devoted to presenting the background
material for our algorithms.
In particular in Subsection \ref{ssec:treewidth} we recall the notion of
treewidth and dynamic programming on tree decompositions, while in
Subsection \ref{ssec:isolation} we introduce the Isolation Lemma.

In Section \ref{sec:illustration} we present the Cut\&Count technique on
two examples: the \steinertree problem and the {\sc Directed} \mincyclecovername problem.
We go into all the details, as we aim to present not only the algorithms
and proofs, but also the intuition behind them.
We use those intuitions in Section \ref{sec:applications}, where we
give sketches of algorithms for all the other problems mentioned in
Theorem \ref{main}.

In Section \ref{sec:negatives:intro} we move to lower bounds.
In Subsection \ref{sec:eth:intro} we present evidence that problems in
which we maximize the number of connected components are unlikely to have
$c^{\tw(G)} |V|^{O(1)}$ algorithms.
In Subsection \ref{sec:seth:intro} we provide arguments that 
several of the algorithms provided in Section \ref{sec:illustration}
have time complexity which is hard to improve.
We finish the paper with a number of conclusions and open problems in
Section \ref{sec:conclusions}

As the reader might have already noticed, there is a quite a large amount of
material covered in this paper.
To keep it readable a considerable number of proofs and analyses was
postponed to the appendix.
Thus, in Appendix \ref{sec:app:alg} we give detailed descriptions of all the
algorithms sketched in Section \ref{sec:applications}.
In particular in Part \ref{ssec:fsc} we describe all the variants of the
Fast Subset Convolution technique we use to decrease the constants in
our algorithms.
Appendix \ref{sec:param-improv} is devoted to the proofs of the results
announced in Subsection \ref{sssec:fpt} --- Theorem \ref{fvs_theorem} and
its analogues for \cvertexcover and \cfvs.
In Appendix \ref{sec:negatives:eth} we turn to the proofs of
the lower bounds stated in Subsection \ref{sec:eth:intro}, while
Appendix \ref{sec:negatives:seth} gives proofs for tight lower
bounds stated in Subsection \ref{sec:seth:intro}.

\section{Preliminaries and notation}\label{sec:preliminaries}

\subsection{Notation}\label{ssec:notation}
\label{subsec:not}
Let $G = (V,E)$ be a graph (possibly directed).
By $V(G)$ and $E(G)$ we denote the sets of vertices and edges of $G$,
respectively.
For a vertex set $X \subset V(G)$ by $G[X]$ we denote the subgraph induced
by $X$.
For an edge set $X \subset E$, we take $V(X)$ to denote the set of the
endpoints of the edges of $X$, and by $G[X]$ --- the subgraph $(V, X)$.
Note that in the graph $G[X]$ for an edge set $X$ the set of vertices remains the same as in the graph $G$.

For an undirected graph $G = (V,E)$, the {\em open neighbourhood} of a 
vertex $v$,  denoted $N(v)$, stands for $\{u \in V : uv \in E\}$, while the 
{\em closed} neighbourhood $N[v]$ is $N(v) \cup \{v\}$.
Similarly, for a set $X \subset V(G)$ by $N[X]$ we mean
$\bigcup_{v \in X} N[v]$ and by $N(x)$ we mean $N[X] \setminus X$.

By a {\em cut} of a set $X\subset V$ we mean a pair $(X_1,X_2)$, with 
$X_1 \cap X_2 = \emptyset$, $X_1 \cup X_2 = X$.
We refer to $X_1$ and $X_2$ as to the (left and right) {\em sides} of the cut.

We denote the degree of a vertex $v$ by $\deg(v)$. $\deg_H(v)$ denotes the 
degree of $v$ in the subgraph $H$. 
For $X \subseteq V$ or $X \subseteq E$, $\deg_X(v)$ is a short for
$\deg_{G[X]}(v)$. 
If $G$ is a directed graph and $X \subseteq V$ or $X \subseteq E$, we 
denote the in- and
out-degree of $v$ in $G[X]$ by $\indeg_{G[X]}(v)$ and $\outdeg_{G[X]}(v)$ 
respectively. 

For an edge $e = uv$ by subdividing it $s$ times (for $s > 0$) we mean the 
following operation: (1) remove the edge $e$, (2) add $s$ vertices $\{x_{e,1}, \ldots, x_{e,s}\}$, (3) add edges $\{ux_{e,1}, x_{e,1}x_{e,2}, \ldots, x_{e,k-1}x_{e,s}, x_{e,s}v\}$.

In a directed graph $G$ by weakly connected components we mean the 
connected components of the underlying undirected graph.
For a (directed) graph $G$, we let $\mathtt{cc}(G)$ denote the number of 
(weakly) connected components of $G$.

For two bags $x,y$ of a rooted tree we say that $y$ is a descendant
of $x$ if it is possible to reach $x$ when starting at $y$ and going
only up the tree.
In particular $x$ is its own descendant.

We denote the symmetric difference of two sets $A$ and $B$ by 
$A \triangle B$.
For two integers $a,b$ we use $a \equiv b$ to indicate that $a$ is 
even if and only if $b$ is even.
We use Iverson's bracket notation: if $p$ is a predicate we let $[p]$ be 
$1$ if $p$ if true and $0$ otherwise.
If $\omega : U \rightarrow \{1,\ldots,N\}$, we shorthand 
$\omega(S)=\sum_{e \in S}\omega(e)$ for $S \subseteq U$. 

For a function $s$ by $s[v \to \alpha]$ we denote the function 
$s \setminus \{(v,s(v))\} \cup \{(v,\alpha)\}$.
Note that this definition works regardless of whether $s(v)$ is already 
defined or not.

\subsection{Treewidth and pathwidth}\label{ssec:treewidth}


\subsubsection{Tree Decompositions}
\label{sec:td}

\begin{definition}[Tree Decomposition, \cite{rs:minors3}]
A \emph{tree decomposition} of a (undirected or directed) graph~$G$ is a tree~$\treedecomp$ in which each 
vertex~$x \in \treedecomp$ has an assigned set of vertices~$B_x \subseteq V$ 
(called a \emph{bag}) such that $\bigcup_{x \in \treedecomp} B_x = V$ with the 
following properties:
\begin{itemize}
\item for any $uv \in E$, there exists an~$x \in \treedecomp$ such that 
$u,v \in B_x$.
\item if $v \in B_x$ and $v \in B_y$, then $v \in B_z$ for all $z$ on 
the path from $x$ to $y$ in $\treedecomp$.
\end{itemize}
\end{definition}

The \emph{treewidth}~$tw(\treedecomp)$ of a tree decomposition~$\treedecomp$ is the size of 
the largest bag of $\treedecomp$ minus one, and the treewidth of a graph $G$ is the 
minimum treewidth over all possible tree decompositions of~$G$.


Dynamic programming algorithms on tree decompositions are often presented 
on nice tree decompositions which were introduced by Kloks~\cite{Kloks94}. 
We refer to the tree decomposition definition given by Kloks as to a {\em standard nice tree decomposition}.

\begin{definition}
A \emph{standard nice tree decomposition} is a tree decomposition where:
\begin{itemize}
\item every bag has at most two children,
\item if a bag $x$ has two children $l,r$, then $B_x = B_l = B_r$,
\item if a bag $x$ has one child $y$, then either $|B_x|=|B_y|+1$ and 
$B_y \subseteq B_x$ or $|B_x|+1=|B_y|$ and $B_x \subseteq B_y$.
\end{itemize}
\end{definition}

We present a slightly different definition of a nice tree decomposition.

\begin{definition}[Nice Tree Decomposition] \label{def:nicetreedecomp}
A \emph{nice tree decomposition} is a tree decomposition with one special 
bag $z$ called the \emph{root} with $B_z = \emptyset$ and in which each 
bag is one of the following types:
\begin{itemize}
\item \textbf{Leaf bag}: a leaf $x$ of $\treedecomp$ with $B_x = \emptyset$.
\item \textbf{Introduce vertex bag}: an internal vertex~$x$ of $\treedecomp$ 
with one child vertex~$y$ for which $B_x = B_y \cup \{v\}$ 
for some $v \notin B_y$. 
This bag is said to \emph{introduce} $v$.
\item \textbf{Introduce edge bag}: an internal vertex~$x$ of $\treedecomp$ labeled 
with an edge $uv \in E$ with one child bag~$y$ for which 
$u,v \in B_x = B_y$. 
This bag is said to \emph{introduce} $uv$.
\item \textbf{Forget bag}: an internal vertex~$x$ of $\treedecomp$ with one child 
bag~$y$ for which $B_x = B_y \setminus \{v\}$ for some $v \in B_y$. 
This bag is said to \emph{forget} $v$.
\item \textbf{Join bag}: an internal vertex $x$ with two child vertices 
$l$ and $r$ with $B_x = B_r = B_l$.
\end{itemize}
We additionally require that every edge in $E$ is introduced exactly once. 
\end{definition}
We note that this definition is slightly different than usual.
In our definition we have the extra requirements that bags 
associated with the leafs and the root are empty. 
Moreover, we added the introduce edge bags.

Given a tree decomposition, a standard nice tree decomposition 
of equal width can be found in polynomial time~\cite{Kloks94} and in the 
same running time, it can easily be modified to meet our extra requirements,
as follows:
add a series of forget bags to the old root, 
and add a series of introduce vertex bags below old leaf bags that are nonempty;
Finally, for every edge $uv \in E$ add an introduce edge bag above the 
first bag with respect to the in-order traversal of $\treedecomp$ that contains $u$ 
and $v$.

By fixing the root of $\treedecomp$, we associate with each bag $x$ 
in a tree decomposition $\treedecomp$ a vertex set $V_x \subseteq V$ where a vertex 
$v$ belongs to $V_x$ if and only if there is a bag $y$
which is a descendant of $x$ in $\treedecomp$ with $v \in B_y$
(recall that $x$ is its own descendant).
We also associate with each bag $x$ of $\treedecomp$ a subgraph of $G$ as follows:

\[ G_x = \Big{(}V_x, E_x = \{e |  \textrm{$e$ is introduced in a descendant of $x$ } \}\Big{)} \]

For an overview of tree decompositions and dynamic programming on tree 
decompositions see~\cite{BodlaenderK08,HicksKK05}.

\subsubsection{Path Decompositions}
A {\em path decomposition} is a tree decomposition that is a path.
The pathwidth of a graph is the minimum width of all path decompositions.
Path decompositions can, similarly as above, be transformed into nice
path decompositions, these obviously contain no join bags.

\subsection{Isolation lemma}\label{ssec:isolation}
An ingredient of our algorithms is the Isolation Lemma:

\begin{definition}
A function $\omega:U \rightarrow \mathbb{Z}$ \emph{isolates} a set family 
$\mathcal{F} \subseteq 2^U$ if there is a unique $S' \in \mathcal{F}$ with 
$\omega(S')= \min_{S \in \mathcal{S}}\omega(S)$. 
\end{definition}

Recall that for $X \subseteq U$, $\omega(X)$ denotes 
$\sum_{u \in X}\omega(u)$.

\begin{lemma}[Isolation Lemma, \cite{isolation}]
\label{lem:iso}
Let $\mathcal{F} \subseteq 2^U$ be a set family over a universe $U$ with 
$|\mathcal{F}|>0$.
For each $u \in U$, choose a weight $\omega(u) \in \{1,2,\ldots,N\}$ 
uniformly and independently at random.
Then
\[
	\mathtt{prob}[\omega \textnormal{ isolates } \mathcal{F}] \geq 1 - \frac{|U|}{N}
\]  
\end{lemma}

It is worth mentioning that in \cite{DBLP:journals/siamcomp/ChariRS95}, 
a lemma using less random bits is shown: 
If $|\mathcal{F}| \leq Z$, then a scheme using 
$O(\log |U| + \log Z)$ random bits to obtain a polynomially 
bounded (in unary) weight function that isolates any set system with high
probability is presented.

The Isolation Lemma allows us to count objects modulo $2$,
since with a large probability it reduces a possibly large number of 
solutions to some problem to a unique one (with an additional weight
constraint imposed). 
This lemma has found many applications \cite{isolation}. 

An alternative method to a similar end is obtained by using 
Polynomial Identity Testing \cite{demillo,schwartz,zippel} over a field of 
characteristic two.
This second method has been already used in the field of exact and 
parameterized algorithms 
\cite{bjorklund-focs,bjorklund-arxiv,DBLP:conf/icalp/Koutis08,DBLP:conf/icalp/KoutisW09,DBLP:journals/ipl/Williams09}. 
The two methods do not have differ much in their consequences:
Both use the same number of random bits (the most randomness efficient 
algorithm are provided in 
\cite{Agrawal:2003:PIT:792538.792540,DBLP:journals/siamcomp/ChariRS95}). 
The challenge of giving a full derandomization seems to be equally 
difficult for both methods 
\cite{DBLP:conf/approx/ArvindM08,DBLP:journals/cc/KabanetsI04}. 
The usage of the Isolation Lemma gives greater polynomial overheads,
however we choose to use it because it requires less preliminary 
knowledge.

\section{Cut\&Count: Illustration of the technique}
\label{sec:illustration}

In this section we present the Cut\&Count technique by 
demonstrating how it applies to the
\steinertree and {\sc Directed} \mincyclecovername problems.
We go through all the details in an expository manner, 
as we aim not only to show the solutions to these particular problems, but
also to show the general workings.




The Cut\&Count technique applies to problems with certain connectivity 
requirements. 
Let $\sols \subseteq 2^\univ$ be a set of solutions; we aim to 
decide whether it is empty. 
Conceptually, Cut\&Count can naturally be split in two parts: 
\begin{itemize}
	\item \textbf{The Cut part}: Relax the connectivity requirement by 
considering the set $\cand \supseteq \sols$ of possibly 
connected candidate solutions. 
Furthermore, consider the set $\objs$ of pairs $(X,C)$ 
where $X \in \cand$ and $C$ is a consistent cut 
(to be defined later) of $X$.
\item \textbf{The Count part}: Compute $|\objs|$ modulo $2$ using a 
sub-procedure. 
Non-connected candidate solutions $X \in \cand \setminus \sols$
cancel since they are consistent with an even number of cuts. 
Connected candidates $x \in \sols$ remain.
\end{itemize}

Note that we need the number of solutions to be odd in order to make the 
counting part work. 
For this we use the Isolation Lemma (Lemma \ref{lem:iso}): 
We introduce uniformly and independently chosen weights $\weight(v)$ 
for every $v \in \univ$ and compute $|\objs_\targetW|$ modulo 2 for every $\targetW$, 
where $\objs_\targetW= \{ (X,C) \in \objs | \weight(X) = \targetW \}$. 
The general setup can thus be summarized as in Algorithm 
\ref{alg:cutandcount}: 

\begin{algorithm}                      
\caption{$\mathtt{cutandcount}(\univ,\treedecomp,\countproc)$}          
\label{alg:cutandcount}                           
\begin{algorithmic}[1]                    
\REQUIRE $\mathtt{cutandcount}(\univ,\treedecomp,\countproc)$
\ENSURE Set $\univ$; nice tree decomposition $\treedecomp$; Procedure $\countproc$ accepting a $\weight: \univ \rightarrow \{1,\ldots,N\}$, $\targetW \in \mathbb{Z}$ and $\treedecomp$.
\FOR{every $v \in \univ$}
	\STATE Choose $\weight(v) \in \{1,\ldots, 2|\univ|\}$ uniformly at random.
\ENDFOR
	\FOR{every $0 \leq \targetW \leq 2|\univ|^2$}
		 \STATE \textbf{if} $\countproc(\weight,\targetW,\treedecomp) \equiv 1$ \textbf{then} \textbf{return yes}
	\ENDFOR
\STATE \textbf{return no} 
\end{algorithmic}
\end{algorithm}

The following corollary that we use throughout the paper
follows from Lemma \ref{lem:iso} by setting 
$\mathcal{F} = \sols$ and $N=2|\univ|$:
\begin{corollary}
\label{cor:cutandcount}
Let $\sols \subseteq 2^{\univ}$ and 
$\objs \subseteq 2^{\univ \times (V \times V)}$. 
Suppose that for every $\targetW \in \Z$:
	\begin{enumerate} 
		\item $|\{(X,C) \in \objs | \weight(X) = \targetW\}| \equiv |\{X \in \sols | \weight(X) = \targetW\}|.$
		\item $\countproc(\weight,\targetW,\treedecomp) \equiv |\{(X,C) \in \objs | \weight(X) = \targetW\}|$ 
	\end{enumerate}
Then Algorithm \ref{alg:cutandcount} returns \textbf{no} if 
$\sols$ is empty and \textbf{yes} with probability at least 
$\frac{1}{2}$ otherwise.
\end{corollary}

When applying the technique, both the Cut and the Count part are 
non-trivial: 
In the Cut part one has to find the proper relaxation of the solution set,
and in the Count part one has to show that the number of non-solutions is 
even for each $W$ and provide an algorithm $\countproc$. 
Usually, as we will see in the expositions of the 
applications, the count part requires more explanation. 
In the next two subsections, we illustrate both parts by giving two 
specific applications.

\subsection{Steiner Tree}
\label{ssec:steiner}
\label{sec:steinertree}

\defproblemu{\steinertree}
{An undirected graph $G = (V,E)$, a set of terminals $T \subseteq V$ and an integer $k$.}
{Is there a set $X \subseteq V$ of cardinality $k$ such that $T \subseteq X$ and $G[X]$ is connected?}

\noindent {\bf{The Cut part.}}
Let us first consider the Cut part of the Cut\&Count technique, 
and start by defining the objects we are going to count.
Suppose we are given a 
weight function $\weight: V \rightarrow \{1,\ldots,N\}$. 
For any integer $\targetW$, let $\cand_\targetW$ be the set of all such 
subsets $X$ of $V$ that $T \subseteq X$, $\weight(X)=\targetW$ 
and $|X|=k$. 
Also, define $\sols_\targetW = \{ X \in \cand_\targetW \ |\ G[X] 
\text{ is connected}\}$. 
The set $\bigcup_\targetW \sols_\targetW$ is our set of solutions --- if for any $\targetW$
this set is nonempty, our problem has a positive answer.
The set $\cand_\targetW$ is the set of candidate solutions, where we
relax the connectivity requirement.
In this easy application the only requirement that remains is that the set
of terminals is contained in the candidate solution.

\begin{definition}
A cut $(V_1,V_2)$ of an undirected graph $G=(V,E)$ is \emph{consistent} if 
$u \in V_1$ and $v \in V_2$ implies $uv \notin E$. 
A {\em consistently cut subgraph} of $G$ is a pair $(X,(X_1,X_2))$ 
such that $(X_1,X_2)$ is a consistent cut of $G[X]$.

Similarly for a directed graph $D=(V,A)$ a cut $(V_1,V_2)$ is consistent
if $(V_1,V_2)$ is a consistent cut in the underlying undirected graph.
A consistently cut subgraph of $D$ is a pair $(X,(X_1,X_2))$ 
such that $(X_1,X_2)$ is a consistent cut of the underlying undirected graph of $D[X]$.
\end{definition}

Let $\anchor$ be an arbitrary terminal.
Define $\objs_\targetW$ to be the set of all consistently cut subgraphs 
$(X,(X_1,X_2))$ such that $X \in \cand_\targetW$ and $\anchor \in X_1$.



Before we proceed with the Count part, let us state 
the following easy combinatorial identity:

\begin{lemma}
\label{lem:evencancel}
Let $G=(V,E)$ be a graph and let $X$ be a subset of vertices such that 
$\anchor \in X \subseteq V$. 
The number of consistently cut subgraphs $(X,(X_1,X_2))$ such that 
$\anchor \in X_1$ is equal to $2^{\mathtt{cc}(G[X])-1}$.
\end{lemma}

\begin{proof}
By definition, we know for every consistently cut subgraph 
$(X,(X_1,X_2))$ and connected component $C$ of $G[X]$ that either 
$C \subseteq X_1$ or $C \subseteq X_2$. 
For the connected component containing $\anchor$, the choice is fixed, 
and for all $\mathtt{cc}(G[X])-1$ other connected components we are free
to choose a side of a cut, which gives $2^{\mathtt{cc}(G[X])-1}$ possibilities 
leading to different consistently cut subgraphs.
\end{proof}

\noindent {\bf{The Count part.}} 
For the Count part, the following lemma shows that the first condition 
of Corollary \ref{cor:cutandcount} is indeed met:
\begin{lemma}
\label{lem:steincancel}
Let $G,\weight, \objs_\targetW$ and $\sols_\targetW$ be as defined above. 
Then for every $\targetW$, $|\sols_\targetW| \equiv |\objs_\targetW|$.
\end{lemma}
\begin{proof}
Let us fix $\targetW$ and omit the subscripts accordingly. 
By Lemma \ref{lem:evencancel}, we know that 
$|\objs|=\sum_{X \in \cand}2^{\mathtt{cc}(G[X])-1}$. 
Thus $|\objs| \equiv \big|\{X \in \cand| \mathtt{cc}(G[X])=1\}
\big| = |\sols|$.
\end{proof}
Now the only missing ingredient left is the sub-procedure $\countproc$.
This sub-procedure, which counts the cardinality of $\objs_\targetW$
modulo 2, is a standard application of dynamic programming:
\begin{lemma}
\label{lem:stein}
Given $G=(V,E)$, $T \subseteq V$, an integer $k$, 
$\weight: V \rightarrow \{1,\ldots,N\}$ and a nice tree decomposition $\treedecomp$, 
there exists an algorithm that can determine $|\objs_\targetW|$ modulo $2$ 
for every $0 \leq \targetW \leq kN$ in $3^tN^2 |V|^{O(1)}$ time.
\end{lemma}
\begin{proof}
We use dynamic programming, but we first need some preliminary definitions.
Recall that for a bag $x \in \treedecomp$ we denoted by $V_x$ the set of vertices 
of all descendants of $x$, while by $G_x$ we denoted the graph composed 
of vertices $V_x$ and the edges $E_x$ introduced by the descendants of $x$.
We now define ``partial solutions'':
For every bag $x \in \treedecomp$, integers $0 \leq \iterk \leq k$, $0 \leq \iterW \leq kN$ 
and $s \in \{\zero,\oneone,\onetwo\}^{\bag{x}}$ define 
\begin{align*}
\cand_x(\iterk,\iterW) &= \Big{\{} X \subseteq V_x 
\big{|} \; 
  (T \cap V_x) \subseteq X \ \wedge\ 
  |X| = \iterk \ \wedge\ \weight(X)= \iterW \Big{\}}\\
	\objs_x(\iterk,\iterW) &= \Big{\{}(X, (X_1,X_2)) \big{|} \; 
	X \in \cand_x(\iterk,\iterW)\ \wedge\ (X, (X_1,X_2)) \text{ is a consistently 
	cut subgraph of } G_x \\
    &\qquad \wedge\ (\anchor \in \subbags{x} \Rightarrow \anchor \in X_1) \Big{\}} \\
	A_x(\iterk,\iterW,s) &= \Big{|}\Big{\{} (X, (X_1,X_2)) \in \objs_x(\iterk,\iterW) 
	\big{|}\; 
	\big(s(v) = \onej \Rightarrow v \in X_j\big)\	\wedge \
	\big(s(v) = \zero \Rightarrow v \notin X\big) \Big{\}}\Big{|}
\end{align*}
The intuition behind these definitions is as follows: the set
$\cand_x(\iterk,\iterW)$ contains all sets $X \subset V_x$ that could
potentially be extended to a candidate solution from $\cand$,
subject to an additional restriction that the cardinality and weight
of the partial solution are equal to $\iterk$ and $\iterW$, respectively.
Similarly, $\objs_x(\iterk,\iterW)$ contains consistently cut subgraphs,
which could potentially be extended to elements of $\objs$, again
with the cardinality and weight restrictions.
The number $A_x(\iterk,\iterW,s)$ counts those elements of $\objs_x(\iterk,\iterW)$
which additionally behave on vertices of $\bag{x}$ in a fashion prescribed by
the sequence $s$.
$\zero, \oneone$ and $\onetwo$ (we refer to them as colours) describe 
the position of any particular vertex with respect to a set $X$ with a
consistent cut $(X_1,X_2)$ of $G[X]$ --- the vertex can either be outside
$X$, in $X_1$ or in $X_2$.
In particular note that
$$\sum_{s \in \{\zero, \oneone, \onetwo\}^{\bag{x}}} A_x(\iterk,\iterW,s)
= |\objs_x(\iterk,\iterW)|$$ --- the various choices of $s$ describe all possible
intersections of an element of $\objs$ with $\bag{x}$.
Observe that since we are interested in values $|\objs_\targetW|$ modulo $2$
it suffices to compute values $A_\rootv(k,\targetW,\emptyset)$ for all $\targetW$
(recall that $\rootv$ is the root of the tree decomposition),
  because $|\objs_\targetW|=|\objs_\rootv(k,\targetW)|$.

We now give the recurrence for $A_x(\iterk,\iterW,s)$ which is used by the 
dynamic programming algorithm. 
In order to simplify the notation, let $v$ denote the vertex introduced and 
contained in an introduce bag, 
and let $y,z$ denote the left and right children of $x$ in $\treedecomp$, if present.
\begin{itemize}
\item \textbf{Leaf bag $x$}:
	\[ A_x(0,0,\emptyset) = 1 \]
	All other values of $A_x(\iterk,\iterW,s)$ are zeroes.
\item \textbf{Introduce vertex $v$ bag $x$}:
	\begin{eqnarray*}
     A_x(\iterk,\iterW,s[v \to \zero]) & = & [v \not\in T]A_y(\iterk,\iterW,s)  \\
	   A_x(\iterk,\iterW,s[v \to \oneone]) & = & A_y(\iterk-1,\iterW-\weight(v),s) \\
	   A_x(\iterk,\iterW,s[v \to \onetwo]) & = & [v \not= \anchor]A_y(\iterk-1,\iterW-\weight(v),s)
  \end{eqnarray*}
	For the first case note that by definition $v$ can not be coloured 
$\zero$ if it is a terminal.
For the other cases, the accumulators have to be updated and we have to make
sure we do not put $s(\anchor)=\mathbf{1}_2$.

\item \textbf{Introduce edge $uv$ bag $x$}:
	\[ A_x(\iterk,\iterW,s) = [s(u) = \zero\ \vee\ s(v) = \zero\ \vee\ s(u) = s(v)]A_y(\iterk,\iterW,s)  \]
   Here we filter table entries inconsistent with the edge $(u,v)$, 
	 i.e., table entries where the endpoints are coloured $\oneone$ and 
	 $\onetwo$.
\item \textbf{Forget vertex $v$ bag $x$}:
	\[ A_x(\iterk,\iterW,s) = \sum_{\alpha \in \{\zero,\oneone,\onetwo\}} A_x(\iterk,\iterW,s[v \to \alpha]) \]
	In the child bag the vertex $v$ can have three states so we sum 
	over all of them.
\item \textbf{Join bag}:
	\[ A_x(\iterk,\iterW,s) = \sum_{\iterk_1 + \iterk_2 = \iterk + |s^{-1}(\{\oneone, \onetwo\})|} \ \  \sum_{\iterW_1 + \iterW_2 = \iterW + \weight(s^{-1}(\{\oneone, \onetwo\}))} A_y(\iterk_1,\iterW_1,s) A_z(\iterk_2,\iterW_2,s) \]
	The only valid combinations to achieve the colouring $s$ is to have the 
	same colouring in both children. 
	Since vertices coloured $\onej$ in $B_x$ are accounted for in the
	accumulated weights of both of the children, 
	we add their contribution to the accumulators.
\end{itemize}
It is easy to see that the Lemma can now be obtained by combining the 
above recurrence with dynamic programming.
Note that as we perform all calculations modulo $2$, we take only constant
time to perform any arithmetic operation.
\end{proof}

We conclude this section with the following theorem.

\begin{theorem}\label{thm:st-main}
 There exists a Monte-Carlo algorithm that given a tree decomposition of width $t$
 solves \steinertree{} in $3^t |V|^{O(1)}$ time.
 The algorithm cannot give false positives and may give false negatives with probability at most $1/2$.
\end{theorem}

\begin{proof}
Run Algorithm \ref{alg:cutandcount} by setting $\univ=V$, 
and $\countproc$ to be the algorithm implied by Lemma \ref{lem:stein}. 
The correctness follows from Corollary \ref{cor:cutandcount} by setting  
$\sols= \bigcup_\targetW \sols_\targetW$ and $\objs= \bigcup_\targetW \objs_\targetW$
and Lemma \ref{lem:steincancel}. 
It is easy to see that the timebound follows from Lemma \ref{lem:stein}.
\end{proof}

\subsection{Directed Cycle Cover}\label{sec:ill:dcc}
\label{ssec:dcc}

\defproblemu{{\sc Directed} \mincyclecovername}
{A directed graph $D = (V,A)$, an integer $k$.}
{Can the vertices of $D$ be covered with at most $k$ vertex disjoint directed cycles?}

This problem is significantly different from the one 
considered in the previous section since the aim is to maximize 
connectivity in a more flexible way: 
in the previous section the solution induced one connected component, 
while it may induce at most $k$ weakly connected components in the context of the 
current section. 
Note that with the Cut\&Count technique as introduced above, 
the solutions we are looking for cancel modulo $2$. 

We introduce a concept called {\em markers}.
A set of solutions contains pairs $(X,M)$, where $X \subseteq A$
is a cycle cover and $M \subseteq V, |M|=k$ is a set of {\em marked}
vertices such that each cycle in $X$ contains at least one marked vertex.
Observe that since $|M|=k$ this ensures that in the set of solutions
in each pair $(X,M)$ the cycle cover $X$ contains at most $k$ cycles.
Note that two different sets of marked vertices of a single
cycle cover are considered to be two \emph{different solutions}.
For this reason we assign random weights \emph{both to the arcs and vertices of $D$}. 
When we relax the requirement that in the pair $(X,M)$ each cycle in $X$
contains at least one vertex from $M$ we obtain a set of candidate solutions.
The objects we count are pairs consisting of $(i)$ a pair $(X,M)$, where $X \subseteq A$ is a cycle cover and $M \subseteq V$ is a set of $k$ markers,
$(ii)$ a cut consistent with $D[X]$, where all the marked vertices $M$
are on the left side of the cut. 
We will see that candidate solutions that contain a cycle without any
marked vertex cancel modulo $2$.
Formal definition follows.

\noindent {\bf{The Cut part.}}
As said before, we assume that we are given a weight function 
$\weight: A \cup V \rightarrow \{1,\ldots,N\}$, where $N = 2|\univ| = 2(|A|+|V|)$.

\begin{definition}
For an integer $\targetW$ we define:
\begin{enumerate}
\item $\cand_\targetW$ to be the family of candidate solutions,
  that is $\cand_\targetW$ is the family of all pairs $(X,M)$,
  such that $X \subseteq A$ is a cycle cover, i.e., 
$\outdeg_X(v)=\indeg_X(v)=1$ for every vertex $v \in V$;
  $M \subseteq V$, $|M|=k$ and $\weight(X \cup M)=W$;
\item $\sols_\targetW$ to be the family of solutions,
  that is $\sols_\targetW$ is the family of all pairs $(X,M)$,
  where $(X,M) \in \cand_\targetW$ and
  every cycle in $X$ contains at least one vertex from the set $M$;
\item $\objs_\targetW$ as all pairs $((X,M),(V_1,V_2))$ such that:
 \[ (X,M) \in \cand_\targetW, \quad \quad 
    (V_1,V_2) \textrm{ is a consistent cut of } D[X], \quad 
		\textrm{  and } \quad M \subseteq V_1. \]
\end{enumerate}
\end{definition}
Observe that the graph $D$ admits a cycle cover with at most $k$ cycles
if and only if there exists $\targetW$ such that $\sols_\targetW$
is nonempty.

\noindent {\bf{The Count part.}} 
We proceed to the Count part by showing that candidate solutions that 
contain an unmarked cycle cancel modulo $2$.

\begin{lemma}
\label{lem:dircyccanc}
Let $D,\weight, \objs_\targetW$ and $\sols_\targetW$ 
be defined as above. 
Then for every $\targetW$, $|\sols_\targetW| \equiv |\objs_\targetW|$.
\end{lemma}
\begin{proof}
For subsets $M \subseteq V$ and $X \subseteq A$, let $\compnomark (M,X)$
denote the number of weakly connected components of $D[X]$ not containing any 
vertex of $M$. Then
\[
	|\objs_\targetW|= \sum_{(X,M) \in \cand_\targetW} \; 2^{\compnomark (M,X)}.
\]
To see this, note that for any $((X,M),(V_1,V_2)) \in \objs_\targetW$ 
and any vertex set $C$ of a cycle from $X$ such that 
$M \cap C = \emptyset$, we have $((X,M),(V_1 \triangle C ,V_2 \triangle C)) 
\in \objs_\targetW$ --- we can move all the vertices of $C$ to the other
side of the cut, also obtaining a consistent cut. 
Thus, for any set of choices of a side of the cut for every cycle
not containing a marker, 
there is an object in $\objs_\targetW$. 
Hence (analogously to Lemma \ref{lem:evencancel}) for any $W$ and $(M,X) \in \cand_\targetW$ there are $2^{\compnomark (M,X)}$ cuts $(V_1,V_2)$ such that 
$((X,M),(V_1,V_2)) \in \objs_\targetW$
and the lemma follows, because:
$$|\objs_\targetW| \equiv
|\{((X,M),(V_1,V_2)) \in \objs_\targetW : \compnomark (M,X) = 0 \}| 
= |\sols_\targetW|.$$ 
\end{proof}

\begin{lemma}
\label{lem:dircycdp}
Given $D=(V,A)$, an integer $k$, a weight function 
$\weight: A \cup V \rightarrow \{1,\ldots,N\}$ and a nice tree decomposition $\treedecomp$, 
there is an algorithm that can determine $|\objs_\targetW|$ modulo 2 for 
every $0 \leq \targetW \leq (k+|V|)N$ in $6^tN^2|V|^{O(1)}$ time.
\end{lemma}
\begin{proof}[Proof sketch] We briefly sketch a 
$64^tN^2|V|^{O(1)}$ time algorithm, whereas the $6^tN^2|V|^{O(1)}$ time algorithm can be found in Appendix~\ref{sec:dirmincyc}.

Let $s \in \{\zzone, \zztwo, \zoone, \zotwo, \ozone, \oztwo, \ooone, 
\ootwo\}^{\bag{x}}$, and for a bag $x \in \treedecomp$ and integers 
$\iterk,\iterW$ let $A_x(\iterk,\iterW,s)$ be
the number of pairs $((X,M),(X_1,X_2))$ such that:
\begin{itemize}
  \item $X \subseteq E_x$, i.e., $X$ is a subset of the set of arcs 
	introduced by $x$ and its descendants,
	\item for every $v \in \bag{x}$ and every $\mathbf{i},\mathbf{o} \in 
	\{0,1\}$ we have $s(v)=\mathbf{io}_j \Rightarrow \indeg_X(v) = \mathbf{i} 
	\ \wedge \ \outdeg_X(v)=\textbf{o} \ \wedge v \ \in X_j$,
	\item for every $v \in V_x \setminus \bag{x}$ we have $\indeg_X(v) = \outdeg_X(v) = 1$.
	\item $(X_1,X_2)$ is a consistent cut of the graph $(V_x,X)$,
  \item $M \subseteq (X_1 \setminus \bag{x})$, $\weight(X) + \weight(M)=\iterW$ and $|M|=\iterk$.
\end{itemize}

To obtain all values $|\objs_\targetW| \mod 2$ it is enough to compute
$A_\rootv(k,\targetW,\emptyset)$ modulo two for all values of $\targetW$,
since $|\objs_\targetW| \equiv A_\rootv(k, \targetW, \emptyset)$.

Note that in the colouring we do not store the information whether a vertex
is a marker or not. This is due to the following observation:
Since the tree decomposition $\treedecomp$ is rooted in an empty bag,
for each each vertex $v \in V$ there exists exactly one bag of the tree 
decomposition which forgets $v$.
Hence if in the forget $v$ bag we have $s(v)=\mathbf{11}_1$
we have an option of making $v$ a marker and updating the accumulator
$\iterk$.

The running time $64^tN^2 |V|^{O(1)}$ can be obtained by using standard 
dynamic programming by using $A_x(\iterk,\iterW,s)$ as the table.
The $64 = 8^2$ comes from the join bags --- the naive way to calculate 
the values of $A_x(\iterk,\iterW,s)$ for a join bag $x$ would be to
iterate over all pairs choices of $(\iterk,\iterW,s)$ for the child bags,
and there are $|V|^{O(1)} N^2 64^t$ such choices to consider.
In order to obtain the claimed time complexity we need to reduce the 
number of states per vertex to six and to handle the join bags more 
efficiently.

To achieve these goals for vertices with colours $\zz$ and $\oo$ we do not 
specify the side of the cut and we use variants of the fast subset 
convolution~\cite{bhkk:fourier-moebius} algorithm.
The details can be found in Appendix \ref{sec:dirmincyc}.
\end{proof}

\begin{theorem}\label{thm:dcc-main}
 There exists a Monte-Carlo algorithm that given a tree decomposition of width $t$
 solves {\sc Directed} \mincyclecovername{} in $6^t |V|^{O(1)}$ time.
 The algorithm cannot give false positives and may give false negatives with probability at most $1/2$.
\end{theorem}

\begin{proof}
Run algorithm Algorithm \ref{alg:cutandcount} by setting $\univ= A \cup V$ 
and $\countproc$ to be the algorithm implied by Lemma 
\ref{lem:dircycdp}. 
The correctness follows from Corollary \ref{cor:cutandcount} by setting 
$\sols= \bigcup_\targetW \sols_\targetW$ and $\objs= \bigcup_\targetW \objs_\targetW$
and Lemma \ref{lem:dircyccanc}. 
It is easy to see that the timebound follows from Lemma \ref{lem:dircycdp}.
\end{proof}

\newcommand{\weifun}{{\omega}}
\newcommand{\tarwei}{{W}}
\newcommand{\tarnum}{{k}}
\newcommand{\cansol}{{\mathcal{R}}}
\newcommand{\consol}{{\mathcal{S}}}
\newcommand{\cutpar}{{\mathcal{C}}}
\newcommand{\maxwei}{{N}}
\newcommand{\solset}{{X}}
\newcommand{\marset}{{M}}
\newcommand{\solcut}{{C}}

\section{Applications of the technique for other problems}
\label{sec:applications}

We now proceed to sketch $|V|^{O(1)}c^{\tw(G)}$ algorithms for
other problems mentioned in the introduction. For the sake of brevity we
present only quick sketches here: for each problem we define
what set of solution candidates do we consider; which candidate-cut pairs
we count; if necessary we argue
why the non-connected candidates are counted an even number of times,
while the connected candidates are counted only once; and
finally we describe what states do we consider for a given bag and
a given weight-sum in the dynamic programming subroutine. We also
briefly mention what techniques do we use to compute the values of the
dynamic programming in the join bags, as this is the most non-trivial
bag to compute efficiently.

For full descriptions of the aforementioned algorithms we refer the reader
to Appendix \ref{sec:app:alg}.

\subsection{\fvs}\label{sec:main:fvs}

\newcommand{\forest}{\mathbf{F}}
\newcommand{\marker}{\mathbf{M}}

\defproblemu{\fvs}
{An undirected graph $G = (V,E)$ and an integer $k$.}
{Does there exist a set $Y \subset V$ of cardinality $k$ so that 
$G[V\setminus Y]$ is a forest?}

Here defining the families $\cansol$ and $\consol$ is somewhat more tricky,
as there is no explicit connectivity requirement in the problem to begin
with.
We proceed by choosing the (presumed) forest left after removing the
candidate solution and using the following simple lemma:

\begin{lemma}
A graph with $n$ vertices and $m$ edges is a forest iff it has at most 
$n-m$ connected components.
\end{lemma}
The simple proof is given in the appendix.

Thus, we ensure that a solution is a forest by counting the vertices
and edges, and ensuring the number of connected components is bounded
from above by using markers, as in Section \ref{sec:ill:dcc}.

Note that here we want to keep track of two distinct vertex sets --- $M$ and
$X$. 
We thus set $\univ = V \times \{\forest,\marker\}$, i.e.,
for each vertex $v \in V$ we choose two weights:
$\weight((v, \forest))$ and $\weight((v,\marker))$,
one for $v \in X$, and the second one for $v \in M$ (the details of this
are presented in the appendix).

The family of solution candidates is defined (as in the \mincyclecovername 
case, where we also used markers)
as the family of pairs $(\solset,\marset)$ such that 
$\marset \subseteq \solset \subseteq V$.
We use accumulators to keep track of the size of $\solset$, the number of 
edges in $G[\solset]$,
the size of $\marset$ and the total weight of
$(\solset \times \{\forest\}) \cup (\marset\times \{\marker\})$.
The solution is a solution candidate $(\solset,\marset)$ with the additional properties
that $G[\solset]$ is a forest and each connected component of $G[\solset]$ contains at least
one vertex from $\marset$.

With a pair $(\solset,\marset)$ we associate consistent cuts of
$G[\solset]$ with all vertices from $\marset$ being on the left side of the
cut.

In this case a pair $(\solset, \marset)$ is consistent with exactly one 
cut iff every connected component of $G[X]$ contains at least one marker;
on the other hand if there are $\compnomark$ components containing no 
markers, we have $2^\compnomark$ consistent cuts.

For each bag we keep states with the following parameters:
\begin{itemize}
\item the number of vertices chosen to be in $\solset$ ($|V|+1$ possiblities);
\item the number of already introduced edges in $G[\solset]$ 
($|V|$ possibilities,
we can discard solution candidates with at least $|V|$ edges in $G[\solset]$);
\item the number of markers already chosen ($|V|+1$ possibilities);
\item the sum of the appropriate weights of vertices already in $\solset$ 
and already in $\marset$;
\item for each vertex in the bag one of the three possible colours: either
it is not in $\solset$, or it is in $\solset$ and on the left, or in
$\solset$ and on the right ($3^{\tw(G)}$ possibilities in total).
\end{itemize}
Note that (as in Section \ref{sec:ill:dcc}) we do not remember whether
a vertex is a marker, instead simply making the choice in the appropriate
forget bag.

Our algorithm answers ``yes'' if for some integer $m$ and for some
weight sum $\tarwei$ there is an odd number of consistent pairs with 
$|V| - \tarnum$ vertices, $m$ edges, $|V| - \tarnum - m$ markers.
The join bag is trivial in this case --- the choices for each vertex have
to match exactly.

The details of the algorithm are given in Appendix \ref{sec:app:fvs}.

\subsection{\cvertexcover}

\defproblemu{\cvertexcover}
{An undirected graph $G = (V,E)$ and an integer $\tarnum$.}
{Does there exist such a connected set $\solset \subset V$ of cardinality 
$\tarnum$ that each edge is incident to at least one vertex from $\solset$?}

We choose one vertex $\anchor$ which we assume to be in the
solution (we can check just two choices, using two endpoints of some edge).
As we choose a vertex set, we randomly select a weight function
$\weifun : V \ra \{1,\ldots,\maxwei\}$.
The family of solution candidates $\cansol_\tarwei$ consists of
vertex covers of size $\tarnum$ with $\weifun(\solset) = \tarwei$ and containing $\anchor$.
The family of solutions $\consol_\tarwei$ contains
elements of $\cansol_\tarwei$ that induce a connected subgraph.
As in the \steinertree problem, we take $\cutpar_\tarwei$
to be the family of pairs $(\solset, \solcut)$, with $\solset \in
\cansol_\tarwei$ and $\solcut$ being a consistent cut of
$G[\solset]$, with $\anchor$ on the left side of the cut
(to break the symmetry).

In the dynamic programming subroutine, for each bag we keep states with
the following parameters:
\begin{itemize}
\item the number of vertices already chosen to be in $\solset$ ($\tarnum+1$ 
possibilities);
\item the sum of the weights of those vertices;
\item for each vertex one of the three states: in $\solset$ and on the left,
in $\solset$ and on the right, or not in $\solset$
($3^{\tw(G)}$ possibilities in total).
\end{itemize}
The vertex cover condition is checked in the introduce edge bags.
We need no tricks in the join bag, the states of all the vertices need
to match.

The details of the algorithm are given in Appendix \ref{sec:app:cx}.

\subsection{\cdomset}
\defproblemu{\cdomset}
{An undirected graph $G = (V,E)$ and an integer $\tarnum$.}
{Does there exist such a connected set $\solset \subset V$ of cardinality at
most $\tarnum$ that $N[\solset] = V$?}

The reasoning here matches the one for the \cvertexcover problem almost exactly.
We fix some vertex $\anchor$, which we require to be a part of the solution
(to obtain a general algorithm we iterate over all possible choices of 
$\anchor$).
As we choose a vertex set, we randomly select a weight function
$\weifun : V \ra \{1,\ldots,\maxwei\}$.
The family $\cansol_\tarwei$ is the family of dominating sets
$\solset$ in $G$ containing $\anchor$
of size $\tarnum$ and weight $\tarwei$,
while $\consol_\tarwei$ is the family of those sets $\solset \in
\cansol_\tarwei$ for which $G[\solset]$ is connected.
The family $\cutpar_\tarwei$ is defined as previously.

For each bag we keep states with the following parameters:
\begin{itemize}
\item the number of vertices already chosen to be in $\solset$ ($\tarnum+1$
possibilities);
\item the sum of the weights of all vertices already in $\solset$ 
($\maxwei\tarnum+1$ possibilities);
\item for each vertex in the bag one of the four possible states: in 
$\solset$ and on the left, in $X$ and on the right, 
not in $\solset$ and adjacent to a vertex already in $X$, or not in
$\solset$ and not adjacent to a vertex
already in $X$ ($4^{\tw(G)}$ possibilities in total).
\end{itemize}

In the join bag the states of the vertices in $X$ have to match, while
the states of the vertices not in $X$ are joined using a standard
Fast Subset Convolution procedure.
The details are given in Appendix \ref{sec:app:cx}.
\subsection{\coct}

\defproblemu{\coct}
{An undirected graph $G = (V,E)$ and an integer $\tarnum$.}
{Does there exist such a connected set $\solset \subset V$ of cardinality 
$\tarnum$ that $G[V \setminus \solset]$ is bipartite?}

The reasoning here matches the one for the previous two problems almost exactly.
We choose one vertex $\anchor$ which we assume to be in the solution.
The only catch is that the naive dynamic programming that solves \oct{} partitions
the vertex set into three (not two) sets: the odd cycle transversal and the bipartion of the
resulting bipartite graph. Thus
$\cansol_\tarwei$ is the family of such partitions of $V$,
  i.e., for one odd cycle transversal $\solset$,
  different bipartitions of $G[V \setminus \solset]$
  result in different solution candidates.
The integer $\tarnum$ given in the input represents the size of the odd cycle transversal,
    whereas the index $\tarwei$ represents the weight of the partition.
As each vertex is in one of three sets in an element of $\cansol_\tarwei$,
we need to generate two weights per vertex, so that each element of $\cansol_\tarwei$
corresponds to a different subset of the domain of the weight function $\weight$.
As previously, $\consol_\tarwei$ are the elements of $\cansol_\tarwei$ where
the odd cycle transversal is connected, and in $\cutpar_\tarwei$ we pair up
candidate solutions with cuts consistent with the odd cycle transversal where $\anchor$
is on a fixed side of the cut.

For each bag we keep states with the following parameters:
\begin{itemize}
\item the number of vertices already chosen to be in $\solset$ ($\tarnum+1$ 
possibilities);
\item the weight of the partition;
\item for each vertex one of the four states: in $\solset$ and on the left,
in $\solset$ and on the right, or in one of two colour classes of $G[V \setminus \solset]$
($4^{\tw(G)}$ possibilities in total).
\end{itemize}
The condition whether the chosen colour classes are correct is checked in the introduce edge bags.
We need no tricks in the join bag, the states of all the vertices need
to match.
The details are given in Appendix \ref{sec:app:cx}.
\subsection{\cfvs}

\defproblemu{\cfvs}
{An undirected graph $G = (V,E)$ and an integer $k$.}
{Does there exist a set $Y \subset V$ of cardinality $k$ so that $G[Y]$ is connected and
$G[V\setminus Y]$ is a forest?}

Here we use the same approach as in the previous three algorithms:
we take an algorithm for \fvs{} and add cuts consistent with the solution.
However, now the base algorithm is not that easy,
 because it is the one described in Section~\ref{sec:main:fvs}, already using the Cut\&Count technique.
We thus need to apply Cut\&Count twice here, but --- as we will see ---
there are no significant difficulties.

As before, we fix one vertex $\anchor$ to be included in the connected feedback vertex set $Y$.

As in the \fvs{} algorithm, we ensure that a solution induces a forest by counting vertices
and edges, and ensuring the number of connected components is bounded
from above by using markers, as in Section~\ref{sec:ill:dcc}.

As we want to keep track of two distinct vertex sets --- $M$ and
$X$, we set $\univ = V \times \{\forest,\marker\}$, i.e.,
for each vertex $v \in V$ choose two weights
$\weight((v, \forest))$ and $\weight((v,\marker))$
one for $v \in X$, and the second one for $v \in M$.

The family of solution candidates is defined just as in the \fvs{} case:
as the family of pairs $(\solset,\marset)$ such that $\marset \subseteq \solset \subseteq V$, with 
accumulators keeping track of the size of $\solset$, the number of edges in $G[\solset]$,
the size of $\marset$ and the total weight of
$\solset \times \{\forest\} \cup \marset\times \{\marker\}$.
The solution is a solution candidate $(\solset,\marset)$ with the additional properties
that $G[\solset]$ is a forest, each connected component of $G[\solset]$ contains at least
one vertex from $\marset$ and, additionally to the \fvs{} case, that $G[V \setminus \solset]$ is connected.

With a pair $(\solset,\marset)$ we associate two consistent cuts:
one of $G[\solset]$ with all vertices from $\marset$ being on the left side of the
cut, and one of $G[V \setminus \solset]$, where $\anchor$ is on the left side
of the cut.

In this case a pair $(\solset, \marset)$ is consistent with exactly one 
cut of $G[X]$ iff every connected component of $G[X]$ contains at least one marker;
on the other hand if there are $c$ components containing no markers,
we have $2^c$ consistent cuts.
Moreover, a pair $(\solset, \marset)$ is consistent with $2^{\conncomp(G[V \setminus \solset])-1}$
cuts of $G[V \setminus \solset]$. Thus, a pair $(\solset,\marset)$ is counted only once
if it is a solution, and an even number of times otherwise.

The dynamic programming proceeds exactly as in the \fvs case.
The details are given in Appendix \ref{sec:app:cx}.



\subsection{{\sc{Undirected}} \mincyclecovername}

We use similar approach as for the directed case in Section~\ref{ssec:dcc}.
However in the undirected case the number of states is decreased
because instead of keeping track of both indegree and
outdegree of a vertex we handle only a single degree. 
Similarly as for the directed case we do not store the side of the cut
for vertices of degree zero and two, which gives exactly $4$ states per vertex
and leads to $4^{\tw(G)} |V|^{O(1)}$ time complexity.

In the join bag we need to use a variant of the Fast Fourier Transform.
Details can be found in Appendix~\ref{sec:dirmincyc}.

\subsection{{\sc(Directed)} \longestcycle and {\sc{(Directed)}} \longestpath}

\defproblemu{{\sc{(Directed)}} \longestcycle}
{An undirected graph $G=(V,E)$ (or a directed graph $D = (V,A)$) and an integer $k$.}
{Does there exist a (directed) simple cycle of length $k$ in $G$ ($D$)?}

\defproblemu{{\sc{(Directed)}} \longestpath}
{An undirected graph $G=(V,E)$ (or a directed graph $D = (V,A)$) and an integer $k$.}
{Does there exist a (directed) simple path of length $k$ in $G$ ($D$)?}

Obviously an algorithm for \longestcycle implies an algorithm of the same time
complexity for the \hamcycle problem.
Moreover in Appendix~\ref{sec:dirmincyc} we show that the \longestpath problem
both in the directed and undirected case may be reduced to
the appropriate variant of the \longestcycle problem.

Observe that for the {\sc{(Directed)}} \longestcycle problem
we can mimic the algorithm for the {\sc{(Directed)}} \mincyclecovername
problem.
It is even easier because we are looking for a connected object
which means that we do not have to use markers.
The only difference between {\sc{(Directed)}} \longestcycle problem
and {\sc{(Directed)}} \mincyclecovername
is that in the {\sc{(Directed)}} \longestcycle problem
we need to count the number of chosen edges, since
we allow vertices of degree zero.

Details can be found in Appendix~\ref{sec:dirmincyc}.

\subsection{\exactleaf and \exactoutbranching}

\defproblemu{\exactleaf}
{An undirected graph $G = (V,E)$ and an integer $\tarnum$.}
{Does there exists a spanning tree of $G$ with exactly $\tarnum$ leaves?}

\defproblemu{\exactoutbranching}
{A directed graph $D = (V,A)$ and an integer $\tarnum$, 
and a root $r \in V$.}
{Does there exist a spanning tree of $D$ with all edges directed away from
the root with exactly $\tarnum$ leaves?}

The above problems generalize the following problems: \maxleaf,
\minleaf, \maxoutbranching and \minoutbranching,
which ask for the number of leaves to be at least $\tarnum$ 
or at most $\tarnum$.

In this subsection we only sketch a natural $8^{\tw(G)}|V|^{O(1)}$
solution to the \exactoutbranching problem 
since it generalizes the \exactleaf problem (simply direct
each edge in both directions and add a root $r$ with a single outgoing arc).
This can be improved to $4^{\tw(G)}|V|^{O(1)}$
for \exactleaf and to $6^{\tw(G)}|V|^{O(1)}$ for
\exactoutbranching using a less intuitive definition of solution 
candidates together with a binomial transform for join bags,
which is described in Appendix~\ref{sec:app:exact-k-stuff}.

As we choose an arc set, we randomly select a weight function
$\weifun : A \ra \{1,\ldots,\maxwei\}$.
The family of solution candidates $\cansol_\tarwei$ consists of
sets of exactly $|V|-1$ arcs of total weight $\tarwei$,
such that exactly $\tarnum$ vertices have no outgoing arc,
each vertex except the root has exactly one incoming arc
and the root $r$ has no incoming arc.
$\consol_\tarwei$ are the solution candidates 
$\solset \in \cansol_\tarwei$ such that 
the underlying undirected graph of $D[\solset]$ is connected,
which is equivalent to $X$ being an outbranching.
We take $\cutpar_\tarwei$ to be the family 
of pairs $(\solset, \solcut)$, with $\solset \in
\cansol_\tarwei$ and $\solcut$ being a consistent cut of
the underlying undirected graph of $D[\solset]$.

For each bag we keep states with the following parameters:
\begin{itemize}
\item the number of already chosen arcs,
\item the sum of the weights of those arcs,
\item the number of already forgotten vertices with no outgoing arcs,
\item for each vertex one of eight states, denoting
the side of the cut (two possibilities), the indegree ($0$ or $1$, two possibilities)
  and whether we have already chosen some outgoing arc from this vertex (two possibilities).
\end{itemize}

For the merging of states we would have to use the Fast
Subset Convolution algorithm (details in Appendix \ref{sec:app:exact-k-stuff}).

\subsection{\maxspantree}
\defproblemu{\maxspantree}
{An undirected graph $G = (V,E)$ and an integer $\tarnum$.}
{Does there exist a spanning tree $T$ of $G$ for which there are at least
$\tarnum$ vertices satisfying $\deg_G(v) = \deg_T(v)$?}

A solution is any set $\solset \subset E$ with the following 
properties:
\begin{itemize}
\item $|\solset| = |V|-1$;
\item there are exactly $\tarnum$ vertices $v$ for which 
$\deg_G(v) = \deg_{G[X]}(v)$;
\item $G[\solset]$ is connected 
(as we have $|V|-1$ edges, this is equivalent to $G[\solset]$ being a tree).
\end{itemize}
We define solution candidates and the pairs to count as usual,
i.e., we count consistent cuts of $G[\solset]$.
Note that we cut the whole set $V$ (instead of only cutting the vertices incident to 
some edge of $\solset$ --- as we want to assure that all the vertices
are connected, and not only $G[V(\solset)]$).

In each bag we parametrize the states as follows:
\begin{itemize}
\item the number of edges already chosen to be in $\solset$ 
($V$ possibilities);
\item the sum of their weights ($\maxwei|V|$ possibilities);
\item the number of already forgotten vertices satisfying $\deg_G(v) = 
\deg_{G[\solset]}(v)$ ($\tarnum+1$ possibilities);
\item for each vertex in the bag, we remember on which side of the cut it
is, and whether there was an introduced edge incident to this vertex which
was not chosen to be included in $\solset$ ($4^{\tw(G)}$ possibilities in total).
\end{itemize}
In the join bag we use the standard Fast Subset Convolution algorithm. For
a precise description, see Appendix \ref{sec:app:mfdst}.

\subsection{\gmtsp}
\defproblemu{\gmtsp}
{An undirected graph $G = (V,E)$ and an integer $\tarnum$.}
{Does there exist a closed walk (possibly repeating edges and vertices) of length
at most $\tarnum$ that visits each vertex of the graph at least once?}

Note that the existence of such a cycle is equivalent to the existence
of a multisubset $\solset$ of edges, which is Eulerian 
(that is $(V,\solset)$ is connected and the degree of each vertex is even).
In particular this means each edge can be assumed to occur in 
$\solset$ at most twice (otherwise we could remove two copies of this edge).

For each edge we have to decide whether we choose it twice, once or not at
all.
To avoid dependency problems, we assign two different independent
weights to an edge, and add the first to the some if the edge is taken
once, and the second if it is taken twice.
Note that we also needed two types of weights in the \fvs problem
where have weights for chosen vertices and for markers.

A solution candidate is a multiset of edges with each edge taken at most
twice, the total cardinality not exceeding $\tarnum$, and each vertex having
an even degree.
The family of solutions consists of such solution candidates $\solset$
that $G[\solset]$ is connected.
We choose the cuts we count with an edge as usual, i.e.,
we count consistent cuts of $G[\solset]$.
As usual we pick one vertex to
be always on the left side of the cut.

In each bag, in addition to the number of edges already chosen and the
sum of their weights, we keep for each vertex the side of the cut and
the parity of the number of already chosen edges incident to this vertex
($4^{\tw(G)}$ possibilites in total).
The values in the states of the join bag are calculated using the
Hadamard transform (details in Appendix \ref{sec:app:gmtsp}).

\section{Lower bounds}\label{sec:negatives:intro}

In this section we describe a bunch of negative results
   concerning the possible time complexities for algorithms for connectivity problems
   parameterized by treewidth or pathwidth.
Our goal is to complement our positive results by showing that
in some situations the known algorithms (including ours) probably cannot be
further improved.

First, let us introduce the complexity assumptions made in this section.
Let $c_k$ be the infimum of the set of the positive reals $c$ that satisfy the following condition:
there exists an algorithm that solves $k$-\sat{} in time $O(2^{cn})$, where $n$ denotes
the number of variables in the input formula. 
The Exponential Time Hypothesis (ETH for short) asserts that $c_3 > 0$,
whereas the Strong Exponential Time Hypothesis (SETH) asserts
  that $\lim_{k \to \infty} c_k = 1$. It is well known that SETH implies ETH \cite{ip:seth}.

The lower bounds presented below are of two different types.
In Section \ref{sec:eth:intro} we discuss several problems
that, assuming ETH, do not admit an algorithm
running in time $2^{o(p \log p)} n^{O(1)}$, where $p$ denotes the pathwidth of the input graph.
In Section \ref{sec:seth:intro} we state that, assuming SETH,
 the base of the exponent in our algorithms for \cvertexcover, \cdomset, \cfvs, \coct,
 \fvs, \steinertree and \exactleaf cannot be improved further.
 All proofs are postponed to Appendices \ref{sec:negatives:eth} and \ref{sec:negatives:seth}
 respectively.

\subsection{Lower bounds assuming ETH}\label{sec:eth:intro}

In Section \ref{sec:applications} we have shown that a lot of well-known
algorithms running in $2^{O(t)} n^{O(1)}$ time can be turned into algorithms
that keep track of the connectivity issues, with only small loss in the base of the exponent.
The problems solved in that manner include \cvertexcover{}, \cdomset{}, \cfvs{} and \coct{}.
Note that using the markers technique introduced in Section \ref{sec:ill:dcc}
we can solve similarly the following artificial generalizations:
given a graph $G$ and an integer $r$,
what is the minimum size of a vertex cover (dominating set, feedback vertex set, odd cycle transversal)
that induces {\bf{at most}} $r$ connected components?

We provide evidence that problems in which we would ask 
to maximize (instead of minimizing)
the number of connected components are harder: they probably
do not admit algorithms running in time $2^{o(p \log p)} n^{O(1)}$, where
$p$ denotes the pathwidth of the input graph.
More precisely, we show that assuming ETH there do not exist
algorithms for \cyclepackingname{}, \maxcyclecovername{} and \maxccdsname{}
running in time $2^{o(p \log p)} n^{O(1)}$.

Let us start with formal problem definitions. The first two problems have undirected and directed
versions.

\defproblemu{\cyclepackingname{}}{A (directed or undirected) graph $G = (V,E)$ and an integer $\ell$}{Does $G$ contain $\ell$ vertex-disjoint cycles?}

\defproblemu{\maxcyclecovername{}}{A (directed or undirected) graph $G = (V,E)$ and an integer $\ell$}{Does $G$ contain a set of at least $\ell$ vertex-disjoint cycles such that each vertex of $G$ is on exactly one cycle?}

The third problem is an artificial problem defined by us that should be compared to \cdomset{}.

\defproblemu{\maxccdsname{}}{An undirected graph $G=(V,E)$ and integers $\ell$ and $r$.}{Does $G$ contain a dominating set of size at most $\ell$ that induces {\bf{at least}} $r$ connected components?}

We prove the following theorem:
\begin{theorem}\label{thm:negative-main}
Assuming ETH, there is no $2^{o(p \log p)} n^{O(1)}$ time algorithm for \cyclepackingname{}, \maxcyclecovername{} (both in the directed and undirected setting) nor for \maxccdsname{}. 
The parameter $p$ denotes the width of a given path decomposition of the input graph.
\end{theorem}

The proofs go along the framework introduced by Lokshtanov et al. \cite{marx:superexp}.
We start our reduction from the \kknphittingsetname{} problem.
By $[k]$ we denote $\{1,2,\ldots,k\}$. In the set $[k] \times [k]$ {\em{a row}} is a set $\{i\} \times [k]$ and {\em{a column}} is a set $[k] \times \{i\}$ (for some $i \in [k]$).

\defproblemu{\kknphittingsetname}{A family of sets $S_1, S_2 \ldots S_m \subseteq [k] \times [k]$,
  such that each set contains at most one element from each row of $[k] \times [k]$.}{
    Is there a set $S$ containing exactly one element from each row
      such that $S \cap S_i \neq \emptyset$ for any $1 \leq i \leq m$?}

We also consider a permutation version of this problem, \kkhittingsetname{}, where
the solution $S$ is also required to contain exactly one vertex from each column.

\begin{theorem}[\cite{marx:superexp}, Theorem 2.4]\label{thm:marx-hs}
Assuming ETH, there is no $2^{o(k \log k)} n^{O(1)}$ time algorithm for \kknphittingsetname{} nor for \kkhittingsetname{}
\end{theorem}
Note that in \cite{marx:superexp} the statement of the above theorem only includes \kknphittingsetname{}.
However, the proof in \cite{marx:superexp} works for the permutation variant as well without any modifications.

All proofs can be found in Section \ref{sec:negatives:eth}.
We first prove the bound for \maxccdsname{} by a simple reduction from \kknphittingsetname{}.
Next we provide a more involved reduction from 
\kkhittingsetname{} to undirected \cyclepackingname{}.
Finally, using rather elementary gadgets, we reduce undirected 
\cyclepackingname{} to the directed case and to both cases of 
\maxcyclecovername{}.

\subsection{Lower bounds assuming SETH}\label{sec:seth:intro}

Following the framework introduced by Lokshtanov et al. 
\cite{treewidth-lower}, we prove that an improvement in the base of the 
exponent in a number of our algorithms
would contradict SETH. Formally, we prove the following theorem.

\begin{theorem}\label{thm:lower-seth-main}
Unless the Strong Exponential Time Hypothesis is false,
there do not exist a constant $\eps>0$ and an algorithm that given an 
instance $(G=(V,E),k)$ or $(G=(V,E),T,k)$
together with a path decomposition of the graph $G$ of width $p$ 
solves one of the following problems:
\begin{enumerate}
\item \cvertexcover in $(3-\eps)^p |V|^{O(1)}$ time,
\item \cdomset in $(4-\eps)^p |V|^{O(1)}$ time,
\item \cfvs in $(4-\eps)^p |V|^{O(1)}$ time,
\item \coct in $(4-\eps)^p |V|^{O(1)}$ time,
\item \fvs in $(3-\eps)^p |V|^{O(1)}$ time,
\item \steinertree in $(3-\eps)^p |V|^{O(1)}$ time,
\item \exactleaf in $(4-\eps)^p |V|^{O(1)}$ time.
\end{enumerate}
\end{theorem}

Note that \vertexcover (without a connectivity requirement) admits a 
$2^t |V|^{O(1)}$ algorithm whereas \domset, \fvs and \oct admit 
$3^t |V|^{O(1)}$ algorithms and those algorithms
are optimal (assuming SETH) \cite{treewidth-lower}.
To use the Cut\&Count technique for the connected versions of these problems
we need to increase the base of the exponent by one to keep the side of 
the cut for vertices
in the solution.
Theorem \ref{thm:lower-seth-main} shows that this is not an artifact of
the Cut\&Count technique, but rather an intrinsic characteristic of these 
problems.

In Appendix \ref{sec:negatives:seth} we provide reductions in spirit of \cite{treewidth-lower}
that prove the first five lower bounds in Theorem \ref{thm:lower-seth-main}.
The result for \exactleaf is immediate from the result for \cdomset, as \cdomset
is equivalent to \maxleaf \cite{maxleaf-cds}. The result for \steinertree follows from the simple observation
that a \cvertexcover instance can be turned into a \steinertree instance by subdividing
each edge with a terminal.

\section{Concluding remarks}
\label{sec:conclusions}

The main consequence of the Cut\&Count technique as presented in this work
could (informally) be stated as the following rule of thumb: 

\newtheorem*{roth}{Rule of thumb}
\begin{roth}
Suppose we are given a graph $G=(V,E)$, a tree decomposition $\treedecomp$ 
of $G$, and implicitly a set family $\mathcal{F}$ of subgraphs of $G$.
Moreover, suppose that for a bag $x \in \treedecomp$ the behaviour of a 
partial subgraph $S \cap G_x$ depends only on a (small) interface $I(S,x)$ 
with bag $V_x$. 
Let $\theta = \max_{x\in \treedecomp} |\{I(S,x): S \in \mathcal{F}\}|$.
Then we can compute $\min_{S \in \mathcal{F}} \conncomp(S)$ in 
$(\theta |V|)^{O(1)}$ time.
\end{roth}
To clarify, note that a standard dynamic programming for determining 
whether $\mathcal{F}$ is empty or not runs in $(\theta |V|)^{O(1)}$ time, 
and usually even in $\theta |V|^{O(1)}$ time.
In fact, the dominant term $\theta^{O(1)}$ in the claimed running time is a
rather cruel upper bound for the Cut\&Count technique as well,
as for many problems we can do a lot better.
If $\theta=c^t |V|^{O(1)}$ with $t$ being the treewidth of $\treedecomp$,
then for many instances of the above we have shown solutions running in
$\theta=(c+c')^t |V|^{O(1)}$ time, where $c'$ is, intuitively, the number 
of states affected by the cuts of the Cut\&Count technique.
Moreover, many problems cannot be solved faster unless the Strong
Exponential Time Hypothesis fails.
We have chosen not to pursue a general theorem in the above spirit, as the
techniques required to get optimal constants seem varied and depend on the
particular problem.

We have also shown that several problems in which one aims to maximize the
number of connected components are not solvable in 
$2^{o(p \log p)}|V|^{O(1)}$ unless the Exponential Time Hypothesis fails. 
Hence, assuming the Exponential Time Hypothesis, there is a marked 
difference between the minimization and maximization of the number of 
connected components in this context.

Finally, we leave the reader with some interesting open questions:
\begin{itemize}
	\item Can Cut\&Count be derandomized?
	For example, can \cvertexcover be solved deterministically in 
	$c^{t}|V|^{O(1)}$ on graphs of treewidth $t$ for some constant $c$?
	\item Since general derandomization seems hard, we ask whether it is 
	possible to derandomize the presented FPT algorithms parameterized
  by the solution size for \fvs, \cvertexcover or \cfvs? Note that the
	tree decomposition considered in these algorithms is of a very specific
	type, which could potentially make this problem easier than the previous
	one.
	\item Do there exist algorithms running in time $c^{t}|V|^{O(1)}$ on graphs 
	of treewidth $t$ that solve counting or weighted variants? For example 
	can the number of Hamiltonian paths be determined, or the Traveling 
	Salesman Problem solved in $c^{t}|V|^{O(1)}$ on graphs of treewidth $t$?
	\item Can exact exponential time algorithms be improved using Cut\&Count 
	(for example for \cdomset , \steinertree and \fvs)?
  \item All our algorithms for directed graphs run in time $6^t |V|^{O(1)}$.
  Can the constant $6$ be improved? Or maybe it is optimal (again, assuming SETH)?
\end{itemize}

\subsubsection*{Acknowledgements}
The second author would like to thank Geevarghese Philip for some useful discussions in a very early stage.

\bibliographystyle{plain}
\bibliography{cut}

\begin{thebibliography}{10}

\bibitem{Agrawal:2003:PIT:792538.792540}
Manindra Agrawal and Somenath Biswas.
\newblock {Primality and identity testing via Chinese remaindering}.
\newblock {\em J. ACM}, 50:429--443, July 2003.

\bibitem{fastFAST}
Noga Alon, Daniel Lokshtanov, and Saket Saurabh.
\newblock Fast fast.
\newblock In Susanne Albers, Alberto Marchetti-Spaccamela, Yossi Matias,
  Sotiris Nikoletseas, and Wolfgang Thomas, editors, {\em Automata, {L}anguages
  and {P}rogramming}, volume 5555 of {\em Lecture Notes in Computer Science},
  pages 49--58. Springer Berlin / Heidelberg, 2009.

\bibitem{DBLP:conf/approx/ArvindM08}
Vikraman Arvind and Partha Mukhopadhyay.
\newblock {Derandomizing the Isolation Lemma and Lower Bounds for Circuit
  Size}.
\newblock In Ashish Goel, Klaus Jansen, Jos{\'e} D.~P. Rolim, and Ronitt
  Rubinfeld, editors, {\em APPROX-RANDOM}, volume 5171 of {\em Lecture Notes in
  Computer Science}, pages 276--289. Springer, 2008.

\bibitem{fvs:4krand}
Ann Becker, Reuven Bar-Yehuda, and Dan Geiger.
\newblock Randomized {A}lgorithms for the {L}oop {C}utset {P}roblem.
\newblock {\em J. Artif. Intell. Res. (JAIR)}, 12:219--234, 2000.

\bibitem{cvc-br}
Daniel Binkele-Raible.
\newblock {\em {Amortized Analysis of Exponential Time and Parameterized
  Algorithms: Measure and Conquer and Reference Search Trees}}.
\newblock PhD thesis, University of Trier, 2010.

\bibitem{bjorklund-focs}
Andreas Bj{\"o}rklund.
\newblock {D}eterminant {S}ums for {U}ndirected {H}amiltonicity.
\newblock In {\em FOCS}, pages 173--182, 2010.

\bibitem{bhkk:fourier-moebius}
Andreas Bj{\"o}rklund, Thore Husfeldt, Petteri Kaski, and Mikko Koivisto.
\newblock Fourier meets m{\"o}bius: fast subset convolution.
\newblock In {\em Proc. of STOC'07}, pages 67--74, 2007.

\bibitem{bjorklund-tsp-bound}
Andreas Bjorklund, Thore Husfeldt, Petteri Kaski, and Mikko Koivisto.
\newblock {T}he {T}ravelling {S}alesman {P}roblem in {B}ounded {D}egree
  {G}raphs.
\newblock In Luca Aceto, Ivan Damgard, Leslie Goldberg, Magnus Halldorsson,
  Anna Ingolfsdottir, and Igor Walukiewicz, editors, {\em Automata, Languages
  and Programming}, volume 5125 of {\em Lecture Notes in Computer Science},
  pages 198--209. Springer Berlin / Heidelberg, 2008.

\bibitem{bjorklund-arxiv}
Andreas Bj{\"o}rklund, Thore Husfeldt, Petteri Kaski, and Mikko Koivisto.
\newblock Narrow sieves for parameterized paths and packings.
\newblock {\em CoRR}, abs/1007.1161, 2010.

\bibitem{fvs1}
Hans~L. Bodlaender.
\newblock On {D}isjoint {C}ycles.
\newblock {\em Int. J. Found. Comput. Sci.}, 5(1):59--68, 1994.

\bibitem{BodlaenderK08}
Hans~L. Bodlaender and Arie M. C.~A. Koster.
\newblock {Combinatorial Optimization on Graphs of Bounded Treewidth}.
\newblock {\em The Computer Journal}, 51(3):255--269, 2008.

\bibitem{BottomUp}
Nicolas Bourgeois, Bruno Escoffier, Vangelis Paschos, and Johan van Rooij.
\newblock A {B}ottom-{U}p {M}ethod and {F}ast {A}lgorithms for max independent
  set.
\newblock In Haim Kaplan, editor, {\em SWAT 2010}, volume 6139 of {\em Lecture
  Notes in Computer Science}, pages 62--73. 2010.

\bibitem{byrka}
Jaroslaw Byrka, Fabrizio Grandoni, Thomas Rothvo{\ss}, and Laura Sanit{\`a}.
\newblock An improved {LP}-based approximation for steiner tree.
\newblock In {\em STOC}, pages 583--592, 2010.

\bibitem{fvs:3.83k}
Yixin Cao, Jianer Chen, and Yang Liu.
\newblock On {F}eedback {V}ertex {S}et {N}ew {M}easure and {N}ew {S}tructures.
\newblock In {\em Proc. of SWAT'10}, pages 93--104, 2010.

\bibitem{DBLP:journals/siamcomp/ChariRS95}
Suresh Chari, Pankaj Rohatgi, and Aravind Srinivasan.
\newblock {Randomness-Optimal Unique Element Isolation with Applications to
  Perfect Matching and Related Problems}.
\newblock {\em SIAM J. Comput.}, 24(5):1036--1050, 1995.

\bibitem{fvs:5k}
Jianer Chen, Fedor~V. Fomin, Yang Liu, Songjian Lu, and Yngve Villanger.
\newblock Improved algorithms for feedback vertex set problems.
\newblock {\em J. Comput. Syst. Sci.}, 74(7):1188--1198, 2008.

\bibitem{nasz-icalp-tcs}
Marek Cygan and Marcin Pilipczuk.
\newblock Exact and approximate bandwidth.
\newblock {\em Theor. Comput. Sci.}, 411(40-42):3701--3713, 2010.

\bibitem{fvs7}
Frank K. H.~A. Dehne, Michael~R. Fellows, Michael~A. Langston, Frances~A.
  Rosamond, and Kim Stevens.
\newblock An {O}(2$^{\mbox{{o}(k)}}$n$^{\mbox{3}}$) {FPT} {A}lgorithm for the
  {U}ndirected {F}eedback {V}ertex {S}et {P}roblem.
\newblock In {\em COCOON}, pages 859--869, 2005.

\bibitem{bidimensionality}
Erik~D. Demaine and Mohammadtaghi Hajiaghayi.
\newblock {The bidimensionality Theory and Its Algorithmic Applications}.
\newblock 2005.

\bibitem{demillo}
Richard~A. DeMillo and Richard~J. Lipton.
\newblock {A Probabilistic Remark on Algebraic Program Testing}.
\newblock {\em Inf. Process. Lett.}, 7(4):193--195, 1978.

\bibitem{beyond-bidimensionality}
Frederic Dorn, Fedor~V. Fomin, Daniel Lokshtanov, Venkatesh Raman, and Saket
  Saurabh.
\newblock Beyond {B}idimensionality: {P}arameterized {S}ubexponential
  {A}lgorithms on {D}irected {G}raphs.
\newblock {\em CoRR}, abs/1001.0821, 2010.

\bibitem{subexponential-bound-genus}
Frederic Dorn, Fedor~V. Fomin, and Dimitrios~M. Thilikos.
\newblock {Fast Subexponential Algorithm for Non-local Problems on Graphs of
  Bounded Genus}.
\newblock In {\em SWAT}, pages 172--183, 2006.

\bibitem{Catalan}
Frederic Dorn, Fedor~V. Fomin, and Dimitrios~M. Thilikos.
\newblock Catalan structures and dynamic programming in {H}-minor-free graphs.
\newblock In {\em Proceedings of the nineteenth annual ACM-SIAM symposium on
  Discrete algorithms}, SODA '08, pages 631--640, Philadelphia, PA, USA, 2008.
  Society for Industrial and Applied Mathematics.

\bibitem{sphere-cut}
Frederic Dorn, Eelko Penninkx, Hans~L. Bodlaender, and Fedor~V. Fomin.
\newblock Efficient exact algorithms on planar graphs: {E}xploiting sphere cut
  branch decompositions.
\newblock In {\em ESA}, pages 95--106. Springer, 2005.

\bibitem{maxleaf-cds}
Robert~James Douglas.
\newblock {NP}-completeness and degree restricted spanning trees.
\newblock {\em Discrete Mathematics}, 105(1-3):41--47, 1992.

\bibitem{fvs2}
Rodney~G. Downey and Michael~R. Fellows.
\newblock Fixed {P}arameter {T}ractability and {C}ompleteness.
\newblock In {\em Complexity Theory: Current Research}, pages 191--225, 1992.

\bibitem{fvs3}
Rodney~G. Downey and Michael~R. Fellows.
\newblock In {\em Parameterized {C}omplexity}, 1999.

\bibitem{eppstein-tsp-cubic}
D.~Eppstein.
\newblock The {T}raveling {S}alesman {P}roblem for {C}ubic {G}raphs.
\newblock {\em J. Graph Algorithms Appl.}, 11:61--81, 2007.

\bibitem{eppstein}
David Eppstein.
\newblock Diameter and treewidth in minor-closed graph families.
\newblock 27:275--291, 2000.

\bibitem{Fomin06ontwo}
Fedor~V. Fomin, Serge Gaspers, Saket Saurabh, and Alexey~A. Stepanov.
\newblock {On Two Techniques of Combining Branching and Treewidth}.
\newblock 54(2):181--207, 2006.

\bibitem{twbound}
Fedor~V. Fomin. and Dimitrios~M. Thilikos.
\newblock {A Simple and Fast Approach for Solving Problems on Planar Graphs}.
\newblock In Volker Diekert and Michel Habib, editors, {\em STACS}, volume 2996
  of {\em Lecture Notes in Computer Science}, pages 56--67, 2004.

\bibitem{gebauer-tsp-bound}
H.~Gebauer.
\newblock On the {N}umber of {H}amilton {C}ycles in {B}ounded {D}egree
  {G}raphs.
\newblock In {\em ANALCO}, 2008.

\bibitem{guo:fvs}
Jiong Guo, Jens Gramm, Falk H{\"u}ffner, Rolf Niedermeier, and Sebastian
  Wernicke.
\newblock Compression-based fixed-parameter algorithms for feedback vertex set
  and edge bipartization.
\newblock {\em J. Comput. Syst. Sci.}, 72(8):1386--1396, 2006.

\bibitem{held-karp-relaxation}
Micheal Held and Richard~M. Karp.
\newblock The {T}raveling {S}alesman {P}roblem and {M}inimum {S}panning
  {T}rees.
\newblock {\em Operations Research}, 18:1138--1162, 1970.

\bibitem{held-karp-relaxation2}
Micheal Held and Richard~M. Karp.
\newblock The {T}raveling-{S}alesman {P}roblem and {M}inimum {S}panning
  {T}rees: {P}art {II}.
\newblock {\em Mathematical Programming}, 1:6--25, 1971.

\bibitem{HicksKK05}
Illya~V. Hicks, Arie M. C.~A. Koster, and Elif Koloto{\v{g}}lu.
\newblock Branch and tree decomposition techniques for discrete optimization.
\newblock {\em Tutorials in Operations Research 2005}, pages 1--19, 2005.

\bibitem{hopcroft-ullman}
John~E. Hopcroft, Rajeev Motwani, and Jeffrey~D. Ullman.
\newblock {\em Introduction to automata theory, languages, and computation -
  (2. ed.)}.
\newblock Addison-Wesley series in computer science. Addison-Wesley-Longman,
  2001.

\bibitem{ip:seth}
Russell Impagliazzo and Ramamohan Paturi.
\newblock On the {C}omplexity of k-{SAT}.
\newblock {\em J. Comput. Syst. Sci.}, 62(2):367--375, 2001.

\bibitem{iwama-tsp-cubic}
K.~Iwama and T.~Nakashima.
\newblock An {I}mproved {E}xact {A}lgorithm for {C}ubic {G}raph {TSP}.
\newblock In {\em COCOON}, pages 108--117, 2007.

\bibitem{DBLP:journals/cc/KabanetsI04}
Valentine Kabanets and Russell Impagliazzo.
\newblock {Derandomizing Polynomial Identity Tests Means Proving Circuit Lower
  Bounds}.
\newblock {\em Computational Complexity}, 13(1-2):1--46, 2004.

\bibitem{fvs6}
Iyad~A. Kanj, Michael~J. Pelsmajer, and Marcus Schaefer.
\newblock Parameterized {A}lgorithms for {F}eedback {V}ertex {S}et.
\newblock In {\em IWPEC}, pages 235--247, 2004.

\bibitem{np-karp}
Richard~M. Karp.
\newblock Reducibility {A}mong {C}ombinatorial {P}roblems.
\newblock {\em Complexity of Computer Computations}, pages 85--103, 1972.

\bibitem{kleinberg-tardos}
Jon Kleinberg and Eva Tardos.
\newblock {\em Algorithm {D}esign}.
\newblock Addison--Wesley, 2005.

\bibitem{Kloks94}
Ton Kloks.
\newblock {\em {Treewidth, Computations and Approximations}}, volume 842 of
  {\em Lecture Notes in Computer Science}.
\newblock Springer, 1994.

\bibitem{DBLP:conf/icalp/Koutis08}
Ioannis Koutis.
\newblock {Faster Algebraic Algorithms for Path and Packing Problems}.
\newblock In Luca Aceto, Ivan Damg{\aa}rd, Leslie~Ann Goldberg, Magn{\'u}s~M.
  Halld{\'o}rsson, Anna Ing{\'o}lfsd{\'o}ttir, and Igor Walukiewicz, editors,
  {\em ICALP (1)}, volume 5125 of {\em Lecture Notes in Computer Science},
  pages 575--586. Springer, 2008.

\bibitem{DBLP:conf/icalp/KoutisW09}
Ioannis Koutis and Ryan Williams.
\newblock {Limits and Applications of Group Algebras for Parameterized
  Problems}.
\newblock In Susanne Albers, Alberto Marchetti-Spaccamela, Yossi Matias,
  Sotiris~E. Nikoletseas, and Wolfgang Thomas, editors, {\em ICALP (1)}, volume
  5555 of {\em Lecture Notes in Computer Science}, pages 653--664. Springer,
  2009.

\bibitem{treewidth-lower}
Daniel Lokshtanov, Daniel Marx, and Saket Saurabh.
\newblock {Known Algorithms on Graphs of Bounded Treewidth are Probably
  Optimal}.
\newblock In {\em Proc. of SODA'11 (to appear)}, 2011.

\bibitem{marx:superexp}
Daniel Lokshtanov, Daniel Marx, and Saket Saurabh.
\newblock {Slightly Superexponential Parameterized Problems}.
\newblock In {\em Proc. of SODA'11 (to appear)}, 2011.

\bibitem{DBLP:conf/walcom/MisraPRSS10}
Neeldhara Misra, Geevarghese Philip, Venkatesh Raman, Saket Saurabh, and
  Somnath Sikdar.
\newblock {FPT} {A}lgorithms for {C}onnected {F}eedback {V}ertex {S}et.
\newblock In Md.~Saidur Rahman and Satoshi Fujita, editors, {\em WALCOM},
  volume 5942 of {\em Lecture Notes in Computer Science}, pages 269--280.
  Springer, 2010.

\bibitem{enumerate-and-expand}
Daniel M\"{o}lle, Stefan Richter, and Peter Rossmanith.
\newblock {Enumerate and Expand: Improved Algorithms for Connected Vertex Cover
  and Tree Cover}.
\newblock 43(2):234--253, 2008.

\bibitem{isolation}
Ketan Mulmuley, Umesh~V. Vazirani, and Vijay~V. Vazirani.
\newblock Matching is as easy as matrix inversion.
\newblock {\em Combinatorica}, 7(1):105--113, 1987.

\bibitem{niedermeier:book}
Rolf Niedermeier.
\newblock {\em Invitation to Fixed Parameter Algorithms (Oxford Lecture Series
  in Mathematics and Its Applications)}.
\newblock {Oxford University Press, USA}, March 2006.

\bibitem{fvs4}
Venkatesh Raman, Saket Saurabh, and C.~R. Subramanian.
\newblock Faster {F}ixed {P}arameter {T}ractable {A}lgorithms for {U}ndirected
  {F}eedback {V}ertex {S}et.
\newblock In {\em ISAAC}, pages 241--248, 2002.

\bibitem{fvs5}
Venkatesh Raman, Saket Saurabh, and C.~R. Subramanian.
\newblock Faster fixed parameter tractable algorithms for finding feedback
  vertex sets.
\newblock {\em ACM Transactions on Algorithms}, 2(3):403--415, 2006.

\bibitem{reed:ic}
Bruce~A. Reed, Kaleigh Smith, and Adrian Vetta.
\newblock Finding odd cycle transversals.
\newblock {\em Oper. Res. Lett.}, 32(4):299--301, 2004.

\bibitem{rs:minors3}
Neil Robertson and Paul~D. Seymour.
\newblock Graph minors. {III}. {P}lanar tree-width.
\newblock {\em J. Comb. Theory}, 36(1):49--64, 1984.

\bibitem{graph-minor-theorem}
Neil Robertson and Paul~D. Seymour.
\newblock Graph {M}inors {XX}. {W}agner's conjecture.
\newblock {\em Journal of Combinatorial Theory, Series B}, 92(2):325--357,
  2004.

\bibitem{schwartz}
Jack.~T. Schwartz.
\newblock {Fast Probabilistic Algorithms for Verification of Polynomial
  Identities}.
\newblock {\em J. ACM}, 27:701--717, October 1980.

\bibitem{mixed-search}
Atsushi Takahashi, Shuichi Ueno, and Yoji Kajitani.
\newblock {Mixed Searching and Proper-Path-Width}.
\newblock {\em Theor. Comput. Sci.}, 137(2):253--268, 1995.

\bibitem{cutwidth-linear}
Dimitrios~M. Thilikos, Maria Serna, and Hans~L. Bodlaender.
\newblock Cutwidth {I}: {A} linear time fixed parameter algorithm.
\newblock 56(1):1--24, 2005.

\bibitem{tutte}
W.~T. Tutte.
\newblock The factorization of linear graphs.
\newblock {\em J. London Math. Soc.}, 22:107--111, 1947.

\bibitem{johan-tw}
Johan M.~M. van Rooij, Hans~L. Bodlaender, and Peter Rossmanith.
\newblock Dynamic {P}rogramming on {T}ree {D}ecompositions {U}sing
  {G}eneralised {F}ast {S}ubset {C}onvolution.
\newblock In {\em ESA}, pages 566--577, 2009.

\bibitem{rooij:inclusion/exclusion}
Johan M.~M. van Rooij, Jesper Nederlof, and Thomas~C. van Dijk.
\newblock {Inclusion/Exclusion Meets Measure and Conquer}.
\newblock In {\em ESA'09}, pages 554--565, 2009.

\bibitem{DBLP:journals/ipl/Williams09}
Ryan Williams.
\newblock Finding paths of length k in {O}$^{\mbox{*}}$(2$^{\mbox{k}}$) time.
\newblock {\em Inf. Process. Lett.}, 109(6):315--318, 2009.

\bibitem{yates-algorithm}
F.~Yates.
\newblock The design and analysis of factorial experiments.
\newblock Technical Communication 35, Common-wealth Bureau of Soils, Harpenden,
  U.K., 1937.

\bibitem{zippel}
Richard Zippel.
\newblock Probabilistic algorithms for sparse polynomials.
\newblock In Edward Ng, editor, {\em Symbolic and Algebraic Computation},
  volume~72 of {\em Lecture Notes in Computer Science}, pages 216--226.
  Springer Berlin / Heidelberg, 1979.

\end{thebibliography}

\appendix
\newpage

\newcommand{\targetn}{A}
\newcommand{\targetm}{B}
\newcommand{\targetcc}{C}
\newcommand{\itern}{a}
\newcommand{\iterm}{b}
\newcommand{\itercc}{c}

\section{Details of all algorithms}\label{sec:app:alg}

In this section we describe in detail all the algorithms sketched in 
Section \ref{sec:applications}.
In all algorithms we assume that we are given a tree decomposition of the
input graph $G$ of width $t$. 
The algorithms all start with constructing a nice tree decomposition, as
in Definition \ref{def:nicetreedecomp}.

\subsection{Fast subset convolution}
\label{ssec:fsc}
\newcommand{\fsc}{\ast}
\newcommand{\fsccov}{\ast_c}
\newcommand{\fscpack}{\ast_p}
\newcommand{\fscxor}{\ast_x}
\newcommand{\fscC}[1]{\fsc^{#1}}
\newcommand{\fscZ}[1]{\fscxor^{#1}}
\newcommand{\fscfour}{\fscZ{4}}
\newcommand{\fscconst}{p}

We first recall the fast subset convolution (and its variants) 
\cite{bhkk:fourier-moebius}, as it is often needed to handle join bags 
efficiently.
We follow notation from \cite{bhkk:fourier-moebius}.
Let $f,g: 2^B \to R$ for some finite set $B$ and ring $R$.
In all our applications the ring $R$ is $\Z_2$, 
thus the ring operations take constant time.

\begin{definition}
The {\em{subset convolution}}, {\em{covering product}}, {\em{packing product}}
  of $f$ and $g$ are defined as functions $f \fsc g, f \fsccov g, f\fscpack g:2^B \to R$ as follows:
\begin{align*}
\left(f \fsc g\right)(T) &= \sum_{T_1,T_2 \subseteq T} [T_1 \cup T_2 = T][T_1 \cap T_2 = \emptyset] f(T_1)g(T_2), \\
\left(f \fsccov g\right)(T) &= \sum_{T_1,T_2 \subseteq T} [T_1 \cup T_2 = T] f(T_1)g(T_2), \\
\left(f \fscpack g\right)(T) &= \sum_{T_1,T_2 \subseteq T} [T_1 \cap T_2 = \emptyset] f(T_1)g(T_2), \\
\end{align*}
\end{definition}

By computing a function $h: 2^B \to R$ we mean determining $h(T)$ for every 
$T \subseteq B$. 
Bj{\"o}rklund et al. \cite{bhkk:fourier-moebius} proved that all three functions
defined above can be computed efficiently.
\begin{theorem}[\cite{bhkk:fourier-moebius}]\label{thm:fsc}
The subset convolution, covering product and packing product of two given functions 
can be computed in $2^{|B|} |B|^{O(1)}$ ring operations.
\end{theorem}

The following generalization of the subset convolution can be found in 
\cite{nasz-icalp-tcs,johan-tw}.
\begin{definition}
Let $\fscconst \geq 2$ be an integer constant and let $B$ be a finite set.
For $t_1,t_2,t \in \{0,1,\ldots,\fscconst-1\}^B$ we say that $t_1+t_2=t$
iff $t_1(b) + t_2(b) = t(b)$ for all $b \in B$.
For functions $f,g: \{0,1,\ldots,\fscconst-1\}^B \to R$ define
$$(f \fscC{\fscconst} g)(t) = \sum_{t_1+t_2=t} f(t_1)g(t_2).$$
\end{definition}
Note that here the addition {\bf is not} evaluated in $\Z_p$ but in $\Z$.
\begin{theorem}[Generalized Subset Convolution \cite{nasz-icalp-tcs,johan-tw}]\label{thm:gfsc}
The generalized subset convolution can be computed in $\fscconst^{|B|} |B|^{O(1)}$
ring operations.
\end{theorem}
Note that in \cite{nasz-icalp-tcs} only the case $R=\Z$ is considered.
However, in our applications ($R=\Z_2$) we can perform calculations
in $\Z$ and at the end take all computed values modulo $2$ within the claimed timebound.

We need also the following variant of subset convolution.
\begin{definition}
Let $\fscconst \geq 2$ be an integer constant and let $B$ be a finite set.
For $t_1,t_2,t \in \Z_\fscconst^B$ we say that $t_1+t_2=t$
if $t_1(b) + t_2(b) = t(b)$ (in $\Z_\fscconst$) for all $b \in B$.
For functions $f,g: \Z_\fscconst^B \to R$ define the {\em{$\Z_\fscconst$ product}} as
$$(f \fscZ{\fscconst} g)(t) = \sum_{t_1+t_2=t} f(t_1)g(t_2).$$
In particular, for $\fscconst=2$ we identify elements
$\Z_2^B$ with subsets of $B$ and define for $f,g: 2^B \to R$ the {\em{xor product}} as:
$$f \fscxor g(T) = \sum_{T_1,T_2 \subseteq B} [T_1 \triangle T_2 = T] f(T_1)g(T_2).$$
\end{definition}
The xor product can be computed in time $2^{|B|}|B|^{O(1)}$ using the well-known Walsh-Hadamard transform.
However, in our applications we need also the case $\fscconst=4$, thus we provide
a proof for both cases below.
We do not state here any general theorem for arbitrary $\fscconst$,
as in that case we would need to use fractional complex numbers during the 
computations, which could lead to rounding problems.

\begin{theorem}\label{thm:fscxor}
Let $R = \Z$ or $R = \Z_q$ for some constant $q$.
For $\fscconst=2$ and $\fscconst=4$ the $\Z_\fscconst$ product of two given functions
$f,g:\Z_\fscconst^B \to R$ can be computed in time $\fscconst^{|B|} |B|^{O(1)}$.
\end{theorem}
\begin{proof}
We will use a simplified version of the Fourier transform. Let us assume $R = \Z$, as otherwise we do calculations in $\Z$ and in the end we take all values
modulo $q$. We use values of order at most $q^{O(1)}\fscconst^{O(|B|)}$, thus all arithmetical operations
in $\Z$ take polynomial time in $|B|$ and $\log q$.

Let us introduce some definitions.
Consider the ring $\Z[i] = \{a+bi: a,b \in \Z\} \subseteq \mathbb{C}$, where the addition and multiplication
operators are inherited from the complex field $\mathbb{C}$.
For $s,t \in \Z_\fscconst^B$ define $s \cdot t$ as $\sum_{b \in B} s(b)t(b) \in \Z$.
Let $\varepsilon$ be the degree-$\fscconst$ root of $1$ in $\Z[i]$, i.e.,
$\varepsilon = -1$ if $\fscconst=2$ and $\varepsilon = i$ if $\fscconst=4$.
Note that $\varepsilon^\fscconst = 1$ in $\Z[i]$.
We somewhat abuse the notation and for $c \in \Z_\fscconst$ use the notation
$\varepsilon^c$
(taking it to mean $\varepsilon^{c'}$, where $c'$ is any integer congruent
to $c$ modulo $p$).
For $f : \Z_\fscconst^B \to \Z$ define $\hat{f}:\Z_\fscconst^B \to \Z[i]$ as
  follows
$$\hat{f}(s) = \sum_{t \in \Z_\fscconst^B} f(t) \varepsilon^{s \cdot t}.$$

We first claim that $\hat{f}$ can be computed in $\fscconst^{|B|} |B|^{O(1)}$ time
using an adjusted Yates algorithm \cite{yates-algorithm}.
We may assume $B = \{1,2,\ldots,|B|\}$ and for $1 \leq b \leq |B|$ we define
$$f_b(s_1,s_2,\ldots,s_b,t_{b+1},\ldots,t_{|B|}) =
\sum_{t_1,t_2,\ldots,t_b \in \Z_\fscconst} f(t_1,t_2,\ldots,t_{|B|}) \varepsilon^{\sum_{\beta=1}^b s_\beta t_\beta}.$$
Furthermore we set $f_0 = f$. 
Note that $f_{|B|} = \hat{f}$ and for $0 \leq b < |B|$
$$f_{b+1}(s_1,s_2,\ldots,s_{b+1},t_{b+2},\ldots,t_{|B|}) = \sum_{t_{b+1} \in \Z_\fscconst} f_b(s_1,s_2,\ldots,s_b,t_{b+1},\ldots,t_{|B|}) \varepsilon^{s_{b+1} t_{b+1}}.$$
Thus, computing all $f_b$ for $0 \leq b \leq |B|$ takes $\fscconst^{|B|} |B|^{O(1)}$ time
and the claim is proven.

Let denote by $f \cdot g$ the pointwise multiplication of functions, i.e. $(f\cdot g)(t)=f(t)\cdot g(t)$. Observe that
\begin{align*}
\left(\widehat{\hat{f} \cdot \hat{g}}\right)(t) &= \sum_{s \in \Z_\fscconst^B} (\hat{f} \cdot \hat{g})(s) \varepsilon^{s \cdot t} \\
    &= \sum_{s \in \Z_\fscconst^B} \left(\sum_{t_1 \in \Z_\fscconst^B} f(t_1) \varepsilon^{t_1 \cdot s}\right) 
                                   \left(\sum_{t_2 \in \Z_\fscconst^B} g(t_2) \varepsilon^{t_2 \cdot s}\right)
                                  \varepsilon^{s \cdot t} \\
    &= \sum_{t_1,t_2 \in \Z_\fscconst^B} f(t_1)g(t_2) \sum_{s \in \Z_\fscconst^B} \varepsilon^{s\cdot (t_1+t_2+t)} \\
    &= \sum_{t_1,t_2 \in \Z_\fscconst^B} f(t_1)g(t_2) \prod_{b \in B} \left(\sum_{s(b) \in \Z_\fscconst} \left(\varepsilon^{t_1(b)+t_2(b)+t(b)}\right)^{s(b)}\right) \\
	 &\left\{ \text{using that } \sum_{s(b) \in \Z_\fscconst} (\epsilon^{\tau})^{s(b)} = \fscconst \text{ if } \tau=0 \text{ and otherwise } \frac{1-(\epsilon^\tau)^\fscconst}{1-\epsilon^\tau}=0 \right\} \\
    &= \sum_{t_1,t_2 \in \Z_\fscconst^B} f(t_1)g(t_2) \prod_{b \in B} \left(\fscconst\cdot[t_1(b)+t_2(b)+t(b) = 0]\right) \\
    &= \sum_{t_1,t_2 \in \Z_\fscconst^B} f(t_1)g(t_2) \fscconst^{|B|} [t_1+t_2+t = 0] \\
    &= \fscconst^{|B|}\left(f \fscZ{\fscconst} g\right)(-t)
\end{align*}
As $\widehat{\hat{f} \cdot \hat{g}}$ can be computed in time $\fscconst^{|B|} |B|^{O(1)}$,
the theorem follows.
\end{proof}

\subsection{\fvs}\label{sec:app:fvs}

In this section we show an algorithm for a more general version of the \fvs problem,
where we are additionally given a set of vertices that have to belong to the solution.

\defproblemu{\constrainedfvs}
{An undirected graph $G = (V,E)$, a subset $\mustset \subseteq V$ and an integer $k$.}
{Does there exist a set $Y \subset V$ of cardinality $k$ such
  that $\mustset \subseteq Y$ and $G[V\setminus Y]$ is a forest?}

This constrained version of the problem is useful when we want to obtain
not only binary output, but also in case of a positive answer a set $Y$.

Here defining a solution candidate with a relaxed connectivity condition to 
work with our technique is somewhat more tricky, as there is no explicit 
connectivity requirement in the problem to begin with.
We proceed by choosing the (presumed) forest left after removing the
candidate solution and using the following simple lemma:

\begin{lemma}
\label{lem:fvssimple}
A graph $G=(V,E)$ with $n$ vertices and $m$ edges is a forest iff it has at most 
$n-m$ connected components.
\end{lemma}

\begin{proof}
Let $E=\{e_1,\ldots,e_m\}$.
Consider a graph $G_0=(V,\emptyset)$ with the same set of vertices and an empty set of edges.
We add edges from the set $E$ to the graph $G_0$ one by one.
Observe that $G$ is a forest iff after adding each edge from $E$ to the graph $G_0$
the number of connected components of $G_0$ decreases.
Since initially $G_0$ has $n$ connected connected the lemma follows.
\end{proof}

\begin{theorem}\label{thm:app:fvs}
 There exists a Monte-Carlo algorithm that given a tree decomposition
 of width $t$ solves the \constrainedfvs problem in $3^t |V|^{O(1)}$ time.
 The algorithm cannot give false positives and may give false negatives with probability at most $1/2$.
\end{theorem}

\begin{proof}
Now we can use the Cut\&Count technique.
As the universe we take the set $\univ = V \times \{\forest,\marker\}$,
where $V \times \{\forest\}$ is used to assign weights to vertices
from the chosen forest and $V \times \{\marker\}$ for markers.
As usual we assume that we are given a weight function $\weight: \univ \rightarrow \{1,...,N\}$, where $N = 2|\univ| = 4|V|$.

\noindent {\bf{The Cut part.}}
For integers $\targetn, \targetm, \targetcc, \targetW$ we define:
  \begin{enumerate}
  \item $\cand^{\targetn, \targetm, \targetcc}_\targetW$ to be the family 
  of solution candidates: marked subgraphs excluding $\mustset$ of size and weight prescribed by super-/sub-scripts,
  i.e., $\cand^{\targetn, \targetm, \targetcc}_\targetW$ is the family of pairs $(X,M)$,
  where $X \subseteq V \setminus \mustset$, $|X| = \targetn$, $G[X]$ contains exactly $\targetm$ edges,
  $M \subseteq X$, $|M| = \targetcc$ and $\weight(X \times \{\forest\}) + \weight(M \times \{\marker\}) = \targetW$;
  \item $\sols^{\targetn, \targetm, \targetcc}_\targetW$ to be the set of solutions:
  the family of pairs $(X, M)$, where $(X,M) \in \cand^{\targetn, \targetm, \targetcc}_{\targetW}$ and $G[X]$ is a forest containing at least one marker from the set $M$ in each connected component;
  \item $\objs_{\targetW}^{\targetn, \targetm, \targetcc}$ 
  to be the family of pairs $((X, M), (X_1,X_2))$,
  where $(X,M) \in \cand^{\targetn, \targetm, \targetcc}_\targetW$, 
  $M \subseteq X_1$, and $(X_1,X_2)$ is a consistent cut of $G[X]$.
  \end{enumerate}
Observe that by Lemma~\ref{lem:fvssimple} the graph $G$ admits 
a feedback vertex set of size $k$ containing $\mustset$ if and only if there exist 
integers $\targetm$,$\targetW$ such that the set 
$\sols^{n-k,\targetm,n-k-\targetm}_\targetW$ is nonempty.

\noindent {\bf{The Count part.}} Similarly as in the case of \mincyclecovername
(analogously to Lemma~\ref{lem:dircyccanc}) note that for any 
$\targetn$,$\targetm$,$\targetcc$,$\targetW$,$(X,M) \in \cand^{\targetn, \targetm, \targetcc}_\targetW$,
there are $2^{\mathtt{cc}(M,G[X])}$ cuts $(X_1,X_2)$ such that
$((X,M),(X_1,X_2)) \in \objs_{\targetW}^{\targetn, \targetm, \targetcc}$,
where by $\mathtt{cc}(M,G[X])$ we denote the number of connected
components of $G[X]$ which do not contain any marker from the set $M$.
Hence by Lemma~\ref{lem:fvssimple}
for every $\targetn$,$\targetm$,$\targetcc$,$\targetW$
satisfying $\targetcc \le \targetn-\targetm$
we have
$|\sols^{\targetn, \targetm, \targetcc}_\targetW| \equiv
 |\objs_\targetW^{\targetn, \targetm, \targetcc}|$. 

Now we describe a procedure
$\countproc(\weight,\targetn,\targetm,\targetcc,\targetW,\treedecomp)$ that,
given a nice tree decomposition $\treedecomp$, weight function $\weight$
and integers $\targetn$,$\targetm$,$\targetcc$,$\targetW$,
computes $|\objs_\targetW^{\targetn, \targetm, \targetcc}|$
modulo $2$ using dynamic programming.

We follow the notation from the \steinertree{} example (see Lemma \ref{lem:stein}).
For every bag $x \in \treedecomp$ of the tree decomposition,
integers $0 \leq \itern \leq |V|$, $0 \leq \iterm < |V|$, $0 \leq \itercc \leq |V|$, 
         $0 \leq \iterW \leq 2\maxw|V|$
and $s \in \{\zero, \oneone, \onetwo\}^{\bag{x}}$ 
(called the colouring) define
\begin{align*}
  \cand_x(\itern, \iterm, \itercc, \iterW) & = \Big{\{} (X,M) \; \big{|}\; X \subseteq \subbags{x} \setminus \mustset \ \wedge\
    |X|=\itern \ \wedge\ |\subedges{x} \cap E(G[X])| = \iterm \\
        & \qquad\ \wedge\ M \subseteq X \setminus \bag{x}\ \wedge\ |M| = \itercc\ \wedge\ \weight(X \times \{\forest\})+\weight(M \times \{\marker\})=\iterW \Big{\}} \\
  \objs_x(\itern, \iterm, \itercc, \iterW) & = \Big{\{} ((X,M),(X_1,X_2)) \big{|}\; 
    (X,M) \in \cand_x(\itern, \iterm, \itercc, \iterW)\\
    & \qquad\ \wedge\ M \subseteq X_1 \wedge (X, (X_1,X_2)) \textrm{ is a consistently cut subgraph of }G_x \Big{\}} \\
    A_x(\itern, \iterm, \itercc, \iterW, s) & = \Big{|} \Big{\{} ((X,M),(X_1,X_2)) 
    \in \objs_x(\itern, \iterm, \itercc, \iterW) \big{|} \\
    &\qquad (s(v) = \onej \Rightarrow v \in X_j)\ \wedge \
   (s(v) = \zero \Rightarrow v \not\in X) \Big{\}} \Big{|}
\end{align*}
Note that we assume $\iterm < |V|$ because otherwise an induced subgraph
containing $\iterm$ edges is definitely not a forest.

Similarly as in the case of \steinertree, $s(v) = \zero$ means 
$v \notin X$, whereas $s(v) = \onej$ corresponds to $v \in X_j$.
The accumulators $\itern$,$\iterm$,$\itercc$ and $\iterW$ keep track of the number of 
vertices and edges in the subgraph induced by vertices from $X$,
number of markers already used and the sum of weights of chosen vertices and markers.
Hence $A_x(\itern,\iterm,\itercc,\iterW,s)$ is the number of pairs from $\objs_x(\itern,\iterm,\itercc,\iterW)$ 
with a fixed interface on vertices from $\bag{x}$.
Note that we ensure that no vertex from $B_x$ is yet marked, because we decide
whether to mark a vertex or not in its forget bag.
Recall that the tree decomposition is rooted in an empty bag
hence for every vertex there exists exactly one forget bag forgetting it.

The algorithm computes $A_x(\itern, \iterm, \itercc, \iterW,s)$ for all bags $x \in\treedecomp$ in a bottom-up fashion
for all reasonable values of $\itern$, $\iterm$, $\itercc$, $\iterW$ and $s$.
We now give the recurrence for $A_x(\itern, \iterm, \itercc, \iterW, s)$ that
is used by the dynamic programming
algorithm.
In order to simplify notation let $v$ the vertex introduced and
contained in an introduce bag, $uv$ the edge introduced in an
introduce edge bag,
and let $y,z$ stand for the left and right child of $x$ in $\treedecomp$ if present.
\begin{itemize}
\item \textbf{Leaf bag}:
\[ A_x(0,0,0,0,\emptyset) = 1 \]
\item \textbf{Introduce vertex bag}:
	\begin{align*}
     A_x(\itern, \iterm, \itercc, \iterW, s \cup \{(v, \zero)\}) & = A_y(\itern, \iterm, \itercc, \iterW, s)  \\
	   A_x(\itern, \iterm, \itercc, \iterW, s \cup \{(v, \onej)\}) & = [v \not \in \mustset] A_y(\itern-1,\iterm, \itercc, \iterW-\omega((v,\forest)), s)  
\end{align*}
\item \textbf{Introduce edge bag}:
\[
     A_x(\itern, \iterm, \itercc, \iterW, s)  = [s(u) = \zero \vee s(v) = \zero \vee s(u) = s(v)] A_y(\itern, \iterm-[s(u) = s(v) \not= \zero], \itercc, \iterW, s) 
\]
 Here we remove table entries not consistent with the edge $uv$,
 and update the accumulator $\iterm$ storing the number of edges
 in the induced subgraph.
\item \textbf{Forget bag}:
\[ A_x(\itern, \iterm, \itercc, \iterW, s) =A_x(\itern, \iterm, \itercc-1, \iterW-\weight((v,\marker)), s[ v \to \oneone] \}) + \sum_{\alpha \in \{\zero, \oneone, \onetwo\}} A_x(\itern, \iterm, \itercc, \iterW, s[v \to \alpha] \} )) \]
If the vertex $v$ was in $X_1$ then we can mark it and update the acumulator $\itercc$.
If we do not mark the vertex $v$ then it can have any of the three states with no additional requirements imposed.
\item \textbf{Join bag}:
\begin{align*}
A_x(\itern, \iterm, \itercc,\iterW, s) &= \sum_{\itern_1 + \itern_2 = \itern + |s^{-1}(\{\mathbf{1_1},\mathbf{1_2}\})|} \ \ \sum_{\iterm_1 + \iterm_2 = \iterm}\ \ \sum_{\itercc_1 + \itercc_2 = \itercc} \\
                      &\qquad \sum_{\iterW_1 + \iterW_2 = \iterW + \weight(s^{-1}(\{\mathbf{1_1},\mathbf{1_2}\}) \times \{\forest\} )}  A_y(\itern_1,\iterm_1,\itercc_1,\iterW_1,s) A_z(\itern_2,\iterm_2,\itercc_2,\iterW_2,s)
\end{align*}
The only valid combinations to achieve the colouring $s$ is to have the 
same colouring in both children. 
Since vertices coloured $\onej$ in $\bag{x}$ are accounted for in both 
tables of the children, we add their contribution to the 
accumulators $\itern$ and $\iterW$.
\end{itemize}

Since $|\objs_\targetW^{\targetn,\targetm,\targetcc}| = A_\rootv(\targetn, \targetm, \targetcc, \targetW, \emptyset)$
the above recurrence leads to a dynamic programming algorithm that
computes the parity of $|\objs_\targetW^{\targetn,\targetm,\targetcc}|$ for all reasonable values of $\targetW,\targetn,\targetm,\targetcc$
in $3^t |V|^{O(1)}$ time.
Consequently we finish the proof of Theorem~\ref{thm:app:fvs}.
\end{proof}

\subsection{Non-connectivity problems with an additional connectivity requirement}\label{sec:app:cx}

In this section we give details on algorithms for problems that are defined 
as standard ``local'' problems with an additional constraint that the 
solution needs to induce a connected subgraph. 
Problems described here are \cvertexcover{}, \cdomset{},
\coct{} and \cfvs{}, but the approach here can be easily carried over 
to similar problems.

Let us start with a short informal description. 
Solving a problem {\sc{Connected X}},
we simply run the easy and well-known algorithm for {\sc{X}} 
(or, in the case of \cfvs{}, we run the algorithm for \fvs{} from Section 
\ref{sec:app:fvs}), but we additionally keep a cut consistent with the 
solution, i.e., we count the number of solution-cut pairs.
Similarly as in the case of \steinertree{}, a solution
to the problem {\sc{X}} that induces $c$ connected components is consistent 
with $2^{c-1}$ cuts, thus all the disconnected solutions cancel out 
modulo $2$.

Similarly as in Section~\ref{sec:app:fvs} we solve more general
versions of problems where additionally as a part of the input
we are given a set $\mustset \subseteq V$ which contains vertices
that must belong to a solution.

\begin{remark}\label{rem:alg:guess-vertex}
  In the algorithms we assume that the set $\mustset \subseteq V$ is nonempty,
  so we can choose one fixed vertex $\anchor \in \mustset$
  that needs to be included in a fixed side of all considered cuts
  (cf. algorithm for \steinertree in Section \ref{sec:steinertree}).
  To solve the problem where $\mustset = \emptyset$, 
  we simply iterate over all possible choices of $\anchor \in V$
  and put $\mustset = \{\anchor\}$.
  Note that this does not increase the probability
  that the (Monte-Carlo) algorithm gives a wrong answer.
  Our algorithms can only give false negatives, 
	so in the case of a YES-instance
  we only need a single run, in which a solution can be found, to give a correct answer.
\end{remark}

Let us now proceed with the formal arguments.
For each problem, we start with a problem definition
and a formal statement of a result.

\subsubsection{\cvertexcover{}}

\defproblemu{\constrainedcvertexcover}
{An undirected graph $G = (V,E)$, a subset $\mustset \subseteq V$ and an integer $k$}
{Does there exist such a subset $X \subseteq V$ of cardinality $k$
  that $\mustset \subseteq X$, $G[X]$ is connected
  and each edge $e \in E$ is incident with at least one vertex from $X$?}

\begin{theorem}\label{thm:app:cvc}
 There exists a Monte-Carlo algorithm that given a tree decomposition of width $t$
 solves \constrainedcvertexcover{} in $3^t |V|^{O(1)}$ time.
 The algorithm cannot give false positives and may give false negatives with probability at most $1/2$.
\end{theorem}
There exists an easy proof of Theorem~\ref{thm:app:cvc} by a reduction 
to the \steinertree{} problem --- we subdivide all edges of the graph $G$
using terminals and add pendant terminals to $S$.
Such a transformation does not change the treewidth of the graph by more 
than one.
Nonetheless we prove the theorem by a direct application of the 
Cut\&Count technique,
in a similar manner as for the \steinertree{} problem in 
Section~\ref{sec:steinertree}.
Our motivation for choosing the second approach is that we
need it to develop an algorithm for \cvertexcover{} parameterized by the 
solution size in Appendix~\ref{sec:param-improv}
which relies on the algorithm we describe in the proof.
\begin{proof}

%
%
%
  We use the Cut\&Count technique.
  As the universe for Algorithm~\ref{alg:cutandcount} we take the vertex set $\univ = V$.
  Recall that we generate a random weight function 
	$\weight:\univ \to \{1,2,\ldots,\maxw\}$, 
	taking $\maxw = 2|\univ| = 2|V|$.
  By Remark \ref{rem:alg:guess-vertex} we may assume that $\mustset \neq \emptyset$
  and we may choose one fixed vertex $\anchor \in \mustset$.

  \noindent {\bf{The Cut part.}} For an integer $\targetW$ we define:
  \begin{enumerate}
  \item $\cand_\targetW$ to be the family of solution candidates
  of size $\targetk$ and weight $\targetW$:
  $\cand_\targetW$ is the family of sets $X \subseteq V$
  such that $\mustset \subseteq X$, $|X|=\targetk$,
  $\weight(X) = \targetW$ and $X$ is a vertex cover of $G$;
  \item $\sols_\targetW$ to be the family of solutions of size $\targetk$
  and weight $\targetW$, that is sets $X \in \cand_\targetW$
  such that $G[X]$ is connected;
  \item $\objs_\targetW$ to be the family of pairs $(X,(X_1,X_2))$,
  where $X \in \cand_\targetW$, $\anchor \in X_1$
  and $(X_1,X_2)$ is a consistent cut of $G[X]$.
  \end{enumerate}

  \noindent {\bf{The Count part.}} Similarly as in the case of \steinertree
	we note that
  by Lemma~\ref{lem:evencancel} for each $X \in \cand_\targetW$
  there exist $2^{\conncomp(G[X])-1}$
  consistent cuts of $G[X]$, 
	thus for any $\targetW$ we have
  $|\sols_\targetW| \equiv |\objs_\targetW|$.

  To finish the proof we need to describe a procedure
  $\countproc(\weight,\targetW,\treedecomp)$ that,
  given a nice tree decomposition $\treedecomp$, weight function $\weight$
  and an integer $\targetW$,
  computes $|\objs_\targetW|$ modulo $2$.

  As usual we use dynamic programming. 
	We follow the notation
  from the \steinertree{} example (see Lemma \ref{lem:stein}).
  For every bag $x \in \treedecomp$ of the tree decomposition,
  integers $0 \leq \iterk \leq |V|$, $0 \leq \iterW \leq \maxw|V|$
  and $s \in \{\zero, \oneone, \onetwo\}^{\bag{x}}$ 
	(called the colouring) define
  \begin{align*}
    \cand_x(\iterk,\iterW) &= \Big\{X \subseteq \subbags{x} \; \big{|}\; 
			(\mustset \cap \subbags{x}) \subseteq X\ \wedge\ |X|=\iterk\ \wedge\ \weight(X)=\iterW \
      \wedge\ X\textrm{ is a vertex cover of }G_x\Big\} \\
    \objs_x(\iterk,\iterW) &= \Big\{(X,(X_1,X_2)) \big{|}\; 
			X \in \cand_x(\iterk,\iterW)\ \wedge\ (X, (X_1,X_2)) 
			\textrm{ is a consistently cut subgraph of }G_x \\
        &\qquad\ \wedge\ (\anchor \in V_x \Rightarrow \anchor \in X_1) \Big\} \\
    A_x(\iterk,\iterW,s) &= \Big| \Big\{(X,(X_1,X_2)) 
			\in \objs_x(\iterk,\iterW) \big{|}\; 
			(s(v) = \onej \Rightarrow v \in X_j)\ \wedge \
			(s(v) = \zero \Rightarrow v \not\in X) \Big\} \Big|
  \end{align*}
  Similarly as in the case of \steinertree, $s(v) = \zero$ means 
	$v \notin X$, whereas $s(v) = \onej$ corresponds to $v \in X_j$.
  The accumulators $\iterk$ and $\iterW$ keep track of the number of 
	vertices in the solution
  and their weights, respectively.
  Hence $A_x(\iterk,\iterW,s)$ is the number of pairs from $\objs$ 
	of candidate solutions and consident cuts on $G_x$,
  with fixed size, weight and interface on vertices from $\bag{x}$.

  The algorithm computes $A_x(\iterk,\iterW,s)$ for all bags $x \in T$ in a bottom-up fashion
  for all reasonable values of $\iterk$, $\iterW$ and $s$.
  We now give the recurrence for $A_x(\iterk,\iterW,s)$ that is used by the dynamic programming
  algorithm.
  In order to simplify notation denote by $v$ the vertex introduced and
  contained in an introduce bag, by $uv$ the edge introduced in an
  introduce edge bag,
  and let $y,z$ be the left and right child of $x$ in $\treedecomp$ if present.
\begin{itemize}
\item \textbf{Leaf bag}:
	\[ A_x(0,0,\emptyset) = 1 \]
\item \textbf{Introduce vertex bag}:
	\begin{eqnarray*}
	 A_x(\iterk,\iterW,s[v \to \zero]) & =  &[v \not \in \mustset]A_y(\iterk,\iterW,s) \\
     A_x(\iterk,\iterW,s[v \to \oneone]) & = & A_y(\iterk-1,\iterW-\weight(v),s)  \\
     A_x(\iterk,\iterW,s[v \to \onetwo]) & = & [v \not= \anchor]A_y(\iterk-1,\iterW-\weight(v),s) 
	\end{eqnarray*}
  We take care of the restrictions imposed by the conditions $(\mustset \cap \subbags{x}) \subseteq X$ and $\anchor \in X_1$.
\item \textbf{Introduce edge bag}:
	\[ A_x(\iterk,\iterW,s) = [s(u) = s(v) \neq \mathbf{0}\ \vee\ (s(u) = \mathbf{0} \wedge s(v) \neq \zero)\ \vee\ (s(u) \neq \mathbf{0} \wedge s(v) = \mathbf{0})] A_y(\iterk,\iterW,s) \]
   Here we remove table entries not consistent with the edge $uv$, i.e., 
	 table entries where the endpoints are colored $\oneone$ and $\onetwo$ 
	 (thus creating an inconsistent cut) or $\zero$ and $\zero$ 
	 (thus leaving an edge that is not covered).
\item \textbf{Forget bag}:
	\[ A_x(\iterk,\iterW,s) = \sum_{\alpha \in \{\zero, \oneone, \onetwo\}} A_y(\iterk,\iterW,s[v \to \alpha]) \]
	In the child bag the vertex $v$ can have three states, and no additional
	requirements are imposed, so we sum over all the three states.
\item \textbf{Join bag}:
	\[ A_x(\iterk,\iterW,s) = \sum_{\iterk_1 + \iterk_2 = \iterk + |s^{-1}(\{\mathbf{1_1},\mathbf{1_2}\})|} \ \  \sum_{\iterW_1 + \iterW_2 = \iterW + \weight(s^{-1}(\{\mathbf{1_1},\mathbf{1_2}\}))} A_y(\iterk_1,\iterW_1,s) A_z(\iterk_2,\iterW_2,s) \]
	The only valid combination to achieve the colouring $s$ is to have the 
	same colouring in both children. 
	Since vertices coloured $\onej$ in $\bag{x}$ are accounted for in both 
	tables of the children, we add their contribution to the 
	accumulators.
\end{itemize}
It is easy to see that the above recurrence leads to a dynamic programming 
algorithm that
computes the parity of $|\sols_\targetW|$ for all values of $\targetW$ in $3^t |V|^{O(1)}$ time,
since $|\objs_\targetW|=A_\rootv(k,\targetW,\emptyset)$ and $|\sols_\targetW| \equiv |\objs_\targetW|$.
Moreover, as we count the parities and not the numbers $A_x$ themselves, 
all arithmetical operations can be done in constant time.
Thus, the proof of Theorem~\ref{thm:app:cvc} is finished. 
\end{proof}

\subsubsection{\cdomset{}}

\defproblemu{\constrainedcdomset}
{An undirected graph $G = (V,E)$, a subset $\mustset \subseteq V$ and an integer $k$.}
{Does there exist such a connected set $\mustset \subseteq X \subseteq V$ of cardinality at
most $k$ that $N[X] = V$?}

\begin{theorem}\label{thm:app:cdomset}
 There exists a Monte-Carlo algorithm that given a tree decomposition of width $t$
 solves \constrainedcdomset{} in $4^t |V|^{O(1)}$ time.
 The algorithm cannot give false positives and may give false negatives with probability at most $1/2$.
\end{theorem}

It is known that \cdomset{} is equivalent to \maxleaf{} \cite{maxleaf-cds},
hence the algorithm for \exactleaf{} can be used to solve \cdomset{}. 
However, here the Cut\&Count application is significantly easier
and more straightforward than in the \exactleaf{} algorithm presented later.
Thus we include the algorithm for \cdomset{} below.

\begin{proof}[Proof of Theorem \ref{thm:app:cdomset}]
  We use the Cut\&Count technique.
  As the universe for Algorithm~\ref{alg:cutandcount} we take the vertex set $\univ = V$.
  Recall that we generate a random weight function 
	$\weight:\univ \to \{1,2,\ldots,\maxw\}$, 
	taking $\maxw = 2|\univ| = 2|V|$.
  By Remark \ref{rem:alg:guess-vertex} we may assume that $\mustset \neq \emptyset$
  and we may choose one fixed vertex $\anchor \in \mustset$.

  \noindent {\bf{The Cut part.}} For an integer $\targetW$ we define:
  \begin{enumerate}
  \item $\cand_\targetW$ to be the family of solution candidates
  of size $\targetk$ and weight $\targetW$:
  $\cand_\targetW$ is the family of sets $X \subseteq V$
  such that $\mustset \subseteq X$, $|X|=\targetk$,
  $\weight(X) = \targetW$ and $N[X] = V$.
  \item $\sols_\targetW$ to be the family of solutions of size $\targetk$
  and weight $\targetW$, that is sets $X \in \cand_\targetW$
  such that $G[X]$ is connected;
  \item $\objs_\targetW$ to be the family of pairs $(X,(X_1,X_2))$,
  where $X \in \cand_\targetW$,
  and $(X_1,X_2)$ is a consistent cut of $G[X]$.
  \end{enumerate}

  \noindent {\bf{The Count part.}} As before
	we note that by Lemma~\ref{lem:evencancel} for each $X \in \cand_\targetW$
  there exist $2^{\conncomp(G[X])-1}$ consistent cuts of $G[X]$, 
	thus for any $\targetW$ we have
	$|\sols_\targetW| \equiv |\objs_\targetW|$. What remains is to describe 
	a procedure $\countproc(\weight,\targetW,\treedecomp)$ that,
  given a nice tree decomposition $\treedecomp$, weight function $\weight$
  and an integer $\targetW$,
  computes $|\objs_\targetW|$ modulo $2$.

  As usual we use dynamic programming. 
	We follow the notation
  from the \steinertree{} example (see Lemma \ref{lem:stein}).
  For every bag $x \in \treedecomp$ of the tree decomposition,
  integers $0 \leq \iterk \leq |V|$, $0 \leq \iterW \leq \maxw|V|$
  and $s \in \{\zeron, \zeroy, \oneone, \onetwo\}^{\bag{x}}$ 
	(called the colouring) define
  \begin{align*}
    \cand_x(\iterk,\iterW) &= \Big\{X \subseteq \subbags{x} \; \big{|}\; 
			(\mustset \cap \subbags{x}) \subseteq X\ \wedge\ |X|=\iterk\ \wedge\ \weight(X)=\iterW \
      \wedge\ N_{G_x}[X] = V_x \setminus s^{-1}(\zeron)  \Big\} \\
    \objs_x(\iterk,\iterW) &= \Big\{(X,(X_1,X_2)) \big{|}\; 
			X \in \cand_x(\iterk,\iterW)\ \wedge\ (X, (X_1,X_2)) 
			\textrm{ is a consistently cut subgraph of }G_x \\
      &\qquad\ \wedge\ (\anchor \in V_x \Rightarrow \anchor \in X_1) \Big\} \\
    A_x(\iterk,\iterW,s) &= \Big| \Big\{(X,(X_1,X_2)) 
			\in \objs_x(\iterk,\iterW) \big{|}\;
			(s(v) = \onej \Rightarrow v \in X_j) \\
      &\qquad\ \wedge \
			(s(v) = \zeroy \Rightarrow v \in N_{G_x}(X))\ \wedge\
      (s(v) = \zeron \Rightarrow v \not\in N_{G_x}[X] \Big\} \Big|
      \end{align*}
  Here $s(v) = \zeroy$ means $v \notin X$ and $v$ is dominated by $X$
  in $G_x$, $s(v) = \zeron$ means $v \notin X$ and $v$ is not dominated
  by $X$ in $G_x$,
	whereas $s(v) = \onej$ corresponds to $v \in X_j$.
  The accumulators $\iterk$ and $\iterW$ keep track of the number of 
	vertices in the solution
  and their weights, respectively.
  Hence $A_x(\iterk,\iterW,s)$ is the number of pairs from $\objs$ 
	of candidate solutions and consistent cuts on $G_x$,
  with fixed size, weight and interface on vertices from $\bag{x}$.

  The algorithm computes $A_x(\iterk,\iterW,s)$ for all bags $x \in T$ in a bottom-up fashion
  for all reasonable values of $\iterk$, $\iterW$ and $s$.
  We now give the recurrence for $A_x(\iterk,\iterW,s)$ that is used by the dynamic programming
  algorithm.
  As usual, $v$ denotes the vertex introduced and
  contained in an introduce bag, $uv$ the edge introduced in an
  introduce edge bag,
  while $y$ and $z$ denote the left and right child of $x$ in $\treedecomp$,
	if present.
\begin{itemize}
\item \textbf{Leaf bag}:
	\[ A_x(0,0,\emptyset) = 1 \]
\item \textbf{Introduce vertex bag}:
	\begin{align*}
	 A_x(\iterk,\iterW,s[v \to \zeron]) & =  [v \not \in \mustset]A_y(\iterk,\iterW,s) \\
	 A_x(\iterk,\iterW,s[v \to \zeroy]) & =   0\\
     A_x(\iterk,\iterW,s[v \to \oneone]) & =  A_y(\iterk-1,\iterW-\weight(v),s)  \\
     A_x(\iterk,\iterW,s[v\to \onetwo]) & =  [v \not= \anchor]A_y(\iterk-1,\iterW-\weight(v),s) 
	\end{align*}
  We take care of restrictions imposed by conditions $(\mustset \cap \subbags{x}) \subseteq X$ and $\anchor \in X_1$.
	Note that at the moment of introducing $v$ there are no edges incident to
	$v$ in $G_x$, thus $v$ cannot dominated, but not chosen.
\item \textbf{Introduce edge bag}:
  \begin{align*}
  A_x(\iterk,\iterW,s) &= A_y(\iterk,\iterW,s) & \textrm{if }s(v),s(u) \in \{\zeron,\zeroy\} \\
  A_x(\iterk,\iterW,s) &= [s(v)=s(u)]A_y(\iterk,\iterW,s) & \textrm{if }s(v),s(u) \in \{\oneone,\onetwo\} \\
  A_x(\iterk,\iterW,s) &= 0& \textrm{if } s(v)=\onej \wedge s(u) = \zeron \\
  A_x(\iterk,\iterW,s) &= 0& \textrm{if } s(v) = \zeron \wedge s(u) = \onej \\
  A_x(\iterk,\iterW,s) &= A_y(\iterk,\iterW,s) + A_y(\iterk,\iterW,s[u \to \zeron]) & \textrm{if }s(v) = \onej \wedge s(u) = \zeroy \\
  A_x(\iterk,\iterW,s) &= A_y(\iterk,\iterW,s) + A_y(\iterk,\iterW,s[v \to \zeron]) & \textrm{if }s(v) = \zeroy \wedge s(u) = \onej \\
  \end{align*}
  Here we perform two operations.
	First, we filter out entries creating an inconsistent cut,
  i.e., ones in which the endpoints are coloured $\oneone$ and $\onetwo$.
  Second, if one of the endpoints becomes
  dominated in $G_x$, its state could be changed from $\zeron$ to $\zeroy$
\item \textbf{Forget bag}:
	\[ A_x(\iterk,\iterW,s) = \sum_{\alpha \in \{\zeroy, \oneone, \onetwo\}} A_y(\iterk,\iterW,s[v \to \alpha]) \]
	In the child bag the vertex $v$ can have four states, but the state where $v$ is not dominated
  ($s(v) = \zeron$) is forbidden (we will have no more chances to dominate
	this vertex, but all vertices need to be dominated).
	Thus we sum over the three remaining states.
\item \textbf{Join bag}:
  For a colouring $s \in \{\zeron, \zeroy, \oneone, \onetwo\}^{\bag{x}}$
  we define its precolouring $\hat{s} \in \{\zero, \oneone, \onetwo\}^{\bag{x}}$
  as
  \begin{align*}
  \hat{s}(v) &= s(v) &\textrm{if }s(v) \in \{\oneone, \onetwo\} \\
  \hat{s}(v) &= \zero &\textrm{if }s(v) \in \{\zeroy, \zeron\}
  \end{align*}
   For a precolouring $\hat{s}$ (or a colouring $s$) and set
  $T \subseteq \hat{s}^{-1}(\zero)$ we define a colouring $s[T]$ as 
  \begin{align*}
  s[T](v) &= \hat{s}(v)  & \textrm{if }\hat{s}(v) \in \{\oneone, \onetwo\} \\
  s[T](v) &= \zeroy & \textrm{if }v \in T \\
  s[T](v) &= \zeron & \textrm{if }v \in \hat{s}^{-1}(\zero) \setminus T
  \end{align*}
  We can now write a recursion formula for join bags.
  \begin{align*}
	A_x(\iterk,\iterW,s) &= \sum_{\iterk_1 + \iterk_2 = \iterk + |s^{-1}(\{\mathbf{1_1},\mathbf{1_2}\})|} \ \  \sum_{\iterW_1 + \iterW_2 = \iterW + \weight(s^{-1}(\{\mathbf{1_1},\mathbf{1_2}\}))} \\
              &\qquad \sum_{T_1,T_2 \subseteq s^{-1}(\{\zeron,\zeroy\})}
                [T_1 \cup T_2 = s^{-1}(\zeroy)]
                A_y(\iterk_1,\iterW_1,s[T_1]) A_z(\iterk_2,\iterW_2,s[T_2])
  \end{align*}
  To achieve the colouring $s$, the precolourings of children have to be the same.
  Moreover, the sets of vertices coloured $\zeroy$ in children have
  to sum up to $s^{-1}(\zeroy)$.
	Since vertices coloured $\onej$ in $\bag{x}$ are accounted for both 
	tables of the children, we add their contribution to the 
	accumulators.

  To compute the recursion formula
  efficiently we need to use the fast evaluation of the covering product.
  For accumulators $\iterk,\iterW$ and a precolouring $\hat{s}$
  we define the following functions on subsets of $\hat{s}^{-1}(\zero)$:
  \begin{align*}
  f^{\iterk,\iterW,\hat{s}}(T) &= A_y(\iterk,\iterW,s[T]), \\
  g^{\iterk,\iterW,\hat{s}}(T) &= A_z(\iterk,\iterW,s[T]).
  \end{align*}
  Now note that
	$$A_x(\iterk,\iterW,s) = \sum_{\iterk_1 + \iterk_2 = \iterk + |s^{-1}(\{\mathbf{1_1},\mathbf{1_2}\})|} \ \  \sum_{\iterW_1 + \iterW_2 = \iterW + \weight(s^{-1}(\{\mathbf{1_1},\mathbf{1_2}\}))}
         (f^{\iterk_1,\iterW_1,\hat{s}} \fsccov g^{\iterk_2,\iterW_2,\hat{s}})(s^{-1}(\zeroy)).$$
  By Theorem \ref{thm:fsc},
  for fixed accumulators $\iterk_1,\iterW_1,\iterk_2,\iterW_2$ and a precolouring $\hat{s}$
  the term
     $$(f^{\iterk_1,\iterW_1,\hat{s}} \fsccov g^{\iterk_2,\iterW_2,\hat{s}})(s^{-1}(\zeroy))$$
   can be computed in time $2^{|\hat{s}^{-1}(\zero)|} |\hat{s}^{-1}(\zero)|^{O(1)}$
   at once for all colourings $s$ with precolouring $\hat{s}$.
   Thus, the total time consumed by the evaluation of $A_x$ is bounded by
   $$|V|^{O(1)} \sum_{\hat{s} \in \{\zero,\oneone,\onetwo\}^{\bag{x}}} 2^{|\hat{s}^{-1}(\zero)|}
    = 4^{|\bag{x}|} |V|^{O(1)}.$$
\end{itemize}
It is easy to see that the above recurrence leads to a dynamic programming 
algorithm that
computes the parity of $|\sols_\targetW|$ for all values of $\targetW$ in $4^t |V|^{O(1)}$ time,
since $|\objs_\targetW|=A_\rootv(k,\targetW,\emptyset)$ and $|\sols_\targetW| \equiv |\objs_\targetW|$.
Moreover, as we count the parities and not the numbers $A_x$ themselves, 
all arithmetical operations (in particular, the ring operations in the convolutions
    used in join bags) can be done in constant time.
Thus, the proof of Theorem~\ref{thm:app:cdomset} is finished. 
\end{proof}

\subsubsection{\coct{}}

\defproblemu{\constrainedcoct}
{An undirected graph $G = (V,E)$, a subset $\mustset \subseteq V$ and an integer $k$}
{Does there exist a subset $X \subseteq V$ of cardinality $k$,
  such that $\mustset \subseteq X$, $G[X]$ is connected
  and $G[V \setminus X]$ is bipartite?}

\newcommand{\weightX}{\mathbf{X}}
\newcommand{\weightL}{\mathbf{L}}

\begin{theorem}\label{thm:app:coct}
 There exists a Monte-Carlo algorithm that given a tree decomposition of width $t$
 solves \constrainedcoct{} in $4^t |V|^{O(1)}$ time.
 The algorithm cannot give false positives and may give false negatives with probability at most $1/2$.
\end{theorem}

\begin{proof}
  We use the Cut\&Count technique.
  As the universe for Algorithm~\ref{alg:cutandcount} we take
  $\univ = V \times \{\weightX,\weightL\}$, i.e., for each vertex $v\in V$
  we generate two weights $\weight((v,\weightX))$ and $\weight((v,\weightL))$.
  Recall that we generate a random weight function 
	$\weight:\univ \to \{1,2,\ldots,\maxw\}$, 
	taking $\maxw = 2|\univ| = 4|V|$.
  By Remark \ref{rem:alg:guess-vertex} we may assume that $\mustset \neq \emptyset$
  and we may choose one fixed vertex $\anchor \in \mustset$.

  \noindent {\bf{The Cut part.}} To make use of the well-known algorithm for \oct{} parameterized
  by treewidth, we need to define a solution not only as a set $X$,
  but we need to add a proof that $G[V \setminus X]$ is bipartite (i.e.,
  a partition of $V \setminus X$ into two independent sets).
  Formally, for an integer $\targetW$ we define:
  \begin{enumerate}
  \item $\cand_\targetW$ to be the family of pairs $(X,L)$,
  where $|X|=\targetk$, $\weight(X \times \{\weightX\} \cup L \times\{\weightL\}) = \targetW$,
  $S \subseteq X$, $X \cap L = \emptyset$
  and $L$ and $V \setminus (X \cup L)$ are independent sets in $G$;
  \item $\sols_\targetW$ to be the family of pairs $(X,L) \in \cand_\targetW$
  such that $G[X]$ is connected;
  \item $\objs_\targetW$ to be the family of pairs $((X,L),(X_1,X_2))$,
  where $(X,L) \in \cand_\targetW$, $\anchor \in X_1$
  and $(X_1,X_2)$ is a consistent cut of $G[X]$.
  \end{enumerate}
  Note that for a single set $X \subseteq V$ there may exist many proofs $L$
  that $G[V \setminus X]$ is bipartite. We consider all pairs $(X,L)$
  as {\em{different}} solutions and solution candidates. To compute weight, each pair $(X,L)$
  is represented as $X \times \{\weightX\} \cup L \times \{\weightL\} \subseteq \univ$, thus
  each pair $(X,L)$ corresponds to a different subset of the weight domain $\univ$.

  \noindent {\bf{The Count part.}} Similarly as in the case of \steinertree
	we note that
  by Lemma~\ref{lem:evencancel} for each $(X,L) \in \cand_\targetW$
  there exist $2^{\conncomp(G[X])-1}$
  consistent cuts of $G[X]$, 
	thus for any $\targetW$ we have
  $|\sols_\targetW| \equiv |\objs_\targetW|$.

  To finish the proof we need to describe a procedure
  $\countproc(\weight,\targetW,\treedecomp)$ that,
  given a nice tree decomposition $\treedecomp$, weight function $\weight$
  and an integer $\targetW$,
  computes $|\objs_\targetW|$ modulo $2$.

  As usual we use dynamic programming. 
	We follow the notation
  from the \steinertree{} example (see Lemma \ref{lem:stein}).
  For every bag $x \in \treedecomp$ of the tree decomposition,
  integers $0 \leq \iterk \leq |V|$, $0 \leq \iterW \leq \maxw|V|$
  and $s \in \{\zerol, \zeror, \oneone, \onetwo\}^{\bag{x}}$ 
	(called the colouring) define
  \begin{align*}
    \cand_x(\iterk,\iterW) &= \Big{\{}(X,L)\; \big{|}\; 
       X,L \subseteq \subbags{x}\ \wedge\ X \cap L = \emptyset\ \wedge\
			(\mustset \cap \subbags{x}) \subseteq X\ \wedge\ 
      |X|=\iterk \\
      &\qquad\ \wedge\ \weight(X \times \{\weightX\} \cup L \times \{\weightL\})=\iterW 
      \ \wedge\ L\textrm{ and }V_x \setminus (X \cup L)\textrm{ are independent sets in }G_x\Big{\}} \\
    \objs_x(\iterk,\iterW) &= \Big\{((X,L),(X_1,X_2)) \big{|}\; 
      (X,L) \in \cand_x(\iterk,\iterW)\ \wedge\ (X, (X_1,X_2)) 
			\textrm{ is a consistently cut subgraph of }G_x\\
      &\qquad\wedge\ (\anchor \in V_x \Rightarrow \anchor \in X_1) \Big\} \\
    A_x(\iterk,\iterW,s) &= \Big| \Big\{((X,L),(X_1,X_2)) 
			\in \objs_x(\iterk,\iterW) \big{|}\; 
			(s(v) = \onej \Rightarrow v \in X_j) \\
        &\qquad\ \wedge\ 
			(s(v) = \zerol \Rightarrow v \in L)\ \wedge\
      (s(v) = \zeror \Rightarrow v \not\in X \cup L)\Big\} \Big|
  \end{align*}
  Here we plan $L$ and $V\setminus (X\cup L)$ to be a bipartition of
	$G[V\setminus X]$; $s(v) = \zerol$ and $\zeror$ mean $v$ is on the left
	or right side of this bipartition, respectively, while $s(v) = \onej$
	means $v$ is in the odd cycle transversal, and on the appropriate side
	of the cut.
  The accumulators $\iterk$ and $\iterW$ keep track of the number of 
	vertices in $X$
  and the weight of the pair $(X,L)$, respectively.
  Hence $A_x(\iterk,\iterW,s)$ is the number of pairs from $\objs$ 
	of candidate solutions and consistent cuts on $G_x$,
  with fixed size, weight and interface on vertices from $\bag{x}$.

  The algorithm computes $A_x(\iterk,\iterW,s)$ for all bags $x \in T$ in a bottom-up fashion
  for all reasonable values of $\iterk$, $\iterW$ and $s$.
  We now give the recurrence for $A_x(\iterk,\iterW,s)$ that is used by the dynamic programming
  algorithm.
  As always let $v$ stand for the vertex introduced and
  contained in an introduce bag, $uv$ for the edge introduced in an
  introduce edge bag,
  and $y,z$ for the left and right child of $x$ in $\treedecomp$ if present.
\begin{itemize}
\item \textbf{Leaf bag}:
	\[ A_x(0,0,\emptyset) = 1 \]
\item \textbf{Introduce vertex bag}:
	\begin{align*}
	 A_x(\iterk,\iterW,s[v \to \zerol]) & = [v \not \in \mustset]A_y(\iterk,\iterW-\weight((v,\weightL)),s) \\
	 A_x(\iterk,\iterW,s[v\to\zeror]) & = [v \not \in \mustset]A_y(\iterk,\iterW,s) \\
     A_x(\iterk,\iterW,s]v\to\oneone]) & = A_y(\iterk-1,\iterW-\weight((v,\weightX)),s)  \\
     A_x(\iterk,\iterW,s[v\to\onetwo]) & = [v \not= \anchor]A_y(\iterk-1,\iterW-\weight((v,\weightX)),s) 
	\end{align*}
  We take care of restrictions imposed by conditions $(\mustset \cap \subbags{x}) \subseteq X$ and $\anchor \in X_1$.
\item \textbf{Introduce edge bag}:
  \begin{align*}
  A_x(\iterk,\iterW,s) &= 0 &\textrm{if }\{s(u),s(v)\} = \{\oneone,\onetwo\}\\
  A_x(\iterk,\iterW,s) &= 0 &\textrm{if }s(u) = s(v) \in \{\zerol,\zeror\}\\
  A_x(\iterk,\iterW,s) &= A_y(\iterk,\iterW,s) &\textrm{otherwise}
  \end{align*}
   Here we remove table entries not consistent with the edge $uv$, i.e., 
	 table entries where the endpoints are coloured $\oneone$ and $\onetwo$ 
	 (thus creating an inconsistent cut) or both coloured $\zerol$ or both
   coloured $\zeror$ (thus introducing an edge in $G_x[L]$ or $G_x[V_x \setminus (X \cup L)]$).
\item \textbf{Forget bag}:
	\[ A_x(\iterk,\iterW,s) = \sum_{\alpha \in \{\zerol, \zeror, \oneone, \onetwo\}} A_y(\iterk,\iterW,s[v \to\alpha]) \]
	In the child bag the vertex $v$ can have four states, and no additional
	requirements are imposed, so we sum over all the four states.
\item \textbf{Join bag}:
	\[ A_x(\iterk,\iterW,s) = \sum_{\iterk_1 + \iterk_2 = \iterk + |s^{-1}(\{\mathbf{1_1},\mathbf{1_2}\})|} \ \  \sum_{\iterW_1 + \iterW_2 = \iterW + \weight(s^{-1}(\{\mathbf{1_1},\mathbf{1_2}\}) \times \{\weightX\} \cup s^{-1}(\zerol) \times \{\weightL\})} A_y(\iterk_1,\iterW_1,s) A_z(\iterk_2,\iterW_2,s) \]
	The only valid combinations to achieve the colouring $s$ is to have the 
	same colouring in both children. 
	Since vertices coloured $\onej$ and $\zerol$ in $\bag{x}$ are accounted for in both 
	tables of the children, we add their contribution to the 
	accumulators.
\end{itemize}
It is easy to see that the above recurrence leads to a dynamic programming 
algorithm that
computes the parity of $|\sols_\targetW|$ for all values of $\targetW$ in $4^t |V|^{O(1)}$ time,
since $|\objs_\targetW|=A_\rootv(k,\targetW,\emptyset)$ and $|\sols_\targetW| \equiv |\objs_\targetW|$.
Moreover, as we count the parities and not the numbers $A_x$ themselves, 
all arithmetical operations can be done in constant time.
Thus, the proof of Theorem~\ref{thm:app:coct} is finished. 
\end{proof}

\subsubsection{\cfvs{}}

\defproblemu{\constrainedcfvs}
{An undirected graph $G = (V,E)$, a subset $\mustset \subseteq V$ and an integer $k$.}
{Does there exist a set $Y \subset V$ of cardinality $k$ such
  that $\mustset \subseteq Y$, $G[Y]$ is connected and $G[V\setminus Y]$ is a forest?}

\begin{theorem}\label{thm:app:cfvs}
 There exists a Monte-Carlo algorithm that given a tree decomposition
 of width $t$ solves the \constrainedcfvs problem in $4^t |V|^{O(1)}$ time.
 The algorithm cannot give false positives and may give false negatives with probability at most $1/2$.
\end{theorem}

\begin{proof}
We use the Cut\&Count technique.
The idea is as in the previous algorithms in this subsection: we use the dynamic programming
for \fvs, additionally keeping a cut consistent with the solution $Y$. However, in the previous
subsections the base dynamic programming algorithms were the easy, naive ones. Here we need to
use the Cut\&Count based algorithm from Section \ref{sec:app:fvs}. Thus, we attach to a solution
candidate two cuts: one of $G[Y]$, and second of $G[V \setminus Y]$.

As a universe we take the set $\univ = V \times \{\forest,\marker\}$,
where $V \times \{\forest\}$ is used to assign weights to vertices
from the chosen forest $G[V \setminus Y]$ and $V \times \{\marker\}$ for markers.
As usual we assume that we are given a weight function $\weight: \univ \rightarrow \{1,...,N\}$, where $N = 2|\univ| = 4|V|$.
By Remark \ref{rem:alg:guess-vertex} we assume $\mustset \neq \emptyset$ and we fix
one vertex $\anchor \in \mustset$.

\noindent {\bf{The Cut part.}}
For integers $\targetn, \targetm, \targetcc, \targetW$ we define:
  \begin{enumerate}
  \item $\cand^{\targetn, \targetm, \targetcc}_\targetW$ to be the family 
  of solution candidates, that is a marked subgraphs excluding $\mustset$ of size and weight prescribed by super-/sub-scripts,
  i.e., $\cand^{\targetn, \targetm, \targetcc}_\targetW$ is the family of pairs $(X,M)$,
  where $X \subseteq V \setminus \mustset$, $|X| = \targetn$, $G[X]$ contains exactly $\targetm$ edges,
  $M \subseteq X$, $|M| = \targetcc$ and $\weight(X \times \{\forest\}) + \weight(M \times \{\marker\}) = \targetW$;
  \item $\sols^{\targetn, \targetm, \targetcc}_\targetW$ to be the set of solutions,
  that is the family of pairs $(X, M)$, where $(X,M) \in \cand^{\targetn, \targetm, \targetcc}_{\targetW}$,
  where $G[X]$ is a forest containing at least one marker from the set $M$ in each connected component
  and $G[V \setminus X]$ is connected;
  \item $\objs_{\targetW}^{\targetn, \targetm, \targetcc}$ 
  to be the family of triples $((X, M), (X_1,X_2), (Y_1,Y_2))$,
  where $(X,M) \in \cand^{\targetn, \targetm, \targetcc}_\targetW$, 
  $M \subseteq X_1$, $(X_1,X_2)$ is a consistent cut of $G[X]$, $\anchor \in Y_1$
  and $(Y_1,Y_2)$ is a consistent cut of $G[V \setminus X]$.
  \end{enumerate}
Observe that by Lemma~\ref{lem:fvssimple} the graph $G$ admits 
a connected feedback vertex set of size $k$ containing $\mustset$ if and only if there exist 
integers $\targetm$,$\targetW$ such that the set 
$\sols^{n-k,\targetm,n-k-\targetm}_\targetW$ is nonempty.

\noindent {\bf{The Count part.}}
Similarly as in the case of \steinertree we note that
by Lemma~\ref{lem:evencancel} for any
$\targetn$,$\targetm$,$\targetcc$,$\targetW$,$(X,M) \in \cand^{\targetn, \targetm, \targetcc}_\targetW$,
  there exist $2^{\conncomp(G[V \setminus X])-1}$
  cuts $(Y_1,Y_2)$ that are consistent cuts of $G[V \setminus X]$ and $\anchor \in Y_1$.
Moreover, similarly as in the case of \mincyclecovername
(analogously to Lemma~\ref{lem:dircyccanc}) note that
there are $2^{\mathtt{cc}(M,G[X])}$ cuts $(X_1,X_2)$ that are consistent with $G[X]$
and $M \subseteq X_1$,
where by $\mathtt{cc}(M,G[X])$ we denote the number of connected
components of $G[X]$ which do not contain any marker from the set $M$.
Thus for any
$\targetn$,$\targetm$,$\targetcc$,$\targetW$,$(X,M) \in \cand^{\targetn, \targetm, \targetcc}_\targetW$,
there are $2^{\conncomp(G[V \setminus X])-1 + \mathtt{cc}(M,G[X])}$
triples
$((X,M),(X_1,X_2),(Y_1,Y_2)) \in \objs_{\targetW}^{\targetn, \targetm, \targetcc}$.
Hence by Lemma~\ref{lem:fvssimple}
for every $\targetn$,$\targetm$,$\targetcc$,$\targetW$
satisfying $\targetcc \le \targetn-\targetm$
we have
$|\sols^{\targetn, \targetm, \targetcc}_\targetW| \equiv
 |\objs_\targetW^{\targetn, \targetm, \targetcc}|$. 

Now we describe a procedure
$\countproc(\weight,\targetn,\targetm,\targetcc,\targetW,\treedecomp)$ that,
given a nice tree decomposition $\treedecomp$, weight function $\weight$
and integers $\targetn$,$\targetm$,$\targetcc$,$\targetW$,
computes $|\objs_\targetW^{\targetn, \targetm, \targetcc}|$
modulo $2$ using dynamic programming.

We follow the notation from the \steinertree{} example (see Lemma \ref{lem:stein}).
For every bag $x \in \treedecomp$ of the tree decomposition,
integers $0 \leq \itern \leq |V|$, $0 \leq \iterm < |V|$, $0 \leq \itercc \leq |V|$, 
         $0 \leq \iterW \leq 2\maxw|V|$
and $s \in \{\zeroone, \zerotwo, \oneone, \onetwo\}^{\bag{x}}$ 
(called the colouring) define
\begin{align*}
  \cand_x(\itern, \iterm, \itercc, \iterW) & = \Big{\{} (X,M) \; \big{|}\; X \subseteq \subbags{x} \setminus \mustset\ \wedge\
    |X|=\itern\ \wedge\ |\subedges{x} \cap E(G[X])| = \iterm \\
        & \qquad \wedge\ M \subseteq X \setminus \bag{x}\ \wedge\ |M| = \itercc\ \wedge\ \weight(X \times \{\forest\})+\weight(M \times \{\marker\})=\iterW \Big{\}} \\
  \objs_x(\itern, \iterm, \itercc, \iterW) & = \Big{\{} ((X,M),(X_1,X_2),(Y_1,Y_2)) \big{|}\; 
    (X,M) \in \cand_x(\itern, \iterm, \itercc, \iterW) \\
    & \qquad \wedge\ M \subseteq X_1\ \wedge\ (X, (X_1,X_2)) \textrm{ is a consistently cut subgraph of }G_x \\
    & \qquad \wedge\ (\anchor \in \subbags{x} \Rightarrow \anchor \in Y_1)\ \wedge\ (\subbags{x} \setminus X, (Y_1,Y_2)) \textrm{ is a consistently cut subgraph of }G_x\Big{\}} \\
    A_x(\itern, \iterm, \itercc, \iterW, s) & = \Big{|} \Big{\{} ((X,M),(X_1,X_2),(Y_1,Y_2)) 
    \in \objs_x(\itern, \iterm, \itercc, \iterW) \big{|} \\ 
    & \qquad (s(v) = \onej \Rightarrow v \in X_j)\ \wedge\ 
   (s(v) = \zeroj v \in Y_j) \Big{\}} \Big{|}
\end{align*}
Note that we assume $\iterm < |V|$ because otherwise an induced subgraph
containing $\iterm$ edges is definitely not a forest.

Similarly as in the case of \steinertree, $s(v) = \zeroj$ means 
$v \in Y_j$, whereas $s(v) = \onej$ corresponds to $v \in X_j$.
The accumulators $\itern$,$\iterm$,$\itercc$ and $\iterW$ keep track of the number of 
vertices and edges in the subgraph induced by vertices from $X$,
number of markers already used and the sum of weights of chosen vertices and markers.
Hence $A_x(\itern,\iterm,\itercc,\iterW,s)$ is the number of triples from $\objs_x(\itern,\iterm,\itercc,\iterW)$ 
with a fixed interface on vertices from $\bag{x}$.
Note that we ensure that no vertex from $B_x$ is yet marked, because we decide
whether to mark a vertex or not in its forget bag.

The algorithm computes $A_x(\itern, \iterm, \itercc, \iterW,s)$ for all bags $x \in\treedecomp$ in a bottom-up fashion
for all reasonable values of $\itern$, $\iterm$, $\itercc$, $\iterW$ and $s$.
We now give the recurrence for $A_x(\itern, \iterm, \itercc, \iterW, s)$ that
is used by the dynamic programming
algorithm.
As in the previous sections by $v$ we denote the vertex introduced and
contained in an introduce bag, by $uv$ the edge introduced in an
introduce edge bag,
and by $y,z$ for the left and right child of $x$ in $\treedecomp$ if present.
\begin{itemize}
\item \textbf{Leaf bag}:
\[ A_x(0,0,0,0,\emptyset) = 1 \]
\item \textbf{Introduce vertex bag}:
  \begin{align*}
   A_x(\itern, \iterm, \itercc, \iterW, s[v\to\zeroone]) & = A_y(\itern, \iterm, \itercc, \iterW, s)  \\
   A_x(\itern, \iterm, \itercc, \iterW, s[v\to\zerotwo]) & = [v \neq \anchor]A_y(\itern, \iterm, \itercc, \iterW, s)  \\
   A_x(\itern, \iterm, \itercc, \iterW, s[v\to\onej]) & = [v \not \in \mustset] A_y(\itern-1,\iterm, \itercc, \iterW-\omega((v,\forest)), s)  
  \end{align*}
  Here we take care of the constraints $S \cap X = \emptyset$ and $\anchor \in Y_1$.
\item \textbf{Introduce edge bag}:
  \begin{align*}
    A_x(\itern, \iterm, \itercc, \iterW, s)  &= 0 &\textrm{if }\{s(v),s(u)\} = \{\zeroone,\zerotwo\} \\
    A_x(\itern, \iterm, \itercc, \iterW, s)  &= 0 &\textrm{if }\{s(v),s(u)\} = \{\oneone,\onetwo\} \\
    A_x(\itern, \iterm, \itercc, \iterW, s)  &= A_y(\itern, \iterm-1, \itercc, \iterW, s) &\textrm{if }s(v)=s(u) \in \{\oneone,\onetwo\} \\
    A_x(\itern, \iterm, \itercc, \iterW, s)  &= A_y(\itern, \iterm, \itercc, \iterW, s) &\textrm{otherwise} \\
  \end{align*}
 Here we remove table entries not consistent with the edge $uv$ (i.e., creating an inconsistent cut, either $(X_1,X_2)$ or $(Y_1,Y_2)$),
 and update the accumulator $\iterm$ storing the number of edges in $G[X]$.
\item \textbf{Forget bag}:
\[ A_x(\itern, \iterm, \itercc, \iterW, s) =A_y(\itern, \iterm, \itercc-1, \iterW-\weight((v,\marker)), s[v \to \oneone]) + \sum_{\alpha \in \{\zeroone, \zerotwo, \oneone, \onetwo\}} A_y(\itern, \iterm, \itercc, \iterW, s[v \to \alpha])) \]
If the vertex $v$ was in $X_1$ then we can mark it and update the accumulator $\itercc$.
If we do not mark the vertex $v$ then it can have any of the four states with no additional requirements imposed.
\item \textbf{Join bag}:
\begin{align*}
A_x(\iterk,\iterW,s) &= \sum_{\itern_1 + \itern_2 = \itern + |s^{-1}(\{\mathbf{1_1},\mathbf{1_2}\})|} \ \ \sum_{\iterm_1 + \iterm_2 = \iterm}\ \ \sum_{\itercc_1 + \itercc_2 = \itercc} \\
                      &\qquad \sum_{\iterW_1 + \iterW_2 = \iterW + \weight(s^{-1}(\{\mathbf{1_1},\mathbf{1_2}\}) \times \{\forest\} )}  A_y(\itern_1,\iterm_1,\itercc_1,\iterW_1,s) A_z(\itern_2,\iterm_2,\itercc_2,\iterW_2,s)
\end{align*}
The only valid combinations to achieve the colouring $s$ is to have the 
same colouring in both children. 
Since vertices coloured $\onej$ in $\bag{x}$ are accounted for in both 
tables of the children, we add their contribution to the 
accumulators $\itern$ and $\iterW$.
\end{itemize}

Since $|\objs_\targetW^{\targetn,\targetm,\targetcc}| = A_\rootv(\targetn, \targetm, \targetcc, \targetW, \emptyset)$
the above recurrence leads to a dynamic programming algorithm that
computes the parity of $|\objs_\targetW^{\targetn,\targetm,\targetcc}|$ for all reasonable values of $\targetW,\targetn,\targetm,\targetcc$
in $4^t n^{O(1)}$ time.
Consequently we finish the proof of Theorem~\ref{thm:app:cfvs}.
\end{proof}

\subsection{Longest Cycles, Paths and Cycle Covers}
\label{sec:dirmincyc}
\label{det:cirlongestcyc}

In this section we consider the following three problems, both in the directed and undirected setting.

\defproblemu{{\sc{(Directed)}} \mincyclecovername}
{An undirected graph $G=(V,E)$ (or a directed graph $D = (V,A)$) and an integer $k$.}
{Can the vertices of $G$ ($D$) be covered with at most $k$ vertex disjoint (directed) cycles?}

\defproblemu{{\sc{(Directed)}} \longestcycle}
{An undirected graph $G=(V,E)$ (or a directed graph $D = (V,A)$) and an integer $k$.}
{Does there exist a (directed) simple cycle of length $k$ in $G$ ($D$)?}

\defproblemu{{\sc{(Directed)}} \longestpath}
{An undirected graph $G=(V,E)$ (or a directed graph $D = (V,A)$) and an integer $k$.}
{Does there exist a (directed) simple path of length $k$ in $G$ ($D$)?}

We capture all three problems in the following artificial one.

\defproblemu{{\sc{(Directed)}} \partialcycles}
{An undirected graph $G=(V,E)$ (or a directed graph $D = (V,A)$) and integers $k$ and $\ell$.}
{Does there exist a family of at most $k$ vertex disjoint (directed) cycles in $G$ ($D$)
  that cover exactly $\ell$ vertices?}

Note that for $k=1$ the above problem becomes \longestcycle, whereas for $\ell = |V|$
it becomes \mincyclecovername. The \longestpath problem can be easily reduced
to \longestcycle, both in the directed and undirected setting.
  Given {\sc{(Directed)}} \longestpath instance $(G,k)$ ($(D,k)$),
   we guess the endpoints $s$ and $t$ of the path in question,
   attach to the graph path of length $|V|+1$ from $t$ to $s$
   and ask for a cycle of length $|V|+1+k$.
   Moreover, given a tree decomposition $\treedecomp$ of $G$ ($D$),
   a tree decomposition for the modified graph can be easily constructed
   by adding $s$ and $t$ to every bag and by covering the attached path
   by a sequence of additional bags of size $3$. The width of the new decomposition
   is larger by a constant than the width of $\treedecomp$.

We now show how to solve \partialcycles using the Cut\&Count technique, in time
$4^t |V|^{O(1)}$ in the undirected case and in time $6^t |V|^{O(1)}$ in the
directed case.

\subsubsection{The undirected case}

\begin{theorem}\label{thm:app:undircycles}
 There exists a Monte-Carlo algorithm that given a tree decomposition of width $t$
 solves \partialcycles{} in $4^t |V|^{O(1)}$ time.
 The algorithm cannot give false positives and may give false negatives with probability at most $1/2$.
\end{theorem}
\begin{proof}
We use the Cut\&Count technique. To count the number of cycles
we use markers. 
However, in this application it is more convenient to take as markers
edges instead of vertices.
The objects we count are subsets of edges,
together with sets of marked edges, thus
we take $\univ = E \times \{\weightX, \marker\}$.
As usual, we assume we are given a weight function $\weight: \univ \to \{1,2,\ldots,N\}$,
where $N = 2|\univ| = 4|E|$.
We also assume $\targetk \leq \ell$.

\noindent {\bf{The Cut part.}} For an integer $\targetW$ we define:
\begin{enumerate}
\item $\cand_\targetW$ to be the family of pairs $(X,M)$, where
$M \subseteq X \subseteq E$, $|X| = \ell$, $|M| = \targetk$,
  $\weight(X \times \{\weightX\} \cup M \times \{\marker\}) = \targetW$ and each vertex $v \in V(X)$
  has degree $2$ in $G[X]$.
\item $\sols_\targetW$ to be the family of pairs $(X,M) \in \cand_\targetW$,
  such that each connected component of $G[X]$ is either an isolated vertex
  or contains an edge from $M$.
\item $\objs_\targetW$ to be the family of pairs $((X,M),(X_1,X_2))$, where
$(X,M) \in \cand_\targetW$ and $(X_1,X_2)$ is a consistent cut of the graph $(V(X), X)$
with $V(M) \subseteq X_1$.
\end{enumerate}
Note that if $|X| = \ell$ and each vertex in $V(X)$ has degree two, then $|V(X)| = \ell$.
Thus if $(X,M) \in \cand_\targetW$ then $X$ is a set of vertex disjoint cycles 
covering exactly $\ell$ vertices of $G$. 
If $(X,M) \in \sols_\targetW$, then the number
of cycles is bounded by $|M| = \targetk$, and if we have an $X$ with at most
$k$ cycles, we can find an $M$ so that $(X,M) \in \sols_\targetW$
for $\targetW = \weight(X \times \{\weightX\} \cup M \times \{\marker\})$ by
taking at least one edge from each cycle.
Thus, we need to check if $\sols_\targetW \neq \emptyset$ for some $\targetW$.

\noindent {\bf{The Count part.}}
Let $((X,M),(X_1,X_2)) \in \objs_\targetW$. Let $\conncomp(X,M)$ denote the number of connected components of $G[X]$
that are not isolated vertices and do not contain an edge from $M$.
If $C \subseteq X$ is the set of edges of such a connected
component of $G[X]$, then $((X,M),(X_1 \triangle V(C),X_2 \triangle V(C))) \in \objs_\targetW$,
i.e., the connected component $C$ can be on either side of the cut $(X_1,X_2)$.
Thus there are $2^{\conncomp(M,X)}$ elements in $\objs_\targetW$ that correspond to
any pair $(X,M) \in \cand_\targetW$, and we infer that
$|\sols_\targetW| \equiv |\objs_\targetW|$.

To finish the proof we need to describe a procedure
$\countproc(\weight,\targetW,\treedecomp)$ that,
given a nice tree decomposition $\treedecomp$, weight function $\weight$
and and an integer $\targetW$,
computes $|\objs_\targetW|$ modulo $2$.

As usual we use dynamic programming. 
We follow the notation
from the \steinertree{} example (see Lemma \ref{lem:stein}).
Let $\Sigma = \{\zero,\oneone,\onetwo,\two\}$.
For every bag $x \in \treedecomp$ of the tree decomposition,
  integers $0 \leq \iterk,b \leq |V|$, $0 \leq \iterW \leq 2\maxw|V|$
  and $s \in \Sigma^{\bag{x}}$ 
	(called the colouring) define
  \begin{align*}
    \cand_x(\iterk,b,\iterW) &= \Big\{(X,M) \; \big{|}\;
      M \subseteq X \subseteq \subedges{x}\ \wedge\ |M| = \iterk\ \wedge\ |X| = b\ \wedge\ \weight(X \times \{\weightX\} \cup M \times \{\marker\}) = \iterW \\
        &\qquad \wedge\ (\forall_{v \in V(X) \setminus\bag{x}} \deg_{G[X]}(v) = 2)\ \wedge\ (\forall_{v \in \bag{x}} \deg_{G[X]}(v) \leq 2) \Big\} \\
    \objs_x(\iterk,b,\iterW) &= \Big\{((X,M),(X_1,X_2)) \big{|}\; 
      (X,M) \in \cand_x(\iterk,b,\iterW)\ \wedge\ V(M) \subseteq X_1\\
        &\qquad\ \wedge\ (X_1,X_2) \textrm{ is a consistent cut of the graph }(V(X),X) \Big\} \\
    A_x(\iterk,b,\iterW,s) &= \Big| \Big\{((X,M),(X_1,X_2)) \in \objs_x(\iterk,b,\iterW) \big{|}\;
      (s(v) = \zero \Rightarrow \deg_{G[X]}(v) = 0) \\
             &\qquad\ \wedge\ (s(v) = \onej \Rightarrow (\deg_{G[X]}(v) = 1 \wedge v \in X_j))
           \  \wedge\ (s(v) = \two \Rightarrow \deg_{G[X]}(v) = 2) \Big\} \Big|
  \end{align*}
  The value of $s(v)$ denotes the degree of $v$ in $G[X]$ and,
  in case of degree one, $s(v)$ also stores information
  about the side of the cut $v$ belongs to.
  We note that we do not need to store the side of the cut for $v$
  if its degree is $0$ and $2$, since it is not yet or no more needed.
	This is a somewhat non-trivial trick --- the natural implementation of
	dynamic programming would use $6$ states for each vertex.
	For vertices of degree 0 this is necessary --- we do not want to count
	isolated vertices as separate connected components, so we do not want
	to have a side of the cut defined for them.
	For vertices of degree 2 the situation is more tricky.
	They are cut (that is, each such vertex is on some side of the cut in
	each counted object in $\objs_x(\iterk,b,\iterW)$), but the information
	about the side of the cut will not be needed --- we have a guarantee that
	no new edges will be added to that vertex (as 2 is the maxmimum degree).
	Note that the fact that we did remember the side of the cut previously
	ensures that when we have a path in the currently constructed solution,
	both endpoints of the path are remembered to be on the same side of the
	cut, even though we no more remember sides for the internal vertices of
	the path.
  The accumulators $\iterk$, $b$ and $\iterW$ keep track of the size of $M$,
  the size of $X$ and the weight of $(X,M)$, respectively.

	Let us spend a moment discussing the choice we made to mark edges (as
	opposed to vertices).
	If we marked vertices, as we did in the previous problems, we would have
	a problem as to when to decide that a given vertex is a marker.
	The natural moment --- in the forget bag --- is unapplicable in this case,
	as (to save on space and time) we do not remember the side of the cut,
	so we do not know whether we can mark a vertex (remember, the whole point
	of marking vertices is to break the symmetry between the sides of the cut,
	so we have to mark only vertices that are on the left).
	The same problem applies to the introduce bag, moreover a vertex is
	introduced more than once, so we could mark it more than once (which would
	cause problems with the application of the Isolation Lemma).
	The best choice would be to mark it when we introduce an edge incident to
	it, but still we could mark it twice, if the introduce edge bags happen
	in two different branches of the tree.
	This can be circumvented by upgrading the nice tree decomposition
	definition, but the way we have chosen --- to mark edges --- is easier and
	cleaner.
	For edges we know that each edge is introduced exactly once, so we have
	a natural place to mark the edge and assure it is marked and counted
	exactly once.

  The algorithm computes $A_x(\iterk,b,\iterW,s)$ for all bags $x \in T$ in a bottom-up fashion
  for all reasonable values of $\iterk$, $b$, $\iterW$ and $s$.
  We now give the recurrence for $A_x(\iterk,b,\iterW,s)$ that is used by the dynamic programming
  algorithm.
  In order to simplify notation denote by $v$ the vertex introduced and
  contained in an introduce bag, by $uv$ the edge introduced in an
  introduce edge bag,
  and by $y,z$ the left and right child of $x$ in $\treedecomp$ if present.
\begin{itemize}
\item \textbf{Leaf bag}:
  \[ A_x(0,0,0,\emptyset) = 1 \]
\item \textbf{Introduce vertex bag}:
  $$A_x(\iterk,b,\iterW,s[v \to\zero]) = A_y(\iterk,b,\iterW,s)$$
  The new vertex has degree zero and we do not impose any other constraints.
\item \textbf{Introduce edge bag}:
For the sake of simplicity of the recurrence formula
let us define a function $\subs: \Sigma \rightarrow 2^\Sigma$.

\begin{center}
  \begin{tabular}{|c | c | c | c | c |}
    \hline
      & $\zero$ & $\oneone$ & $\onetwo$ & $\two$ \\
    \hline
      $\subs$ & $\emptyset$ & $\{\zero\}$ & $\{\zero\}$ & $\{\oneone,\onetwo\}$ \\
    \hline
  \end{tabular}
\end{center}
Intuitively, for a given state $\alpha \in \Sigma$ the value $\subs(\alpha)$
is the set of possible states a vertex can have before adding an incident edge.

We can now write the recurrence for the introduce edge bag.
\begin{align*}
  A_x(\iterk,b,\iterW,s) &= A_y(\iterk,b,\iterW, s) 
                       + \sum_{\alpha_u \in \subs(s(u))}\ \ \sum_{\alpha_v \in \subs(s(v))}\ \ \sum_{j \in \{1,2\}}\\
                      &\qquad[(\alpha_u = \onej  \vee s(u) = \onej) \wedge (\alpha_v = \onej  \vee s(v) = \onej)]\\
                      &\qquad\quad \bigg(A_y(\iterk,b-1,\iterW-\weight((uv,\weightX)),s[u\to\alpha_u,v\to\alpha_v]) \\
                      &\qquad\qquad + [j=1] A_y(\iterk-1,b-1,\iterW-\weight((uv,\weightX)) - \weight((uv,\marker)),s[u\to\alpha_u,v\to\alpha_v])\bigg)
\end{align*}
To see that all cases are handled correctly, first notice that we can always choose not to use the introduced edge.
Observe that in order to add the edge $uv$ by the definition of $\subs$
we need to have $\alpha_u \in \subs(s(u))$ and $\alpha_v \in \subs(s(v))$.
We use the integer $j$ to iterate over two sides of the cut the edge $uv$ can
be contained in.
Finally we check whether $j=1$ before we make $uv$ a marker.
\item \textbf{Forget bag}:
  $$A_x(\iterk,b,\iterW,s) = A_y(\iterk,b,\iterW,s[v \to \two]) + A_y(\iterk,b,\iterW,s[v\to\zero])$$
  The forgotten vertex must have degree two or zero in $G[X]$.
\item \textbf{Join bag}:
  For colourings $s_1,s_2,s \in \{\zero,\oneone,\onetwo,\two\}^{\bag{x}}$ we say that $s_1+s_2 = s$ if for each $v \in \bag{x}$ at least one of the following holds:
  \begin{align*}
  s_1(v) = \zero &\ \wedge\ s(v) = s_2(v) \\
  s_2(v) = \zero &\ \wedge\ s(v) = s_1(v) \\
  s_1(v) = s_2(v) = \onej &\ \wedge\ s(v) = \two
  \end{align*}
  We can now write the recurrence for the join bags.
  \begin{align*}
  A_x(\iterk,b,\iterW,s) &= \sum_{\iterk_1 + \iterk_2 = \iterk} \ \ \sum_{b_1+b_2=b} \ \ \sum_{\iterW_1 + \iterW_2 = \iterW} \ \ \sum_{s_1 + s_2 = s} A_y(\iterk_1,b_1,\iterW_1,s_1)A_z(\iterk_2,b_2,\iterW_2,s_2)
  \end{align*}
  The accumulators in the children bags need to sum up to the accumulators in the parent bag.
  Also the degrees need to sum up and the sides of the cut need to match, which is ensured by the constraint $s_1+s_2=s$.

  A straightforward computation of the above recurrence leads to $16^t |V|^{O(1)}$ time.
  We now show how to use the $\Z_4$ product to obtain a better time complexity.

  Let $\phi:\{\zero,\oneone,\onetwo,\two\} \to \Z_4$ be defined as
  \begin{align*}
  \phi(\zero) &= 0 & \phi(\oneone) &= 1 &
  \phi(\onetwo) &= 3 & \phi(\two)  &= 2
  \end{align*}
  Let $\phi:\{\zero,\oneone,\onetwo,\two\}^{\bag{x}} \to \Z_4^{\bag{x}}$ be obtained by extending $\phi$ in the natural way.
  Note that $\phi$ is a bijection.
  Define $\rho:\{\zero,\oneone,\onetwo,\two\} \to \Z$ as
  \begin{align*}
  \rho(\zero) &= 0 & \rho(\oneone) &= 1 &
  \phi(\onetwo) &= 1 & \phi(\two)  &= 2
  \end{align*}
  and let $\rho(s) = \sum_{v \in \bag{x}} \rho(s(v))$ for colouring $s$,
  i.e., $\rho(s)$ is the sum of degrees of all vertices in $\bag{x}$.
  Let
  \begin{align*}
  f^{\iterk,b,\iterW}_m(\phi(s)) &= [\rho(s)=m]A_y(\iterk,b,\iterW,s) \\
  g^{\iterk,b,\iterW}_m(\phi(s)) &= [\rho(s)=m]A_z(\iterk,b,\iterW,s) \\
  h^{\iterk,b,\iterW}_m(\phi(s)) &= \sum_{\iterk_1 + \iterk_2 = \iterk} \ \ \sum_{b_1+b_2=b} \ \ \sum_{\iterW_1 + \iterW_2 = \iterW} \ \ \sum_{m_1+m_2=m} (f^{\iterk_1,b_1,\iterW_1}_{m_1} \fscZ{4} g^{\iterk_2,b_2,\iterW_2}_{m_2})(\phi(s))
  \end{align*}
  We claim that
  $$A_x(\iterk,b,\iterW,s) = h^{\iterk,b,\iterW}_{\rho(s)}(\phi(s)).$$
  First notice that the values of accumulators are divided among the children,
and that no vertex or edge is accounted for twice by the definition of $A_x$. Hence, it suffices to prove that values in the expansion of $h^{\iterk,b,\iterW}_{\rho(s)}(\phi(s))$ corresponding to a choice of $s_1,s_2$ and possibly having a contribution to $A_x(\iterk,b,\iterW,s)$ are exactly those, for which $s_1+s_2=s$ holds. To see this, first note that
  $$\rho(\phi^{-1}(\phi(s_1) + \phi(s_2))) \leq \rho(s_1) + \rho(s_2).$$
  Observe that the above inequality is an equality iff $s_1+s_2=s$. Thus when counting $h^{\iterk,b,\iterW}_m(\phi(s))$ we sum non-zero values only for such $s_1,s_2$ where $s=s_1+s_2$. As the addition operator on colourings corresponds to the addition operator in $\Z_4$, the claim follows.

   By Theorem \ref{thm:fscxor}, the function $h^{\iterk,b,\iterW}_m$ can be computed in $4^t |V|^{O(1)}$ time and the time bound for the join bags follows.
\end{itemize}
It is easy to see that the above recurrence leads to a dynamic programming 
algorithm that
computes the parity of $|\sols_\targetW|$ for all values of $\targetW$ in $4^t |V|^{O(1)}$ time,
since $|\objs_\targetW|=A_\rootv(k,\ell,\targetW,\emptyset)$ and $|\sols_\targetW| \equiv |\objs_\targetW|$.
Moreover, as we count the parities and not the numbers $A_x$ themselves, 
all arithmetical operations can be done in constant time.
Thus, the proof of Theorem~\ref{thm:app:undircycles} is finished. 
\end{proof}

\subsubsection{The directed case}

\begin{theorem}\label{thm:app:dircycles}
 There exists a Monte-Carlo algorithm that given a tree decomposition of width $t$
 solves \dirpartialcycles{} in $6^t |V|^{O(1)}$ time.
 The algorithm cannot give false positives and may give false negatives with probability at most $1/2$.
\end{theorem}
\begin{proof}
We use the Cut\&Count technique. To count the number of cycles
we use markers. As in the undirected case, in this application it is more convenient to take as markers
arcs instead of vertices. The objects we count are subsets of arcs,
together with sets of marked arcs, thus
we take $\univ = A \times \{\weightX, \marker\}$.
As usual, we assume we are given a weight function $\weight: \univ \to \{1,2,\ldots,N\}$,
where $N = 2|\univ| = 4|A|$.
We also assume $\targetk \leq \ell$.

\noindent {\bf{The Cut part.}} For an integer $\targetW$ we define:
\begin{enumerate}
\item $\cand_\targetW$ to be the family of pairs $(X,M)$, where
$M \subseteq X \subseteq A$, $|X| = \ell$, $|M| = \targetk$,
  $\weight(X \times \{\weightX\} \cup M \times \{\marker\}) = \targetW$ and each vertex $v \in V(X)$
  has indegree and outdegree $1$ in $G[X]$.
\item $\sols_\targetW$ to be the family of pairs $(X,M) \in \cand_\targetW$,
  such that each connected component of $G[X]$ is either an isolated vertex
  or contains an arc from $M$.
  \item $\objs_\targetW$ to be the family of pairs $((X,M),(X_1,X_2))$, where
$(X,M) \in \cand_\targetW$ and $(X_1,X_2)$ is a consistent cut of the graph $(V(X), X)$
with $V(M) \subseteq X_1$.
\end{enumerate}
Note that if $|X| = \ell$ and each vertex in $V(X)$ has indegree and outdegree one, then $|V(X)| = \ell$.
Thus similarly as before we need to check if $\sols_\targetW \neq \emptyset$ for some $\targetW$.

\noindent {\bf{The Count part.}}
Let $((X,M),(X_1,X_2)) \in \objs_\targetW$. Let $\conncomp(X,M)$ denote the number of weakly\footnote{We stress this for clarity: in $G[X]$ weakly connected components are always strongly connected components due to the requirements imposed on $X$.} connected components of $G[X]$
that are not isolated vertices and do not contain an arc from $M$.
If $C \subseteq X$ is the set of arcs of such a weakly connected
component of $G[X]$, then $((X,M),(X_1 \triangle V(C),X_2 \triangle V(C))) \in \objs_\targetW$,
i.e., the weakly connected component $C$ can be on either side of the cut $(X_1,X_2)$.
Thus there are $2^{\conncomp(M,X)}$ elements in $\objs_\targetW$ that correspond to
any pair $(X,M) \in \cand_\targetW$, and we infer that
$|\sols_\targetW| \equiv |\objs_\targetW|$.

To finish the proof we need to describe a procedure
$\countproc(\weight,\targetW,\treedecomp)$ that,
given a nice tree decomposition $\treedecomp$, weight funtion $\weight$
and an integer $\targetW$,
computes $|\objs_\targetW|$ modulo $2$.

As usual we use dynamic programming. 
We follow the notation
from the \steinertree{} example (see Lemma \ref{lem:stein}).
Let $\Sigma = \{\zz,\zoone,\zotwo,\ozone,\oztwo,\oo\}$.
For every bag $x \in \treedecomp$ of the tree decomposition,
  integers $0 \leq \iterk,b \leq |V|$, $0 \leq \iterW \leq 2\maxw|V|$
  and $s \in \Sigma^{\bag{x}}$ 
	(called the colouring) define
  \begin{align*}
    \cand_x(\iterk,b,\iterW) &= \Big\{(X,M) \; \big{|}\;
      M \subseteq X \subseteq \subarcs{x}\ \wedge\ |M| = \iterk\ \wedge\ |X| = b\ \wedge\ \weight(X \times \{\weightX\} \cup M \times \{\marker\}) = \iterW \\
        &\qquad \wedge\ (\forall_{v \in V(X) \setminus \bag{x}} \indeg_{G[X]}(v) = \outdeg_{G[X]}(v) = 1)\ \wedge\ (\forall_{v \in \bag{x}} \indeg_{G[X]}(v),\outdeg_{G[X]}(v) \leq 1) \Big\} \\
    \objs_x(\iterk,b,\iterW) &= \Big\{((X,M),(X_1,X_2)) \big{|}\; 
      (X,M) \in \cand_x(\iterk,b,\iterW)\ \wedge\ V(M) \subseteq X_1\\
        &\qquad \wedge\ (X_1,X_2) \textrm{ is a consistent cut of the graph} (V(X), X) \Big\} \\
    A_x(\iterk,b,\iterW,s) &= \Big| \Big\{((X,M),(X_1,X_2)) \in \objs_x(\iterk,b,\iterW) \big{|}\; 
      (s(v) = \mathbf{io}_j \Rightarrow v \in X_j) \\
    &\qquad\ \wedge\
      ((s(v) = \mathbf{io} \vee s(v) = \mathbf{io}_j) \Rightarrow (\indeg_{G[X]}(v) = \mathbf{i} \wedge \outdeg_{G[X]}(v) = \mathbf{o})) \Big\} \Big| \\
  \end{align*}
  The value of $s(v)$ contains an information about
  the indegree and outdegree of $v$ and, in case 
  when the degree of $v$ is one, $s(v)$ also stores information
  about the side of the cut $v$ belongs to.
  We note that we do not need to store the side of the cut for $v$
  if its degree is $0$ and $2$, since it is not yet or no more needed.
  The accumulators $\iterk$, $b$ and $\iterW$ keep track of the size of $M$,
  the size of $X$ and the weight of $(X,M)$, respectively.

  The algorithm computes $A_x(\iterk,b,\iterW,s)$ for all bags $x \in T$ in a bottom-up fashion
  for all reasonable values of $\iterk$, $b$, $\iterW$ and $s$.
  We now give the recurrence for $A_x(\iterk,b,\iterW,s)$ that is used by the dynamic programming
  algorithm.
  In order to simplify notation let $v$ be the vertex introduced and
  contained in an introduce bag, $(u,v)$ the arc introduced in an
  introduce edge (arc) bag,
  and $y,z$ the left and right child of $x$ in $\treedecomp$ if present.

\begin{itemize}
\item \textbf{Leaf bag}:
  \[ A_x(0,0,0,\emptyset) = 1 \]
\item \textbf{Introduce vertex bag}:
  \[ A_x(\iterk,b,\iterW,s[v \to \zz])  =  A_y(\iterk,b,\iterW,s) \]
     The new vertex has indegree and outdegree zero.
\item \textbf{Introduce edge (arc) bag}:
For the sake of simplicity of the recurrence formula
let us define functions $\insubs,\outsubs: \Sigma \rightarrow 2^\Sigma$.

\begin{center}
  \begin{tabular}{|c | c | c | c | c | c | c |}
    \hline
      & $\zz$ & $\zoone$ & $\zotwo$ & $\ozone$ & $\oztwo$ & $\oo$ \\
    \hline
      $\insubs$ & $\emptyset$ & $\emptyset$ & $\emptyset$ & $\{\zz\}$ & $\{\zz\}$ & $\{\zoone, \zotwo\}$ \\
    \hline
      $\outsubs$ & $\emptyset$ & $\{\zz\}$ & $\{\zz\}$ & $\emptyset$ & $\emptyset$ & $\{\ozone, \oztwo\}$ \\
    \hline
  \end{tabular}
\end{center}
Intuitively, for a given state $\alpha \in \Sigma$ the values $\insubs(\alpha)$ and $\outsubs(\alpha)$ 
are the sets of possible states a vertex can have before adding an incoming and respectively outgoing arc.

We can now write the recurrence for the introduce arc bag.
\begin{align*}
  A_x(\iterk,b,\iterW,s) &= A_y(\iterk,b,\iterW, s) 
                       + \sum_{\alpha_u \in \outsubs(s(u))}\ \ \sum_{\alpha_v \in \insubs(s(v))}\ \ \sum_{j \in \{1,2\}}\\
                      &\qquad[(\alpha_u = \ozj  \vee s(u) = \zoj) \wedge (\alpha_v = \zoj  \vee s(v) = \ozj)]\\
                      &\qquad\quad \bigg(A_y(\iterk,b-1,\iterW-\weight(((u,v),\weightX)),s[u\to\alpha_u,v\to\alpha_v]) \\
                      &\qquad\qquad + [j=1] A_y(\iterk-1,b-1,\iterW-\weight((u,v),\weightX)-\weight((u,v),\marker),s[u\to\alpha_u,v\to\alpha_v])\bigg)
\end{align*}
To see that all cases are handled correctly, first notice that we can always choose not to use the introduced arc. 
Observe that in order to add the arc $(u,v)$ by the definition of $\insubs$ and $\outsubs$
we need to have $\alpha_u \in \outsubs(s(u))$ and $\alpha_v \in \insubs(s(v))$.
We use the integer $j$ to iterate over two sides of the cut the arc $(u,v)$ can
be contained in.
Finally we check whether $j=1$ before we make $(u,v)$ a marker.
\item \textbf{Forget vertex $v$ bag $x$}:
  $$A_x(\iterk,b,\iterW,s) = A_y(\iterk,b,\iterW,s[v \to \oo]) + A_y(\iterk,b,\iterW,s[v\to\zz])$$
  The forgotten vertex must have degree zero or two.
\item \textbf{Join bag}:
We have two children $y$ and $z$. Figure \ref{fig:jointablesmdcc} shows how two individual states of a vertex in $y$ and $z$ combine to a state of $x$. XX indicates that two states do not combine. The correctness of the table is easy to check.
\begin{figure}[htb]
	\begin{center}

	\begin{tabular}{c||c|c|c|c|c|c|} 
		  & $\zz$ & $\zoone$ & $\zotwo$ & $\oztwo$ & $\ozone$ & $\oo$\\
\hline
$\zz$   & $\zz$ & $\zoone$ & $\zotwo$ & $\oztwo$ & $\ozone$ & $\oo$ \\
$\zoone$ & $\zoone$ & $\mathbf{XX}$ & $\mathbf{XX}$ & $\mathbf{XX}$ & $\oo$ & $\mathbf{XX}$ \\
$\zotwo$ & $\zotwo$ & $\mathbf{XX}$ & $\mathbf{XX}$ & $\oo$ & $\mathbf{XX}$ & $\mathbf{XX}$ \\
$\oztwo$ & $\oztwo$ & $\mathbf{XX}$ & $\oo$ & $\mathbf{XX}$ & $\mathbf{XX}$ & $\mathbf{XX}$ \\
$\ozone$ & $\ozone$ & $\oo$ & $\mathbf{XX}$ & $\mathbf{XX}$ & $\mathbf{XX}$ & $\mathbf{XX}$ \\
$\oo$ & $\oo$ & $\mathbf{XX}$ & $\mathbf{XX}$ & $\mathbf{XX}$ & $\mathbf{XX}$ & $\mathbf{XX}$ 
	\end{tabular}

	\end{center}
	\caption{The join table of \dirpartialcycles where it is indicated which states combine to which other states.}
	\label{fig:jointablesmdcc}
\end{figure}

For colourings $s_1,s_2,s \in \Sigma^{\bag{x}}$ we say that $s_1+s_2=s$ if for each vertex $v \in \bag{x}$ the values
of $s_1(v)$ and $s_2(v)$ combine into $s(v)$ as in Figure \ref{fig:jointablesmdcc}.
We can now write the recurrence formula for join bags.
\begin{align*}
A_x(\iterk,b,\iterW,s) &= \sum_{\iterk_1 + \iterk_2 = \iterk}\ \ \sum_{b_1+b_2=b} \ \ \sum_{\iterW_1 + \iterW_2 = \iterW} \ \ \sum_{s_1 + s_2 = s}
                         A_y(\iterk_1,b_1,\iterW_1,s_1)A_z(\iterk_2,b_2,\iterW_2,s_2) 
\end{align*}
A straightforward computation of the above formula leads to $36^t |V|^{O(1)}$ time complexity. We now show how to use
Generalized Subset Convolution to obtain a better time bound.

Let $\phi,\rho: \Sigma \to \{0,1,2,3,4,5\}$ where
\begin{align*}
\phi(\zz) &= 0 & \phi(\zoone) &= 1 & \phi(\zotwo) &= 2 & \phi(\oztwo) &= 3 & \phi(\ozone) &= 4 & \phi(\oo) &= 5 \\
\rho(\zz) &= 0 & \rho(\zoone) &= 1 & \rho(\zotwo) &= 1 & \rho(\oztwo) &= 1 & \rho(\ozone) &= 1 & \rho(\oo) &= 2 
\end{align*}
Let $\phi: \Sigma^{\bag{x}} \rightarrow \{0,1,2,3,4,5\}^{\bag{x}}$ be obtained by extending $\phi$ in the natural way. 
Define $\rho: \Sigma^{\bag{x}} \rightarrow \Z$ as $\rho(s)=\sum_{e \in \bag{x}}\rho(e)$.
Hence $\rho$ reflects the total number of $\one$'s in a state $s$, i.e., the sum of all degrees of vertices in $\bag{x}$. Then, define
\begin{align*}
	f^{\iterk,b,\iterW}_{m}(\phi(s)) &= [\rho(s) = m] A_y(\iterk,b,\iterW, s)\\
	g^{\iterk,b,\iterW}_{m}(\phi(s)) &= [\rho(s) = m] A_z(\iterk,b,\iterW, s)\\
	h^{\iterk,b,\iterW}_{m}(\phi(s)) &=  \sum_{\iterk_1 + \iterk_2 = \iterk}\ \ \sum_{b_1+b_2=b} \ \ \sum_{\iterW_1 + \iterW_2 = \iterW} \ \ \sum_{m_1 + m_2 = m} (f^{\iterk_1,b_1,\iterW_1}_{m_1} \fscC{6} g^{\iterk_2,b_2,\iterW_2}_{m_2})(\phi(s))
\end{align*}
We claim that
\[ A_x(\iterk,b,\iterW,s) = h^{\iterk,b,\iterW}_{\rho(s)}(\phi(s)) \]
To see this, first notice that the values of accumulators are divided among the children,
and that no vertex or edge is accounted for twice by the definition of $A_x$. 
Hence, it suffices to prove that exactly all combinations of table entries from 
$A_y$ and $A_z$ that combine to state $s$ acccording to Table \ref{fig:jointablesmdcc} 
contribute to $A_x(\iterk,b,\iterW,s)$. 
Notice that if $\alpha,\beta \in \Sigma$ and 
$\gamma = \phi^{-1}(\phi(\alpha) + \phi(\beta))$, then 
$\rho(\gamma) \leq \rho(\alpha)+\rho(\beta)$. 
This implies that the only pairs that contribute to 
$h^{\iterk,b,\iterW}_{m}(\phi(s))$ are the pairs not leading to crosses in Table 
\ref{fig:jointablesmdcc} since for the other pairs we have 
$\rho(\gamma) < \rho(\alpha)+\rho(\beta)$. 
Finally notice that for every such pair we have that 
$\gamma$ is the correct state, and hence correctness follows.

Finally we obtain that, by Theorem \ref{thm:gfsc}, the values $A_x(\iterk,b,\iterW,s)$ for a join bag $x$ can be computed in time $6^t |V|^{O(1)}$.
\end{itemize}
It is easy to see that the above recurrence leads to a dynamic programming 
algorithm that
computes the parity of $|\sols_\targetW|$ for all values of $\targetW$ in $6^t |V|^{O(1)}$ time,
since $|\objs_\targetW|=A_\rootv(k,\ell,\targetW,\emptyset)$ and $|\sols_\targetW| \equiv |\objs_\targetW|$.
Moreover, as we count the parities and not the numbers $A_x$ themselves, 
all arithmetical operations can be done in constant time.
Thus, the proof of Theorem~\ref{thm:app:dircycles} is finished. 
\end{proof}

\newcommand{\objsx}{\overline{\objs}}
\newcommand{\objsxk}{\ell}
\newcommand{\iterxk}{\ell}

\subsection{Spanning trees with a prescribed number of leaves}\label{sec:app:exact-k-stuff}

In this section we provide algorithms that solve
\exactleaf{} and \exactoutbranching{} in time $4^t n^{O(1)}$ and $6^t n^{O(1)}$, respectively.
The algorithms are very similar and use the same tricks, thus we gather them together in this subsection.

Both algorithms use almost the same Cut part that is very natural for the considered problems.
However, a quite straightforward realization of the accompanying Count part would lead
to running times $6^t n^{O(1)}$ and $8^t n^{O(1)}$, respectively.
To obtain better time bounds we need to count objects in a more
ingenious way.

\subsubsection{\exactleaf}\label{sec:app:exact-k-leaf}

\defproblemu{\exactleaf}
{An undirected graph $G = (V,E)$ and an integer $k$.}
{Does there exists a spanning tree of $G$ with exactly $k$ leaves?}

\begin{theorem}\label{thm:app:exactleaf}
 There exists a Monte-Carlo algorithm that given a tree decomposition of width $t$
 solves \exactleaf{} in $4^t |V|^{O(1)}$ time.
 The algorithm cannot give false positives and may give false negatives with probability at most $1/2$.
\end{theorem}

\begin{proof}
We assume that $G$ is connected, as otherwise we can safely answer NO.
We also assume $|V| \geq 3$, and therefore any spanning tree of $G$ contains some internal bags (i.e., bags of degree
at least $2$). Using similar arguments as in Remark \ref{rem:alg:guess-vertex},
we may assume that we are given vertex $\anchor \in V$ that is required to be
an internal bag of the spanning tree in question. Thus, we can look
for spanning trees of $G$ that are rooted in the given vertex $\anchor$.

We use the Cut\&Count technique.
Our solutions and solution candidates are subsets of edges, thus we
take $\univ = E$ and generate random weight function $\weight:\univ \to \{1,2,\ldots,\maxw\}$,
where $\maxw = 2|\univ| = 2|E|$.

\noindent {\bf{The Cut part.}} For integers $\targetW$ and $\targetk$ we define:
\begin{enumerate}
\item $\cand_\targetW^\targetk$ to be the family of sets $X \subseteq E$, such that
$\weight(X) = \targetW$, $|X| = |V|-1$, $G[X]$ contains exactly $\targetk$ vertices of degree one,
and the degree of $\anchor$ in $G[X]$ is at least $2$.
\item $\sols_\targetW^\targetk$ to be the family of sets $X \in \cand_\targetW^\targetk$,
  such that $G[X]$ is connected;
\item $\objs_\targetW^\targetk$ to be the family of pairs $(X,(X_1,X_2))$, where
$X \in \cand_\targetW^\targetk$ and $(X_1,X_2)$ is a consistent cut of $G[X]$ with $\anchor \in X_1$.
\end{enumerate}
The condition that for $X \in \sols_\targetW^\targetk$ the graph $G[X]$ is connected, together
with $|X| = |V|-1$ gives us that each $X \in \sols_\targetW^\targetk$ induces a spanning tree.
Thus, $\sols_\targetW^\targetk$ is indeed a family of spanning trees of exactly $k$ leaves
with root $\anchor$.

Note that, unlike in other algorithms, we use the superscript $\targetk$ in the definitions. 
To achieve claimed the running time we need to do computations for many values of $\targetk$.

\noindent {\bf{The Count part.}}
To use Algorithm \ref{alg:cutandcount} we need to formally prove that for any $\targetW$ and $\targetk$ we have $|\sols_\targetW^\targetk| \equiv |\objs_\targetW^\targetk|$.
Similarly as in the case of \steinertree
	we note that
  by Lemma~\ref{lem:evencancel} for each $X \in \cand_\targetW^\targetk$
  there exist $2^{\conncomp(G[X])-1}$
  consistent cuts of $G[X]$, and the claim follows.

To finish the proof we need to show how to compute $|\objs_\targetW^\targetk|$ modulo
$2$ in time $4^t n^{O(1)}$.
A straightforward dynamic programming algorithm would lead to a $6^t n^{O(1)}$ time complexity
(we encourage the reader to sketch this algorithm to see why the steps
introduced below are needed),
thus we need to be a bit more ingenious here.

Let us define the set $\objsx_\targetW^\objsxk$ to be a family of triples $(X,R,(Y_1,Y_2))$
such that
\begin{enumerate}
\item $X \in \bigcup_{\targetk=0}^{|V|-1} \cand_\targetW^\targetk$,
(i.e., we do not impose any constraint on the number of vertices of degree one in $G[X]$),
\item $R \subseteq V \setminus \{\anchor\}$ and $|R| = \objsxk$,
\item Each vertex $v \in R$ has degree one in $G[X]$ and the unique neighbour of $v$ in $G[X]$
is not an element of $R$ (i.e., $G[X]$ does not contain a connected component that consists
    of two vertices from $R$ connected by an edge).
\item Let $G(V\setminus R, X)$ denote the graph with the vertex set $V
\setminus R$ and the edge set consisting of those edges of $X$ that have
both endpoints in $V\setminus R$. Then we require that $(Y_1,Y_2)$ is a consistent cut of $G(V \setminus R,X)$ with $\anchor \in Y_1$.
\end{enumerate}
Informally speaking, there are two differences between $\objs_\targetW^\targetk$ and $\objsx_\targetW^\objsxk$.
First, instead of requiring a prescribed number of vertices of degree one,
we distinguish a fixed number of vertices that have to be of degree one, and we do not care about the degrees of the other vertices.
Second, we consider only consistent cuts of $G(V \setminus R, X)$, not of 
whole $G[X]$.

We first note that the second difference is somewhat illusory.
Let $X \in \cand_\targetW^\targetk$ and $R \subseteq V \setminus \{\anchor\}$
be as in the definition of $\objsx_\targetW^\objsxk$, i.e., $|R| = \objsxk$, each $v \in R$ is of degree one in $G[X]$ and its unique neighbour in $G[X]$
is not in $R$. If $(X,(X_1,X_2)) \in \objs_\targetW^\targetk$, then the cut $(Y_1,Y_2)$ where $Y_j = X_j \setminus R$ is consistent with $G(V \setminus R, X)$
and $(X,R,(Y_1,Y_2)) \in \objsx_\targetW^\objsxk$. In the other direction, observe that if $(X,R,(Y_1,Y_2)) \in \objsx_\targetW^\objsxk$, then
there exists exactly one cut $(X_1,X_2)$ of $G[X]$, such that $Y_1 \subseteq X_1$ and $Y_2 \subseteq X_2$, namely $X_j = Y_j \cup (N_{G[X]}(Y_j) \cap R)$.
Thus, in the analysis that follows we can assume that $(Y_1,Y_2)$ is in fact a consistent cut of $G[X]$ with $\anchor \in Y_1$.

Let $(X,R,(Y_1,Y_2)) \in \objsx_\targetW^\objsxk$. As there exists $2^{\conncomp(G[X])-1}$
consistent cuts of $G[X]$, for fixed $X$ and $R$ we have
$(X,R,C) \in \objsx_\targetW^\objsxk$ for exactly $2^{\conncomp(G[X])-1}$ cuts $C$.
Thus, if $X \notin \bigcup_{\targetk=0}^{|V|-1} \sols_\targetW^\targetk$,
all elements of $\objsx_\targetW^\objsxk$ with $X$ cancel out modulo $2$.

Otherwise, if $G[X]$ induces a spanning tree of $G$ with root $\anchor$, 
  there exists exactly one cut $(V,\emptyset)$ consistent with $G[X]$, such that
  $\anchor$ is in the first set in the cut.
  Moreover, note that there are no two adjacent vertices of degree one in $G[X]$ (as $|V|\geq 3$).
  Thus, if $X \in \sols_\targetW^\targetk$, then
there exist $\binom{\targetk}{\objsxk}$ choices of the set $R$ and one choice of a cut $C=(V,\emptyset)$
such that $(X,R,C) \in \objsx_\targetW^\objsxk$ (we use the convention that
$\binom{\targetk}{\objsxk} = 0$ if $\targetk < \objsxk$).

Summing up, we obtain that in $\Z_2$
$$|\objsx_\targetW^\objsxk| \equiv
\sum_{\targetk=\objsxk}^{|V|-1} \binom{\targetk}{\objsxk} |\sols_\targetW^\targetk| \equiv
\sum_{\targetk=\objsxk}^{|V|-1} \binom{\targetk}{\objsxk} |\objs_\targetW^\targetk|.$$
Note that, operating over the field $\Z_2$,
we have obtained a linear operator that transforms a vector
$(|\objs_\targetW^\targetk|)_{\targetk=0}^{|V|-1}$ into a vector
$(|\objsx_\targetW^\objsxk|)_{\objsxk=0}^{|V|-1}$. Moreover, the matrix of this operator can be computed in polynomial time and is upper triangular
with ones on the diagonal. Thus, this operator can be easily inverted, and we can compute (in $\Z_2$)
all values $(|\objs_\targetW^\targetk|)_{\targetk=0}^{|V|-1}$
knowing all values 
$(|\objsx_\targetW^\objsxk|)_{\objsxk=0}^{|V|-1}$.

To finish the proof we need to describe a procedure
$\overline{\countproc}(\weight,\targetW,\objsxk,\treedecomp)$ that,
given a nice tree decomposition $\treedecomp$, weight funtion $\weight$
and integers $\targetW$ and $\objsxk$
computes $|\objsx_\targetW^\objsxk|$ modulo $2$.
Now we can use dynamic programming on the tree decomposition.

Recall that in the definition of $\objsx_\targetW^\objsxk$ the cut $(Y_1,Y_2)$
was a consistent cut of only $G(V \setminus R, X)$. We make use of this fact
to reduce the size of the table in the dynamic programming, as we do not need to remember
side of the cut for vertices in $R$.


	We follow the notation
  from the \steinertree{} example (see Lemma \ref{lem:stein}).
  For every bag $x \in \treedecomp$ of the tree decomposition,
  integers $0 \leq \iterxk \leq |V|$, $0 \leq \iterW \leq \maxw|E|$, $0 \leq m,d < |V|$,
  and $s \in \{\oneone,\onetwo,\zerozero,\zeroone\}^{\bag{x}}$
  (called the colouring) define
  \begin{align*}
    \cand_x(\iterxk,\iterW,m,d) &= \Big\{(X,R) \; \big{|}\;
      X \subseteq \subedges{x}\ \wedge\ R \subseteq \subbags{x}\ \wedge\ |X|=m\ \wedge\ |R| = \iterxk\ \wedge\
      \weight(X)=\iterW \\
     &\qquad \wedge\ \deg_{G[X]}(\anchor) = d 
     \ \wedge\ (v \in R \setminus \bag{x} \Rightarrow \deg_{G[X]}(v) = 1)
     \ \wedge\ (v \in R \cap \bag{x} \Rightarrow \deg_{G[X]}(v)  \leq 1)
        \Big\}\\
    \objs_x(\iterxk,\iterW,m,d) &= \Big\{(X,R,(Y_1,Y_2)) \big{|}\; 
      (X,R) \in \cand_x(\iterxk,\iterW,m,d)\ \wedge\ (\anchor \in \subbags{x} \Rightarrow \anchor \in Y_1)\\
        &\qquad \wedge\ (Y_1,Y_2)
      \textrm{ is a consistent cut of }G(\subbags{x} \setminus R, X)
			\Big\} \\
    A_x(\iterxk,\iterW,m,d,s) &= \Big| \Big\{(X,R,(Y_1,Y_2)) \in \objs_x(\iterxk,\iterW,m,d)
      \big{|}\; \\
      &\qquad (s(v) = \onej \Rightarrow v \in Y_j)\ \wedge\
      (s(v) = \zeroj \Rightarrow (v \in R \wedge \deg_{G[X]}(v) = j))
         \Big\} \Big|
  \end{align*}
  Here $s(v) = \zeroj$ denotes that $v \in R$ and $\deg_{G[X]}(v) = j$,
  whereas $s(v) = \onej$ denotes that $v \in Y_j$ (and thus $v \notin R$).
  The accumulators $\iterxk$, $m$ and $\iterW$ keep track of the size of $R$,
  size of $X$ and the weight of $X$, respectively.
  The accumulator $d$ keeps track of the degree of $\anchor$ in $G[X]$,
  since we need to ensure that in the end it is at least $2$.
  Hence $A_x(\iterxk,\iterW,m,d,s)$ reflects the number of
  partial objects from $\objsx$ 
  with fixed sizes of $R$, $X$, weight of $X$, 
  degree of $\anchor$ and interface on vertices from $\bag{x}$.

  The algorithm computes $A_x(\iterxk,\iterW,m,d,s)$
  for all bags $x \in T$ in a bottom-up fashion
  for all reasonable values of $\iterxk$, $\iterW$, $m$, $d$ and the colouring $s$.
  We now give the recurrence for $A_x(\iterxk,\iterW,m,d,s)$
  that is used by the dynamic programming
  algorithm.
  In order to simplify notation we denote by $v$ the vertex introduced and
  contained in an introduce bag, by $uv$ the edge introduced in an
  introduce edge bag,
  and by $y,z$ the left and right child of $x$ in $\treedecomp$ if present.
\begin{itemize}
\item \textbf{Leaf bag}:
  \[ A_x(0,0,0,0,\emptyset) = 1 \]
\item \textbf{Introduce vertex bag}:
  \begin{align*}
    A_x(\iterxk,\iterW,m,d,s[v\to\zerozero]) &= [v \neq \anchor]A_y(\iterxk-1,\iterW,m,d,s) \\
    A_x(\iterxk,\iterW,m,d,s[v\to\zeroone]) &= 0 \\
    A_x(\iterxk,\iterW,m,d,s[v\to\oneone]) &= A_y(\iterxk,\iterW,m,d,s) \\
    A_x(\iterxk,\iterW,m,d,s[v\to\onetwo]) &= [v \neq \anchor]A_y(\iterxk,\iterW,m,d,s)
  \end{align*}
  If the new vertex is in $R$, it has degree zero and cannot be equal to $\anchor$.
  Otherwise, we need to ensure that we do not put $\anchor$ into $Y_2$.
\item \textbf{Introduce edge bag}:
\begin{align*}
 A_x(\iterxk,\iterW,m,d,s) &= A_y(\iterxk,\iterW,m,d,s) &\\
                              &\quad + A_y(\iterxk,\iterW-\weight(uv),m-1,d-[\anchor=u \vee \anchor=v],s) &\textrm{if }s(u)=s(v)=\onej \\
 A_x(\iterxk,\iterW,m,d,s) &= A_y(\iterxk,\iterW,m,d,s) &\\
 &\quad + A_y(\iterxk,\iterW-\weight(uv),m-1,d-[\anchor=u],s[v\to \zerozero]) &\textrm{if }s(v) = \zeroone \wedge s(u) = \onej \\
 A_x(\iterxk,\iterW,m,d,s) &= A_y(\iterxk,\iterW,m,d,s) &\\
 &\quad + A_y(\iterxk,\iterW-\weight(uv),m-1,d-[\anchor=v],s[u\to \zerozero]) &\textrm{if }s(v) = \onej \wedge s(u) = \zeroone \\
 A_x(\iterxk,\iterW,m,d,s) &= A_y(\iterxk,\iterW,m,d,s) &\textrm{otherwise}
  \end{align*}
  Here we consider adding $uv$ to $X$. This is possible in two cases.
  First, if $u,v \notin R$ and $s(u)=s(v)$. Second, if exactly one of $u$ and $v$ is in $R$
  (recall that we forbid edges connecting two vertices in $R$). In the second case we need
  to update the degree of the vertex in $R$. In both cases we need to update the degree of $\anchor$,
  if needed.
\item \textbf{Forget bag}:
\begin{align*}
  A_x(\iterxk,\iterW,m,d,s) &= [d \geq 2]A_y(\iterxk,\iterW,d,s[v\to\oneone]) &\textrm{if }v = \anchor\\
  A_x(\iterxk,\iterW,m,d,s) &= \sum_{\alpha \in \{\zeroone,\oneone,\onetwo\}} A_y(\iterxk,\iterW,m,d,s[v\to\alpha])&\textrm{otherwise}
  \end{align*}
  If we forget $v=\anchor$, we require that its degree is at least two and $v \in Y_1$.
  Otherwise, we require only that if $v \in R$ then $\deg_{G[X]}(v) = 1$.

\item \textbf{Join bag}:
  We proceed similarly as in the case of join bags in the \cdomset{} problem.
  For a colouring $s \in \{\zerozero, \zeroone, \oneone, \onetwo\}^{\bag{x}}$
  we define its precolouring $\hat{s} \in \{\zero, \oneone, \onetwo\}^{\bag{x}}$
  as
  \begin{align*}
  \hat{s}(v) &= s(v) &\textrm{if }s(v) \in \{\oneone, \onetwo\} \\
  \hat{s}(v) &= \zero &\textrm{if }s(v) \in \{\zerozero, \zeroone\}
  \end{align*}
   For a precolouring $\hat{s}$ (or a colouring $s$) and a set
  $T \subseteq \hat{s}^{-1}(\zero)$
  we define the colouring $s[T]$
  \begin{align*}
  s[T](v) &= \hat{s}(v)  & \textrm{if }\hat{s}(v) \in \{\oneone, \onetwo\} \\
  s[T](v) &= \zeroone & \textrm{if }v \in T \\
  s[T](v) &= \zerozero & \textrm{if }v \in \hat{s}^{-1}(\zero) \setminus T
  \end{align*}
  We can now write a recursion formula for the join bags.
  \begin{align*}
	A_x(\iterxk,\iterW,m,d,s) &= \sum_{\iterxk_1 + \iterxk_2 = \iterxk + |s^{-1}(\{\zerozero,\zeroone\})|} \ \  \sum_{\iterW_1 + \iterW_2 = \iterW} \ \ \sum_{m_1+m_2 = m} \ \ \sum_{d_1+d_2=d}\ \ 
  \sum_{T_1,T_2 \subseteq s^{-1}(\{\zerozero,\zeroone\})} \\
    &\qquad\qquad
                [T_1 \cup T_2 = s^{-1}(\zeroone)][T_1 \cap T_2 = \emptyset]
                A_y(\iterxk_1,\iterW_1,m_1,d_1,s[T_1]) A_z(\iterxk_2,\iterW_2,m_2,d_2,s[T_2])
  \end{align*}
  To achieve the colouring $s$, the precolourings of the children have to be the same.
  Moreover, a vertex $v \in R$ has degree one only if it has degree one in exactly one
  of the children bags. Thus the sets of vertices coloured $\zeroone$ in children have
  to be disjoint and sum up to $s^{-1}(\zeroone)$.
	Since vertices coloured $\zeroj$ in $\bag{x}$ are accounted for in both 
	tables of the children, we add their contribution to the 
	accumulator $\iterxk$.

  To compute the recursion formula
  efficiently we need to use the fast subset convolution.
  For accumulators $\iterxk,\iterW,m,d$ and a precolouring $\hat{s}$
  we define the following functions on subsets of $\hat{s}^{-1}(\zero)$:
  \begin{align*}
  f^{\iterxk,\iterW,m,d,\hat{s}}(T) &= A_y(\iterxk,\iterW,m,d,s[T]), \\
  g^{\iterxk,\iterW,m,d,\hat{s}}(T) &= A_z(\iterxk,\iterW,m,d,s[T]).
  \end{align*}
  Now note that
	\begin{align*}
  A_x(\iterxk,\iterW,m,d,s) &= \sum_{\iterxk_1 + \iterxk_2 = \iterxk + |s^{-1}(\{\zerozero,\zeroone\})|} \ \  \sum_{\iterW_1 + \iterW_2 = \iterW} \ \ \sum_{m_1+m_2 = m} \ \ \sum_{d_1+d_2=d} \\
                               &\qquad\qquad 
         (f^{\iterxk_1,\iterW_1,m_1,d_1,\hat{s}} \fsc g^{\iterxk_2,\iterW_2,m_2,d_2,\hat{s}})(s^{-1}(\zeroone)).
  \end{align*}
  By Theorem \ref{thm:fsc},
  for fixed accumulators $\iterxk_1,\iterW_1,m_1,d_1,\iterxk_2,\iterW_2,m_2,d_2$ and a precolouring $\hat{s}$
  the term
     $$(f^{\iterxk_1,\iterW_1,m_1,d_1,\hat{s}} \fsc g^{\iterxk_2,\iterW_2,m_2,d_2,\hat{s}})(s^{-1}(\zeroone))$$
   can be computed in time $2^{|\hat{s}^{-1}(\zero)|} |\hat{s}^{-1}(\zero)|^{O(1)}$
   at once for all colourings $s$ with precolouring $\hat{s}$.
   Thus, the total time consumed by the evaluation of $A_x$ is bounded by
   $$|V|^{O(1)} \sum_{\hat{s} \in \{\zero,\oneone,\onetwo\}^{\bag{x}}} 2^{|\hat{s}^{-1}(\zero)|}
    = 4^{|\bag{x}|} |V|^{O(1)}.$$
\end{itemize}
It is easy to see that the above recurrence leads to a dynamic programming 
algorithm that
computes the parity of $|\objsx_\targetW^\objsxk|$ for all values of $\targetW$ and $\objsxk$
in $4^t |V|^{O(1)}$ time,
since $|\objsx_\targetW^\objsxk|=\sum_{d\geq 2}A_\rootv(\objsxk,\targetW,|V|-1,d,\emptyset)$.
Moreover, as we count the parities and not the numbers $A_x$ themselves, 
all arithmetical operations (in particular the ring operations
    in the fast subset convolution) can be done in constant time.
As discussed before, knowing in $\Z_2$ the values of $|\objsx_\targetW^\objsxk|$ for $0 \leq \objsxk \leq |V|-1$ we can compute all values of $|\objs_\targetW^\targetk|$ and $|\sols_\targetW^\targetk|$ modulo $2$.
Thus, the proof of Theorem~\ref{thm:app:exactleaf} is finished. 
\end{proof}

\subsubsection{\exactoutbranching}

\defproblemu{\exactoutbranching}
{A directed graph $D = (V,A)$ and an integer $\tarnum$, and a root $r \in V$.}
{Does there exist a spanning tree of $D$ with all edges directed away from
the root with exactly $\tarnum$ leaves?}

\begin{theorem}\label{thm:app:exactoutbranching}
 There exists a Monte-Carlo algorithm that given a tree decomposition of width $t$
 solves \exactoutbranching{} in $6^t |V|^{O(1)}$ time.
 The algorithm cannot give false positives and may give false negatives with probability at most $1/2$.
\end{theorem}

\begin{proof}
We use the Cut\&Count technique in a very similar manner as for the \exactleaf problem.
Our solutions and solution candidates are subsets of arcs, thus we
take $\univ = A$ and generate random weight function $\weight:\univ \to \{1,2,\ldots,\maxw\}$,
where $\maxw = 2|\univ| = 2|A|$. We set $\anchor = r$.

\noindent {\bf{The Cut part.}} For integers $\targetW$ and $\targetk$ we define:
\begin{enumerate}
\item $\cand_\targetW^\targetk$ to be the family of sets $X \subseteq A$, such that
$\weight(X) = \targetW$, $\indeg_{G[X]}(\anchor) = 0$ and $\indeg_{G[X]}(v) = 1$ if $v \neq \anchor$,
  and $G[X]$ contains exactly $\targetk$ vertices of outdegree zero.
\item $\sols_\targetW^\targetk$ to be the family of sets $X \in \cand_\targetW^\targetk$,
  such that $G[X]$ is weakly connected;
\item $\objs_\targetW^\targetk$ to be the family of pairs $(X,(X_1,X_2))$, where
$X \in \sols_\targetW^\targetk$ and $(X_1,X_2)$ is a consistent cut of $G[X]$ with $\anchor \in X_1$.
\end{enumerate}
The condition that for $X \in \sols_\targetW^\targetk$ the graph $G[X]$ is weakly connected, together
with the condition on the indegrees of vertices gives us that each $X \in \sols_\targetW^\targetk$ is of size $|V|-1$ and induces an spanning tree of $G$,
     rooted in $\anchor$, with all edges directed away from the root.

As in the case of \exactleaf, we use the superscript $\targetk$ in the definitions,
since we perform a similar trick in the Count part and 
we do computations for many values of $\targetk$.

\noindent {\bf{The Count part.}}
To use Algorithm \ref{alg:cutandcount} we need to formally prove that for any $\targetW$ and $\targetk$ we have $|\sols_\targetW^\targetk| \equiv |\objs_\targetW^\targetk|$.
Similarly as in the case of \steinertree
	we note that
  by Lemma~\ref{lem:evencancel} for each $X \in \cand_\targetW^\targetk$
  there exist $2^{\conncomp(G[X])-1}$
  consistent cuts of $G[X]$, and the claim follows
  (here $\conncomp(G[X])$ denotes the number of weakly connected components of $G[X]$).

To finish the proof we need to show now to compute $|\objs_\targetW^\targetk|$ modulo
$2$ in time $6^t n^{O(1)}$.
A straightforward dynamic programming algorithm would lead to $8^t n^{O(1)}$ time complexity,
thus again we need to be more careful.

Let us define the set $\objsx_\targetW^\objsxk$ to be a family of triples $(X,R,(Y_1,Y_2))$
such that
\begin{enumerate}
\item $X \in \bigcup_{\targetk=0}^{|V|-1} \cand_\targetW^\targetk$,
(i.e., we do not impose any constraint on the number of outdegree zero vertices in $G[X]$),
\item $R \subseteq V$ and $|R| = \objsxk$,
\item each vertex $v \in R$ has outdegree zero in $G[X]$,
\item $(Y_1,Y_2)$ is a consistent cut of $G(V \setminus R, X)$ (defined as in the previous subsection) with $\anchor \notin Y_2$.
\end{enumerate}
Note that, unlike the \exactleaf case, we do not need to require that the vertices from $R$ are not connected by edges from $X$,
as this is guaranteed by the outdegree condition. Moreover, we allow $\anchor \in R$.

Informally speaking, there are two differences between $\objs_\targetW^\targetk$ and $\objsx_\targetW^\objsxk$.
First, instead of requiring a prescribed number of vertices of outdegree zero,
we distinguish a fixed number of vertices that have to be of outdegree zero, and we do not care about the outdegrees of the other vertices.
Second, we consider only consistent cuts of $G(V \setminus R, X)$, not whole $G[X]$.

We first note that (again) the second difference is only apparent. Let $X \in \cand_\targetW^\targetk$ and $R \subseteq V$
be as in the definition of $\objsx_\targetW^\objsxk$, i.e., $|R| = \objsxk$, each $v \in R$ is of outdegree zero in $G[X]$.
If $(X,(X_1,X_2)) \in \objs_\targetW^\targetk$, then the cut $(Y_1,Y_2)$ where $Y_j = X_j \setminus R$ is consistent with $G(V \setminus R, X)$
and $(X,R,(Y_1,Y_2)) \in \objsx_\targetW^\objsxk$. In the other direction, observe that if $(X,R,(Y_1,Y_2)) \in \objsx_\targetW^\objsxk$, then
there exists exactly one cut $(X_1,X_2)$ of $G[X]$, such that $\anchor \in X_1$, $Y_1 \subseteq X_1$ and $Y_2 \subseteq X_2$, namely $X_1 = Y_1 \cup (N_{G[X]}(Y_1) \cap R) \cup \{\anchor\}$
and $X_2 = Y_2 \cup (N_{G[X]}(Y_2) \cap R)$ (in particular, if $\anchor \in R$, then $\anchor$ is isolated in $G[X]$, and can be put safely to $X_1$).
Thus, in the analysis that follows we can silently assume that $(Y_1,Y_2)$ is in fact a consistent cut of $G[X]$ with $\anchor \in Y_1$.

Let $(X,(X_1,X_2)) \in \objs_\targetW^\targetk$. Note that we have exactly $\binom{\targetk}{\objsxk}$ choices
of the set $R$, such that $(X,R,(X_1 \setminus R,X_2 \setminus R)) \in \objsx_\targetW^\objsxk$, as any choice
of $\objsxk$ vertices of outdegree zero in $G[X]$ can be used as $R$. Thus we have that modulo $2$
$$|\objsx_\targetW^\objsxk| \equiv
\sum_{\targetk=\objsxk}^{|V|-1} \binom{\targetk}{\objsxk} |\sols_\targetW^\targetk| \equiv
\sum_{\targetk=\objsxk}^{|V|-1} \binom{\targetk}{\objsxk} |\objs_\targetW^\targetk|.$$
As in the case of \exactleaf, operating over the field $\Z_2$,
we have obtained a linear operator that transforms a vector
$(|\objs_\targetW^\targetk|)_{\targetk=0}^{|V|-1}$ into a vector
$(|\objsx_\targetW^\objsxk|)_{\objsxk=0}^{|V|-1}$. 
Again the matrix of that operator can be computed in polynomial time and is easily inverted (as it is upper triangular
with ones on the diagonal). Thus we can compute (in $\Z_2$)
all values $(|\objs_\targetW^\targetk|)_{\targetk=0}^{|V|-1}$
knowing all values 
$(|\objsx_\targetW^\objsxk|)_{\objsxk=0}^{|V|-1}$.

To finish the proof we need to describe a procedure
$\overline{\countproc}(\weight,\targetW,\objsxk,\treedecomp)$ that,
given a nice tree decomposition $\treedecomp$, weight funtion $\weight$
and integers $\targetW$ and $\objsxk$
computes $|\objsx_\targetW^\objsxk|$ modulo $2$.
Now we can use dynamic programming on the tree decomposition.

Recall that in the definition of $\objsx_\targetW^\objsxk$ the cut $(Y_1,Y_2)$
was a consistent cut of only $G(V \setminus R, X)$. We make use of this fact
to reduce the size of the table in the dynamic programming, as we do not need to remember
side of the cut for vertices in $R$.

	We follow the notation
  from \steinertree{} example (see Lemma \ref{lem:stein}).
  For every bag $x \in \treedecomp$ of the tree decomposition,
  integers $0 \leq \iterxk \leq |V|$, $0 \leq \iterW \leq \maxw|A|$,
  and $s \in \{\oneone,\onetwo,\zero\}^{\bag{x}}$, $s_\textrm{in} \in \{\zero,\one\}^{\bag{x}}$
  (called the colourings) define
  \begin{align*}
    \cand_x(\iterxk,\iterW) &= \Big\{(X,R) \; \big{|}\;
      X \subseteq \subarcs{x}\ \wedge\ R \subseteq \subbags{x}\ \wedge\ |R| = \iterxk\ \wedge\
      \weight(X)=\iterW \\
     &\qquad\ \wedge\ (v \in \subbags{x} \setminus \bag{x} \Rightarrow \indeg_{G[X]}(v) = [v \neq \anchor])
        \ \wedge\ (v \in \bag{x} \Rightarrow \indeg_{G[X]}(v) \leq [v \neq \anchor]) \\
      &\qquad   \wedge\ (v \in R \Rightarrow \outdeg_{G[X]}(v) = 0)
        \Big\}\\
    \objs_x(\iterxk,\iterW) &= \Big\{(X,R,(Y_1,Y_2)) \big{|}\; 
      (X,R) \in \cand_x(\iterxk,\iterW)\ \wedge\ \anchor \notin Y_2 \\
        &\qquad \wedge\ (Y_1,Y_2)
      \textrm{ is a consistent cut of }G(\subbags{x} \setminus R, X)
			\Big\} \\
    A_x(\iterxk,\iterW,s,s_\textrm{in}) &= \Big| \Big\{(X,R,(Y_1,Y_2)) \in \objs_x(\iterxk,\iterW)
      \big{|}\; \\
      &\qquad (s(v) = \onej \Rightarrow v \in Y_j)\ \wedge\
      (s(v) = \zero \Rightarrow v \in R)\ \wedge\ (\forall_{v\in\bag{x}} s_\textrm{in}(v) = \indeg_{G[X]}(v))
         \Big\} \Big|
  \end{align*}
  Here $s(v) = \zero$ denotes that $v \in R$,
  whereas $s(v) = \onej$ denotes that $v \in Y_j$ (and thus $v \notin R$).
  The value $s_\textrm{in}(v)$ denotes the indegree of $v$ in $G[X]$.
  The accumulators $\iterxk$ and $\iterW$ keep track of the size of $R$ and the weight of $X$, respectively.
  Hence $A_x(\iterxk,\iterW,s,s_\textrm{in})$ reflects the number of
  partial objects from $\objsx$ 
  with fixed size of $R$, weight of $X$ and interface on vertices from $\bag{x}$.

  The algorithm computes $A_x(\iterxk,\iterW,s,s_\textrm{in})$
  for all bags $x \in T$ in a bottom-up fashion
  for all reasonable values of $\iterxk$, $\iterW$ and colourings $s$, $s_\textrm{in}$.
  We now give the recurrence for $A_x(\iterxk,\iterW,s,s_\textrm{in})$
  that is used by the dynamic programming
  algorithm.
  In order to simplify notation we denote by $v$ the vertex introduced and
  contained in an introduce bag, by $(u,v)$ the arc introduced in an
  introduce edge bag,
  and by $y,z$ for the left and right child of $x$ in $\treedecomp$ if present.
\begin{itemize}
\item \textbf{Leaf bag}:
  \[ A_x(0,0,\emptyset,\emptyset) = 1 \]
\item \textbf{Introduce vertex bag}:
  \begin{align*}
    A_x(\iterxk,\iterW,s[v \to \alpha], s_\textrm{in}[v \to \one]) &= 0 \\
    A_x(\iterxk,\iterW,d,s[v\to\zero], s_\textrm{in}[v\to\zero]) &= A_y(\iterxk-1,\iterW,s,s_\textrm{in}) \\
    A_x(\iterxk,\iterW,d,s[v\to\oneone], s_\textrm{in}[v\to\zero]) &= A_y(\iterxk,\iterW,s,s_\textrm{in}) \\
    A_x(\iterxk,\iterW,d,s[v\to\onetwo], s_\textrm{in}[v\to\zero]) &= [v \neq \anchor]A_y(\iterxk,\iterW,s,s_\textrm{in})
  \end{align*}
  The new vertex has indegree zero and $\anchor$ cannot be put into $Y_2$.
\item \textbf{Introduce edge bag}:
\begin{align*}
 A_x(\iterxk,\iterW,s,s_\textrm{in}) &= A_y(\iterxk,\iterW,s,s_\textrm{in}) + A_y(\iterxk,\iterW-\weight((u,v)),s,s_\textrm{in}[v\to \zero]) &\\
                                 &\quad\textrm{if }(s(u)=s(v)=\onej \vee (s(u) = \onej \wedge s(v) = \zero))\ \wedge\ v \neq \anchor\ \wedge\ s_\textrm{in}(v) = \one &\\
 A_x(\iterxk,\iterW,s,s_\textrm{in}) &= A_y(\iterxk,\iterW,s,s_\textrm{in}) &\textrm{otherwise}
  \end{align*}
  Here we consider adding the arc $(u,v)$ to $X$. First, we need that $v \neq \anchor$.
  Second, we need that $u,v \in Y_j$ or $u \in Y_j$ and $v \in R$. Moreover,
  we need to update the indegree of $v$ and the accumulator keeping the weight of $X$.
\item \textbf{Forget bag}:
\begin{align*}
  A_x(\iterxk,\iterW,s,s_\textrm{in}) &= \sum_{\alpha \in \{\zero,\oneone\}} A_y(\iterxk,\iterW,s[v \to\alpha],s_\textrm{in}[v\to\zero]) &\textrm{if }v = \anchor\\
  A_x(\iterxk,\iterW,s,s_\textrm{in}) &= \sum_{\alpha \in \{\zero,\oneone,\onetwo\}} A_y(\iterxk,\iterW,s[v\to\alpha],s_\textrm{in}[v\to\one])&\textrm{otherwise}
  \end{align*}
  If we forget $v=\anchor$, we require that its indegree is zero and $v \notin Y_2$.
  Otherwise, we require that the indegree of the forgotten vertex is one.

\item \textbf{Join bag}:
  Let us define for $T \subseteq \bag{x}$ a colouring $s_\textrm{in}[T]$ as $s_\textrm{in}[T](v) = \one$ if $v \in T$ and $s_\textrm{in}[T](v) = \zero$ otherwise. Then
  \begin{align*}
	A_x(\iterxk,\iterW,s,s_\textrm{in}) &= \sum_{\iterxk_1 + \iterxk_2 = \iterxk + |s^{-1}(\zero)|} \ \  \sum_{\iterW_1 + \iterW_2 = \iterW} \ \  
  \sum_{T_1,T_2 \subseteq \bag{x}} \\
    &\qquad\qquad
                [T_1 \cup T_2 = s_\textrm{in}^{-1}(\one)][T_1 \cap T_2 = \emptyset]
                A_y(\iterxk_1,\iterW_1,s,s_\textrm{in}[T_1]) A_z(\iterxk_2,\iterW_2,s,s_\textrm{in}[T_2])
  \end{align*}
  The colourings $s$ in children need to be the same, whereas the colourings $s_\textrm{in}$ in the children
  need to sum up to the colouring $s_\textrm{in}$ in the bag $x$, i.e., a vertex $v$ has indegree one
  only if it has indegree one in exactly one of the children bags.
	Since vertices coloured $\zero$ in $\bag{x}$ are accounted for in both 
	tables of the children, we add their contribution to the 
	accumulator $\iterxk$.

  To compute the recursion formula
  efficiently we need to use fast subset convolution.
  For accumulators $\iterxk,\iterW$ and a colouring $s$
  we define the following functions on subsets $\bag{x}$
  \begin{align*}
  f^{\iterxk,\iterW,s}(T) &= A_y(\iterxk,\iterW,s,s_\textrm{in}[T]), \\
  g^{\iterxk,\iterW,s}(T) &= A_z(\iterxk,\iterW,s,s_\textrm{in}[T]).
  \end{align*}
  Now note that
	$$
	A_x(\iterxk,\iterW,s,s_\textrm{in}) = \sum_{\iterxk_1 + \iterxk_2 = \iterxk + |s^{-1}(\zero)|} \ \  \sum_{\iterW_1 + \iterW_2 = \iterW}
         (f^{\iterxk_1,\iterW_1,s} \fsc g^{\iterxk_2,\iterW_2,s})(s_\textrm{in}^{-1}(\one))$$
  By Theorem \ref{thm:fsc},
  for fixed accumulators $\iterxk_1,\iterW_1,\iterxk_2,\iterW_2$ and a colouring $s$
  the term
     $$(f^{\iterxk_1,\iterW_1,s} \fsc g^{\iterxk_2,\iterW_2,s})(s_\textrm{in}^{-1}(\one))$$
   can be computed in time $2^t t^{O(1)}$
   at once for all colourings $s_\textrm{in}$.
   Thus, the total time consumed by the evaluation of $A_x$ is bounded by $6^t n^{O(1)}$.
\end{itemize}
It is easy to see that the above recurrence leads to a dynamic programming 
algorithm that
computes the parity of $|\objsx_\targetW^\objsxk|$ for all values of $\targetW$ and $\objsxk$
in $6^t |V|^{O(1)}$ time,
since $|\objsx_\targetW^\objsxk|=A_\rootv(\objsxk,\targetW,\emptyset,\emptyset)$.
Moreover, as we count the parities and not the numbers $A_x$ themselves, 
all arithmetical operations (in particular the ring operations
    in the fast subset convolution) can be done in constant time.
As discussed before, knowing in $\Z_2$ the values of $|\objsx_\targetW^\objsxk|$ for $0 \leq \objsxk \leq |V|-1$ we can compute all values of $|\objs_\targetW^\targetk|$ and $|\sols_\targetW^\targetk|$ modulo $2$.
Thus, the proof of Theorem~\ref{thm:app:exactoutbranching} is finished. 
\end{proof}

\subsection{\maxspantree{}}
\label{sec:app:mfdst}
In this subsection we solve a bit more general version of \maxspantree{}
where the tree in question needs to contain exactly the prescribed number
of vertices of full degree.

\defproblemu{\exactspantree}
{An undirected graph $G = (V,E)$ and an integer $\tarnum$.}
{Does there exist a spanning tree $T$ of $G$ for which there are exactly 
$k$ vertices satisfying $\deg_G(v) = \deg_T(v)$?}

\begin{theorem}\label{thm:app:mfdst}
 There exists a Monte-Carlo algorithm that given a tree decomposition
 of width $t$ solves the \exactspantree problem in $4^t |V|^{O(1)}$ time.
 The algorithm cannot give false positives and may give false negatives with probability at most $1/2$.
\end{theorem}

\begin{proof}
We use the Cut\&Count technique.
As a universe we take the set of edges $\univ = E$.
As usual we assume that we are given a weight function $\weight: \univ \rightarrow \{1,...,N\}$, where $N = 2|\univ| = 2|E|$.
Let $\anchor$ be an arbitrary vertex.

\noindent {\bf{The Cut part.}}
For an integer $\targetW$ we define:
  \begin{enumerate}
  \item $\cand_\targetW$ to be the family 
  of solution candidates of weight $\targetW$, that is subsets of exactly $|V|-1$ edges $X \subseteq E$, $|X|=|V|-1$, $\weight(X)=\targetW$,
  such that there are exactly $k$ vertices $v$ satisfying $\deg_G(v) = \deg_{G[X]}(v)$,
  \item $\sols_\targetW$ to be the set of solutions,
  that is solution candidates $X \in \cand_\targetW$ such that the graph $G[X]$ is connected;
  \item $\objs_\targetW$ to be the family of pairs $(X, (X_1,X_2))$,
  where $X \in \cand_\targetW$, 
  $\anchor \in X_1$, and $(X_1,X_2)$ is a consistent cut of the graph $G[X]$.
  \end{enumerate}
Observe that for $X \in \cand_\targetW$ the graph $G[X]$ is connected iff $G[X]$ is a tree, since $|X|=|V|-1$.
Thus $\sols_\targetW$ is indeed a family of solutions of weight $\targetW$.

\noindent {\bf{The Count part.}} 
To use Algorithm \ref{alg:cutandcount} we need to formally prove that 
for any $\targetW$ we have $|\sols_\targetW| \equiv |\objs_\targetW|$.
Similarly as in the case of \steinertree we note that
by Lemma~\ref{lem:evencancel} for each $X \in \cand_\targetW$
there exist $2^{\conncomp(G[X])-1}$ consistent cuts of the graph $G[X]$,
and the claim follows.

To finish the proof we need to show how to compute $|\objs_\targetW|$ modulo
$2$ in time $4^t n^{O(1)}$ using dynamic programming.
In a state we store 
the number of vertices that are already forgotten and have all their incident edges chosen,
the number of already chosen edges,
the sum of weights of already chosen edges,
and moreover for each vertex of the bag we remember the side of the cut and
one bit of information whether there exists some already introduced edge
incident with that vertex that was not chosen.
Formal definition follows.

We follow the notation from \steinertree{} example (see Lemma \ref{lem:stein}).
For a bag $x \in \treedecomp$ of the tree decomposition,
integers $0 \leq \iterk \leq |V|$, $0 \leq \iterm < |V|$, $0 \leq \iterW \leq \maxw(|V|-1)$,
$s_{\cut} \in \{\one, \two\}^{\bag{x}}$ and  $s_{\deg} \in \{\zero, \one\}^{\bag{x}}$
(called the colouring) define
\begin{align*}
  \cand_x(\iterk, \iterm, \iterW) & = \Big{\{} X \subseteq \subedges{x} \; \big{|}\; |X| = \iterm\ \wedge\ \weight(X)=\iterW 
          \ \wedge\ |\{v \in \subbags{x} \setminus \bag{x}\ :\ \deg_{G_x}(v) = \deg_X(v) \}| = \iterk \Big{\}} \\
  \objs_x(\iterk, \iterm, \iterW) & = \Big{\{} (X,(X_1,X_2))  \big{|}\; X \in \cand_x(\iterk, \iterm, \iterW)\
 \wedge\ \anchor \in X_1\ \wedge\ (X_1,X_2) \textrm{ is a consistent cut of }(V_x,X) \Big{\}} \\
    A_x(\iterk, \iterm, \iterW, s_{\cut}, s_{\deg}) & = \Big{|} \Big{\{} (X,(X_1,X_2)) 
    \in \objs_x(\iterk, \iterm, \iterW) \big{|}  (v \in X_j \cap \bag{x} \Rightarrow s_{\cut}(v) = j) \\
     & \qquad \wedge\ (s_{\deg}(v) = \zero \Rightarrow \deg_{G_x}(v) = \deg_X(v))\ \wedge\ (s_{\deg}(v) = \one \Rightarrow \deg_{G_x}(v) > \deg_X(v)) \Big{\}} \Big{|}
\end{align*}

By $s_{\cut}(v) = j$ we denote $v \in X_j$,
whereas $s_{\deg}(v)$ is equal to one iff there exists
an edge in $\subedges{x} \setminus X$ that is incident with $v$.
Hence $A_x(\iterk,\iterm,\iterW,s_{\cut},s_{\deg})$ is the number of pairs from $\objs_x(\iterk, \iterm, \iterW)$ 
with a fixed interface on vertices from $\bag{x}$.

The algorithm computes $A_x(\iterk, \iterm, \iterW, s_{\cut}, s_{\deg})$ for all bags $x \in\treedecomp$ in a bottom-up fashion
for all reasonable values of $\iterk$, $\iterm$, $\iterW$, $s_{\cut}$ and $s_{\deg}$.
We now give the recurrence for $A_x(\iterk, \iterm, \iterW, s_{\cut}, s_{\deg})$ that
is used by the dynamic programming algorithm.
As usual $v$ denotes the vertex introduced and
contained in an introduce bag, $uv$ the edge introduced in an
introduce edge bag,
and $y,z$ the left and right child of $x$ in $\treedecomp$ if present.
\begin{itemize}
\item \textbf{Leaf bag}:
\[ A_x(0,0,0,\emptyset,\emptyset) = 1 \]
\item \textbf{Introduce vertex bag}:
	\begin{align*}
     A_x(\iterk, \iterm, \iterW, s_{\cut} [v\to \one], s_{\deg} [v\to \zero]) & = A_y(\iterk, \iterm, \iterW, s_{\cut}, s_{\deg})  \\
     A_x(\iterk, \iterm, \iterW, s_{\cut} [v\to \two], s_{\deg} [v\to \zero]) & = [v \not= \anchor]A_y(\iterk, \iterm, \iterW, s_{\cut}, s_{\deg})  \\
     A_x(\iterk, \iterm, \iterW, s_{\cut} [v\to \alpha], s_{\deg} [v\to \one]) & = 0
\end{align*}
We make sure that $\anchor$ belongs to $X_1$.
\item \textbf{Introduce edge bag}:
	\begin{align*}
     A_x(\iterk, \iterm, \iterW, s_{\cut}, s_{\deg})  & = [s_{\cut}(u) = s_{\cut}(v)] A_y(\iterk, \iterm-1, \iterW-\weight(uv), s_{\cut}, s_{\deg}) \\
                                                   & \qquad + \sum_{\alpha_u,\alpha_v \in \{\zero,\one\}} A_y(\iterk, \iterm, \iterW, s_{\cut}, s_{\deg}[u\to\alpha_u, v\to \alpha_v]) \\
     &\qquad \textrm{ if } s_{\deg}(u) = s_{\deg}(v) = \one\\
     A_x(\iterk, \iterm, \iterW, s_{\cut}, s_{\deg})  & = [s_{\cut}(u) = s_{\cut}(v)] A_y(\iterk, \iterm-1, \iterW-\weight(uv), s_{\cut}, s_{\deg}) \\
 & \qquad \textrm{ otherwise}
  \end{align*}
If $s_{\deg}(u) = s_{\deg}(v) = \one$ then we have an option of not taking the edge $uv$ to the set $X$. In this case, the previous values of $s_{\deg}(u)$ and $s_{\deg}(v)$ can be arbitrary.
\item \textbf{Forget bag}:

	\begin{align*}
     A_x(\iterk, \iterm, \iterW, s_{\cut}, s_{\deg}) & = \sum_{j \in \{\one, \two\}} \quad \sum_{\alpha \in \{\zero,\one\}} 
             A_y(\iterk-[\alpha = \zero], \iterm, \iterW, s_{\cut} [v\to j], s_{\deg} [v\to \alpha]) 
  \end{align*}
If the vertex $v$ had all incident edges chosen ($\alpha = \zero$) then we update the accumulator $\iterk$.

\item \textbf{Join bag}:

The only valid combinations to achieve the colouring $s_{\cut}$ is to have the 
same colouring in both children. 
However we have $s_{\deg}(v)=\one$ in $x$ if and only if $s_{\deg}(v)=\one$
in $y$ or in $z$. Hence we use a covering product.
We somewhat abuse the notation and identify a function $s_{\deg}$ with a subset
$s_{\deg}^{-1}(\one) \subseteq \bag{x}$.
Let us define
\begin{align*}
	f^{\iterk,\iterm,\iterW,s_{\cut}}(s_{\deg}) &= A_y(\iterk,\iterm,\iterW,s_{\cut},s_{\deg}) \\
	g^{\iterk,\iterm,\iterW,s_{\cut}}(s_{\deg}) &= A_z(\iterk,\iterm,\iterW,s_{\cut},s_{\deg}) \\
	h^{\iterk,\iterm,\iterW,s_{\cut}}(s_{\deg}) &= \sum_{\iterk_1 + \iterk_2 = \iterk}\ \ 
          \sum_{\iterm_1 + \iterm_2 = \iterm}\ \ 
          \sum_{\iterW_1 + \iterW_2 = \iterW}
            (f^{\iterk_1,\iterm_1,\iterW_1,s_{\cut}} \fsccov g^{\iterk_2,\iterm_2,\iterW_2,s_{\cut}})(s_{\deg})
\end{align*}
Consequently we have
\begin{align*}
  A_x(\iterk,\iterm,\iterW,s_{\cut},s_{\deg}) = h^{\iterk,\iterm,\iterW,s_{\cut}}(s_{\deg})
\end{align*}
\end{itemize}

It is easy to see that we can combine the 
above recurrence with dynamic programming. 
For each of the $2^t |V|^{O(1)}$ argument values of $\iterk$,$\iterm$,$\iterW$ and $s_{\cut}$
the covering product $\fsccov$ can be computed in $2^t |V|^{O(1)}$ by Theorem \ref{thm:fsc}.
Note that as we perform all calculations modulo $2$, we take only constant
time to perform any arithmetic operation.

Since $|\objs_\targetW| = A_\rootv(k, |V|-1, \targetW, \emptyset, \emptyset)$
the above recurrence leads to a dynamic programming algorithm that
computes the parity of $|\objs_\targetW|$ (and thus of $|\sols_\targetW|$ as well) for all reasonable values of $\targetW$
in $4^t |V|^{O(1)}$ time.
Consequently we finish the proof of Theorem~\ref{thm:app:mfdst}.
\end{proof}

\subsection{\gmtsp}
\label{sec:app:gmtsp}
\defproblemu{\gmtsp}
{An undirected graph $G = (V,E)$ and an integer $\tarnum$.}
{Does there exist a closed walk (possibly repeating edges and vertices) of length
at most $\tarnum$ that visits each vertex of the graph at least once?}

\begin{theorem}\label{thm:app:gmtsp}
 There exists a Monte-Carlo algorithm that given a tree decomposition
 of width $t$ solves the \gmtsp problem in $4^t |V|^{O(1)}$ time.
 The algorithm cannot give false positives and may give false negatives with probability at most $1/2$.
\end{theorem}

\begin{proof}
We use the Cut\&Count technique.
Observe that we may assume that $G$ is connected and $k \le 2(|V|-1)$ because taking twice all edges of any spanning tree gives a solution.

Since we want to distinguish the case when we take an edge once or twice to the solution,
as a universe we take the set $\univ = E \times \{1, 2\}$,
where we use $(e,1)$ if an edge is chosen once
and $(e,2)$ in case we use $e$ twice.
As usual we assume that we are given a weight function $\weight: \univ \rightarrow \{1,...,N\}$, where $N = 2|\univ| = 4|E|$.
Let $\anchor$ be an arbitrary vertex.

\noindent {\bf{The Cut part.}}
For integers $\iterk$ and $\targetW$ we define:
  \begin{enumerate}
  \item $\cand_\targetW^{\iterk}$ to be the family 
  of solution candidates of size $\iterk$ and weight $\targetW$, 
  that is functions $\phi \in \{0,1,2\}^E$ where $\sum_{e \in E} \phi(e) = \iterk$, $\sum_{e \in E, \phi(e) > 0} \weight((e,\phi(e)))=\targetW$,
  such that for each vertex $v \in V$ its degree is even, i.e., $|\{uv \in E : \phi(e)=1\}| \equiv 0$;
  \item $\sols_\targetW^{\iterk}$ to be the set of solutions,
  that is solution candidates $\phi \in \cand_\targetW^{\iterk}$ such that the graph $G[\phi^{-1}(\{1,2\})]$ is connected;
  \item $\objs_\targetW^{\iterk}$ to be the family of pairs $(\phi, (X_1,X_2))$,
  where $\phi \in \cand_\targetW^{\iterk}$, 
  $\anchor \in X_1$, and $(X_1,X_2)$ is a consistent cut of the graph $G[\phi^{-1}(\{1,2\})]$.
  \end{enumerate}
We want to check whether there exist numbers $\targetW$ and $\iterk \leq \targetk$ such that
$\sols_\targetW^\iterk \not= \emptyset$.

\noindent {\bf{The Count part.}} 
To use Algorithm \ref{alg:cutandcount} we need to formally prove that 
for any $\iterk$ and $\targetW$ we have $|\sols_\targetW^\iterk| \equiv |\objs_\targetW^\iterk|$.
Similarly as in the case of \steinertree we note that
by Lemma~\ref{lem:evencancel} for each $\phi \in \cand_\targetW^\iterk$
there exist $2^{\conncomp(G')-1}$ consistent cuts of the graph $G'=G[\phi^{-1}(\{1,2\})]$,
and the claim follows.

To finish the proof we need to show how to compute $|\objs_\targetW^\iterk|$ modulo
$2$ in time $4^t n^{O(1)}$ using dynamic programming.

We follow the notation from \steinertree{} example (see Lemma \ref{lem:stein}).
For a bag $x \in \treedecomp$ of the tree decomposition,
integers $0 \leq \iterk \leq k$, $0 \leq \iterW \leq k\maxw$,
$s_{\cut} \in \{\one, \two\}^{\bag{x}}$ and  $s_{\deg} \in \{\zero, \one\}^{\bag{x}}$
(called the colourings) define
\begin{align*}
  \cand_x(\iterk, \iterW) & = \Big{\{} \phi \in \{0,1,2\}^{\subedges{x}} \; \big{|}\; \sum_{e \in \subedges{x}} \phi(e) = \iterk\ \wedge\
  \sum_{e \in E_x, \phi(e) > 0} \weight((e,\phi(e)))=\iterW \\
    & \qquad \wedge\ \forall_{v \in \subbags{x} \setminus \bag{x}} |\{uv \in E_x:\phi(uv)=1\}| \mod 2 = 0 \Big{\}} \\
  \objs_x(\iterk, \iterW) & = \Big{\{} (\phi,(X_1,X_2))  \big{|}\; \phi \in \cand_x(\iterk, \iterW)
 \ \wedge\ \anchor \in X_1\ \wedge\ (X_1,X_2) \textrm{ is a consistent cut of }G_x[\phi^{-1}(\{1,2\})] \Big{\}} \\
    A_x(\iterk, \iterW, s_{\cut}, s_{\deg}) & = \Big{|} \Big{\{} (\phi,(X_1,X_2)) 
    \in \objs_x(\iterk, \iterW) \big{|} (s_{\cut}(v) = j \Rightarrow v \in X_j) \\
     & \qquad \wedge\ \forall_{v \in \bag{x}} s_{\deg}(v) \equiv |\{uv \in \subedges{x} : \phi(uv) = 1\}| \Big{\}} \Big{|}
\end{align*}

The accumulators $\iterk$ and $\iterW$ keep track of the number of 
edges chosen (with multiplicities) and the appropriate sum of weights.
In the sequence $s_\cut$ we store the information about the side of the cut of each vertex from $B_x$,
whereas $s_{\deg}$ is used to remember whether a vertex has an odd or even degree.
Hence $A_x(\iterk,\iterW,s_\cut,s_{\deg})$ is the number of pairs from $\objs_x(\iterk,\iterW)$ 
with a fixed interface on vertices from $\bag{x}$.

The algorithm computes $A_x(\iterk, \iterW, s_{\cut}, s_{\deg})$ for all bags $x \in\treedecomp$ in a bottom-up fashion
for all reasonable values of $\iterk$, $\iterW$, $s_{\cut}$ and $s_{\deg}$.
We now give the recurrence for $A_x(\iterk, \iterW, s_{\cut}, s_{\deg})$ that
is used by the dynamic programming algorithm.
As usual let $v$ stand for the vertex introduced and
contained in an introduce bag, $uv$ for the edge introduced in an
introduce edge bag,
and $y,z$ for the left and right child of $x$ in $\treedecomp$ if present.
\begin{itemize}
\item \textbf{Leaf bag}:
\[ A_x(0,0,\emptyset,\emptyset) = 1 \]
\item \textbf{Introduce vertex bag}:
	\begin{align*}
     A_x(\iterk, \iterW, s_{\cut}[v\to \one], s_{\deg}[v\to \zero]) & = A_y(\iterk, \iterW, s_{\cut}, s_{\deg})  \\
     A_x(\iterk, \iterW, s_{\cut}[v\to\two], s_{\deg} [v\to \zero]) & = [v \not= \anchor]A_y(\iterk, \iterW, s_{\cut}, s_{\deg})  \\
     A_x(\iterk, \iterW, s_{\cut} [v\to \alpha], s_{\deg} [v\to \one]) & = 0 
\end{align*}
We make sure that $\anchor$ belongs to $X_1$.
\item \textbf{Introduce edge bag}:
	\begin{align*}
     A_x(\iterk, \iterW, s_{\cut}, s_{\deg})  & = A_y(\iterk, \iterW, s_{\cut}, s_{\deg}) + \\
     & \qquad [s_{\cut}(u) = s_{\cut}(v)] A_y(\iterk-1, \iterW-\weight((uv,1)), s_{\cut}, s_{\deg}') + \\
     & \qquad [s_{\cut}(u) = s_{\cut}(v)] A_y(\iterk-2, \iterW-\weight((uv,2)), s_{\cut}, s_{\deg}) 
  \end{align*}
  where by $s_{\deg}'$ we denote $s_{\deg}$ with changed values for $u$ and $v$, formally $$s_{\deg}'=s_{\deg}[u\to 1-s_{\deg}(u),v\to 1-s_{\deg}(v)].$$

  We can either not take an edge or take it once or twice. 
\item \textbf{Forget bag}:

	\begin{align*}
     A_x(\iterk, \iterW, s_{\cut}, s_{\deg}) & = \sum_{j \in \{\one, \two\}}
             A_y(\iterk, \iterW, s_{\cut} [v\to j], s_{\deg} [v\to \zero]) 
  \end{align*}
  We simply check the parity of the vertex which we are about to forget,
  both sides of the cut are allowed.

\item \textbf{Join bag}:

The only valid combinations to achieve the colouring $s_{\cut}$ is to have the 
same colouring in both children. 
However for $s_{\deg}$ we have to calculate the xor product of $s_{\deg}(v)$
for $y$ and $z$.
We somewhat abuse the notation and identify a function $s_{\deg}$ with a subset
$s_{\deg}^{-1}(\one) \subseteq \bag{x}$.
Let us define
\begin{align*}
	f^{\iterk,\iterW,s_{\cut}}(s_{\deg}) &= A_y(\iterk,\iterW,s_{\cut},s_{\deg}) \\
	g^{\iterk,\iterW,s_{\cut}}(s_{\deg}) &= A_z(\iterk,\iterW,s_{\cut},s_{\deg}) \\
	h^{\iterk,\iterW,s_{\cut}}(s_{\deg}) &= \sum_{\iterk_1 + \iterk_2 = \iterk}
          \sum_{\iterW_1 + \iterW_2 = \iterW}
            (f^{\iterk_1,\iterW_1,s_{\cut}} \fscxor g^{\iterk_2,\iterW_2,s_{\cut}})(s_{\deg})
\end{align*}
Consequently we have
\begin{align*}
  A_x(\iterk,\iterW,s_{\cut},s_{\deg}) = h^{\iterk,\iterW,s_{\cut}}(s_{\deg})
\end{align*}
\end{itemize}

It is easy to see that we can combine the 
above recurrence with dynamic programming. 
For each of the $2^t |V|^{O(1)}$ argument values of $\iterk$,$\iterW$ and $s_{\cut}$
the xor product can be computed in $2^t |V|^{O(1)}$ by Theorem \ref{thm:fscxor}.
Note that as we perform all calculations modulo $2$, we take only constant
time to perform any arithmetic operation.

Since for each $\iterk$ we have $|\objs_\targetW^{\iterk}| = A_\rootv(\iterk, \targetW, \emptyset, \emptyset)$
the above recurrence leads to a dynamic programming algorithm that
computes the parity of $|\objs_\targetW^{\iterk}|$ (and thus of $|\sols_\targetW^{\iterk}|$ as well) for all reasonable values of $\targetW$ and $\iterk$
in $4^t |V|^{O(1)}$ time.
Consequently we finish the proof of Theorem~\ref{thm:app:gmtsp}.
\end{proof}

\section{Improvements in FVS, CVC and CFVS parameterized by the solution size}\label{sec:param-improv}

Our technique gives rise to an improvement of several parameterized complexity upper bounds for vertex deletion problems in which the remaining graph has to be of constant treewidth.
These problems are \fvs, \cvertexcover and \cfvs. 
The main idea behind the new results is the combination of the {\em iterative compression} technique, developed by Reed et al. \cite{reed:ic}, and the Cut\&Count technique. 

We begin with the \fvs problem, as it was exhaustively studied by the parameterized complexity community. 
Let us recall that previously best algorithm, due Cao, Chen and Liu, runs in $3.83^k n^{O(1)}$ time \cite{fvs:3.83k}.

\begin{theorem}[Theorem \ref{fvs_theorem}, restated]\label{thm:fvs:fpt}
There exists a Monte-Carlo algorithm for the \fvs problem in a graph with
$n$ vertices in $3^k n^{O(1)}$ time and polynomial space.
The algorithm cannot give false positives and may give false negatives with probability at most $1/2$.
\end{theorem}
\begin{proof}
Let $v_1,v_2,\ldots,v_n$ be an arbitrary ordering of the vertices of the given graph $G=(V,E)$. 
Let us denote $G_i=G[\{v_1,v_2,\ldots,v_i\}]$ for all $1\leq i\leq n$. 
Observe that if $G$ admits a feedback vertex set of size at most $k$, i.e. there is a set $A\subseteq V, |A|\leq k$ such that $G[V\setminus A]$ is a forest, then so do all the graphs $G_i$, because $G_i[\{v_1,v_2,\ldots,v_i\}\setminus A]$ is a forest as well and $|A\cap \{v_1,v_2,\ldots,v_i\}|\leq |A|\leq k$.

We construct feedback vertex sets $A_1,A_2,\ldots,A_n$ of size at most $k$ consecutively in $G_1,G_2,\ldots,G_n=G$. 
If at any step the algorithm finds out that the set we seek does not exist (with high probability), we answer NO. 
We begin with $A_1=\emptyset$, which is a feasible solution in graph $G_1$
(we ignore the trivial case $k = 0$). 
The idea of iterative compression is that when we are to construct the set 
$A_{i+1}$, we can use the previously constructed set $A_i$. 
Let $B_{i+1}=A_i\cup \{v_{i+1}\}$. 
Observe that $B_{i+1}$ is a feedback vertex set in $G_{i+1}$. 
If $|B_{i+1}|\leq k$, then we take $A_{i+1}=B_{i+1}$. 
Thus we are left with the case in which, given a feedback vertex set of size
$k+1$, we need to construct a feedback vertex set of size at most $k$ 
or determine that none such exists. 
Denote this given feedback vertex set by $B$.

As $B$ is a feedback vertex set, the graph induced by the rest of the vertices is a forest. 
Thus we can construct a tree decomposition of the graph $G_{i+1}$ of width at most $k+2$ by creating a tree decomposition of the forest of width $1$ and adding the whole set $B$ to each bag. 
To begin with, we test whether $G_{i+1}$ admits a feedback vertex set of size at most $k$. 
We apply (using the tree decomposition obtained above as the input) 
the dynamic programming algorithm described in Section 
\ref{sec:app:fvs}, running in $3^kn^{O(1)}$ time, 
which tests whether the graph admits a feedback vertex set of size at most $k$. 
Observe that this algorithm, as described in the proof of Theorem \ref{thm:app:fvs}, uses exponential space. 
However, in each step when computing $A_x(a,b,c,w,s)$ the algorithm refers only to values $A_y(a',b',c',w',s')$, where $s'=s$ on the intersection of the domains of $s$ and $s'$. 
In our case the intersection of every two bags of the tree decomposition contains $B$. 
Therefore we can reorder the computation in the following manner: 
for every evaluation $\overline{s}: B\to \{\mathbf{0},\mathbf{1}_1, \mathbf{1}_2\}$ we fix it as the ,,core'' evaluation for every bag in the decomposition and run the algorithm to compute all the values $A_x(a,b,c,w,s)$, where $s|_B=\overline{s}$. 
Such a computation takes polynomial time and space. 
As there are $3^{k+1}$ such possible evaluations $\overline{s}$, the algorithm runs in $3^kn^{O(1)}$ time and in polynomial space. 
We make $n$ independent runs of the algorithm in order to assure that the probability of a false negative is at most $\frac{1}{2^n}$.

Once we have done this, we already tested with high probability whether the desired feedback vertex set exists or not. 
If the answer is negative, we answer NO. 
Otherwise we need to explicitly construct the set $A_{i+1}$ in order to use it in the next step of the iterative compression. 
We make use of the algorithm for \constrainedfvs, given by Theorem \ref{thm:app:fvs}. 
The algorithm considers the vertices of $G_{i+1}$ one by one, building a set $K$ which at the end will be the constructed $A_{i+1}$. 
We begin with $K=\emptyset$ and preserve an invariant that at each step there is a feedback vertex set of size at most $k$ containing the set $K$. 
When considering the vertex $v$, we test in $3^kn^{O(1)}$ time whether the graph admits a constrained feedback vertex set of size at most $k$ with $S=K\cup\{v\}$, making $n$ independent runs of the algorithm given by Theorem \ref{thm:app:fvs} in order to reduce the probability of a false negative to at most $\frac{1}{2^n}$. 
If the answer is positive, we can safely add $v$ to $K$ as we know that there is a feedback vertex set of size at most $k$ containing $K\cup\{v\}$ (recall our algorithms do not return false positives). 
Otherwise we simply proceed to the next vertex. 
The computation terminates when $K$ is already a feedback vertex set or when we have exhausted all vertices. 
Observe that if $G_{i+1}$ admits a feedback vertex set of size at most $k$, this construction will terminate building a feedback vertex set $A_{i+1}$ of size at most $k$ unless there was an error in at least one of the tests. 
If we exhaust all vertices, we answer NO, as an error has occured.
Note that in each run of the algorithm for \constrainedfvs we can reorder the computation in the same way as in the previous paragraph to reduce space usage to polynomial.

Observe that the described algorithm at most $n^2+n$ times makes $n$ independent runs of the algorithm from Theorem \ref{thm:app:fvs} as a subroutine: 
in each of $n$ steps of the iterative compression at most $n+1$ times. 
Each of these groups of runs has a probability of a false negative bounded by $\frac{1}{2^n}$, thus the probability of a false negative is bounded by $\frac{n^2+n}{2^n}$, which is lower than $\frac{1}{2}$ for large enough $n$.
\end{proof}

Now we proceed to the algorithm for \cvertexcover. 
The previously best FPT algorithm is due to Binkele-Raible \cite{cvc-br}, and runs in $2.4882^kn^{O(1)}$ time complexity. 
The following algorithm is also an application of iterative compression, 
however we make use of the connectivity requirement in order to reduce the complexity from $3^k n^{O(1)}$ down to $2^k n^{O(1)}$.

\begin{theorem}\label{thm:cvc:fpt}
There exists a Monte-Carlo algorithm for the \cvertexcover problem in a graph with
$n$ vertices in $2^k n^{O(1)}$ time and polynomial space.
The algorithm cannot give false positives and may give false negatives with probability at most $1/2$.
\end{theorem}
\begin{proof}
Firstly observe that \cvertexcover problem is {\emph{contraction}}--closed.
This means that if a graph $H$ admits a connected vertex cover $A$ of size at most $k$, 
then $H'$ obtained from $H$ by contracting an edge of $H$ (and reducing possible multiedges to simple edges) also admits a connected vertex cover $A'$ of size at most $k$. 
Indeed, the contracted edge $uv$ needs to be covered by $A$, so $u\in A$ or $v\in A$. 
Thus we can construct $A'$ by removing $u$ and $v$ from $A$ and adding the vertex obtained from the contracted edge. 
It can be easily seen that $A'$ is a connected vertex cover of $H'$ of size at most $k$.

Therefore, we can consider a sequence of graphs $G_1,G_2,\ldots,G_{n}=G$ ($G$ is the connected graph given in the input), where $G_i$ is obtained from $G_{i+1}$ by contracting any edge and reducing possible multiedges to simple edges, and $G_1$ is a graph composed of a single vertex. 
The argument from the last paragraph ensures that we can proceed as in the proof of Theorem \ref{thm:fvs:fpt}, namely construct connected vertex covers for $G_1,G_2,\ldots,G_{n}$ consecutively, 
and the only thing we have to to show is how to construct a connected vertex cover of size $k$ in $G_{i+1}$ given a connected vertex cover $A_i$ of size $k$ in $G_i$, or determine that none exists.

Let $G_i$ be constructed from $G_{i+1}$ by contracting an edge $uv$. 
We construct $B$ from $A_i$ by removing the vertex obtained in the contraction (if it was contained in $A_i$) and inserting both $u$ and $v$. 
Observe that $B$ is of size at most $k+2$ and it is a vertex cover of $G_{i+1}$. 
As $V(G_{i+1})\setminus B$ is an independent set, then we can construct a path decomposition of $G_{i+1}$ of width at most $k+2$: for every vertex from $V(G_{i+1})\setminus B$ we introduce a bag, connect the bags in any order and then add the set $B$ to every bag.

Now we are going to test whether $G_{i+1}$ admits a connected vertex cover of size at most $k$. 
We could apply the algorithm from Theorem \ref{thm:app:cvc}. 
As in the proof of Theorem \ref{thm:fvs:fpt}, this dynamic programming algorithm during computation of $A_x(i,w,s)$ also refers only to values $A_y(i',w',s')$ for $s'$ such that $s=s'$ on the intersection of domains of $s$ and $s'$. 
Therefore, similarly as before, we would iterate through all possible evaluations $\overline{s}:B\to \{\mathbf{0},\mathbf{1}_1, \mathbf{1}_2\}$, each time computing all the values $A_x(i,W,s)$ such that $s|_B=\overline{s}$ in polynomial time, thus using only polynomial space in the whole algorithm. 
Unfortunately, the algorithm given by Theorem \ref{thm:app:cvc} runs in $3^kn^{O(1)}$ time.

We can, however, reduce the complexity by bounding the number of reasonable evaluations $\overline{s}:B\to \{\mathbf{0},\mathbf{1}_1, \mathbf{1}_2\}$ by $3^3\cdot 2^{k-1}$. 
$B$ induces in $G_{i+1}$ a graph with consisting of a single large connected
component (coming from $A_i$), and at most two additional vertices. 
Take any spanning tree of the large component and root it at some vertex $r$. 
We present the evaluation $\overline{s}$ in the following manner. 
For the root $r$ and the two additional vertices we choose for $\overline{s}$ any value from $\{\mathbf{0},\mathbf{1}_1, \mathbf{1}_2\}$, giving 
$3^3$ choices in total. 
Now consider the rest of the tree (containing all the remaining vertices from $B$) in a top--down manner. 
Observe that every vertex $v$ from the tree has only two possible evaluation, depending on the evaluation of its parent $u$:
\begin{itemize}
\item if $\overline{s}(u) = \mathbf{0}$, the two possible options are $\mathbf{1}_1, \mathbf{1}_2$, as otherwise the edge connecting $v$ with its parent would not be covered;
\item if $\overline{s}(u) = \mathbf{1}_j$, the two possible options are $\mathbf{0}$ and $\mathbf{1}_j$, as otherwise the evaluation $\overline{s}$ would not describe any consistent cut.
\end{itemize}
Thus each of $k+2$ elements of $B$ has only two options, except from the starting $3$, which have $3$ options each. 
This means we only need to consider $3^3\cdot 2^{k-1}$ possible ,,core'' evaluations $\overline{s}$, which yields an algorithm with running time $2^k n^{O(1)}$, using polynomial space. 
As previously, we make $n$ independent runs of the algorithm in order to reduce the probability of a false negative to at most $\frac{1}{2^n}$.

Once we have tested whether $G_{i+1}$ admits a connected vertex cover of size at most $k$, we can construct it explicitly similarly as in the proof of Theorem \ref{thm:fvs:fpt} using the algorithm for \constrainedcvertexcover. 
We consider vertices one by one, each time determining whether the vertex can be inserted into the constructed connected vertex cover by running the algorithm from Theorem \ref{thm:app:cvc} $n$ times. 
Observe that all these runs can be done in $2^kn^{O(1)}$ time and polynomial space complexity using the same technique as in the testing. 
Thus we succeed in constructing $A_{i+1}$ unless at least one of the tests returns a false negative.

The algorithm makes at most $n^2+n$ groups of $n$ independent runs of algorithm from Theorem \ref{thm:app:cvc}. Therefore the probability of a false negative is bounded by $\frac{n^2+n}{2^n}$ which is less than $\frac{1}{2}$ for large enough $n$.
\end{proof}

Finally, we use a similar technique to obtain an algorithm for \cfvs. 
The previously best FPT algorithm is due to Misra et al.~\cite{DBLP:conf/walcom/MisraPRSS10}, and runs in $46.2^kn^{O(1)}$ time complexity.

\begin{theorem}
There exists a Monte-Carlo algorithm solving the \cfvs problem in a graph with
$n$ vertices in $3^k n^{O(1)}$ time and polynomial space.
The algorithm cannot give false positives and may give false negatives with probability at most $1/2$.
\end{theorem}
\begin{proof}
Similarly as in the proof of Theorem \ref{thm:cvc:fpt}, the \cfvs problem is also {\emph{contraction}}--closed. Consider any graph $H$ and obtain $H'$ by contracting an edge $uv$ into a vertex $w$. 
Consider a connected feedback vertex set $A$ of a graph $H$ of size at most $k$ and construct a set $A'\subseteq V(H')$ as following:
\begin{itemize}
\item if $u,v\notin A$ then $A'=A$;
\item otherwise $A'=(A\cup \{w\})\setminus \{u,v\}$.
\end{itemize}
It can be easily seen that $A'$ is a connected feedback vertex set of $H'$ of size at most $k$.

This observation enables us to use iterative compression approach, similarly as in the proof of Theorem \ref{thm:cvc:fpt}. 
Namely we consider a sequence of graphs $G_1,G_2,\ldots,G_{n}=G$ ($G$ is the connected graph given in the input), where $G_i$ is obtained from $G_{i+1}$ by contracting any edge and reducing possible multiedges to simple edges. 
For every $G_i$ we try to construct a connected feedback vertex set $A_i$ in a consecutive manner and if at any step we fail, we can safely answer NO. 
Thus we need to show a way of constructing a connected feedback vertex set of size $k$ in $G_{i+1}$ given a connected feedback vertex set $A_i$ of size $k$ in $G_i$, or determining that none exists.

Let $G_i$ be constructed from $G_{i+1}$ by contracting an edge $uv$. 
We construct $B$ from $A_i$ by removing the vertex obtained in the contraction (if it was included in $A_i$) and inserting both $u$ and $v$. 
Observe that $B$ is a feedback vertex set of $G_{i+1}$ of size at most $k+2$ containing a connected component of size at least $|B|-2$. 
Therefore, we can construct a tree decomposition of graph $G_{i+1}$ of width $k+3$ by constructing the tree decomposition of width $1$ of the forest $G_{i+1}\setminus B$ and including $B$ into every bag.

Now we are going to test whether $G_{i+1}$ admits a connected feedback vertex set of size at most $k$. 
A straightforward application of the algorithm from Theorem \ref{thm:app:cfvs} would yield an algorithm with running time $4^kn^{O(1)}$. 
This algorithm once again has the property of referring only to previously computed values with the same evaluation on the intersection of the domains, so we can also apply the method already used in proofs of Theorems \ref{thm:fvs:fpt} and \ref{thm:cvc:fpt} to reduce the space usage to polynomial.

Once again, using the special structure of the set $B$ we can also reduce the time complexity down to $3^kn^{O(1)}$ by bounding the number of reasonable evaluations $\overline{s}:B\to \{\mathbf{0}_1, \mathbf{0}_2,\mathbf{1}_1, \mathbf{1}_2\}$ by $4^3 3^{k-1}$. 
$B$ is a graph consisting of a large connected component and at most two additional vertices. 
Take any spanning tree of the connected component and root it in a vertex $r$. 
Each reasonable evaluation $\overline{s}$ can be coded in the following manner: vertex $r$ and the two possible additional vertices have $4$ possibilities of the value in $\overline{s}$, but every other vertex in the tree has only three possibilities, depending on the value $\overline{s}(u)$, where
$u$ denotes the parent of $v$:
\begin{itemize}
\item if $\overline{s}(u) = \mathbf{0}_j$, the possibilities are $\mathbf{1}_1$, $\mathbf{1}_2$ and $\mathbf{0}_j$;
\item if $\overline{s}(u) = \mathbf{1}_j$, the possibilities are $\mathbf{0}_1$, $\mathbf{0}_2$ and $\mathbf{1}_j$;
\end{itemize}
as otherwise the cut could not be consistent. 
Thus every vertex from $B$ has only $3$ possibilities, apart from at most $3$, which have $4$ possibilities. 
So the number of reasonable evaluations $\overline{s}$ is bounded by $4^33^{k-1}$, thus the testing algorithm runs in time complexity $3^kn^{O(1)}$ and uses polynomial space. 
Again we make $n$ independent runs of the algorithm in order to reduce the probability of a false negative to at most $\frac{1}{2^n}$.

The idea of reconstructing the solution is the same as in the proofs of Theorems \ref{thm:fvs:fpt} and \ref{thm:cvc:fpt}. 
We consider vertices one by one iteratively constructing a connected feedback vertex set. 
At each step we determine whether the considered vertex can or cannot be taken as the next vertex of the so far built part of the solution, using the algorithm for \constrainedcfvs obtained in Theorem \ref{thm:app:cfvs}. 
If it can, we take it, otherwise we just proceed to the next vertex. 
At each step we make $n$ independent runs to reduce the probability of a false negative to at most $\frac{1}{2^n}$. 
If the graph admitted a connected feedback vertex set of size at most $k$, we will construct it in this manner unless at least one test gives a false negative. 
Again, using previous observations the computation in each of the runs can be reordered so that the running time is $3^kn^{O(1)}$ and the space usage is polynomial.

Again, the union bound proves that the probability of obtaining a false negative in any of the tests is bounded by $\frac{n^2+n}{2^n}$ which for large enough $n$ is lower than $\frac{1}{2}$, as we make at most $n^2+n$ groups of $n$ independent runs of the algorithm from Theorem \ref{thm:app:cfvs}.
\end{proof}

\section{Negative results under ETH}\label{sec:negatives:eth}

In this section we provide an evidence that the problems,
   where we want not to minimize, but to maximize
the number of connected components, are harder, in the sense, that they probably
do not admit algorithms running in time $2^{o(p \log p)} n^{O(1)}$, where
$p$ denotes the pathwidth of the input graph.
More precisely, we show that assuming ETH there do not exist
algorithms for \cyclepackingname{}, \maxcyclecovername{} and \maxccdsname{}
running in time $2^{o(p \log p)} n^{O(1)}$.

Let us recall formal problem definitions. The first two problems have undirected and directed
versions.

\defproblemu{\cyclepackingname{}}{A (directed or undirected) graph $G = (V,E)$ and an integer $\ell$}{Does $G$ contain $\ell$ vertex-disjoint cycles?}

\defproblemu{\maxcyclecovername{}}{A (directed or undirected) graph $G = (V,E)$ and an integer $\ell$}{Does $G$ contain a set of at least $\ell$ vertex-disjoint cycles such that each vertex of $G$ is contained in exactly one cycle?}

\defproblemu{\maxccdsname{}}{An undirected graph $G=(V,E)$ and integers $\ell$ and $r$}{Does $G$ contain a dominating set of size at most $\ell$ that induces {\bf{at least}} $r$ connected components?}

We prove the following theorem.
\begin{theorem}[Theorem \ref{thm:negative-main} restated]\label{thm:negative-main-restate}
Assuming ETH, there is no $2^{o(p \log p)} n^{O(1)}$ time algorithm for \cyclepackingname{}, \maxcyclecovername{} (both in the directed and undirected setting) nor for \maxccdsname{}. 
The parameter $p$ denotes the width of a given path decomposition of the input graph.
\end{theorem}

We start our reductions from \kknphittingsetname{} and \kkhittingsetname{} problems.
As we discussed in Section \ref{sec:eth:intro}, the non-permutation version was introduced and analyzed
by Lokshtanov et al. \cite{marx:superexp}. We denote $[k] = \{1,2,\ldots,k\}$. In the set $[k] \times [k]$ {\em{a row}} is a set $\{i\} \times [k]$ and {\em{a column}} is
a set $[k] \times \{i\}$ (for some $i \in [k]$). We include formal definitions for sake
of completeness.

\defproblemu{\kknphittingsetname}{A family of sets $S_1, S_2 \ldots S_m \subseteq [k] \times [k]$,
  such that each set contains at most one element from each row of $[k] \times [k]$.}{
    Is there a set $S$ containing exactly one element from each row
      such that $S \cap S_i \neq \emptyset$ for any $1 \leq i \leq m$?}
\defproblemu{\kkhittingsetname}{A family of sets $S_1, S_2 \ldots S_m \subseteq [k] \times [k]$,
  such that each set contains at most one element from each row of $[k] \times [k]$.}{
    Is there a set $S$ containing exactly one element from each row and exactly one element from each column
      such that
      $S \cap S_i \neq \emptyset$ for any $1 \leq i \leq m$?}

\begin{theorem}[\cite{marx:superexp}, Theorem 2.4]\label{thm:marx-hs:app}
Assuming ETH, there is no $2^{o(k \log k)} n^{O(1)}$ time algorithm for \kknphittingsetname{}
nor for \kkhittingsetname{}.
\end{theorem}
Note that in \cite{marx:superexp} the statement of the above theorem only includes \kknphittingsetname{}.
However, the proof in \cite{marx:superexp} works for the permutation variant as well without any modifications.

We first prove the bound for \maxccdsname{} by quite simple reduction from \kknphittingsetname{}.
This is done in Section \ref{sec:eth:maxccds}.
Then, in Section \ref{sec:eth:ucycles} we prove the bound for undirected \cyclepackingname{},
  by quite involved reduction from
\kkhittingsetname{}. In Section \ref{sec:neg:undir2dir}
we provide a reduction to directed \cyclepackingname{} and in Section \ref{sec:neg:packing2cover}
we provide a reduction to \maxcyclecovername{} in both variants.

\subsection{\maxccdsname{}}\label{sec:eth:maxccds}

In this subsection we provide a reduction from \kknphittingsetname{} to \maxccdsname{}.
We are given an instance $(k, S_1, \ldots, S_m)$ of \kknphittingsetname{},
called the initial instance,
and we are to construct an equivalent instance $(G,\ell,r)$ of \maxccdsname{}.

We first set $\ell := 3k+m$ and $r := k$.

\subsubsection{Gadgets}

We introduce a few simple gadgets used repeatedly in the construction.
In all definitions $H=(V,E)$ is an undirected graph, and the parameters $\ell$ and $r$ are fixed.
\begin{definition}
 By adding a {\em{force gadget}} for vertex $v \in V$ we mean the following construction:
 we introduce $\ell+1$ new vertices of degree one, connected to $v$.
\end{definition}
\begin{lemma}
 If graph $G$ is constructed from graph $H=(V,E)$ by adding a force gadget to vertex $v \in V$,
 then $v$ is contained in each dominating set in $G$ of size at most $\ell$.
\end{lemma}
\begin{proof}
 If $D$ is a dominating set in $G$, and $v \notin D$, then all new vertices added in the force gadget
 need to be included in $D$. Thus $|D| \geq \ell+1$.
\end{proof}

\begin{definition}
  By adding a {\em{one-in-many gadget}} to vertex set $X \subseteq V$ we mean the following construction:
  we introducte $\ell+1$ new vertices of degree $|X|$, connected to all vertices in $X$.
\end{definition}
\begin{lemma}\label{lem:maxccds:gadget2}
 If graph $G$ is constructed from graph $H=(V,E)$ by adding a one-in-many gadget to vertex set $X \subseteq V$,
 then each dominating set in $G$ of size at most $\ell$ contains a vertex from $X$.
\end{lemma}
\begin{proof}
 If $D$ is a dominating set in $G$, and $X \cap D = \emptyset$, then all new vertices added in the one-in-many gadget
 need to be included in $D$. Thus $|D| \geq \ell+1$.
\end{proof}

We conclude with the pathwidth bound.
\begin{lemma}\label{lem:maxccds:gadget-pw}
 Let $G$ be a graph and let $G'$ be a graph constructed from $G$ by adding multiple force and one-in-many
 gadgets. Assume we are given a path decomposition of $G$ of width $p$ with the following property: for each one-in-many gadget,
 attached to vertex set $X$, there exists a bag in the path decomposition that contains $X$.
 Then, in polynomial time, we can construct a path decomposition of $G'$ of width at most $p+1$.
\end{lemma}
\begin{proof}
 Let $w$ be a vertex in $G'$, but not in $G$, i.e., a vertex added in one of the gadgets. By the assumptions of the lemma,
 there exists a bag $V_w$ in the path decomposition of $G$ that contains $N(w)$. For each such vertex $w$, we introduce a new bag $V_w' = V_w \cup \{w\}$
 and we insert it into the path decomposition after the bag $V_w$. If $V_w$ is multiplied for many vertices $w$, we insert all the new bags
 after $V_w$ in an arbitrary order.

 It is easy to see that the new path decomposition is a proper path decomposition of $G'$, as $V_w'$ covers all edges incident to $w$.
 Moreover, we increased the maximum size of bags by at most one, thus the width of the new decomposition is at most $p+1$.
\end{proof}

\subsubsection{Construction}

Let $S_i^{\texttt{row}} = \{i\} \times [k]$ be a set containing all elements in the $i$-th row
in the set $[k] \times [k]$. We denote $\mathcal{S} = \{S_s: 1 \leq s \leq m\} \cup \{S_i^{\texttt{row}}: 1 \leq i \leq k\}$. Note that for each $A \in \mathcal{S}$ we have $|A| \leq k$, as each set
$S_i$ contains at most one element from each row.

First let us define the graph $H$.
We start by introducing vertices $p_i^L$ for $1 \leq i \leq k$
and vertices $p_j^R$ for $1 \leq j \leq k$.
Then, for each set $A \in \mathcal{S}$ we introduce vertices $x_{i,j}^A$ for all $(i,j) \in A$
and edges $p_i^Lx_{i,j}^A$ and $p_j^Rx_{i,j}^A$. Let $X^A = \{x_{i,j}^A: (i,j) \in A\}$.

To construct graph $G$, we attach the following gadgets to graph $H$.
For each $1 \leq i \leq k$ and $1 \leq j \leq k$ we attach force gadgets to vertices $p_i^L$ and $p_j^R$.
Moreover, for each $A \in \mathcal{S}$ we attach one-in-many gadget to the set $X^A$.

We now provide a pathwidth bound on the graph $G$.
\begin{lemma}
 The pathwidth of $G$ is at most $3k$.
\end{lemma}
\begin{proof}
 First consider the following path decomposition of $H$.
 For each $A \in \mathcal{S}$ we create a bag
 $$V_A = \{p_i^L : 1 \leq i \leq k\} \cup \{p_j^R: 1 \leq j \leq k\} \cup \{x_{i,j}^A: (i,j) \in A\}.$$
 The path decomposition of $H$ consists of all bags $V_A$ for $A \in \mathcal{S}$ in an
 arbitrary order. Note that the above path decomposition is a proper path decomposition
 of $H$ of width at most $3k-1$ (as $|A| \leq k$ for each $A \in \mathcal{S}$)
 and it satisfies conditions for Lemma \ref{lem:maxccds:gadget-pw}.
\end{proof}

\subsubsection{From hitting set to dominating set}

\begin{lemma}
 If the initial \kknphittingsetname{} instance was a YES-instance, then there exists a dominating set $D$ in the graph $G$,
 such that $|D| = \ell$ and $D$ induces exactly $r$ connected components.
\end{lemma}

\begin{proof}
 Let $S$ be a solution to the initial \kknphittingsetname{} instance $(k,S_1, \ldots, S_m)$.
 For each $A \in \mathcal{S}$ fix an element $(i_A,j_A) \in S \cap A$.
 Recall that $S$ contains exactly one element from each row, thus $S \cap A \neq \emptyset$
  for all sets $A \in \mathcal{S}$.
 Let us define:
 $$D = \{p_i^L: 1 \leq i \leq k\} \cup \{p_j^R: 1 \leq j \leq k\} \cup \{x_{i_A,j_A}^A: A \in \mathcal{S}\}.$$
 
 First note that $|D| = 3k+m$, as there are $k$ vertices $p_i^L$, $k$ vertices $p_j^R$,
 and $|\mathcal{S}| = k+m$, since $\mathcal{S}$ consists of $m$ sets $S_s$ and $k$ sets
 $S_i^{\texttt{row}}$.

 Let us now check whether $D$ is a dominating set in $G$. Vertices $p_i^L$ and
 $p_j^R$ for $1 \leq i,j \leq k$ dominate all vertices of the graph $H$ and all vertices
 added in the attached force gadgets. Moreover, $D \cap X^A = \{x_{i_A,j_A}^A\}$
 for each $A \in \mathcal{S}$, thus $D$ dominates all vertices added in one-in-many
 gadgets attached to sets $X^A$.

We now prove that $G[D]$ contains exactly $r = k$ connected components.
Let us define for each $1 \leq j \leq k$:
$$D_j = \{p_j^R\} \cup \{p_i^L: (i,j) \in S\} \cup \{x_{i_A,j_A}^A: A \in \mathcal{S}, j_A = j\}.$$
Note that $D_j$ is a partition of $D$ into $k$ pairwise disjoint sets.
Moreover, observe that $G[D_j]$ is connected and, since $S$ contains exactly one element
from each row, no vertices from $D_j$ and $D_{j'}$ are adjacent, for $j \neq j'$.
This finishes the proof of the lemma.
\end{proof}

\subsubsection{From dominating set to hitting set}

\begin{lemma}
 If there exists a dominating set $D$ in the graph $G$, such that $|D| \leq \ell$ and $D$ induces at least $r$ connected components,
 then the initial \kknphittingsetname{} instance was a YES-instance.
\end{lemma}
\begin{proof}
 By the properties of the force gadget, $D$ needs to include all forced vertices, i.e.,
 vertices $p_i^L$ and $p_j^R$ for $1 \leq i,j \leq k$.
  There are $2k$ forced vertices, thus we have $\ell-2k = k+m$ vertices left.

 By the properties of one-in-many gadgets, $D$ needs to include at least one vertex
 from each set $X^A$, $A \in \mathcal{S}$. But $|\mathcal{S}| = k+m$
 and sets $X^A$ are pairwise disjoint. Thus, $D$ consist of  all forced vertices
 and exactly one vertex from each set $X^A$, $A \in \mathcal{S}$.

 For each $1 \leq i \leq k$ let $x_{i,f(i)}^{S_i^\texttt{row}}$ be the unique vertex
 in $D \cap X^{S_i^\texttt{row}}$. Let $S = \{(i,f(i)): 1 \leq i \leq k\}$.
 We claim that $S$ is a solution to the initial \kknphittingsetname{} instance.
 It clearly contains exactly one element from each row.

 Let $D_j$ be the vertex set of the connected component of $G[D]$ that contains
 $p_j^R$. Note that $p_i^L \in D_j$ whenever $j = f(i)$, i.e., $(i,j) \in S$.
 This implies that $\bigcup_{j=1}^k D_j$ contains all vertices $p_i^L$. Moreover,
 as each vertex in $X^A$ for $A \in \mathcal{S}$ is adjacent to some vertex $p_j^R$,
 the sets $D_j$ are the only connected components of $G[D]$. As $G[D]$ contains at least
 $r=k$ connected components, $D_j \neq D_{j'}$ for $j \neq j'$.

 Let $1 \leq s \leq m$ and let us focus on set $S_s \in \mathcal{S}$.
 Let $x_{i,j}^{S_s}$ be the unique vertex in $D \cap X^{S_s}$.
 Note that $x_{i,j}^{S_s}$ connects $p_i^L \in D_{f(i)}$ with
 $p_j^R \in D_j$. As sets $D_j$ are pairwise distinct, this implies
 that $j = f(i)$ and $(i,j) \in S \cap S_s$. Thus the components
 of $G[D]$ are exactly the sets $D_j$, $1 \leq j \leq k$.
\end{proof}

\subsection{Undirected \cyclepackingname{}}\label{sec:eth:ucycles}

\subsubsection{Proof overview and preliminaries}

First note that for \cyclepackingname{}, we can assume that the input graph may be a multigraph,
i.e., it may contain multiple edges and loops. The following lemma summarizes this observation.
\begin{lemma}\label{lem:neg:multi}
Let $(G,\ell)$ be an instance of (directed or undirected) \cyclepackingname{},
    where $G$ may contain multiple edges and loops. Then we can construct in polynomial
    time an equivalent (directed or undirected, respectively) instance $(G',\ell)$,
    such that $G'$ does not contain multiple edges nor loops.
   Moreover, given a path decomposition of $G$ of width $p$, we can construct
   in polynomial time a path decomposition of $G'$ of width at most $p+2$.
\end{lemma}
\begin{proof}
To construct $G'$, we replace each edge $e \in E(G)$ with a path of length three,
   i.e., we insert vertices $u_e^1$ and $u_e^2$ in the middle of edge $e$.
Clearly $G$ is a simple graph and vertex-disjoint cycle families in $G$ and $G'$ naturally translates into each other.

We are left with the pathwidth bound. Assume we have a path decomposition of $G$ of width $p$.
For each edge $e \in E(G)$ we fix a bag $V_e$ that covers $e$. We introduce a new bag $V_e' = V_e \cup \{u_e^1, u_e^2\}$
and insert $V_e'$ near the bag $V_e$ in the path decomposition.
It is easy to see that the new decomposition is a proper path decomposition of $G'$ and its width is at most $p+2$.
\end{proof}

Let us introduce some extra notation. We say that a vertex is covered by a cycle (or a family of cycles)
if the vertex belongs to the cycle (or belongs to at least one cycle in the family).
A graph is covered by a cycle family if every its vertex is covered by the family.
By $(v_1, \ldots, v_r)$ we denote a path (or a cycle) consisting of vertices $v_1, v_2, \ldots, v_r$ in this order.

We now present an overview of the proof of Theorem \ref{thm:negative-main} for undirected \cyclepackingname{}.
We provide a construction that, given an instance $(k, S_1, S_2, \ldots, S_m)$ of \kkhittingsetname{} (called an {\em{initial instance}}),
produces in polynomial time an undirected graph $G$, an integer $\ell$ and a path decomposition of $G$ with the following properties:
\begin{enumerate}
\item The path decomposition of $G$ has width $O(k)$.\label{neg:prop:pw}
\item If the initial instance of \kkhittingsetname{} is a YES-instance, then there exists a family of $\ell$ vertex-disjoint cycles
 in $G$. In other words, $(G,\ell)$ is a YES-instance of undirected \cyclepackingname{}.\label{neg:prop:hs2cycle}
\item If there exist a family of $\ell$ vertex-disjoint cycles in $G$, then the initial \kkhittingsetname{} instance
is a YES-instance.\label{neg:prop:cycle2hs}
\end{enumerate}

In Section \ref{sec:neg:prelims} we describe the {\em{$r$-in-many gadget}}, a tool used widely in the construction. In Section \ref{sec:neg:construction} we give the construction of the graph $G$
and show the pathwidth bound, i.e., Point \ref{neg:prop:pw}.
Points \ref{neg:prop:hs2cycle} and \ref{neg:prop:cycle2hs} are proven in Sections \ref{sec:neg:hs2cycle} and \ref{sec:neg:cycle2hs} respectively.
Reductions to directed \cyclepackingname{} and to \maxcyclecovername{} are in Sections \ref{sec:neg:undir2dir} and \ref{sec:neg:packing2cover} respectively.

\subsubsection{$r$-in-many gadget}\label{sec:neg:prelims}

In this section we describe the {\em{$r$-in-many gadget}}, a tool used in further sections.
Informally speaking, the $r$-in-many gadget attached to vertex set $X$ ensures that at most $r$ vertices from $X$ are used in a solution (a family of cycles).
\begin{definition}\label{def:$r$-in-many}
Let $H$ be a multigraph and $X$ be an arbitrary subset of vertices of $H$.
By a {\em{$r$-in-many gadget attached to $X$}} ($1 \leq r < |X|$) we mean the following construction:
we introduce $(|X|-r)$ new vertices $\{u_i: 1 \leq i \leq |X|-r\}$ and 
for each $1 \leq i \leq |X|-r$ and $x \in X$ we add two edges $xu_i$.
In other words, we introduce $|X|-r$ vertices connected to the set $X$ via double edges.
The set $X$ is called an {\em{attaching point}} of the gadget.
A cycle of length two, constisting of two edges $xu_i$
for some $1 \leq i \leq |X|-r$ and $x \in X$, is called a {\em{gadget short cycle}}.
\end{definition}
\begin{definition}\label{def:gadget-ext}
Let $H$ be a multigraph, let $X_1, \ldots, X_d$ be pairwise disjoint subsets of vertices of $H$
and let $r_1, \ldots, r_d$ be integers satisfying $1 \leq r_i < |X_i|$ for $1 \leq i \leq d$.
Let $G$ be a multigraph constructed from $H$ by attaching to $X_i$ a $r_i$-in-many gadget, for all $1 \leq i \leq d$.
We say that $G$ is a {\em{gadget extension}} of $H$.
The maximum number of gadget short cycles that can be packed in $G$ is denoted by $\ell_G$, i.e., $\ell_G := \sum_{i=1}^d |X_i|-r_i$.
\end{definition}
\begin{definition}
Let $G$ be a gadget extension of $H$, and let $X_i$ and $r_i$ be as in Definition \ref{def:gadget-ext}.
If $\mathcal{C}_H$ is a family of vertex-disjoint cycles in $H$ satisfying the following property: for each $1 \leq i \leq d$
at most $r_i$ vertices from $X_i$ are covered by $\mathcal{C}_H$, then we say that $\mathcal{C}_H$ is {\em{gadget safe in $H$}}.
If $\mathcal{C}_G$ is a family of vertex-disjoint cycles in $G$ containing $\ell_G$ gadget short cycles,
then we say that $\mathcal{C}_G$ is {\em{gadget safe in $G$}}.
\end{definition}

\begin{figure}[htbp]
\begin{center}
\includegraphics{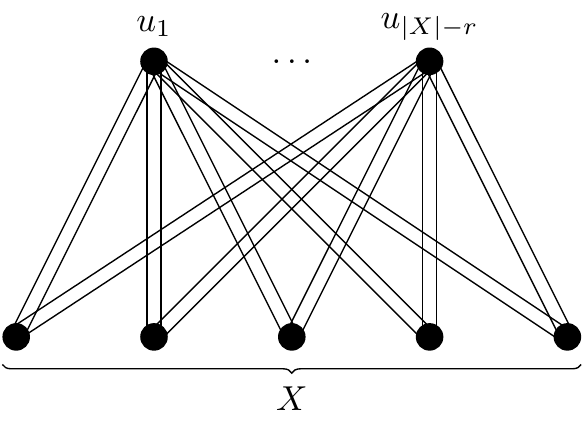}
\end{center}
\caption{The $r$-in-many gadget attached to set $X$.}
\label{fig:r-in-many}
\end{figure}

The following lemma shows how the $r$-in-many gadget is intended to be used.
\begin{lemma}\label{lem:r-in-many:use}
Let $G$ be a gadget extension of $H$, and let $X_i$ and $r_i$ be as in Definition \ref{def:gadget-ext}.
Let $\mathcal{C}_H$ be a family of cycles that is gadget safe in $H$.
Then $\mathcal{C}_H$ can be extended to gadget safe in $G$ family $\mathcal{C}_G$ of size $|\mathcal{C}_H|+ \ell_G$
\end{lemma}
\begin{proof}
For each $1 \leq i \leq d$, let $Y_i \subseteq X_i$ be a set of (arbitrarily chosen)
  $|X_i|-r_i$ vertices not covered by $\mathcal{C}_H$. 
  Assign $\mathcal{C}_G := \mathcal{C}_H$.
  For each $1 \leq i \leq d$ we add to $\mathcal{C}_G$
  a set of $|X_i|-r_i$ gadget short cycles, each consisting of one vertex in $Y_i$ and one vertex $u_j$ from the gadget attached to $X_i$.
\end{proof}

The next lemma shows that we can safetely assume that the $r$-in-many gadgets are
used as in the proof of Lemma \ref{lem:r-in-many:use}.
\begin{lemma}\label{lem:r-in-many:rev}
Let $G$ be a gadget extension of $H$, and let $X_i$ and $r_i$ be as in Definition \ref{def:gadget-ext}.
Let $\ell$ be the maximum possible cardinality of a family of vertex-disjoint cycles
in $G$. Then there exists a gadget safe in $G$ family $\mathcal{C}$ of size $\ell$.
Moreover, after removing from $\mathcal{C}$ all $\ell_G$ gadget short cycles,
  we obtain a gadget safe in $H$ family of cycles.
\end{lemma}
\begin{proof}
Let $\mathcal{C}$ be a family of vertex-disjoint cycles in $G$
of size $\ell$ that maximizes the number of gadget short cycles. By contradiction,
assume that $\mathcal{C}$ is not gadget safe in $G$. That means it contains less than $\sum_{i=1}^d |X_i|-r_i$ gadget short cycles, 
i.e., there exists $1 \leq i \leq d$, such that less than $|X_i|-r_i$ gadget short cycles in the gadget
attached to $X_i$ are in $\mathcal{C}$.

Let $u$ be a vertex in the gadget attached to $X_i$ that does not lie on a gadget short cycle
in $\mathcal{C}$. If $u$ is covered by a cycle $C \in \mathcal{C}$, then there exists $x \in X_i$
also covered by $C$. We can replace $C$ with a gadget short cycle $(u,x)$, increasing the number
of gadget short cycles in $\mathcal{C}$, a contradiction.

Thus, $u$ is not covered in $\mathcal{C}$. Let $x \in X_i$ be a vertex that is not covered by a
gadget short cycle in $\mathcal{C}$ (there exists, as $r_i < |X_i|$ and sets $X_i$ are pairwise disjoint). If $x$ is not covered by
$\mathcal{C}$, we can add gadget short cycle $(u,x)$ to $\mathcal{C}$, increasing its size, a contradiction.
Otherwise, we can replace the cycle with $x$ with gadget short cycle $(u,x)$, a contradiction too.
Thus, $\mathcal{C}$ contains $\ell_G = \sum_{i=1}^d |X_i|-r_i$ gadget short cycles.
It is a straightforward corollary from the definitions that
after removing these $\ell_G$ cycles, we obtain a gadget safe in $H$ family of cycles.
\end{proof}

Finally, we show that attaching a $r$-in-many gadget may not influence much the pathwidth of the graph.
\begin{lemma}\label{lem:r-in-many:pw}
Let $G$ be a gadget extension of $H$, and let $X_i$ and $r_i$ be as in Definition \ref{def:gadget-ext}.
Assume that we are given a path decomposition of $H$ of width $p$, such that for each $1 \leq i \leq d$ there
exists a bag $V_i$ that contains the whole $X_i$. Then in polynomial time we can construct a path decomposition
of $G$ of width at most $p+1$.
\end{lemma}
\begin{proof}
Let $1 \leq i \leq d$ and let $V_i$ be a bag containing $X_i$.
We introduce bags $V_i^j$, $1 \leq j \leq |X_i|-r_i$, taking $V_i^j = V_i \cup \{u_j\}$.
We insert the newly created bags $V_i^j$ near the bag $V_i$ in the path decomposition.
As all bags $V_i^j$ contain $V_i$, this modification does not spoil the properties
of the path decomposition of $H$. Bag $V_i^j$ covers all edges incident to $u_j$.
Thus, the new path decomposition is a proper path decomposition of $G$ and has width at most $p+1$, 
  as $|V_i^j| = |V_i| + 1$.
\end{proof}

\subsubsection{Construction}\label{sec:neg:construction}

Let $(k, S_1, S_2, \ldots, S_m)$ be an instance of \kkhittingsetname{}.
W.l.o.g. we may assume that each set $S_i$ is nonempty.
We first construct a graph $H$ as follows:
\begin{enumerate}
\item The vertex set $V(H)$ consists of
  \begin{enumerate}
   \item vertices $p_i^Z, p_i^R, q_i^Z, q_i^R$ for $1 \leq i \leq k$;
   \item vertices $p_{i,j}^C, q_{i,j}^C$ for $1 \leq i,j \leq k$; for each $1 \leq i \leq k$ we denote $X_i^p = \{p_{i,j}^C: 1 \leq j \leq k\}$ and $X_i^q = \{q_{i,j}^C: 1 \leq j \leq k\}$;
   \item vertices $x_{i,s}^L, x_{i,s}^R, y_{i,s}^L, y_{i,s}^R$ for $1 \leq i \leq k$ and $1 \leq s \leq m$;
   \item and vertices $x_{i,s}^C, y_{i,s}^C, x_{i,s}^Z, y_{i,s}^Z, z_{i,s}^C$
     for $1 \leq s \leq m$ and $(i,j) \in S_s$ (recall that there is at most one element in each row in $S_s$); we denote $X_s^x = \{x_{i,s}^C: (i,j) \in S_s\}$,
     $X_s^y = \{y_{i,s}^C: (i,j) \in S_s\}$ and $X_s^z = \{z_{i,s}^C: (i,j) \in S_s\}$.
 \end{enumerate}
 The vertex set is partitioned into four parts $L$, $R$, $C$ and $Z$, according to the superscripts (the first three are acronyms for left, right and centre, the last one should be
     seen as an important separator between left and centre).
\item Vertices $p_i^Z$ and $p_j^R$ are connected into full bipartite graph with vertices $p_{i,j}^C$ inserted
into the middle of each edge, i.e., for all $1 \leq i,j \leq k$ we add edges $p_i^Zp_{i,j}^C$ and $p_{i,j}^Cp_j^R$.
  Similar construction is performed for vertices $q_i^Z$, $q_j^R$ and $q_{i,j}^C$,
  i.e., for all $1 \leq i,j \leq k$ we add edges $q_i^Zq_{i,j}^C$ and $q_{i,j}^Cq_j^R$.
\item For each $1 \leq i \leq k$, vertices $x_{i,s}^L$ and $y_{i,s}^L$ are arranged into path from $p_i^Z$ to $q_i^Z$,
  i.e., $x_{i,s}^Ly_{i,s}^L \in E$ for $1 \leq s \leq m$, $y_{i,s}^Lx_{i,s+1}^L \in E$ for $1 \leq s < m$
  and $q_i^Zy_{i,m}^L,p_i^Zx_{i,1}^L \in E$. By $\mathcal{P}_i^L$ we denote the path from $p_i^Z$ to $q_i^Z$
\item For each $1 \leq i \leq k$, vertices $x_{i,s}^R$ and $y_{i,s}^R$ are arranged into path from $p_i^R$ to $q_i^R$,
  i.e., $x_{i,s}^Ry_{i,s}^R$ for $1 \leq s \leq m$, $y_{i,s}^Rx_{i,s+1}^R \in E$ for $1 \leq s < m$
  and $p_i^Rx_{i,1}^R,q_i^Ry_{i,m}^R \in E$. By $\mathcal{P}_i^R$ we denote the path from $p_i^R$ to $q_i^R$.
\item For each $1 \leq s \leq m$, if $(i,j) \in S_s$, we add a path $(x_{i,s}^L,x_{i,s}^Z,x_{i,s}^C,x_{j,s}^R)$.
\item Similarly, for each $1 \leq s \leq m$, if $(i,j) \in S_s$, we add a path $(y_{i,s}^L,y_{i,s}^Z,y_{i,s}^C,y_{j,s}^R)$.
\item Moreover, for each $1 \leq s \leq m$ and $(i,j) \in S_s$ we add a cycle $(x_{i,s}^Z, z_{i,s}^C, y_{i,s}^Z)$.
\end{enumerate}
The graph $G$ is defined as a gadget extension of $H$ by attaching to $H$ the following $r$-in-many gadgets:
\begin{enumerate}
\item to vertex sets $X_i^p$ and $X_i^q$ for $1 \leq i \leq k$ we attach $1$-in-many gadgets;
\item to vertex sets $X_s^x$ and $X_s^y$ for $1 \leq s \leq m$ we attach $1$-in-many gadgets;
\item for $1 \leq s \leq m$ we attach a $(|X_s^z|-1)$-in-many gadget to $X_s^z$.
\end{enumerate}
  
\begin{figure}[htbp]
\begin{center}
  \includegraphics{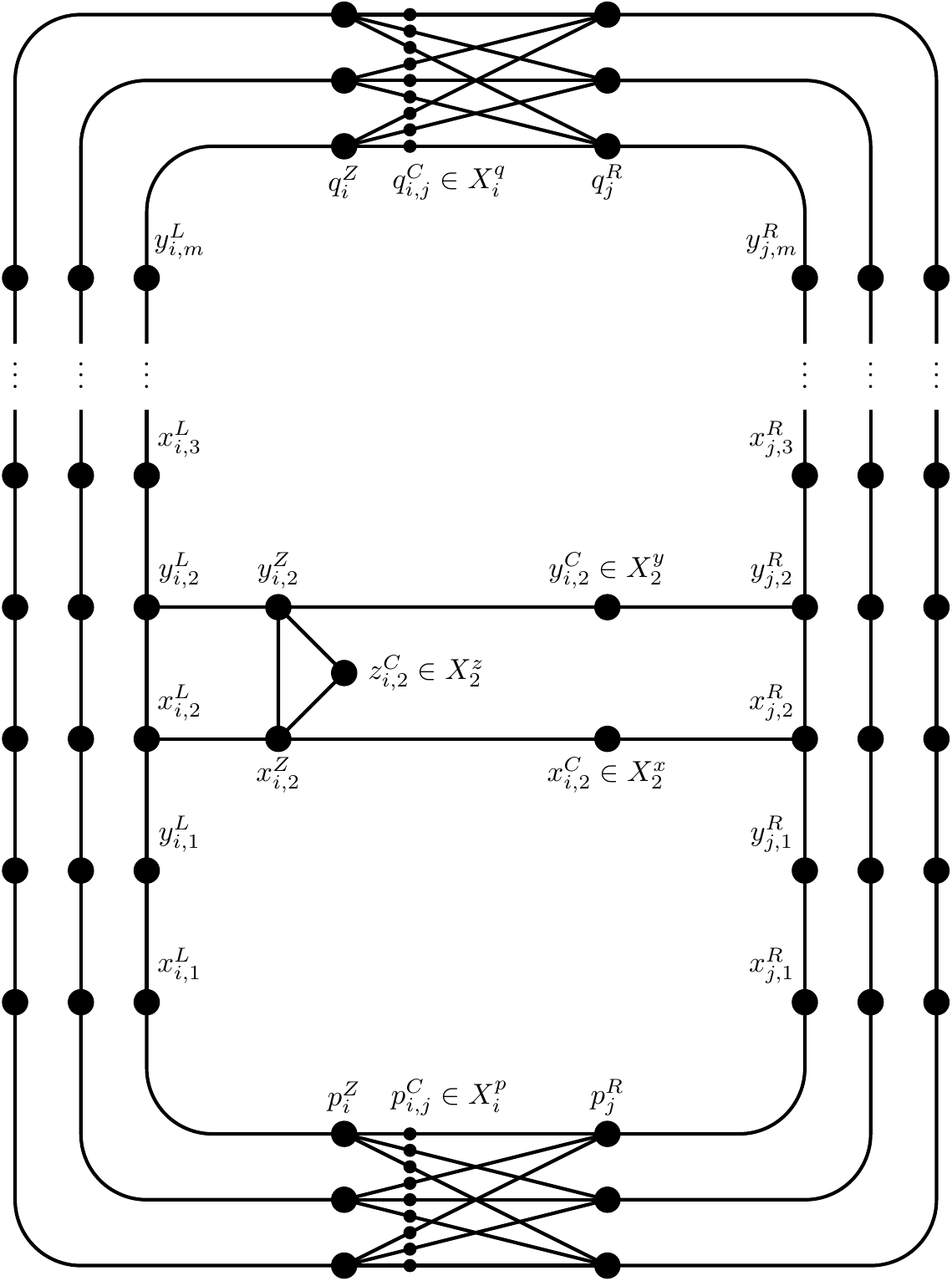}
\end{center}
\caption{The part of graph $H$ with the main frame and the part for element $(i,j) \in S_s$. Recall that
the gadget safe in $H$ family may cover at most one element of $X_i^p$, $X_i^q$, $X_s^x$ and $X_s^y$ and
cannot cover whole $X_s^z$.}
\label{fig:neg:constrution}
\end{figure}

Clearly, the above construction can be done in polynomial time.
Note that we can pack the following number of gadget short cycles in $G$:
\begin{align*}
\ell_G &:= \sum_{i=1}^k (|X_i^p|-1 + |X_i^q|-1) + \sum_{s=1}^m (|X_s^x| -1 + |X_s^y|-1 + |X_s^z| - (|X_s^z|-1)) \\
       &= 2k^2-2k-m + 2\sum_{s=1}^m |S_s|.
\end{align*}
We take $\ell := k + \sum_{s=1}^m |S_s| + \ell_G$, i.e., we ask for $\ell$ vertex-disjoint cycles in $G$.

The following lemma shows the pathwidth bound of $G$, i.e., proves Point \ref{neg:prop:pw}.
\begin{lemma}\label{lem:neg-dir-pw}
In polynomial time we can construct a path decomposition of $G$ of width $11k$.
\end{lemma}
\begin{proof}
  By Lemma \ref{lem:r-in-many:pw}, it is sufficient to show
  a path decomposition of $H$ of width $11k-1$ such that each set $X_i^p$, $X_i^q$, $X_s^x$, $X_s^y$, $X_s^z$ is contained
  in some bag.

  The path decomposition consists of bags $V_i^p$ for $1 \leq i \leq k$, $V_s$ for $0 \leq s \leq m$
  and $V_i^q$ for $1 \leq i \leq k$, arranged in a path in this order. We define:
  \begin{enumerate}
    \item $V_i^p = \{p_j^Z, p_j^R, p_{i,j}^C : 1 \leq j \leq k\}$ for $1 \leq i \leq k$.
    \item $V_0 = \{p_i^Z, p_i^R, x_{i,1}^L, x_{i,1}^R: 1 \leq i \leq k\}$.
    \item $V_s = \{x_{i,s}^L, x_{i,s}^R, y_{i,s}^L, y_{i,s}^R, x_{i,s+1}^L, x_{i,s+1}^R: 1 \leq i \leq k\} \cup \{x_{i,s}^C, y_{i,s}^C, x_{i,s}^Z, y_{i,s}^Z, z_{i,s}^C: (i,j) \in S_s\}$
       for $1 \leq s < m$.
    \item $V_m = \{x_{i,m}^L, x_{i,m}^R, y_{i,m}^L, y_{i,m}^R, q_i^Z, q_i^R: 1 \leq i \leq k\} \cup \{x_{i,m}^C, y_{i,m}^C, x_{i,m}^Z, y_{i,m}^Z, z_{i,m}^C: (i,j) \in S_m\}$.
    \item $V_i^q = \{q_j^Z, q_j^R, q_{i,j}^C : 1 \leq j \leq k\}$ for $1 \leq i \leq k$.
  \end{enumerate}
  It is easy to see that this is a proper path decomposition of graph $H$.
  Moreover, $X_i^p \subset V_i^p$, $X_i^q \subset V_i^q$ and $X_s^x,X_s^y,X_s^z \subset V_s$.
  As for the size bound, note that $|V_i^p| = |V_i^q| = 3k$ for $1 \leq i \leq k$,
  $|V_0| = 4k$ and  $|V_s| \le 11k$ for $1 \leq s \leq m$.
\end{proof}
Let us note that the above bound is not optimal, but we need only $O(k)$ bound.

\subsubsection{From hitting set to cycle cover}\label{sec:neg:hs2cycle}

We prove Point \ref{neg:prop:hs2cycle} by the following lemma:
\begin{lemma}\label{lem:neg-dir-hs2cycles}
If the initial \kkhittingsetname{} instance is a YES-instance, then the graph $G$ contains $\ell$ vertex-disjoint cycles.
\end{lemma}
\begin{proof}
Let $S = \{(i,f(i)): 1 \leq i \leq k\}$ be the solution to the \kkhittingsetname{} instance. Recall that $S$ contains exactly
one element from each row and exactly one element from each column, thus $f$ is a permutation of $[k]$.

By Lemma \ref{lem:r-in-many:use}, it is sufficient to show a family of cycles $\mathcal{C}$ in $H$ that is gadget safe in $H$
and is of size $k + \sum_{s=1}^m |S_s|$.

For each $1 \leq s \leq m$, fix an index $1 \leq i_s \leq k$, such that $(i_s, f(i_s)) \in S \cap S_s$.
Let
$$\mathcal{C}_1 = \{(x_{i,s}^Z, y_{i,s}^Z, z_{i,s}^C): 1 \leq s \leq m, 1 \leq i \leq k, i \neq i_s\}.$$
Note that $\mathcal{C}_1$ is a family of $\sum_{s=1}^m (|S_s|-1)$ vertex-disjoint cycles, thus we need to find $k+m$ more.

Fix $i$, $1\leq i \leq k$, and let $\{1 \leq s \leq m: i_s = i\} = \{s_1, s_2, \ldots, s_{h(i)}\}$ and $s_1 < s_2 < \ldots  < s_{h(i)}$. Consider the following family of $h(i)+1$ cycles $\{C(i,j): 0 \leq j \leq h(i)\}$:
\begin{enumerate}
\item $C(i,0)$ consists of the path $(p_i^Z, p_{i,f(i)}^C, p_{f(i)}^R)$, the subpath of $\mathcal{P}_{f(i)}^R$ from $p_{f(i)}^R$ to $x_{f(i),s_1}^R$,
  the path $(x_{f(i),s_1}^R, x_{i,s_1}^C, x_{i,s_1}^Z, x_{i,s_1}^L)$ and the subpath of $\mathcal{P}_i^L$ from $x_{i,s_1}^L$ to $p_i^Z$;
\item $C(i,j)$ for $1 \leq j < h(i)$ consists of the path $(y_{i,s_j}^L, y_{i,s_j}^Z, y_{i,s_j}^C,y_{f(i),s_j}^R)$, the subpath of $\mathcal{P}_{f(i)}^R$ from $y_{f(i),s_j}^R$ to $x_{f(i),s_{j+1}}^R$,
  the path $(x_{f(i),s_{j+1}}^R, x_{i,s_{j+1}}^C, x_{i,s_{j+1}}^Z, x_{i,s_{j+1}}^L)$ and the subpath of $\mathcal{P}_i^L$ from $x_{i,s_{j+1}}^L$ to $y_{i,s_j}^L$;
\item $C(i,h(i))$ consists of the path $(y_{i,s_{h(i)}}^L, y_{i,s_{h(i)}}^Z, y_{i,s_{h(i)}}^C,y_{f(i),s_{h(i)}}^R)$, the subpath of $\mathcal{P}_{f(i)}^R$ from $y_{f(i),s_j}^R$ to $q_{f(i)}^R$,
  the path $(q_{f(i)}^R, q_{i,f(i)}^C, q_i^Z)$ and the subpath of $\mathcal{P}_i^L$ from $q_i^Z$ to $y_{i,s_{h(i)}}^L$.
\end{enumerate}
Note that
$$\mathcal{C}_2 = \{C(i,j): 1 \leq i \leq k, 0 \leq j \leq h(i)\}$$
is a family of $k+m$ vertex-disjoint cycles in $H$ and they are disjoint with $\mathcal{C}_1$.
Moreover, $\mathcal{C} := \mathcal{C}_1 \cup \mathcal{C}_2$ does not cover:
\begin{enumerate}
\item $X_i^p \setminus \{p_{i,f(i)}^C\}$ and $X_i^q \setminus \{q_{i,f(i)}^C\}$ for $1 \leq i \leq k$;
\item $X_s^x \setminus \{x_{i_s,s}^C\}$ and $X_s^y \setminus \{y_{i_s,s}^C\}$ for $1 \leq s \leq m$;
\item $z_{i_s,s}^C \in X_s^z $ for $1 \leq s \leq m$.
\end{enumerate}
Thus $\mathcal{C}$ is gadget safe in $H$.
An example showing packing of three cycles for $(i,f(i)) \in S$ and $h(i) = 2$ can be found in Fig.~\ref{fig:hs2cycles}.
\end{proof}

\begin{figure}[htbp]
\begin{center}
  \includegraphics{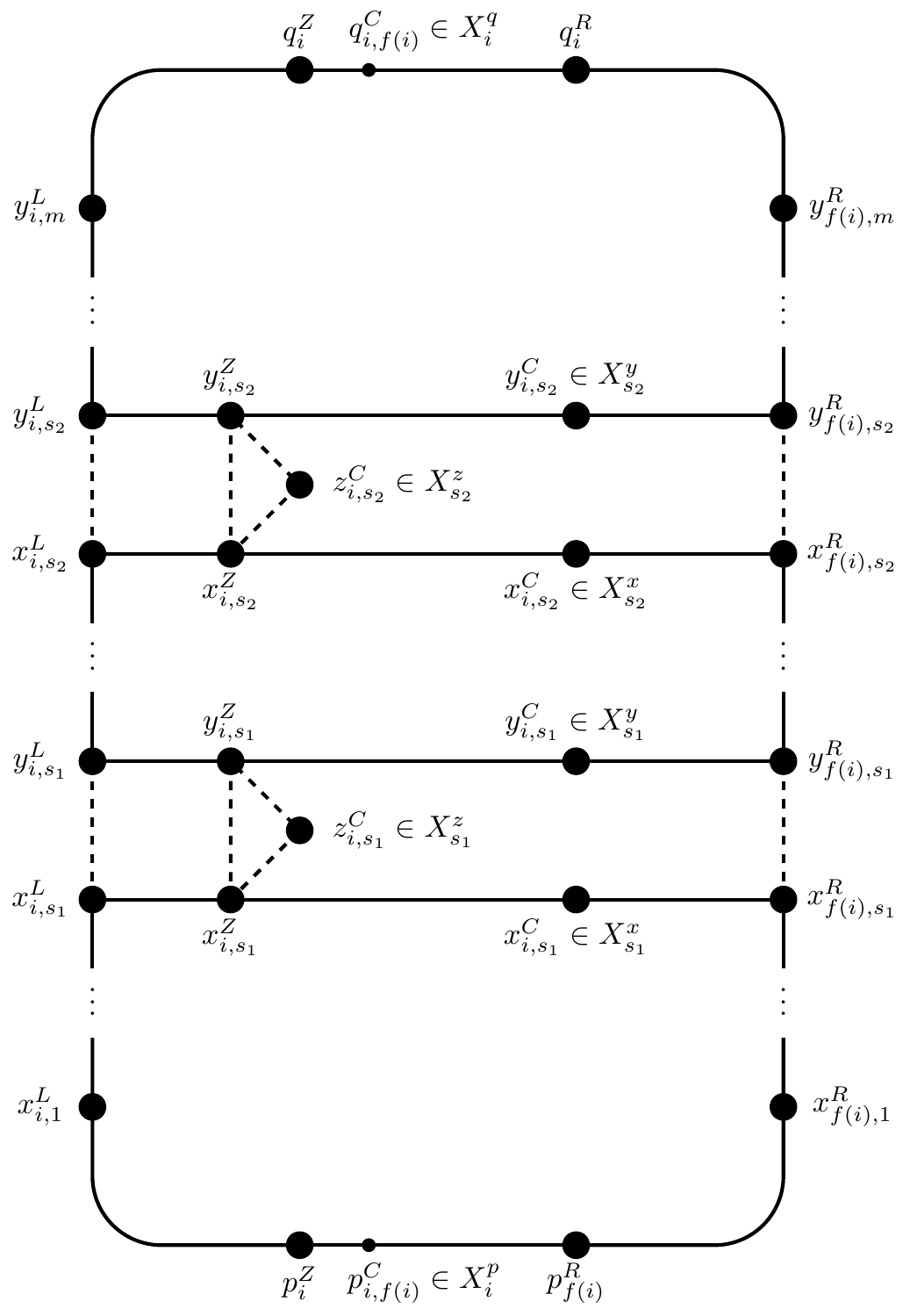}
\end{center}
\caption{An example how to pack three cycles for $(i,f(i)) \in S$ and $h(i) = 2$.}
\label{fig:hs2cycles}
\end{figure}

\subsubsection{From disjoint cycles to hitting set}\label{sec:neg:cycle2hs}

In this section we prove Point \ref{neg:prop:cycle2hs} by the following lemma:
\begin{lemma}\label{lem:neg-dir-cycles2hs}
If the graph $G$ contains at least $\ell$ vertex-disjoint cycles, then the initial \kkhittingsetname{} instance is a YES-instance.
\end{lemma}
\begin{proof}
Let $\mathcal{C}_G$ be a family of vertex-disjoint cycles in $G$ with maximum possible number of cycles.
By the assumption, $|\mathcal{C}_G| \geq \ell$. By Lemma \ref{lem:r-in-many:rev} we can assume that
$\mathcal{C}_G$ is gadget safe in $G$ and let $\mathcal{C} \subset \mathcal{C}_G$ be the
gadget safe in $H$ family of size $|\mathcal{C}_G| - \ell_G \geq k+ \sum_{s=1}^m |S_s|$,
       i.e., $\mathcal{C}$ consists of those cycles in $\mathcal{C}_G$ that are not gadget short cycles.

We now analyze the family $\mathcal{C}$.
Informally speaking, we are going to show that $\mathcal{C}$ can be placed only in the way
as in the proof of Lemma \ref{lem:neg-dir-hs2cycles}.

First let us analyze subgraph $H[L \cup R \cup C]$. Note that this subgraph is a forest containing:
\begin{enumerate}
\item $k$ paths $\mathcal{P}_i^L$ for $1 \leq i \leq k$ without the endpoints, i.e., without $p_i^Z$ and $q_i^Z$;
\item $k$ trees consisting of paths $\mathcal{P}_j^R$ ($1 \leq j \leq k$) with attached leaves $p_{i,j}^C$, $q_{i,j}^C$ for $1 \leq i \leq k$
and $x_{i,s}^C$, $y_{i,s}^C$ for $1 \leq s \leq m$, $(i,j) \in S_s$.
\item $\sum_{s=1}^m |S_s|$ isolated vertices $z_{i,s}^C$, $1 \leq s \leq m$, $(i,j) \in S_s$.
\end{enumerate}

Consider now a subgraph of $H$ induced by $L \cup R \cup C \cup \{a\}$, where $a$ is an arbitrary vertex in $Z$. Note that this graph
is a forest. Indeed, each vertex in $Z$ has at most one edge incident to each connected component of $H[L \cup R \cup C]$:
\begin{enumerate}
\item for $1 \le i \le k$ the vertex $p_i^Z$ is adjacent to $x_{i,1}^L$ on path $\mathcal{P}_i^L$ and vertices $p_{i,j}^C$ for $1 \leq j \leq k$;
\item similarly for $1 \le i \le k$ the vertex $q_i^Z$ is adjacent to $y_{i,m}^L$ on path $\mathcal{P}_i^L$ and vertices $q_{i,j}^C$ for $1 \leq j \leq k$;
\item for $1 \leq s \leq m$, $(i,j) \in S_s$ the vertex $x_{i,s}^Z$ is adjacent to $z_{i,s}^C$, $x_{i,s}^L$ and $x_{i,s}^C$;
\item similarly for $1 \leq s \leq m$, $(i,j) \in S_s$ the vertex $y_{i,s}^Z$ is adjacent to $z_{i,s}^C$, $y_{i,s}^L$ and $y_{i,s}^C$.
\end{enumerate}
  Thus, each cycle from $\mathcal{C}$ contains at least two vertices from $Z$.
  But, $|Z| = 2k + 2\sum_{s=1}^m |S_s| = 2|\mathcal{C}|$.
  Thus, each cycle in $\mathcal{C}$ contains exactly two vertices from $Z$ and $\mathcal{C}$ covers $Z$.

Let $C_i \in \mathcal{C}$ be the cycle that covers $p_i^Z$.
The vertex $p_i^Z$ has neighbours $x_{i,1}^L$ and $X_i^p$. As we are allowed to choose only one
vertex from $X_i^p$, the cycle $C_i$ contains a path $(x_{i,1}^L, p_i^Z, p_{i,f(i)}^C, p_{f(i)}^R)$ for some $1 \leq f(i) \leq k$.
If the cycle $C_i$ contains the edge $p_{f(i)}^Rp_{j,f(i)}^C$ for $j \neq i$, it contains also $p_j^Z$ and $x_{j,1}^L$.
But $x_{j,1}^L$ and $x_{i,1}^L$ are in different connected components of $H[L \cup R \cup C]$, thus $C_i$ needs to contain a third vertex in $Z$,
a contradiction. Thus, $C_i$ contains the path $(x_{i,1}^L, p_i^Z, p_{i,f(i)}^C, p_{f(i)}^R, x_{f(i),1}^R)$. Note that this in particular
implies that $f$ is a permutation of $[k]$.

We claim that $S = \{(i,f(i)): 1 \leq i \leq k\}$ is a hitting set in the initial \kkhittingsetname{} instance. It clearly
contains exactly one element from each row and from each column. We now show that $S \cap S_s \neq \emptyset$ for each
$1 \leq s \leq m$.

Let
$$Z_s = \{p_i^Z : 1 \leq i \leq k\}  \cup \{x_{i,t}^Z, y_{i,t}^Z : 1 \leq t \leq s, (i,j)\in S_t\} \subset Z$$
for $0 \leq s \leq m$ and let
$$E_s = \{y_{i,s}^Lx_{i,s+1}^L : 1 \leq i \leq k\} \cup \{y_{j,s}^Rx_{j,s+1}^R: 1 \leq j \leq k\} \subset E(H)$$
for $1 \leq s < m$ and let
$$E_0 = \{p_i^Zx_{i,1}^L: 1 \leq i \leq k\} \cup \{p_j^Rx_{j,1}^R: 1\leq i \leq k\} \subset E(H).$$
Note that for $0 \leq s < m$ the set $E_s$ is a set of $2k$ edges that separate $Z_s$ from $Z \setminus Z_s$.

We now select cycles $C(i,s) \in \mathcal{C}$ for $1 \leq i \leq k$ and $0 \leq s \leq m$ with the following property:
for $1 \leq i \leq k$ and $1 \leq s < m$ the edges $y_{i,s}^Lx_{i,s+1}^L$ and $y_{f(i),s}^Rx_{f(i),s+1}^R$ lie on $C(i,s)$
and for $1 \leq i \leq k$ the edges $p_i^Lx_{i,1}^L$ and $p_{f(i)}^Rx_{f(i),1}^R$ lie on $C(i,0)$.
Note that cycles $C(i,0) = C_i$ satisfy the above property. We select cycles $C(i,s)$ by an induction on $s$
and in the $s$-th step of the induction we prove that there exists $1 \le i \le k$ such that $(i,f(i)) \in S_s$.

Let us fix $s$, $0 \leq s < m$. We show some more properties of cycles $\{C(i,s): 1\leq i \leq k\}$ that we use
in the induction step. Note that each cycle $C(i,s)$ contains two edges from $E_s$ and each edge from $E_s$
is contained in some $C(i,s)$. Moreover, each edge in $E_s$ is in different connected component of $H[L \cup R \cup C]$.
Thus, each cycle $C(i,s)$ contains: a subpath of $\mathcal{P}_i^L$, a subpath of the connected component of $H[L \cup R \cup C]$
containing $\mathcal{P}_{f(i)}^R$, a vertex in $Z_s$ and a vertex in $Z \setminus Z_s$.
This in particular implies that cycles $\{C(i,s): 1 \leq i \leq k\}$ are pairwise different.
As $E_s$ separates $Z_s$ from $Z \setminus Z_s$,
 each cycle in $\mathcal{C} \setminus \{C(i,s): 1 \leq i \leq k\}$ contains either two vertices in $Z_s$, or two vertices in $Z \setminus Z_s$.

We now perform an induction step. Let $1 \leq s \leq m$ and assume that we have selected cycles $C(i,s-1)$ for $1 \leq i \leq k$.

Let $z_{i,s}^C$ be a (arbitrarily chosen) vertex not covered by $\mathcal{C}$, where $(i,j) \in S_s$.
Let us focus on the vertex $x_{i,s}^Z$. It has three neighbours apart from $z_{i,s}^C$: vertices $x_{i,s}^L$, $x_{i,s}^C$ and $y_{i,s}^Z$.
The vertex $x_{i,s}^C$ has degree two, and the other neighbour is $x_{j,s}^R$. As $x_{i,s}^L \in C(i,s-1)$ and
$x_{j,s}^R \in C(f^{-1}(j),s-1)$, the vertex $x_{i,s}^Z$ lies on $C(i,s-1)$ or $C(f^{-1}(j),s-1)$.
Both $C(i,s-1)$ and $C(f^{-1}(j),s-1)$ are not allowed to cover two vertices from $Z \setminus Z_{s-1}$, thus $x_{i,s}^Z$ and
$y_{i,s}^Z$ lie on different cycles in $\mathcal{C}$. But this means that $C(i,s-1)$ or $C(f^{-1}(j),s-1)$ contains
the path $(x_{i,s}^L, x_{i,s}^Z, x_{i,s}^C, x_{j,s}^R)$, thus $C(i,s-1) = C(f^{-1}(j),s-1)$. Since
$\{C(i,s-1): 1 \leq i \leq k\}$ are pairwise different, $j = f(i)$ and $(i,f(i)) \in S_s$.

Now focus on vertex $y_{i,s}^Z$. The vertex $x_{i,s}^Z$ is used on cycle $C(i,s-1)$, $y_{i,s}^Z \notin C(i,s-1)$ and we assumed the vertex $z_{i,s}^C$ is not covered by $\mathcal{C}$.
Thus, $y_{i,s}^Z$ lies on a cycle $C$ with path $(y_{i,s}^L, y_{i,s}^Z, y_{i,s}^C, y_{f(i),s}^R)$. As $x_{i,s}^L, x_{f(i),s}^R \in C(i,s-1)$,
  and we are allowed to cover only one vertex from $X_s^y$,
  $C$ contains a path $(x_{i,s+1}^L, y_{i,s}^L, y_{i,s}^Z, y_{i,s}^C, y_{f(i),s}^R, y_{f(i),s+1}^R)$ (if $s < m$) or
  $(q_i^Z, y_{i,s}^L, y_{i,s}^Z, y_{i,s}^C, y_{f(i),s}^R, q_{f(i)}^R)$ (if $s = m$).

Now focus on vertices $x_{i',s}^Z$ for $i' \neq i$. As $x_{i,s}^C$ is covered by $\mathcal{C}$, the vertex $x_{i',s}^C$ cannot be
covered too. Thus, $x_{i',s}^Z$ and $y_{i',s}^Z$ lie on the same cycle, say $C^{i'}$. As $C(i',s-1)$ is not allowed to cover two vertices
from $Z \setminus Z_{s-1}$, the vertex $x_{i',s}^L$ does not lie on $C^{i'}$, thus $C^{i'} = (x_{i',s}^Z, y_{i',s}^Z, z_{i',s}^C)$.

As vertices $x_{i',s}^Z$ and $y_{i',s}^Z$ are covered by the cycle $(x_{i',s}^Z, z_{i',s}^C, y_{i',s}^Z)$, the cycle $C(i',s-1)$ contains
the path $(x_{i',s}^L, y_{i',s}^L, x_{i',s+1}^L)$ (if $s<m$) or $(x_{i',s}^L, y_{i',s}^L, q_{i'}^Z)$ (if $s=m$).
As vertices $x_{i,s}^C$ and $y_{i,s}^C$ are covered by cycles $C(i,s-1)$ and $C$, the cycle $C(i',s-1)$ contains
path $(x_{f(i'),s}^R, y_{f(i'),s}^R, x_{f(i'),s+1}^R)$ (if $s<m$) or $(x_{f(i'),s}^R, y_{f(i'),s}^R, q_{f(i')}^R)$ (if $s=m$).

Thus we can put $C(i,s) = C$ and $C(i',s) = C(i',s-1)$ for $i' \neq i$ and the induction step is performed.
As we maintain the induction step up to $s=m$, for each $1 \leq s \leq m$ we prove that $(i,f(i)) \in S_s$, thus the initial \kkhittingsetname{} instance
is a YES-instance.
\end{proof}

\subsection{From undirected to directed \cyclepackingname{}}\label{sec:neg:undir2dir}

In this section we provide a reduction from undirected to directed \cyclepackingname{},
proving Theorem \ref{thm:negative-main} for directed \cyclepackingname{}.
\begin{lemma}\label{lem:neg:undir2dir}
Let $(G,\ell)$ be an instance of undirected \cyclepackingname{}.
Then we can construct in polynomial time an equivalent instance $(G',\ell')$ of directed \cyclepackingname{}.
Moreover, given a path decomposition of $G$ of width $p$, in polynomial time we can construct
a path decomposition of $G'$ of width at most $p+3$.
\end{lemma}
\begin{proof}
In many problems, a reduction from an undirected version to a directed one is performed
by simply changing each edge $uv$ into pair of arcs $(u,v)$ and $(v,u)$. However, in the
case of \cyclepackingname{}, such a reduction introduces many $2$-cycles $(u,v)$ that
do not have a counterpart in the original undirected graph. We circumvent this problem
by adding a directed version of $1$-in-many gadget.

To construct graph $G'$, for each edge $e = uv \in E(G)$ we introduce three extra vertices
$x_e^{uv}$, $x_e^{vu}$ and $z_e$, and we replace the edge $e$ with three cycles: $(u, x_e^{uv}, v, x_e^{vu})$,
  $(x_e^{uv}, z_e)$ and $(x_e^{vu}, z_e)$. In graph $G'$ we ask for $\ell' := \ell + |E(G)|$ vertex-disjoint cycles.
  The cycles $(x_e^{uv}, z_e)$ and $(x_e^{vu}, z_e)$ are called {\em{short cycles}}.

\begin{figure}[htbp]
\begin{center}
  \includegraphics{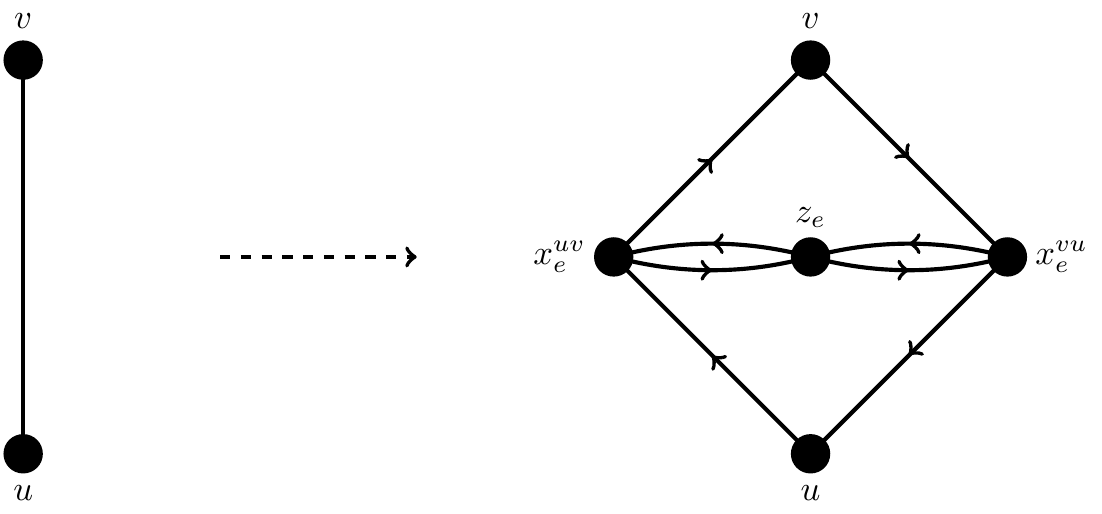}
\end{center}
\caption{The construction of the gadget replacing edge $uv$.}
\label{fig:undir2dir}
\end{figure}

First assume that we have a family $\mathcal{C}$ of $\ell$ vertex-disjoint cycles in $G$.
For each cycle $C \in \mathcal{C}$, we orient it in an arbitrary way, and translate it 
into a cycle $C'$ in $G'$ (i.e., if $C$ goes from $u$ to $v$ via edge $e=uv$, then in $G'$
    the cycle $C'$ uses arcs $(u, x_e^{uv})$ and $(x_e^{uv},v)$). In this way we create
a family $\mathcal{C'}$ of $\ell$ vertex-disjoint cycles in $G'$. This family has a property
that for each $e \in E(G)$ at least one vertex $x_e^{uv}$ and $x_e^{vu}$ is not covered. Thus
we can add a cycle $(x_e^{uv}, z_e)$ or $(x_e^{vu}, z_e)$ to $\mathcal{C}'$, obtaining a family
of $\ell'$ cycles.

In the other direction, let $\mathcal{C'}$ be a family of vertex-disjoint cycles in $G'$
that contains maximum possible number of short cycles
among families of vertex-disjoint cycles of maximum possible size. Assume $|\mathcal{C'}| \geq \ell'$.

We claim that $\mathcal{C'}$ contains $|E(G)|$ short cycles. As there are $|E(G)|$ vertices $z_e$,
it may not contain more. Assume that it contains less than $|E(G)|$ short cycles.
Let $e = uv$ be an edge, such that $z_e$ is not covered by a short cycle.
If $z_e$ is not covered by $\mathcal{C}'$, we can add the short cycle $(x_e^{uv},z_e)$ to $\mathcal{C}'$,
   possibly deleting a cycle covering $x_e^{uv}$. Otherwise, if $z_e$ is covered by a cycle $C$, then $x_e^{uv}$ or $x_e^{vu}$ (say $x_e^{uv}$)
   also belongs to $C$. But then we can replace $C$ with the cycle $(z_e, x_e^{uv})$.
   In both cases, we increase the number of short cycles in $\mathcal{C}'$ while not decreasing
   its size, a contradiction.

Let $\mathcal{C}$ be the other $\ell$ cycles in $\mathcal{C}'$ that are not short cycles.
Each such cycle $C'$ does not cover vertices $z_e$, thus if it covers $x_e^{uv}$, it contains
a subpath $(u,x_e^{uv}, v)$. Moreover, either $x_e^{uv}$ or $x_e^{vu}$ is covered by a short cycle in
$\mathcal{C}'$. Thus $C'$ translates into a cycle $C$ in $G$, by taking
an edge $e$ for each vertex $x_e^{uv}$ visited by $C'$. In this way we obtain
$\ell$ vertex-disjoint cycles in $G$.

We are left with the pathwidth bound. Assume we have a path decomposition of $G$
of width $p$. We construct a path decomposition of $G'$ in the following way.
For each $e \in E(G)$, we pick a bag $V_e$ that covers $e$. We create
a new bag $V_e' := V_e \cup \{x_e^{uv}, x_e^{vu}, z_e\}$ and insert it into the
path decomposition near $V_e$. It is easy to see that this is a proper path decomposition
of $G'$ and its width is at most $p+3$.
\end{proof}

\subsection{From \cyclepackingname{} to \maxcyclecovername{}}\label{sec:neg:packing2cover}

In this section we provide a reduction from \cyclepackingname{} to \maxcyclecovername{}
that proves Theorem \ref{thm:negative-main} for \maxcyclecovername{}, both in directed and undirected setting.
Formally, we prove the following lemma
\begin{lemma}\label{lem:neg:packing2cycle}
Let $(G,\ell)$ be an instance of (directed or undirected) \cyclepackingname{}.
Then we can construct in polynomial time an equivalent instance $(G', \ell')$ of (directed or undirected, respectively) \maxcyclecovername{}.
Moreover, given a path decomposition of $G$ of width $p$, in polynomial time we can construct a path decomposition of $G$ of width at most $p+3$.
\end{lemma}
\begin{proof}
To construct $G'$, we take $G$ and for each $v \in V(G)$ we introduce three new vertices $a_v$, $b_v$ and $c_v$ and
five new arcs $(a_v,b_v)$, $(b_v,c_v)$, $(c_v, a_v)$, $(c_v, v)$ and $(v, a_v)$
(in the undirected setting, these arcs are edges without direction). We ask for a cycle cover with at least $\ell' := \ell + |V(G)|$ cycles.

\begin{figure}[htbp]
\begin{center}
  \includegraphics{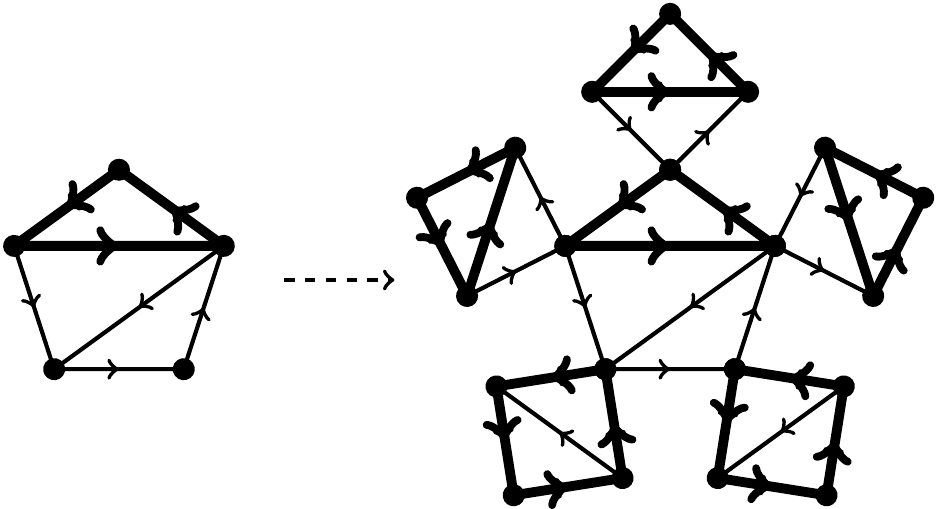}
\end{center}
\caption{An example of the reduction from \cyclepackingname{} to \maxcyclecovername{}
  in the directed setting, together with the conversion of cycle families.}
\label{fig:packing2cover}
\end{figure}

Let $\mathcal{C}$ be a set of $\ell$ vertex-disjoint cycles in $G$. To construct a cycle cover of size $\ell' = \ell + |V(G)|$ in $G'$, we take the cycles in $\mathcal{C}$
and for each vertex $v \in V(G)$: if $v$ is covered by $\mathcal{C}$, we take the cycle $(a_v, b_v, c_v)$, and otherwise we take the cycle $(a_v, b_v, c_v, v)$.

In the other direction, let $\mathcal{C}'$ be a cycle cover of $G'$ with at least $\ell'$ cycles. For each $v \in V(G)$ let $C_v$ be the cycle that covers $b_v$.
Note that $a_v, c_v \in C_v$. Thus $\mathcal{C}' \setminus \{C_v: v \in V(G)\}$ is a family of at least $\ell' - |V(G)| = \ell$ vertex-disjoint cycles in $G$.

We are left with the pathwidth bound. Assume we have a path decomposition of graph $G$ of width $p$. For each vertex $v \in V(G)$, we pick a single bag $V_v$
that contains $v$. We introduce a new bag $V_v' = V_v \cup \{a_v, b_v, c_v\}$ and insert it near $V_v$ in the path decomposition.
It is easy to see that this is a proper path decomposition of $G'$ and its width is at most $p+3$.
\end{proof}

\section{Negative results under SETH}\label{sec:negatives:seth}

\newcommand{\blocksize}{\eta} 
\newcommand{\blockvars}{\beta} 
\newcommand{\numcopies}{TODO} 
\newcommand{\numcopiessmall}{TODO} 
\newcommand{\iF}{t} 
\newcommand{\iv}{\ell} 
\newcommand{\ivup}{\alpha} 
\newcommand{\iB}{k} 
\newcommand{\iC}{i} 
\newcommand{\ib}{j} 
\newcommand{\groupg}{\mathbf{B}} 
\newcommand{\goalsize}{K} 
\newcommand{\terms}{T} 
\newcommand{\ivupseq}{\mathbf{A}} 

In this section we
 provide an evidence that our algorithms for \cvertexcover, \cdomset, \cfvs, \coct, \fvs, \steinertree and \exactleaf are probably optimal,
i.e., the exponential in treewith term has an optimal base of the exponent. In other words,
  we prove here Theorem \ref{thm:lower-seth-main}. The lower bounds are proven under the
assumption of the Strong Exponential Time Hypothesis, that is, we prove that a faster algorithm solving one of the considered problems
would imply a faster than exhaustive search algorithm for SAT.

We first note that the result for \exactleaf follows from the result for \cdomset,
as \exactleaf is not easier than \maxleaf, which is equivalent to \cdomset \cite{maxleaf-cds}.
Second, note that the lower bound for \steinertree follows from the lower bound for \cvertexcover
by the following simple reduction.
\begin{lemma}
Given a \cvertexcover instance $(G,k)$ together with a tree decomposition of $G$ of width $t$
(or path decomposition of $G$ of width $p$), one can construct in polynomial time an
equivalent instance $(G',T,k')$ of \steinertree together with a tree decomposition
of $G'$ of width $t$ (or path decomposition of $G$ of width at most $p+1$, respectively).
\end{lemma}
\begin{proof}
  Given the instance $(G,k)$, we construct the instance $(G',T,k')$ by subdiving
  each edge of $G$ with a terminal and setting $k' = |T| + k$
  (i.e., we replace each edge $e=uv$ in $G$ with a path $u,w_e,v$ and we take
   $T = \{w_e: e \in E(G)\}$).
	It is easy to see that $G$ contains a connected vertex cover
  of size at most $k$ if and only if $(G',T)$ contains a Steiner tree of 
	cardinality at most
  $k+|T| = k+|E(G)|$.

  Moreover, given a tree decomposition $\treedecomp$ of $G$ of width $t$, 
	we can construct a tree decomposition
  $\treedecomp'$ of $G'$ of width at most $t$ in the following manner. 
	If $G$ is a forest, $G'$ is a forest
  too and we construct the optimal tree decomposition of $G'$ in 
	polynomial time.
  Otherwise, $t \geq 2$ and for each $e \in E(G)$ we pick one fixed node 
	$x(e)$ whose bag covers $e$.
  We first set $\treedecomp' := \treedecomp$ and then, for each $uv=e \in E(G)$, 
	we attach a leaf node $y(e)$ to the node $x(e)$
  with the bag $\bag{y(e)}=\{u,v,w_e\}$.

  In the case of a path decomposition $\treedecomp$ of $G$ of width $p$,
	we can construct a path decomposition
  $\treedecomp'$ of $G'$ of width at most $p+1$ in the following manner. 
  For each $e \in E(G)$ we pick one fixed node 
	$x(e)$ whose bag covers $e$.
  We first set $\treedecomp' := \treedecomp$ and then, sequentially for each $uv=e \in E(G)$, 
	we create a new node $y(e)$ with the bag $\bag{y(e)}=\bag{x(e)} \cup \{w_e\}$
  and insert it immediately after the node $x(e)$ on the path decomposition.
\end{proof}

Thus we are left with the first five problems of Theorem \ref{thm:lower-seth-main}.
In the proofs we follow the same approach as Lokshtanov et al~\cite{treewidth-lower} in the lower bounds for problems without the connectivity requirement.
We mostly base on the \vertexcover and \domset lower bounds of \cite{treewidth-lower}, however our gadgets are adjusted to the considered problems.
For each problem we show a polynomial-time construction that, given a SAT instance with $n$ variables, constructs an equivalent instance of the considered problem,
together with a path decomposition of the underlying graph of width roughly $\log_3(2^n) = n/\log 3$
in the case of \cvertexcover and \fvs and of width roughly $\log_4(2^n) = n/2$
in the case of the other problems.
Thus, each of the following subsections consists of three parts. First, we give a construction procedure that, given a SAT formula $\Phi$, produces an instance of the considered problem.
Second, we prove that the constructed instance is equivalent to the formula $\Phi$. Finally, we show the claimed pathwidth bound.

Similarly as in \cite{treewidth-lower}, we prove pathwidth bounds for the constructed graphs using mixed search game.
Let us recall the definition from \cite{treewidth-lower}.
\begin{definition}[\cite{mixed-search,treewidth-lower}]
In a {\em{mixed search game}}, a graph $G$ is considered as a system of tunnels. Initially, all edges are
contaminated by a gas. An edge is cleared by placing searchers at both its end-points simultaneously
or by sliding a searcher along the edge. A cleared edge is re-contaminated if there is a path from
an uncleared edge to the cleared edge without any searchers on its vertices or edges. A search is a
sequence of operations that can be of the following types: (a) placement of a new searcher on a vertex;
(b) removal of a searcher from a vertex; (c) sliding a searcher on a vertex along an incident edge and
placing the searcher on the other end. A {\em{search strategy}} is winning if after its termination all edges are
cleared. The {\em{mixed search number}} of a graph G, denoted $\textrm{ms}(G)$, is the minimum number of searchers
required for a winning strategy of mixed searching on $G$.
\end{definition}
\begin{proposition}[\cite{mixed-search}]
  For a graph G, $\textrm{pw}(G) \leq  \textrm{ms}(G) \leq \textrm{pw}(G) + 1$.
\end{proposition}
Moreover, in each case considered by us, the presented cleaning strategy easily yield a polynomial time
algorithm that constructs a path decomposition of $G$ of width not greater than the number of searchers used.

\renewcommand{\numcopies}{\ensuremath{m(2\blocksize n'+1)}\xspace}
\renewcommand{\numcopiessmall}{\ensuremath{(2\blocksize n'+1)}\xspace}

\subsection{\cvertexcover}

\begin{theorem}\label{thm:cvc-seth}
Assuming SETH, there cannot exist a constant $\eps>0$ and an algorithm that given an instance $(G=(V,E), k)$ together with a path decomposition of the graph $G$ of width $p$ solves the \cvertexcover problem in $(3-\eps)^p |V|^{O(1)}$ time.
\end{theorem}

\paragraph{Construction}
Given $\eps > 0$ and an instance $\Phi$ of SAT with $n$ variables and $m$ clauses
we construct a graph $G$ as follows.
We first choose a constant integer $\blocksize$, which value depends on $\eps$ only.
The exact formula for $\blocksize$ is presented later.
We partition variables of $\Phi$ into groups $F_1,\ldots,F_{n'}$,
each of size at most $\blockvars = \lfloor \log 3^\blocksize \rfloor$,
hence $n'=\lceil n/\blockvars \rceil$.
Note that now $\blocksize n' \sim n/\log 3$, the pathwidth of $G$ will be roughly $\blocksize n'$.

First, we add to the graph $G$ two vertices $r$ and $r^\ast$, connected by an edge.
In the graph $G$ the vertex $r^\ast$ will be of degree one, thus any connected vertex cover of $G$
needs to include $r$. The vertex $r$ is called a {\em{root}}.

Second, we take $a = \numcopies$ and for each $1 \leq \iF \leq n'$
and $1 \leq \iv \leq \blocksize$ we create a path $\mathcal{P}_{\iF,\iv}$ consisting
of $2a$ vertices $v_{\iF,\iv,\iB}^\ivup$, $0 \leq \iB < a$ and $1 \leq \ivup \leq 2$,
   arranged in the following order:
   $$v_{\iF,\iv,0}^1, v_{\iF,\iv,0}^2,v_{\iF,\iv,1}^1,\ldots,v_{\iF,\iv,a-1}^1,v_{\iF,\iv,a-1}^2.$$
Furthermore we connect all vertices
$v_{\iF,\iv,\iB}^2$ ($0 \leq \iB < a$) and $v_{\iF,\iv,0}^1$ to the root $r$.
To simplify further notation we denote $v_{\iF,\iv,a}^1 = r$.
Let $\mathcal{V}$ be the set of all vertices on all paths $\mathcal{P}_{\iF,\iv}$.

We now provide a description of a group gadget $\groupg_{\iF,\iB}$,
which will enable us to encode $2^\blockvars$ possible assignments of one group
of $\blockvars$ variables.
Fix a block $F_\iF$, $1 \leq \iF \leq n'$, and a position $\iB$, $0 \leq \iB < a$.
For each $1 \leq \iv \leq \blocksize$ we create three vertices $h_{\iF,\iv,\iB}^\ivup$,
$1 \leq \ivup \leq 3$ that are pairwise adjacent and all are adjacent to the root $r$.
Moreover, we add edges $h_{\iF,\iv,\iB}^1v_{\iF,\iv,\iB}^1$, $h_{\iF,\iv,\iB}^2v_{\iF,\iv,\iB}^2$
and $h_{\iF,\iv,\iB}^3v_{\iF,\iv,\iB+1}^1$. 
Let $\mathcal{H}_{\iF,\iv,\iB} = \{h_{\iF,\iv,\iB}^\ivup: 1 \leq \ivup \leq 3\}$,
 $\mathcal{H}_{\iF,\iB} = \bigcup_{\iv=1}^\blocksize \mathcal{H}_{\iF,\iv,\iB}$.
 and $\mathcal{H} = \bigcup_{\iF=1}^{n'} \bigcup_{\iB=0}^{a-1} \mathcal{H}_{\iF,\iB}$.
Note that each (connected) vertex cover
in $G$ needs to include at least two out of three vertices from each set
$\mathcal{H}_{\iF,\iv,\iB}$.

In order to encode $2^\blockvars$ assignments we consider subsets of $\mathcal{H}_{\iF,\iB}$
that contain {\em{exactly one}} vertex out of each set $\mathcal{H}_{\iF,\iv,\iB}$.
For a sequence $S=(s_1,\ldots,s_\blocksize) \in \{1,2,3\}^\blocksize$ by $S(\mathcal{H}_{\iF,\iB})$
we denote
the set $\{h_{\iF,\iv,\iB}^{s_\iv} : 1 \le \iv \le \blocksize\}$.
For each sequence $S \in \{1,2,3\}^\blocksize$ we add three vertices
$x_{\iF,\iB}^S$, $x_{\iF,\iB}^{S\ast}$ and $y_{\iF,\iB}^S$,
where $x_{\iF,\iB}^S$ is also adjacent to all the vertices of $S(\mathcal{H}_{\iF,\iB})$
(recall that $\blocksize$ and $\blockvars$ are constants depending only on $\eps$).
We add edges $x_{\iF,\iB}^Sx_{\iF,\iB}^{S\ast}$, $x_{\iF,\iB}^Sy_{\iF,\iB}^S$ and $y_{\iF,\iB}^Sr$.
In the graph $G$ the vertices $x_{\iF,\iB}^{S\ast}$ are of degree one, thus any connected
vertex cover in $G$ needs to include all vertices $x_{\iF,\iB}^S$.
Let $\mathcal{Y}_{\iF,\iB} = \{y_{\iF,\iB}^S: S \in \{1,2,3\}^\blocksize\}$
for $1 \leq \iF \leq n'$, $0 \leq \iB < a$
and $\mathcal{Y} = \bigcup_{\iF=1}^{n'} \bigcup_{\iB=0}^{a-1} \mathcal{Y}_{\iF,\iB}$.

Additionally we add two adjacent vertices $z_{\iF,\iB}$ and $z_{\iF,\iB}^\ast$ and
connect $z_{\iF,\iB}$ to vertices $y_{\iF,\iB}^S$ for all $S \in \{1,2,3\}^\blocksize$.
Again, the vertex $z_{\iF,\iB}^\ast$ is of degree one in $G$ and forces $z_{\iF,\iB}$
to be included in any connected vertex cover of $G$.

\begin{figure}
\begin{center}
\includegraphics{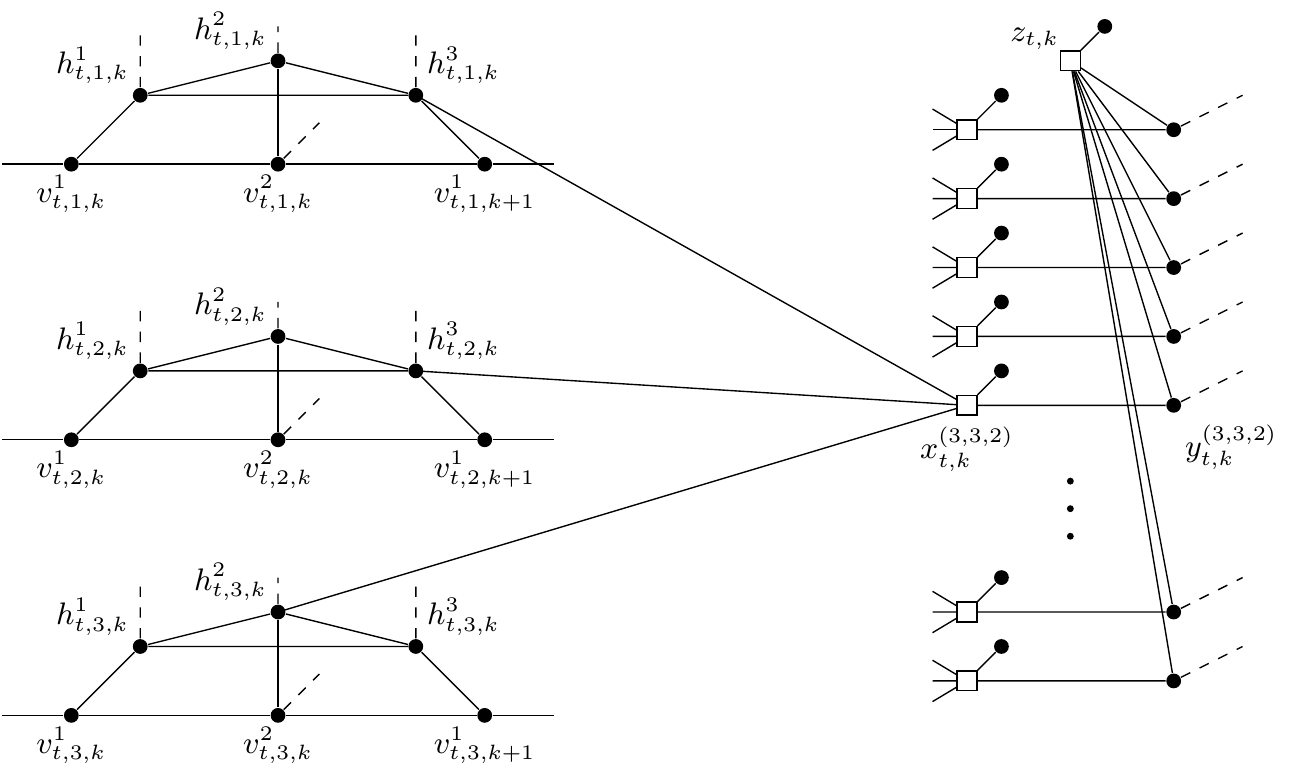}
\caption{Group gadget $\groupg_{\iF,\iB}$ for $\blocksize=3$. Dashed edges are connecting a vertex with the root $r$. Vertices that have a pendant neighbour and thus
  need to be included in any connected vertex cover of $G$
  are presented as squares.  }
  \label{fig:seth-cvc}
  \end{center}
  \end{figure}

The above step finishes the construction of the group gadgets needed to encode
an assignment and now we add vertices used
to check the satisfiability of the formula $\Phi$.
Observe that for a group of variables $F_\iF$
there are at most $2^{\blockvars}$ possible assignments
and there are $3^\blocksize \ge 2^{\blockvars}$ vertices $x_{\iF,\iB}^S$
for sequences $S$ from the set $\{1,2,3\}^\blocksize$ in each group gadget $\groupg_{\iF,\iB}$,
hence we can assign a {\em unique} sequence $S$ to each assignment.
Let $C_0,\ldots,C_{m-1}$ be the clauses of the formula $\Phi$.
For each clause $C_\iC$ we create $\numcopiessmall$ pairs of adjacent vertices
$c_{\iC,\ib}$ and $c_{\iC,\ib}^\ast$, one for each $0 \le \ib < \numcopiessmall$.
The vertex $c_{\iC,\ib}^\ast$ is of degree one in $G$ and therefore forces any connected vertex
cover of $G$ to include $c_{\iC,\ib}$.
The vertex $c_{\iC,\ib}$ can only be connected to
gadgets $\groupg_{\iF,m\ib+\iC}$ for $1 \le \iF \le n'$.
For each group of variables $F_\iF$ we consider all sequences
$S \in \{1,2,3\}^\blocksize$ that correspond 
to an assignment of $F_\iF$ satisfying the clause $C_\iC$ (i.e., one of the variables of $F_\iF$ is assigned a value such that $C_\iC$ is already satisfied).
For each such sequence $S$ and for each $0 \le \ib < \numcopiessmall$
we add an edge $y_{\iF,m\ib+\iC}^S c_{\iC,\ib}$.

We can view the whole construction as a matrix of group gadgets, where each row corresponds to some group of variables $F_\iF$
and each column is devoted to some clause in such a way that each clause gets $\numcopiessmall$ private columns (but not consecutive)
of the group gadget matrix, as in Figure \ref{fig:seth-cvcmatrix}.

\begin{figure}[htbp]
\begin{center}
\includegraphics{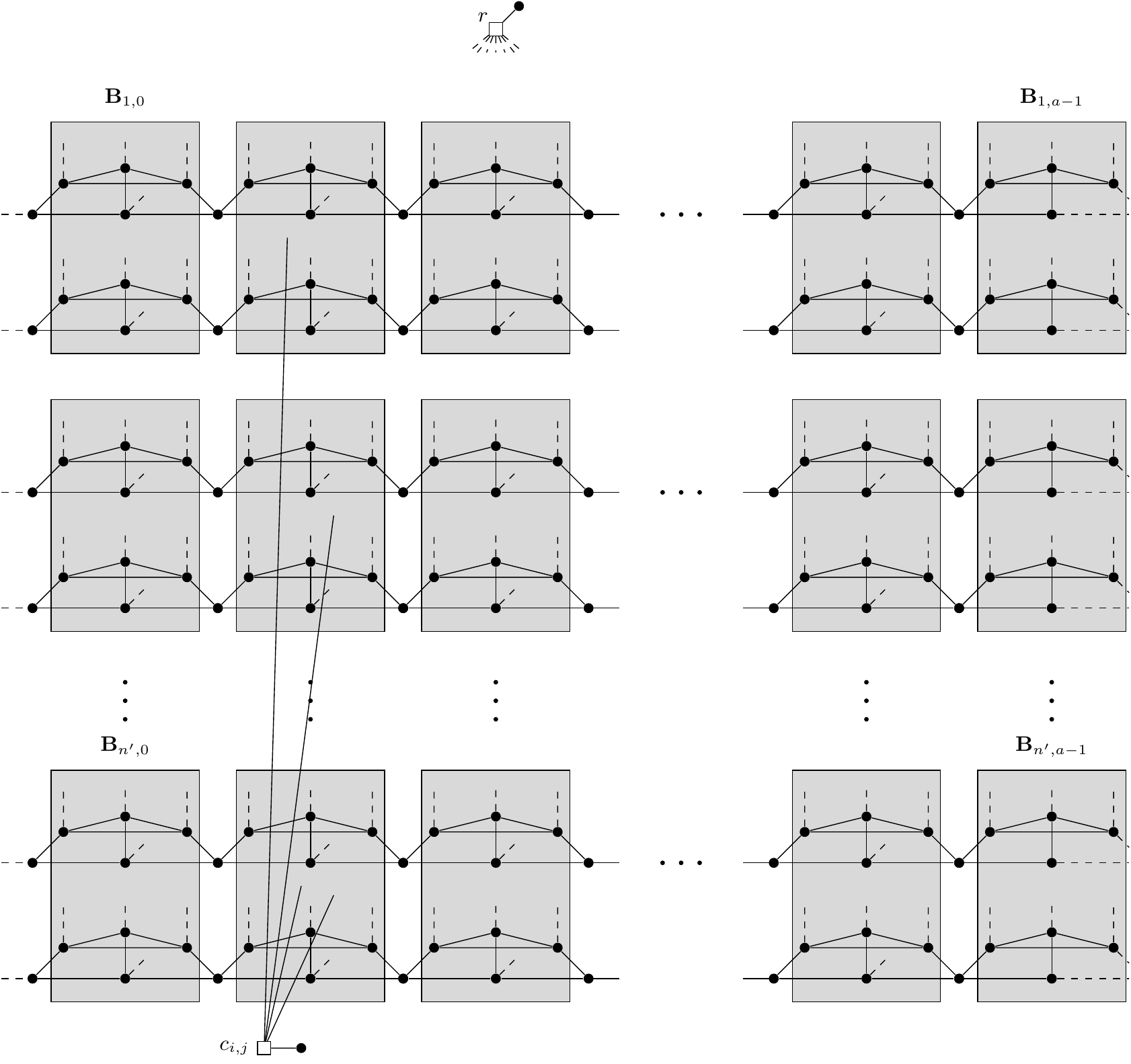}
\caption{Gray rectangles represent group gadgets.
    Dashed edges connect have
  the root vertex $r$ as one of the endpoints. 
  Vertices that have a pendant neighbour and thus
    need to be included in any connected vertex cover of $G$
    are presented as squares.
}
\label{fig:seth-cvcmatrix}
\end{center}
\end{figure}

Finally, let $\goalsize = \blocksize n' \cdot 3a + (3^\blocksize + 2)n'a + a + 1$ be the size
of the vertex cover we ask for.

\paragraph{Correctness}

\begin{lemma}\label{lem:cvc-seth:corr1}
If $\Phi$ has a satisfying assignment, then there exists a connected vertex cover in $G$ of size $\goalsize$.
\end{lemma}

\begin{proof}
Given a satisfying assignment $\phi$ of the formula $\Phi$ we construct a connected
vertex cover $X \subseteq V$ as follows. 
Let $X_{\textrm{force}}$ be the set of vertices 
that are forced to be in any connected vertex cover of $G$,
that is $r$, $c_{\iC,\ib}$ for $0 \leq \iC < m$, $0 \leq \ib < \numcopiessmall$,
$x_{\iF,\iB}^S$ for $1 \leq \iF \leq n'$,
$0 \leq \iB < \numcopies$, $S \in \{1,2,3\}^\blocksize$
and $z_{\iF,\iB}$ for $1 \leq \iF \leq n'$, $0 \leq \iB < \numcopies$.
Note that $|X_{\textrm{force}}| = 1 + a + (3^\blocksize+1)n'a$.

For each group of variables $F_\iF$ we consider the sequence $S_\iF \in \{1,2,3\}^\blocksize$
which corresponds to the restriction of the assignment $\phi$ to the variables of $F_\iF$.
Let
$$X_h = \bigcup_{\iF=1}^{n'} \bigcup_{\iB=0}^{a-1} \bigcup_{\iv=1}^\blocksize \{h_{\iF,\iv,\iB}^\ivup: \ivup \neq S_\iF(\iv)\}.$$
The set $X_h$ includes exactly two vertices out of each set $\mathcal{H}_{\iF,\iv,\iB}$, thus
$|X_h| = \blocksize n' \cdot 2a$. Moreover, we define the set $X_v$ to contain all vertices
$v_{\iF,\iv,\iB}^2$
if $S_\iF(\iv) = 2$, and all vertices $v_{\iF,\iv,\iB}^1$ otherwise
($1 \leq \iF \leq n'$, $1 \leq \iv \leq \blocksize$, $0 \leq \iB < a$). The set $X_v$ includes every other vertex on each path $\mathcal{P}_{\iF,\iv}$,
thus $|X_v| = \blocksize n' a$.

Finally, let us define the set $X_y$ to be the set of all vertices
$y_{\iF,\iB}^{S_\iF}$ for all $1 \leq \iF \leq n'$ and $0 \leq \iB < a$.
Let $X = X_{\textrm{force}} \cup X_h \cup X_v \cup X_y$. 
Note that as $|X_y| = n'a$ we have $|X| = \goalsize$. We now verify that $X$ is a
connected vertex cover of $G$.

First, we verify that $G \setminus X$ is an edgeless graph.
\begin{enumerate}
\item The vertices $r^\ast$, $c_{\iC,\ib}^\ast$, $x_{\iF,\iB}^{S\ast}$
and $z_{\iF,\iB}^\ast$ are isolated in $G \setminus X$, as their single neighbours in $G$
are included in $X_{\textrm{force}}$.
\item The vertices $v_{\iF,\iv,\iB}^\ivup$ that are not in $X$ are isolated in $G \setminus X$,
as $X_v$ contains every other vertex on each path $\mathcal{P}_{\iF,\iv}$
and we chose $\ivup$ in such a manner that the neighbours of $v_{\iF,\iv,\iB}^\ivup$
from $\mathcal{H}$ are in $X_h$.
\item The vertices $h_{\iF,\iv,\iB}^{S_\iF(\iv)}$
are isolated in $G \setminus X$, since their single neighours on paths $\mathcal{P}_{\iF,\iv}$
are in $X_v$ and all other neighbours of $h_{\iF,\iv,\iB}^{S_\iF(\iv)}$ are
in $X_{\textrm{force}}$ and in $X_h$.
\item Finally, the vertices $y_{\iF,\iB}^S$ are isolated in $G \setminus X$
since their neighbourhoods are contained in $X_{\textrm{force}}$.
\end{enumerate}

To finish the proof we need to verify that $G[X]$ is connected.
We ensure it by showing that in $G[X]$ each vertex in $X$ is connected to the root $r$.
\begin{enumerate}
\item The claim is obvious for $X_h$ and $X_y$, as they are contained in the neighbourhood of $r$.
\item If vertices $v_{\iF,\iv,\iB}^2$ belong to $X_v$, they are connected to $r$ by
direct edges. Otherwise, the vertices $v_{\iF,\iv,\iB}^1$
are connected via vertices $h_{\iF,\iv,\iB}^1$ or
$h_{\iF,\iv,\iB-1}^3$ (with the exception of vertices $v_{\iF,\iv,0}^1$ that are connected directly).
\item Each vertex $x_{\iF,\iB}^S$ for $S \neq S_\iF$ is connected to $r$
via any vertex $h_{\iF,\iv,\iB}^{S(\iv)}$ for which $S(\iv) \neq S_\iF(\iv)$.
\item Vertices $x_{\iF,\iB}^{S_\iF}$ and $z_{\iF,\iB}$ are connected to $r$ via $y_{\iF,\iB}^{S_\iF}$.
\item Finally, each vertex $c_{\iC,\iB}$ is connected to $r$ via any vertex $y_{\iF,\iB}^{S_\iF}$
for which the assignment $\phi$ on variables from $F_\iF$ satisfies the clause $C_\iC$.
\end{enumerate}
\end{proof}

\begin{lemma}\label{lem:cvc-seth:corr2}
If there exists a connected vertex cover $X$ of size at most $\goalsize$ in the graph $G$,
then $\Phi$ has a satisfying assignment.
\end{lemma}

\begin{proof}
As in the previous lemma, let $X_{\textrm{force}}$ be the set of vertices 
that are forced to be in any connected vertex cover of $G$,
that is $r$, $c_{\iC,\ib}$ for $1 \leq \iC \leq m$, $0 \leq \ib < \numcopiessmall$,
$x_{\iF,\iB}^S$ for $1 \leq \iF \leq n'$,
$0 \leq \iB < \numcopies$, $S \in \{1,2,3\}^\blocksize$
and $z_{\iF,\iB}$ for $1 \leq \iF \leq n'$, $0 \leq \iB < \numcopies$.
Note that $|X_{\textrm{force}}| = 1 + a + (3^\blocksize+1)n'a$ and $X_{\textrm{force}} \subseteq X$.

Let $X_v = X \cap \mathcal{V}$, $X_h = X \cap \mathcal{H}$ and $X_y = X \cap \mathcal{Y}$.
Note that
\begin{enumerate}
\item $X_v$ needs to include at least $a$ vertices from each path
$\mathcal{P}_{\iF,\iv}$, thus $|X_v| \geq \blocksize n' a$.
\item $X_h$ needs to include at least two vertices out of each set $\mathcal{H}_{\iF,\iv,\iB}$,
  thus $|X_h| \geq \blocksize n' \cdot 2a$.
\item $X_y$ needs to include at least one vertex $y_{\iF,\iB}^S$ for each $1\leq \iF \leq n'$
and $0 \leq \iB < a$ to ensure that the vertex $z_{\iF,\iB}$ is connected to the root $r$
in $G[X]$. Thus $|X_y| \geq n'a$.
\end{enumerate}
As $|X| \leq \goalsize$, we have $|X_v| = \blocksize n' a$,
   $|X_h| = \blocksize n' \cdot 2a$, $|X_y| = n'a$ and $|X| = \goalsize$.

As $|X_h| = \blocksize n' \cdot 2a$, for each $1 \leq \iF \leq n'$, $1 \leq \iv \leq \blocksize$
and $0 \leq \iB < a$ we have $|X_h \cap \mathcal{H}_{\iF,\iv,\iB}| = 2$.
This allows us to define a sequence $S_{\iF,\iB} \in \{1,2,3\}^\blocksize$ 
satisfying $h_{\iF,\iv,\iB}^{S_{\iF,\iB}(\iv)} \notin X_h$ for $1 \leq \iv \leq \blocksize$.

Note that the vertex $x_{\iF,\iB}^{S_{\iF,\iB}} \in X_{\textrm{force}}$ does not have
a neighbour in $X_h$, thus, to connect it to the root $r$, we need to have
$y_{\iF,\iB}^{S_{\iF,\iB}} \in X_y$. As $|X_y| = n'a$, we infer that
$|\mathcal{Y}_{\iF,\iB} \cap X| = 1$ for all $1 \leq \iF \leq n'$ and $0 \leq \iB < a$,
  i.e., $y_{\iF,\iB}^S \in X$ if and only if $S = S_{\iF,\iB}$.

We now show that for fixed $\iF$ the sequences $S_{\iF,\iB}$ cannot differ much for $0 \leq \iB < a$.
As $|X_v| = \blocksize n' a$, for each $1 \leq \iF \leq n'$, $1 \leq \iv \leq \blocksize$
and $0 \leq \iB \leq a$ we have that $|X_v \cap \{v_{\iF,\iv,\iB}^\ivup :1 \leq \ivup \leq 2\}| = 1$.
Let $\ivup(\iF,\iv,\iB)$ be such that  $v_{\iF,\iv,\iB}^{\ivup(\iF,\iv,\iB)} \in X_v$.
Now note that for $1 \leq \iF \leq n'$, $1 \leq \iv \leq \blocksize$ and $0 \leq \iB < a-1$:
\begin{enumerate}
\item If $S_{\iF,\iB}(\iv) = 3$, then $\ivup(\iF,\iv,\iB+1) = 1$, as otherwise
the edge $v_{\iF,\iv,\iB+1}^1h_{\iF,\iv,\iB}^3$ is not covered by $X$.
Moreover, $h_{\iF,\iv,\iB+1}^1 \in X_h$, as otherwise $v_{\iF,\iv,\iB+1}^1$ is 
isolated in $G[X]$, and $h_{\iF,\iv,\iB+1}^2 \in X_h$, as otherwise the
edge $v_{\iF,\iv,\iB+1}^2h_{\iF,\iv,\iB+1}^2$ is not covered by $X$.
Thus $S_{\iF,\iB+1}(\iv) = 3$ as well.
\item If $S_{\iF,\iB}(\iv) = 1$ then $\ivup(\iF,\iv,\iB) = 1$, as otherwise the 
    edge $v_{\iF,\iv,\iB}^1h_{\iF,\iv,\iB}^1$ is not covered by $X$.
    Thus $\ivup(\iF,\iv,\iB+1) = 1$, as otherwise the edge $v_{\iF,\iv,\iB}^2v_{\iF,\iv,\iB+1}^1$
    is not covered by $X$, and $S_{\iF,\iB+1}(\iv) \neq 2$,
    as otherwise the edge $v_{\iF,\iv,\iB+1}^2h_{\iF,\iv,\iB+1}^2$ is not covered by $X$.
\end{enumerate}

For fixed $1 \leq \iF \leq n'$ and $1 \leq \iv \leq \blocksize$ define the sequence
$\hat{S}_{\iF,\iv}(\iB) = S_{\iF,\iB}(\iv)$. From the above arguments
we infer that the sequence $\hat{S}_{\iF,\iv}(\iB)$ cannot change more than twice.
As $a = \numcopies$, we conclude that there exists an index $0 \leq \ib < \numcopiessmall$
such that for all $1 \leq \iF \leq n'$, $1 \leq \iv \leq \blocksize$ the sequence
$\hat{S}_{\iF,\iv}(m\ib + \iC)$ is constant for $0 \leq \iC < m$.

We create now an assignment $\phi$ by taking, for each group of variables $F_\iF$,
an assignment corresponding to the sequence $S_{\iF,m\ib}$.
Now we prove that $\phi$ satisfies $\Phi$. Take any clause $C_\iC$, $0 \leq \iC < m$,
and focus on the vertex $c_{\iC,\ib} \in X_{\textrm{force}}$. The vertex $c_{\iC,\ib}$
needs to be connected to the root $r$ in $G[X]$, thus for some $1 \leq \iF \leq n'$ and $S \in \{1,2,3\}^\blocksize$
we have $y_{\iF,m\ib+\iC}^S \in X$ and the assignment of $F_\iF$
that corresponds to $S$ satisfies $C_\iC$. However, we know that
$y_{\iF,m\ib+\iC}^S \in X$ implies that $S = S_{\iF,m\ib+\iC}$, and thus
$\phi$ satisfies $C_\iC$.
\end{proof}

\paragraph{Pathwidth bound}

\begin{lemma}\label{lem:cvc-seth:pathwidth}
Pathwidth of the graph $G$ is at most $\blocksize n'+O(3^\blocksize)$.
Moreover a path decomposition of such width can be found in polynomial time.
\end{lemma}

\begin{proof}
We give a mixed search strategy to clean the graph with $\blocksize n'+O(3^\blocksize)$ searchers.
First we put a searcher in the vertex $r^\ast$ and slide it to the root $r$.
This searcher remains there till the end of the cleaning process.

For a gadget $\groupg_{\iF,\iB}$ we call the vertices 
$v_{\iF,\iv,\iB}^1$ and $v_{\iF,\iv,\iB+1}^1$, $1 \le \iv \le \blocksize$,
  as {\em entry vertices} and {\em exit vertices} respectively.
We search the graph in $a = \numcopies$ rounds.
At the beginning of round $\iB$ ($0 \leq \iB < a$)
there are searchers on the entry vertices of the gadget $\groupg_{\iF,\iB}$ for every
$1 \leq \iF \le n'$.
Let $0 \le \iC < m$ and $0 \le \ib < \numcopiessmall$ be integers such that $\iB = \iC + m\ib$.
We place a searcher on $c_{\iC,\ib}^\ast$ and slide it to $c_{\iC,\ib}$.
Then, for each $1 \le \iF \le n'$ in turn we:
\begin{itemize}
  \item put $O(3^\blocksize)$ searchers on all vertices of the group gadget $\groupg_{\iF,\iB}$,
  \item put $2\blocksize$ searchers on all vertices $v_{\iF,\iv,\iB}^2$ and $v_{\iF,\iv,\iB+1}^1$,
  $1 \leq \iv \leq \blocksize$,
  \item remove searchers from all vertices of the group gadget $\groupg_{\iF,\iB}$
  and $v_{\iF,\iv,\iB}^\ivup$ for $1 \leq \iv \leq \blocksize$ and $1 \leq \ivup \leq 2$.
\end{itemize}
The last step of the round is removing a searcher from the vertex $c_{\iC,\ib}$.
After the last round the whole graph $G$ is cleaned.
Since we reuse $O(3^\blocksize)$ searchers for cleaning group gadgets,
$\blocksize n'+O(3^\blocksize)$ searchers suffice to clean the graph.

Using the above graph cleaning process a path decomposition of width
$\blocksize n'+O(3^\blocksize)$ can be constructed in polynomial time.
\end{proof}

\begin{proof}[Proof of Theorem \ref{thm:cvc-seth}]
Suppose \cvertexcover can be solved in $(3-\eps)^p |V|^{O(1)}$
time provided that we are given a path decomposition of $G$ of width $p$.
Let $\lambda = \log_3(3-\eps) < 1$.
We choose $\blocksize$ large enough such that
$\frac{\log 3^\blocksize}{\lfloor \log 3^\blocksize \rfloor} < \frac{1}{\lambda}$.
Given an instance of SAT we construct an instance of \cvertexcover
using the above construction and the chosen value of $\blocksize$.
Next we solve \cvertexcover using the $3^{\lambda p} |V|^{O(1)}$ time algorithm.
Lemmata~\ref{lem:cvc-seth:corr1}, \ref{lem:cvc-seth:corr2} ensure correctness,
whereas Lemma~\ref{lem:cvc-seth:pathwidth}
implies that running time of our algorithm is $3^{\lambda \blocksize n'} |V|^{O(1)}$,
however we have 
$$3^{\lambda \blocksize n'}=2^{\lambda \blocksize n'\log 3}=2^{\lambda n'\log 3^\blocksize}\leq 2^C \cdot2^{\lambda n \log 3^\blocksize / \lfloor \log 3^\blocksize \rfloor } = 2^C \cdot 2^{\lambda' n}$$ 
for some $\lambda' < 1$ and $C=\lambda \log 3^\blocksize$. This concludes the proof.
\end{proof}

\renewcommand{\numcopies}{\ensuremath{m(n+1)}\xspace}
\renewcommand{\numcopiessmall}{\ensuremath{n+1}\xspace}

\subsection{\cdomset}

\begin{theorem}\label{thm:cds-seth}
Assuming SETH, there cannot exist a constant $\eps>0$ and an algorithm that given an instance $(G=(V,E),k)$ together 
with a path decomposition of the graph $G$ of width $p$ solves the \cdomset problem in $(4-\eps)^p |V|^{O(1)}$ time.
\end{theorem}

\paragraph{Construction}

Given $\eps > 0$ and an instance $\Phi$ of SAT with $n$ variables and $m$ clauses we construct a graph $G$ as follows.
We assume that the number of variables $n$ is even, otherwise we add a single dummy variable.
We partition variables of $\Phi$ into groups $F_1,\ldots,F_{n'}$, each of size two,
hence $n'= n/2$.
The pathwidth of $G$ will be roughly $n'$.

First, we add to the graph $G$ two vertices $r$ and $r^\ast$, connected by an edge.
In the graph $G$ the vertex $r^\ast$ is of degree one, thus any connected dominating set of $G$
needs to include $r$. The vertex $r$ is called a {\em{root}}.

Second, we take $a = \numcopies$ and for each $1 \leq \iF \leq n'$
we create a path $\mathcal{P}_\iF$ consisting
of $4a$ vertices $v_{\iF,\iB}^\ivup$ and $h_{\iF,\iB}^\ivup$,
$0 \leq \iB < a$ and $1 \leq \ivup \leq 2$.
On the path $\mathcal{P}_\iF$ the vertices are arranged in the following order:
$$v_{\iF,0}^1, h_{\iF,0}^1, v_{\iF,0}^2, h_{\iF,0}^2, v_{\iF,1}^1,\ldots, h_{\iF,a-1}^2.$$
Let $\mathcal{V}$ and $\mathcal{H}$ be the sets of all vertices $v_{\iF,\iB}^\ivup$ and
$h_{\iF,\iB}^\ivup$ ($1 \leq \iF \leq n'$, $0 \leq \iB < a$, $1 \leq \ivup \leq 2$), respectively.
We connect vertices $v_{\iF,0}^1$ and all vertices in $\mathcal{H}$ to the root $r$.
To simplify further notation we denote $v_{\iF,a}^1 = r$, note that $h_{\iF,a-1}^2v_{\iF,a}^1 \in E$.

Third, for each $1 \leq \iF \leq n'$ and $0 \leq \iB < a$
we introduce guard vertices $p_{\iF,\iB}^1$, $p_{\iF,\iB}^2$ and
$q_{\iF,\iB}$. Each guard vertex is of degree two in $G$, namely
$p_{\iF,\iB}^1$ is adjacent to $v_{\iF,\iB}^1$ and $v_{\iF,\iB}^2$,
  $p_{\iF,\iB}^2$ is adjacent to $v_{\iF,\iB}^2$ and $v_{\iF,\iB+1}^1$
  and $q_{\iF,\iB}$ is adjacent to $h_{\iF,\iB}^1$ and $h_{\iF,\iB}^2$.
  Thus, each guard vertex ensures that at least one of its neighbours is contained
  in any connected dominating set in $G$.

The intuition of the construction made so far is as follows. For each two-variable block $F_\iF$
we encode any assignment of the variables in $F_\iF$ as a choice whether to take
$v_{\iF,\iB}^1$ or $v_{\iF,\iB}^2$ and $h_{\iF,\iB}^1$ or $h_{\iF,\iB}^2$ to the connected
dominating set in $G$.

We have finished the part of the construction needed to encode an assignment
and now we add vertices used
to check the satisfiability of the formula $\Phi$.
Let $C_0,\ldots,C_{m-1}$ be the clauses of the formula $\Phi$.
For each clause $C_\iC$ we create $(\numcopiessmall)$
vertices $c_{\iC,\ib}$, one for each $0 \le \ib < \numcopiessmall$.
Consider a clause $C_\iC$ and a group of variables $F_\iF=\{x_\iF^1, x_\iF^2\}$.
If $x_\iF^1$ occurs positively in $C_\iC$ then we connect $c_{\iC,\ib}$ with
$v_{\iF,m\ib+\iC}^1$ 
and if $x_\iF^1$ occurs negatively in $C_\iC$ then we connect $c_{\iC,\ib}$ with
$v_{\iF,m\ib+\iC}^2$.
Similarly if $x_\iF^2$ occurs positively in $C_\iC$ then we connect $c_{\iC,\ib}$
with $h_{\iF,m\ib+\iC}^1$ 
and if $x_\iF^2$ occurs negatively in $C_\iC$ then we connect $c_{\iC,\ib}$ with $h_{\iF,m\ib+\iC}^2$.
Intuitively taking the vertex $v_{\iF,m\ib+\iC}^1$
into a connected dominating set corresponds to setting $x_\iF^1$ to true,
whereas taking the vertex $h_{\iF,m\ib+\iC}^1$ into a connected dominating set corresponds to setting $x_\iF^2$ to true.

We can view the whole construction as a matrix, where each row corresponds to some group of variables $F_\iF$
and each column is devoted to some clause in such a way that each clause gets $(\numcopiessmall)$ private columns (but not consecutive)
of the matrix.

\begin{figure}
\begin{center}
\includegraphics{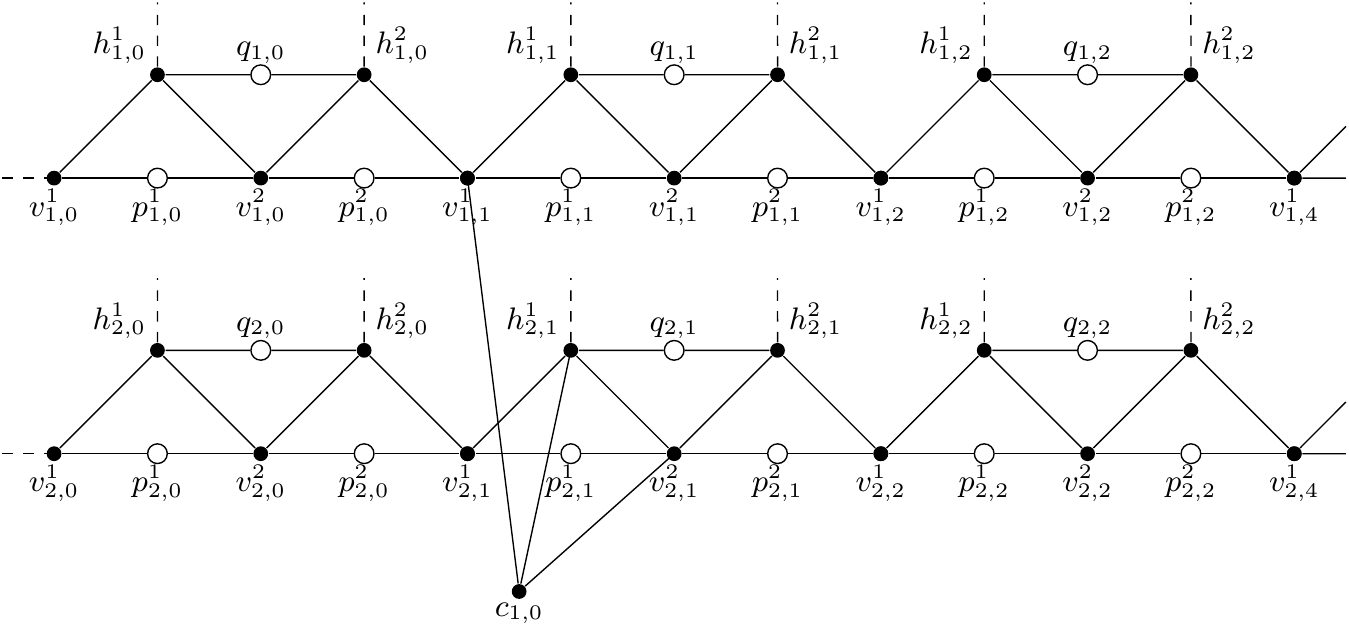}
\caption{Part of the construction for \cdomset.
  Dashed edges are connecting a vertex with the root $r$.
  Empty circles represent guard vertices.}
  \label{fig:seth-cds}
  \end{center}
  \end{figure}

Finally, let $\goalsize = 1 + n' \cdot 2a$ be the size of the connected dominating set we ask for.

\paragraph{Correctness}

\begin{lemma}\label{lem:cds-seth:corr1}
If $\Phi$ has a satisfying assignment, then there exists a connected dominating 
set $X$ in the graph $G$ of size $\goalsize$.
\end{lemma}

\begin{proof}
Given a satisfying assignment $\phi$ of the formula $\Phi$
we construct a connected dominating set $X$ as follows.
For each block $F_\iF = \{x_\iF^1, x_\iF^2\}$ and for
each $0 \leq \iB < a$ we include into $X$:
\begin{enumerate}
\item the vertex $v_{\iF,\iB}^1$ if $\phi(x_\iF^1)$ is true, and $v_{\iF,\iB}^2$ otherwise;
\item the vertex $h_{\iF,\iB}^1$ if $\phi(x_\iF^2)$ is true, and $h_{\iF,\iB}^2$ otherwise.
\end{enumerate}
Finally, we put $r$ into $X$. Note that $|X| = 1 + n' \cdot 2a = \goalsize$.
We now verify that $X$ is a connected dominating set in $G$.
First, we verify that $X$ dominates all vertices in $G$.
\begin{enumerate}
\item $r^\ast$ and all vertices in $\mathcal{H}$ are dominated by the root $r$.
\item All guards $p_{\iF,\iB}^1$, $p_{\iF,\iB}^2$ and $q_{\iF,\iB}^1$ are dominated
by $X \cap (\mathcal{H} \cup \mathcal{V})$
(with the possible exception of $p_{\iF,a-1}^2$ that is dominated by $r$).
\item All vertices in $\mathcal{V}$ are dominated by $X \cap \mathcal{H}$
(with the possible exception of $v_{\iF,0}^1$ that is dominated by $r$).
\item Finally, each clause vertex $c_{\iC,\ib}$ is dominated by any vertex $v_{\iF,m\ib+\iC}^\ivup$
or $h_{\iF,m\ib+\iC}^\ivup$ that corresponds to a variable that satisfies $C_\iC$ in the
assignment $\phi$.
\end{enumerate}

To finish the proof we need to ensure that $G[X]$ is connected.
We prove this by showing that each vertex in $X$ is connected to the root $r$ in $G[X]$.
This is obvious for vertices in $X \cap \mathcal{H}$, as $\mathcal{H} \subseteq N_G(r)$.
Moreover, for each $1 \leq \iF \leq n'$
and $0 \leq \iB < a$:
\begin{enumerate}
\item if $v_{\iF,\iB}^1 \in X$, then $v_{\iF,\iB}^1$ is connected to the root
via $h_{\iF,\iB-1}^2$ or $h_{\iF,\iB}^1$, with the exception of $v_{\iF,0}^1$,
    that is connected to $r$ directly;
\item if $v_{\iF,\iB}^2 \in X$, then $v_{\iF,\iB}^2$ is connected to the root
via $h_{\iF,\iB}^1$ or $h_{\iF,\iB}^2$.
\end{enumerate}
\end{proof}

\begin{lemma}\label{lem:cds-seth:corr2}
If there exists a connected dominating set $X$ of size at most $\goalsize$ in the graph $G$,
then $\Phi$ has a satisfying assignment.
\end{lemma}

\begin{proof}
First note that the vertex $r^\ast$ ensures that $r \in X$.
Moreover, the guard vertices $p_{\iF,\iB}^1$ and $q_{\iF,\iB}$
ensure that for each $1 \leq \iF \leq n'$ and $0 \leq \iB < a$ at least one vertex
$v_{\iF,\iB}^\ivup$ and at least one vertex $h_{\iF,\iB}^\ivup$ ($1 \leq \ivup \leq 2$)
  belongs to $X$. As $|X| \leq 1 + n' \cdot 2a$ and we have
  $n' \cdot 2a$ aforementioned guards with disjoint neighbourhoods,
for each $1 \leq \iF \leq n'$ and $0 \leq \iB < a$ exactly one vertex
$v_{\iF,\iB}^\ivup$ and exactly one vertex $h_{\iF,\iB}^\ivup$ belongs to $X$.
Moreover, $X \subseteq \{r\} \cup \mathcal{V} \cup \mathcal{H}$.

For each $0 \leq \iB < a$ we construct an assignment $\phi_\iB$ as follows.
For each block $F_\iF = \{x_\iF^1,x_\iF^2\}$ we define:
\begin{enumerate}
\item $\phi_\iB(x_\iF^1)$ to be true if $v_{\iF,\iB}^1 \in X$ and
false if $v_{\iF,\iB}^2 \in X$;
\item $\phi_\iB(x_\iF^2)$ to be true if $h_{\iF,\iB}^1 \in X$ and
false if $h_{\iF,\iB}^2 \in X$.
\end{enumerate}

We now show that the assignments $\phi_\iB$ cannot differ much for all indices $0 \leq \iB < a$.
Note that for each block $F_\iF = \{x_\iF^1, x_\iF^2\}$ and $0 \leq \iB < a-1$:
\begin{enumerate}
\item if $\phi_\iB(x_\iF^1)$ is true, then $\phi_{\iB+1}(x_\iF^1)$ is also true,
  as otherwise $v_{\iF,\iB}^2,v_{\iF,\iB+1}^1 \notin X$ and the guard $p_{\iF,\iB}^2$
  is not dominated by $X$;
\item if $\phi_\iB(x_\iF^2)$ is true, then $\phi_{\iB+1}(x_\iF^2)$ is also true,
  as otherwise $h_{\iF,\iB}^2,h_{\iF,\iB+1}^1 \notin X$ and the vertex $v_{\iF,\iB}^2$
  is either not dominated by $X$ (if $v_{\iF,\iB}^2 \notin X$) or
  isolated in $G[X]$ (if $v_{\iF,\iB}^2 \in X$).
\end{enumerate}
For each variable $x$ we define a sequence $\widehat{\phi}_x(\iB) = \phi_\iB(x)$, $0 \leq \iB < a$.
From the reasoning above we infer that for each variable $x$ the sequence
$\widehat{\phi}_x(\iB)$ can change its value at most once, from false to true.
Thus, as $a = \numcopies$, we conclude that there exists $0 \leq \ib < \numcopiessmall$
such that for all $0 \leq \iC < m$ the assignments $\phi_{m\ib + \iC}$ are equal.

We claim that the assigment $\phi = \phi_{m\ib}$ satisfies $\Phi$.
Consider a clause $C_\iC$ and focus on the vertex $c_{\iC,\ib}$.
It is not contained in $X$, thus one of its neighbour is contained in $X$.
As this neighbour corresponds to an assignment of one
variable that both satisfies $C_\iC$ (by the construction process)
and is consistent with $\phi_{m\ib+\iC} = \phi$ (by the definition of $\phi_{m\ib+\iC}$),
the assignment $\phi$ satisfies $C_\iC$ and the proof is finished.
\end{proof}

\paragraph{Pathwidth bound}

\begin{lemma}\label{lem:cds-seth:pathwidth}
Pathwidth of the graph $G$ is at most $n'+O(1)$.
Moreover a path decomposition of such width can found in polynomial time.
\end{lemma}

\begin{proof}
We give a mixed search strategy to clean the graph with $n'+9$ searchers.
First we put a searcher in the vertex $r^\ast$ and slide it to the root $r$.
This searcher remains there till the end of the cleaning process.

We search the graph in $a = \numcopies$ rounds.
At the beginning of round $\iB$ ($0 \leq \iB < a$)
there are searchers on all vertices
$v_{\iF,\iB}^1$ for $1 \leq \iF \leq n'$.
Let $0 \le \iC < m$ and $0 \le \ib < \numcopiessmall$ be integers such that $\iB = \iC + m\ib$.
We place a searcher on $c_{\iC,\ib}$.
Then, for each $1 \le \iF \le n'$ in turn
we put $7$ searchers on vertices
$p_{\iF,\iB}^1$, $v_{\iF,\iB}^2$, $p_{\iF,\iB}^2$,
$v_{\iF,\iB+1}^1$, $h_{\iF,\iB}^1$, $h_{\iF,\iB}^2$ and $q_{\iF,\iB}$,
and then remove $7$ searchers from vertices
$v_{\iF,\iB}^1$, $p_{\iF,\iB}^1$, $v_{\iF,\iB}^2$, $p_{\iF,\iB}^2$,
$h_{\iF,\iB}^1$, $h_{\iF,\iB}^2$ and $q_{\iF,\iB}$.
The last step of the round is removing a searcher from the vertex $c_{\iC,\ib}$.
After the last round the whole graph $G$ is cleaned.
Since we reuse $8$ searchers in the cleaning process, $n' + 9$
searchers suffice to clean the graph.

Using the above graph cleaning process a path decomposition of width
$n'+O(1)$ can be constructed in polynomial time.
\end{proof}

\begin{proof}[Proof of Theorem \ref{thm:cds-seth}]
Suppose \cdomset can be solved in $(4-\eps)^p |V|^{O(1)}$ time
provided that we are given a path decomposition of $G$ of width $p$.
Given an instance of SAT we construct an instance of \cdomset
using the above construction and solve it
using the $(4-\eps)^p |V|^{O(1)}$ time algorithm.
Lemmata~\ref{lem:cds-seth:corr1}, \ref{lem:cds-seth:corr2} ensure correctness,
whereas Lemma~\ref{lem:cds-seth:pathwidth}
implies that running time of our algorithm is $(4-\eps)^{n/2} |V|^{O(1)}$,
however we have 
$(4-\eps)^{n/2} = (\sqrt{4-\eps})^n$ and $\sqrt{4-\eps} < 2$.
This concludes the proof.
\end{proof}

\renewcommand{\numcopies}{\ensuremath{m(n+1)}\xspace}
\renewcommand{\numcopiessmall}{\ensuremath{n+1}\xspace}

\subsection{\cfvs and \coct}

\begin{theorem}\label{thm:cfvs-seth}
Assuming SETH, there cannot exist a constant $\eps>0$ and an algorithm that given an instance $(G=(V,E),k)$ together 
with a path decomposition of the graph $G$ of width $p$ solves the \cfvs problem in $(4-\eps)^p |V|^{O(1)}$ time.
\end{theorem}
\begin{theorem}\label{thm:coct-seth}
Assuming SETH, there cannot exist a constant $\eps>0$ and an algorithm that given an instance $(G=(V,E),k)$ together 
with a path decomposition of the graph $G$ of width $p$ solves the \coct problem in $(4-\eps)^p |V|^{O(1)}$ time.
\end{theorem}

In this section we prove Theorems \ref{thm:cfvs-seth} and \ref{thm:coct-seth} at once.
That is, we provide a single reduction that, given an instance $\Phi$ of SAT with $n$
variables and $m$ clauses, produces a graph $G$ together
with path decomposition of width roughly $n/2$ and integer $\goalsize$,
such that (1) if $\Phi$ is satisfiable then $G$ admits a connected feedback vertex set
of size at most $\goalsize$ (2) if $G$ admits a connected odd cycle transversal of size at most
$\goalsize$, then $\Phi$ is satisfiable. As any connected feedback vertex set
is a connected odd cycle transversal as well, this is sufficient to prove lower bounds
for \cfvs and \coct.

\paragraph{Construction}

Before we start, let us introduce one small gadget. By {\em{introducing a pentagon edge $vw$}}
we mean the following construction: we add three new vertices $u_{vw}^1,u_{vw}^2,u_{vw}^3$ and edges $vu_{vw}^1$, $u_{vw}^1w$, $vu_{vw}^2$, $u_{vw}^2u_{vw}^3$, $u_{vw}^3w$.
In the graph $G$ the vertices $u_{vw}^\alpha$ are of degree two and thus the created pentagon ensures that any connected feedback vertex set or connected
odd cycle transversal of $G$ includes $v$ or $w$. We call $u_{vw}^\alpha$ {\em{guard vertices}}.

Given $\eps > 0$ and an instance $\Phi$ of SAT with $n$ variables and $m$ clauses we construct a graph $G$ as follows.
We assume that the number of variables $n$ is even, otherwise we add a single dummy variable.
We partition variables of $\Phi$ into groups $F_1,\ldots,F_{n'}$, each of size two,
hence $n'= n/2$.
The pathwidth of $G$ will be roughly $n'$.

First, we add to the graph $G$ five vertices $r^1$, $r^2$, $r$, $r^\ast$ and $r^{\ast\ast}$
and edges $r^1r^2$, $rr^\ast$, $r^\ast r^{\ast\ast}$ and $r^{\ast\ast}r$.
In the graph $G$ the vertices $r^\ast$ and $r^{\ast\ast}$ are of degree two,
thus any connected feedback vertex set or connected odd cycle transversal of $G$
needs to include $r$. The vertex $r$ is called a {\em{root}}.

Second, we take $a = \numcopies$ and for each $1 \leq \iF \leq n'$
we create a path $\mathcal{P}_\iF$ consisting
of $4a$ vertices $v_{\iF,\iB}^\ivup$ and $h_{\iF,\iB}^\ivup$,
$0 \leq \iB < a$ and $1 \leq \ivup \leq 2$.
On the path $\mathcal{P}_\iF$ the vertices are arranged in the following order:
$$v_{\iF,0}^1, h_{\iF,0}^1, v_{\iF,0}^2, h_{\iF,0}^2, v_{\iF,1}^1,\ldots, h_{\iF,a-1}^2.$$
Let $\mathcal{V}$ and $\mathcal{H}$ be the sets of all vertices $v_{\iF,\iB}^\ivup$ and
$h_{\iF,\iB}^\ivup$ ($1 \leq \iF \leq n'$, $0 \leq \iB < a$, $1 \leq \ivup \leq 2$), respectively.
We connect vertices $v_{\iF,0}^1$ and all vertices in $\mathcal{H}$ to the root $r$.
Moreover, we connect all vertices $h_{\iF,\iB}^\ivup$ ($1 \leq \iF \leq n'$, $0 \leq \iB < a$, $1 \leq \ivup \leq 2$)
to the vertex $r^\ivup$.
To simplify further notation we denote $v_{\iF,a}^1 = r$, note that $h_{\iF,a-1}^2v_{\iF,a}^1 \in E$.

Third, for each $1 \leq \iF \leq n'$ and $0 \leq \iB < a$
we introduce pentagon edges $v_{\iF,\iB}^1v_{\iF,\iB}^2$, $v_{\iF,\iB}^2v_{\iF,\iB+1}^1$ and $h_{\iF,\iB}^1h_{\iF,\iB}^2$.

The intuition of the construction made so far is as follows. For each two-variable block $F_\iF$
we encode any assignment of the variables in $F_\iF$ as a choice whether to take
$v_{\iF,\iB}^1$ or $v_{\iF,\iB}^2$ and $h_{\iF,\iB}^1$ or $h_{\iF,\iB}^2$ to the connected
feedback vertex set or connected odd cycle transversal in $G$.

We have finished the part of the construction needed to encode an assignment
and now we add vertices used
to check the satisfiability of the formula $\Phi$.
Let $C_0,\ldots,C_{m-1}$ be the clauses of the formula $\Phi$.
For each clause $C_\iC$ we create $(\numcopiessmall)$
triples of vertices $c_{\iC,\ib}$, $c_{\iC,\ib}^\ast$, $c_{\iC,\ib}^{\ast\ast}$, one for each $0 \le \ib < \numcopiessmall$.
Each such triple is connected into a triangle. The vertices $c_{\iC,\ib}^\ast$ and $c_{\iC,\ib}^{\ast\ast}$ are of degree two
in $G$, thus they ensure that each vertex $c_{\iC,\ib}$ is contained in any connected feedback vertex set or connected odd cycle transversal in $G$.
Let $\mathcal{C}$ be the set of all vertices $c_{\iC,\ib}$, $|\mathcal{C}| = a$.

Consider a clause $C_\iC$ and a group of variables $F_\iF=\{x_\iF^1, x_\iF^2\}$.
If $x_\iF^1$ occurs positively in $C_\iC$ then we connect $c_{\iC,\ib}$ with
$v_{\iF,m\ib+\iC}^1$ via a path of length two, that is we add a vertex $c_{\iC,\ib,\iF,1}^v$
and edges $v_{\iF,m\ib+\iC}^1c_{\iC,\ib,\iF,1}^v$,$c_{\iC,\ib,\iF,1}^vc_{\iC,\ib}$.
If $x_\iF^1$ occurs negatively in $C_\iC$ then we connect $c_{\iC,\ib}$ with
$v_{\iF,m\ib+\iC}^2$ via a path of length two, that is we add a vertex $c_{\iC,\ib,\iF,2}^v$
and edges $v_{\iF,m\ib+\iC}^2c_{\iC,\ib,\iF,2}^v$,$c_{\iC,\ib,\iF,2}^vc_{\iC,\ib}$.
Similarly if $x_\iF^2$ occurs positively in $C_\iC$ then we connect $c_{\iC,\ib}$
with $h_{\iF,m\ib+\iC}^1$ via a path of length two, that is we add a vertex $c_{\iC,\ib,\iF,1}^h$
and edges $h_{\iF,m\ib+\iC}^1c_{\iC,\ib,\iF,1}^h$,$c_{\iC,\ib,\iF,1}^hc_{\iC,\ib}$.
If $x_\iF^2$ occurs negatively in $C_\iC$ then we connect $c_{\iC,\ib}$ with $h_{\iF,m\ib+\iC}^2$
via a path of length two, that is we add a vertex $c_{\iC,\ib,\iF,2}^h$
and edges $h_{\iF,m\ib+\iC}^2c_{\iC,\ib,\iF,2}^h$,$c_{\iC,\ib,\iF,2}^hc_{\iC,\ib}$.

Intuitively, taking the vertex $v_{\iF,m\ib+\iC}^1$
into a connected feedback vertex set or connected odd cycle transversal corresponds to setting $x_\iF^1$ to true,
whereas taking the vertex $h_{\iF,m\ib+\iC}^1$ corresponds to setting $x_\iF^2$ to true.

We can view the whole construction as a matrix, where each row corresponds to some group of variables $F_\iF$
and each column is devoted to some clause in such a way that each clause gets $(\numcopiessmall)$ private columns (but not consecutive)
of the matrix.

\begin{figure}
\begin{center}
\includegraphics{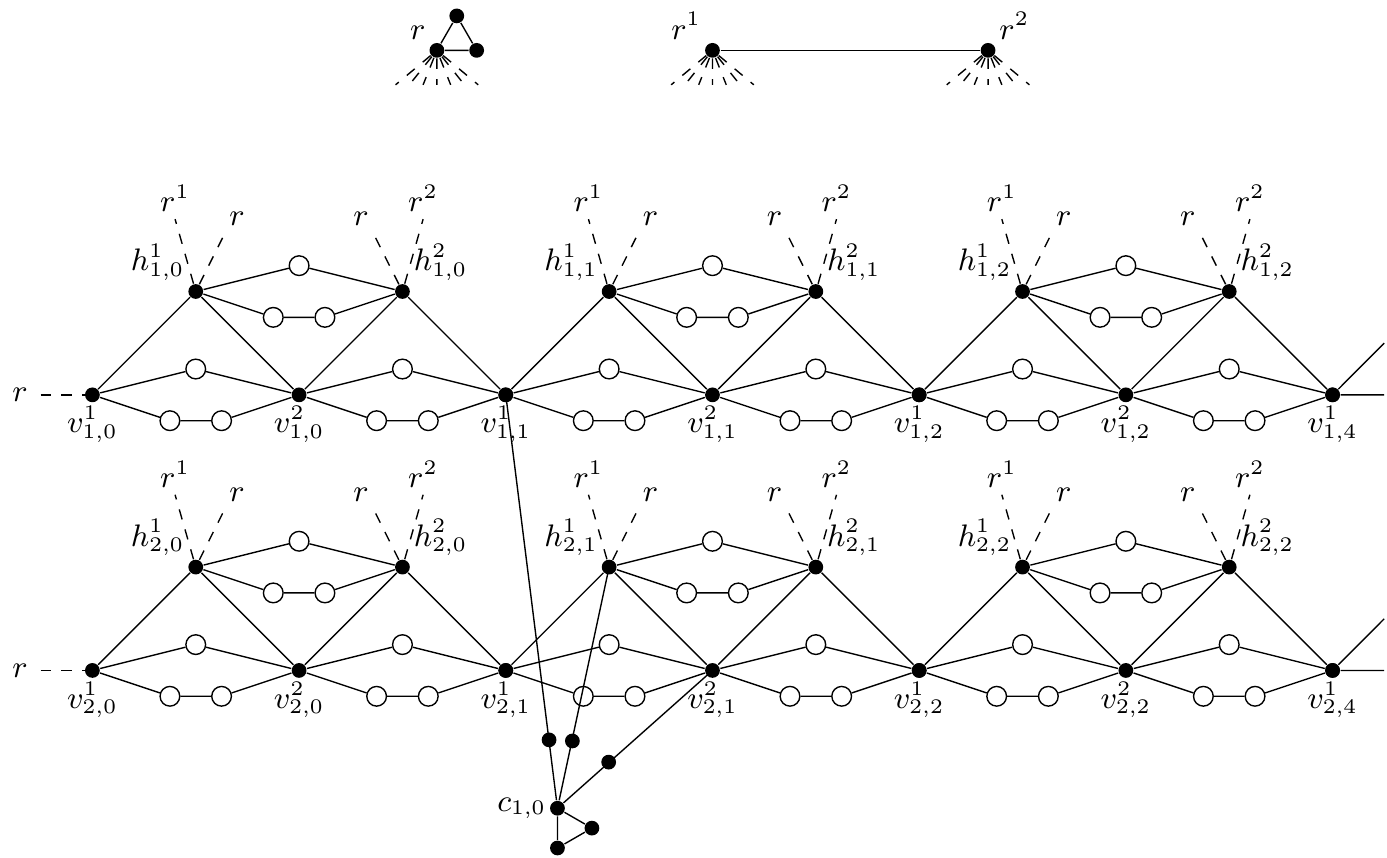}
\caption{Part of the construction for \cfvs and \coct.
  Dashed edges have one endpoint $r$, $r^1$ or $r^2$.
  Empty circles represent guard vertices in the pentagon edges.}
  \label{fig:seth-cfvs-coct}
  \end{center}
  \end{figure}

Finally, let $\goalsize = 1 + 2a + n' \cdot 2a$ be the size of the connected dominating set we ask for.

\paragraph{Correctness}

\begin{lemma}\label{lem:cfvs-coct-seth:corr1}
If $\Phi$ has a satisfying assignment, then there exists a connected feedback vertex set $X$ in the graph $G$ of size $\goalsize$.
\end{lemma}

\begin{proof}
Given a satisfying assignment $\phi$ of the formula $\Phi$
we construct a connected feedback vertex set $X$ as follows.
For each block $F_\iF = \{x_\iF^1, x_\iF^2\}$ and for
each $0 \leq \iB < a$ we include into $X$:
\begin{enumerate}
\item the vertex $v_{\iF,\iB}^1$ if $\phi(x_\iF^1)$ is true, and $v_{\iF,\iB}^2$ otherwise;
\item the vertex $h_{\iF,\iB}^1$ if $\phi(x_\iF^2)$ is true, and $h_{\iF,\iB}^2$ otherwise.
\end{enumerate}
Moreover, we put $r$ and all vertices in $\mathcal{C}$ into $X$.
Finally, for each clause $C_\iC$ 
let $x_\iF^\alpha \in F_\iF$ be any fixed variable satisfying $C_\iC$ in the assignment $\phi$.
For each $0 \le \ib < \numcopiessmall$
we add to the set $X$ exactly one neighbour of $c_{\iC,\ib}$, namely we add the vertex $c_{\iC,\ib,\iF,\beta}^\gamma$,
where $\gamma=v$ if $\alpha=1$ and $\gamma=h$ otherwise, whereas
$\beta=1$ if $\phi(x_\iF^\alpha)$ is true and $\beta=2$ otherwise.

Note that $|X| = 1 + 2a + n' \cdot 2a = \goalsize$.
We now verify that $X$ is a connected feedback vertex set in $G$.
First, we verify that $G \setminus X$ is a forest.
\begin{enumerate}
\item $r^\ast$, $r^{\ast\ast}$, $c_{\iC,\ib}^\ast$, $c_{\iC,\ib}^{\ast\ast},c_{\iC,\ib,\iF,\beta}^\gamma$ ($0 \leq \iC < m$, $0 \leq \ib < \numcopiessmall$, $1 \le \iF \le n'$, $1 \le \beta \le 2$, $\gamma \in \{v,h\}$) are either contained in $X$ or of degree one in $G \setminus X$.
\item Each guard vertex in $G \setminus X$ is either of degree at most one
  or is of degree two and has a leaf as a neighbour,
  as $X$ includes at least one endpoint of each pentagon edge.
\item In $G\setminus X$ the vertices from $\mathcal{V} \setminus X$ are connected to guard vertices, vertices $c_{\iC,\ib,\iF,\beta}^\gamma$,  and at most one vertex from $\mathcal{H} \setminus X$.
\item In $G\setminus X$ the vertices from $\mathcal{H} \setminus X$ are connected to guard vertices, vertices $c_{\iC,\ib,\iF,\beta}^\gamma$, at most one vertex from $\mathcal{V} \setminus X$, and exactly one vertex from the set $\{r^1,r^2\}$.
\item In $G\setminus X$ the vertices $r^1$ and $r^2$ are connected to each other and to some vertices in $\mathcal{H} \setminus X$, but no vertex in $\mathcal{H} \setminus X$ 
can reach both $r^1$ and $r^2$ in $G \setminus X$ without using the edge $r^1r^2$.
\end{enumerate}

To finish the proof we need to ensure that $G[X]$ is connected.
We prove this by showing that each vertex in $X$ is connected to the root $r$ in $G[X]$.
This is obvious for vertices in $X \cap \mathcal{H}$, as $\mathcal{H} \subseteq N_G(r)$.
For each $1 \leq \iF \leq n'$ and $0 \leq \iB < a$:
\begin{enumerate}
\item if $v_{\iF,\iB}^1 \in X$, then $v_{\iF,\iB}^1$ is connected to root
via $h_{\iF,\iB-1}^2$ or $h_{\iF,\iB}^1$, with the exception of $v_{\iF,0}^1$,
    that is connected to $r$ directly;
\item if $v_{\iF,\iB}^2 \in X$, then $v_{\iF,\iB}^2$ is connected to root
via $h_{\iF,\iB}^1$ or $h_{\iF,\iB}^2$.
\end{enumerate}
We are left with the vertices $c_{\iC,\ib}$ and their neighbours chosen to $X$ for $0 \leq \iC < m$, $0 \leq \ib < \numcopiessmall$.
However, by the definition of $X$ the only neighbour of $c_{i,j}$ chosen to $X$ connects it to a vertex $w \in \mathcal{V} \cup \mathcal{H}$ corresponding to a choice of the value $\phi(x)$ for some variable $x$. Therefore $w\in X$, so $c_{i,j}$ along with its only neighbour from $X$ are also connected to the root.
\end{proof}

\begin{lemma}\label{lem:cfvs-coct-seth:corr2}
If there exists a connected odd cycle transversal $X$ of size at most $\goalsize$ in the graph $G$,
then $\Phi$ has a satisfying assignment.
\end{lemma}

\begin{proof}
First note that $r \in X$ and $\mathcal{C} \subseteq X$.
  Moreover, $X$ needs to contain at least one endpoint of each pentagon edge $h_{\iF,\iB}^1h_{\iF,\iB}^2$
  and $v_{\iF,\iB}^1v_{\iF,\iB}^2$ ($1 \leq \iF \leq n'$, $0 \leq \iB < a$) and these pentagon edges
  are pairwise disjoint.
  Furthermore, each vertex in $\mathcal{C}$ needs to have a neighbour in $X$, but $\mathcal{C}$
  is an independent set and the neighbourhoods of vertices from $\mathcal{C}$ are pairwise disjoint and disjoint from $\mathcal{H} \cup \mathcal{V} \cup \{r\}$.
So far we have one vertex $r$, $a$ vertices in $\mathcal{C}$, $n' \cdot 2a$ endpoints of pentagon edges
and $a$ neighbours of vertices from $\mathcal{C}$, thus, as $|X| \leq \goalsize = 1 + 2a +n' \cdot 2a$,
$X$ contains $r$, $\mathcal{C}$, exactly one endpoint of each pentagon edge, 
exactly one neighbour of each vertex from $\mathcal{C}$ and nothing more.
In particular, $r^1,r^2 \notin X$.

For each $0 \leq \iB < a$ we construct an assignment $\phi_\iB$ as follows.
For each block $F_\iF = \{x_\iF^1,x_\iF^2\}$ we define:
\begin{enumerate}
\item $\phi_\iB(x_\iF^1)$ to be true if $v_{\iF,\iB}^1 \in X$ and
false if $v_{\iF,\iB}^2 \in X$;
\item $\phi_\iB(x_\iF^2)$ to be true if $h_{\iF,\iB}^1 \in X$ and
false if $h_{\iF,\iB}^2 \in X$.
\end{enumerate}

We now show that the assignments $\phi_\iB$ cannot differ much for all indices $0 \leq \iB < a$.
Note that for each block $F_\iF = \{x_\iF^1, x_\iF^2\}$ and $0 \leq \iB < a-1$:
\begin{enumerate}
\item if $\phi_\iB(x_\iF^1)$ is true, then $\phi_{\iB+1}(x_\iF^1)$ is also true,
  as otherwise $X$ contains no endpoint of the pentagon edge $v_{\iF,\iB}^2v_{\iF,\iB+1}^1$.
\item if $\phi_\iB(x_\iF^2)$ is true, then $\phi_{\iB+1}(x_\iF^2)$ is also true,
  as otherwise $h_{\iF,\iB}^2,h_{\iF,\iB+1}^1 \notin X$ and
  either the vertex $v_{\iF,\iB}^2$ is not connected to the root in $G[X]$ (if $v_{\iF,\iB}^2 \in X$, since vertices from $\mathcal{C}$ are leaves in $G[X]$)
  or $G \setminus X$ contains a cycle of length five consisting of vertices $v_{\iF,\iB}^2$, $h_{\iF,\iB}^2$, $r^2$, $r^1$ and $h_{\iF,\iB+1}^1$ (if $v_{\iF,\iB}^2 \notin X$).
\end{enumerate}
For each variable $x$ we define a sequence $\widehat{\phi}_x(\iB) = \phi_\iB(x)$, $0 \leq \iB < a$.
From the reasoning above we infer that for each variable $x$ the sequence
$\widehat{\phi}_x(\iB)$ can change its value at most once, from false to true.
Thus, as $a = \numcopies$, we conclude that there exists $0 \leq \ib < \numcopiessmall$
such that for all $0 \leq \iC < m$ the assignments $\phi_{m\ib + \iC}$ are equal.

We claim that the assigment $\phi = \phi_{m\ib}$ satisfies $\Phi$.
Consider a clause $C_\iC$ and focus on the vertex $c_{\iC,\ib} \in X$.
As $G[X]$ is connected, there exists a vertex $x$ in the set $N(N(c_{\iC,\ib}))$
that belongs to $X \cap (\mathcal{V} \cup \mathcal{H})$.
This vertex $x$ corresponds to an assignment of one variable that both satisfies $C_\iC$ (by the construction process)
and is consistent with $\phi_{m\ib+\iC} = \phi$ (by the definition of $\phi_{m\ib+\iC}$).
Thus the assignment $\phi$ satisfies $C_\iC$ and the proof is finished.
\end{proof}

\paragraph{Pathwidth bound}

\begin{lemma}\label{lem:cfvs-coct-seth:pathwidth}
Pathwidth of the graph $G$ is at most $n'+O(1)$.
Moreover, a path decomposition of such width can found in polynomial time.
\end{lemma}

\begin{proof}
We give a mixed search strategy to clean the graph with $n'+O(1)$ searchers.
First we put five searchers on the vertices $r^1$, $r^2$, $r$, $r^\ast$ and $r^{\ast\ast}$
and then remove the searchers from the vertices $r^\ast$ and $r^{\ast\ast}$.
The searchers on the vertices $r^1$, $r^2$ and $r$ remain till the end of the cleaning process.

We search the graph in $a = \numcopies$ rounds.
At the beginning of round $\iB$ ($0 \leq \iB < a$)
there are searchers on all vertices
$v_{\iF,\iB}^1$ for $1 \leq \iF \leq n'$.
Let $0 \le \iC < m$ and $0 \le \ib < \numcopiessmall$ be integers such that $\iB = \iC + m\ib$.
We first place three searchers on $c_{\iC,\ib}$, $c_{\iC,\ib}^\ast$ and $c_{\iC,\ib}^{\ast\ast}$
and afterwards we remove the searchers from $c_{\iC,\ib}^\ast$ and $c_{\iC,\ib}^{\ast\ast}$.

Then, for each $1 \le \iF \le n'$ in turn
we put $O(1)$ searchers on vertices
$v_{\iF,\iB}^2$, $v_{\iF,\iB+1}^1$, $h_{\iF,\iB}^1$, $h_{\iF,\iB}^2$, the guard
vertices of pentagon edges $v_{\iF,\iB}^1v_{\iF,\iB}^2$, $v_{\iF,\iB}^2v_{\iF,\iB+1}^1$,$h_{\iF,\iB}^1h_{\iF,\iB}^2$, and all vertices 
$c_{\iC,\ib,\iF,\beta}^\gamma$,
and then remove searchers from the vertices $v_{\iF,\iB}^1$, $v_{\iF,\iB}^2$, $h_{\iF,\iB}^1$, $h_{\iF,\iB}^2$, all aforementioned guard vertices
and vertices $c_{\iC,\ib,\iF,\beta}^\gamma$.
The last step of the round is removing a searcher from the vertex $c_{\iC,\ib}$.
After the last round the whole graph $G$ is cleaned.
Since we reuse searchers in the cleaning process, $n' + O(1)$
searchers suffice to clean the graph.

Using the above graph cleaning process a path decomposition of width
$n'+O(1)$ can be constructed in polynomial time.
\end{proof}

\begin{proof}[Proof of Theorems \ref{thm:cfvs-seth} and \ref{thm:coct-seth}]
Suppose \cfvs or \coct can be solved in $(4-\eps)^p |V|^{O(1)}$
provided that we are given a path decomposition of $G$ of width $p$.
Given an instance of SAT we construct an instance of \cfvs or \coct
using the above construction and solve it
using the $(4-\eps)^p |V|^{O(1)}$ time algorithm.
Lemmata~\ref{lem:cfvs-coct-seth:corr1}, \ref{lem:cfvs-coct-seth:corr2}, together with an observation that any connected feedback vertex set is also a connected odd cycle transversal in $G$,
ensure correctness, whereas Lemma~\ref{lem:cfvs-coct-seth:pathwidth}
implies that running time of our algorithm is $(4-\eps)^{n/2} |V|^{O(1)}$,
however we have 
$(4-\eps)^{n/2} = (\sqrt{4-\eps})^n$ and $\sqrt{4-\eps} < 2$.
This concludes the proof.
\end{proof}

\renewcommand{\numcopies}{\ensuremath{m(2\blocksize n'+1)}\xspace}
\renewcommand{\numcopiessmall}{\ensuremath{2\blocksize n'+1}\xspace}

\subsection{\fvs}

\begin{theorem}\label{thm:fvs-seth}
Assuming SETH, there cannot exist a constant $\eps>0$ and an algorithm that given an instance $(G=(V,E), k)$ together with a path decomposition of the graph $G$ of width $p$ solves the \fvs problem in $(3-\eps)^p |V|^{O(1)}$ time.
\end{theorem}

\paragraph{Construction}
The construction here is a bit different than in the previous subsections,
as we do not have a constraint that the solution needs to induce a connected subgraph.
As one may notice, this connectivity constraint is used intensively in the previous subsections,
in particular, it gives quite an easy way to verify the correctness of the assignment
encoded in the solution. In the case of \fvs we need to do it in a different and
more complicated way. Parts of the construction here resembles the lower bound
proof for \oct by Lokshtanov, Marx and Saurabh \cite{treewidth-lower}.

Before we start, let us introduce one small gadget,
      used already in the proof of the lower bounds for \cfvs and \coct.
      By {\em{introducing a triangle edge $vw$}}
we mean the following construction: we add a new vertex $u_{vw}$ and edges $vw$, $vu_{vw}$, $wu_{vw}$.
We call $u_{vw}$ {\em{a guard vertex}} and in the graph $G$ its degree equals two.
Note that any feedback vertex set $X$ in $G$ needs to
intersect the triangle composed of vertices $\{v,w,u_{vw}\}$ and, moreover, if $u_{vw} \in X$
then $X \setminus \{u_{vw}\} \cup \{v\}$ is also a feedback vertex set in $G$ of not greater size.
Thus we may focus only on feedback vertex sets in $G$ that do not contain guard vertices.
Each such feedback vertex set needs to include at least one endpoint of each triangle edge.

Given $\eps > 0$ and an instance $\Phi$ of SAT with $n$ variables and $m$ clauses
we construct a graph $G$ as follows.
We first choose a constant integer $\blocksize$, which value depends on $\eps$ only.
The exact formula for $\blocksize$ is presented later.
We partition variables of $\Phi$ into groups $F_1,\ldots,F_{n'}$,
each of size at most $\blockvars = \lfloor \log 3^\blocksize \rfloor$,
hence $n'=\lceil n/\blockvars \rceil$.
Note that now $\blocksize n' \sim n/\log 3$, the pathwidth of $G$ will be roughly $\blocksize n'$.

First, we add to the graph $G$ a vertex $r$, called a {\em{root}}.

Second, we take $a = \numcopies$ and for each $1 \leq \iF \leq n'$
and $1 \leq \iv \leq \blocksize$ we create a path $\mathcal{P}_{\iF,\iv}$ consisting
of $3a$ vertices $v_{\iF,\iv,\iB}^\ivup$, $0 \leq \iB < a$ and $1 \leq \ivup \leq 3$,
   arranged in the following order:
   $$v_{\iF,\iv,0}^1, v_{\iF,\iv,0}^2,v_{\iF,\iv,0}^3,v_{\iF,\iv,1}^1,\ldots,v_{\iF,\iv,a-1}^1,v_{\iF,\iv,a-1}^2,v_{\iF,\iv,a-1}^3.$$
Let $\mathcal{V}_{\iF,\iv,\iB} = \{v_{\iF,\iv,\iB}^\ivup : 1 \leq \ivup \leq 3\}$,
$\mathcal{V}_{\iF,\iB}= \bigcup_{\iv=1}^\blocksize \mathcal{V}_{\iF,\iv,\iB}$
and $\mathcal{V} = \bigcup_{\iF=1}^{n'} \bigcup_{\iB=0}^{a-1} V_{\iF,\iB}$.

Third, for each two consecutive vertices $v_{\iF,\iv,\iB}^\ivup$, $v_{\iF,\iv,\iB'}^{\ivup'}$
on the path $\mathcal{P}_{\iF,\iv}$ we introduce vertices $h_{\iF,\iv,\iB}^{\ivup,1}$
and $h_{\iF,\iv,\iB}^{\ivup,2}$ connected by a triangle edge.
Furthermore, we add edges
$h_{\iF,\iv,\iB}^{\ivup,1}v_{\iF,\iv,\iB}^\ivup$, $h_{\iF,\iv,\iB}^{\ivup,1}r$, 
$h_{\iF,\iv,\iB}^{\ivup,2}v_{\iF,\iv,\iB'}^{\ivup'}$, $h_{\iF,\iv,\iB}^{\ivup,2}r$.
Let $\mathcal{H}$ be the set of all vertices $h_{\iF,\iv,\iB}^{\ivup,\gamma}$.

We now provide a description of a group gadget $\groupg_{\iF,\iB}$,
which will enable us to encode $2^\blockvars$ possible assignments of one group
of $\blockvars$ variables.
Fix a block $F_\iF$, $1 \leq \iF \leq n'$, and a position $\iB$, $0 \leq \iB < a$.
The group gadget $\groupg_{\iF,\iB}$ includes (already created) vertices
$v_{\iF,\iv,\iB}^\ivup$, $h_{\iF,\iv,\iB}^{\ivup,\gamma}$
($1 \leq \iv \leq \blocksize$, $1 \leq \ivup \leq 3$, $1 \leq \gamma \leq 2$)
and all guard vertices in the triangle edges between them.
Moreover, we perform the following construction.

For each $1 \leq \iv \leq \blocksize$ we introduce three vertices
$p_{\iF,\iv,\iB}^\ivup$ ($1 \leq \ivup \leq 3$),
  pairwise connected by triangle edges. Moreover,
for each $1 \leq \ivup \leq 3$ we connect $p_{\iF,\iv,\iB}^\ivup$ and $v_{\iF,\iv,\iB}^\ivup$
by a triangle edge.
Let $\mathcal{P}$ be the set of all vertices $p_{\iF,\iv,\iB}^\ivup$ in the whole graph $G$.

In order to encode $2^\blockvars$ assignments we consider subsets of $\mathcal{V}_{\iF,\iB}$
that contain {\em{exactly one}} vertex out of each set
$\mathcal{V}_{\iF,\iv,\iB}$.
For each sequence $S=(s_1,\ldots,s_\blocksize) \in \{1,2,3\}^\blocksize$ we perform the following construction.
First, for each $1 \leq \iv \leq \blocksize$ we introduce three vertices
$q_{\iF,\iv,\iB}^{S,\ivup}$ ($1 \leq \ivup \leq 3$), pairwise connected by triangle edges.
Second, we connect $q_{\iF,\iv,\iB}^{S,\ivup}$ with $v_{\iF,\iv,\iB}^\ivup$
($1 \leq \ivup \leq 3$) with a triangle edge.
Third, we introduce a vertex $x_{\iF,\iB}^S$ and connect all vertices $q_{\iF,\iv,\iB}^{S,s_\iv}$
($1 \leq \iv \leq \blocksize$) and the vertex $x_{\iF,\iB}^S$ into a cycle $\mathcal{Q}_{\iF,\iB}^S$.
Fourth, we connect all vertices $x_{\iF,\iB}^S$ for $S \in \{1,2,3\}^\blocksize$
into a cycle $\mathcal{X}_{\iF,\iB}$. 
Finally, for each $S \in \{1,2,3\}^\blocksize$ we introduce two new vertices
$y_{\iF,\iB}^S$ and $z_{\iF,\iB}^S$ and triangle edges $x_{\iF,\iB}^Sy_{\iF,\iB}^S$
and $y_{\iF,\iB}^Sz_{\iF,\iB}^S$.
Let $\mathcal{X}$, $\mathcal{Y}$ and $\mathcal{Z}$ be the sets of all vertices
$x_{\iF,\iB}^S$, $y_{\iF,\iB}^S$ and $z_{\iF,\iB}^S$, respectively.
This finishes the construction of the group gadget $\groupg_{\iF,\iB}$.

\begin{figure}
\begin{center}
\includegraphics{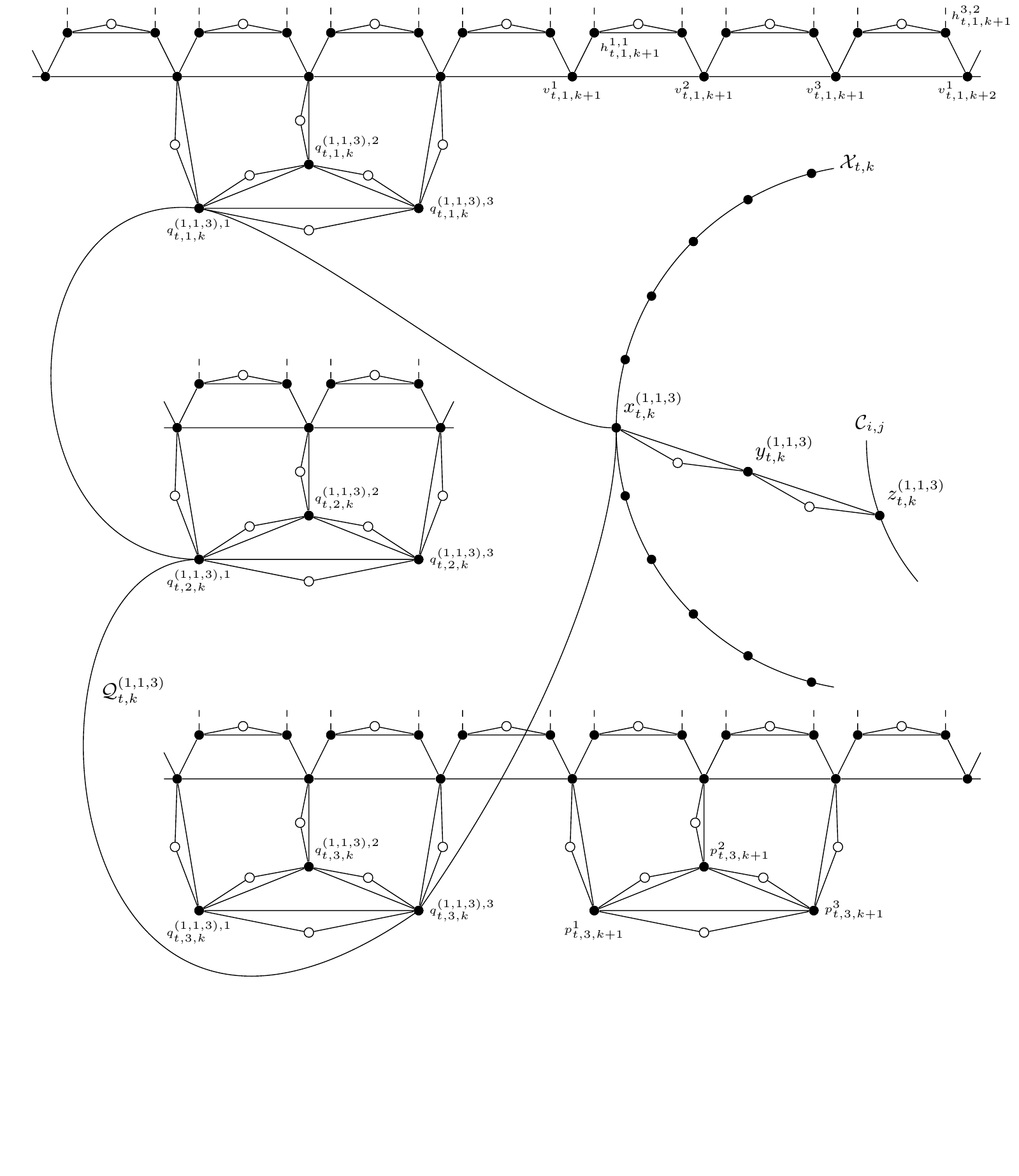}
\caption{Part of the construction around group gadget $\groupg_{\iF,\iB}$ for $\blocksize=3$.
  Dashed edges are connecting a vertex with the root $r$.
  Empty circles represent guard vertices.}
  \label{fig:seth-fvs}
  \end{center}
  \end{figure}

We add vertices used to check the satisfiability of the formula $\Phi$.
Observe that for a group of variables $F_\iF$
there are at most $2^{\blockvars}$ possible assignments
and there are $3^\blocksize \ge 2^{\blockvars}$ vertices $x_{\iF,\iB}^S$
for sequences $S$ from the set $\{1,2,3\}^\blocksize$ in each group gadget $\groupg_{\iF,\iB}$,
hence we can assign a {\em unique} sequence $S$ to each assignment.
Let $C_0,\ldots,C_{m-1}$ be the clauses of the formula $\Phi$.
For each clause $C_\iC$ and for each $0 \leq \ib < \numcopiessmall$
we perform the following construction that uses
gadgets $\groupg_{\iF,m\ib+\iC}$ for $1 \le \iF \le n'$.
For each group of variables $F_\iF$ we consider the set $\mathcal{S}_{\iF,\iC}$ of all sequences
$S \in \{1,2,3\}^\blocksize$ that correspond 
to an assignment of $F_\iF$ satisfying the clause $C_\iC$
(i.e., one of the variables of $F_\iF$ is assigned a value such that $C_\iC$ is already satisfied).
We connect all vertices $z_{\iF,m\ib+\iC}^S$ for $1 \leq \iF \leq n'$, $S \in \mathcal{S}_{\iF,\iC}$
into a cycle $\mathcal{C}_{\iC,\ib}$. The vertices on the cycle $\mathcal{C}_{\iC,\ib}$ are sorted
in the order of increasing value of $\iF$.

We can view the whole construction as a matrix of group gadgets, where each row corresponds to some group of variables $F_\iF$
and each column is devoted to some clause in such a way that each clause gets $(\numcopiessmall)$ private columns (but not consecutive)
of the group gadget matrix, as in Figure \ref{fig:seth-fvsmatrix}.

\begin{figure}[htbp]
\begin{center}
\includegraphics{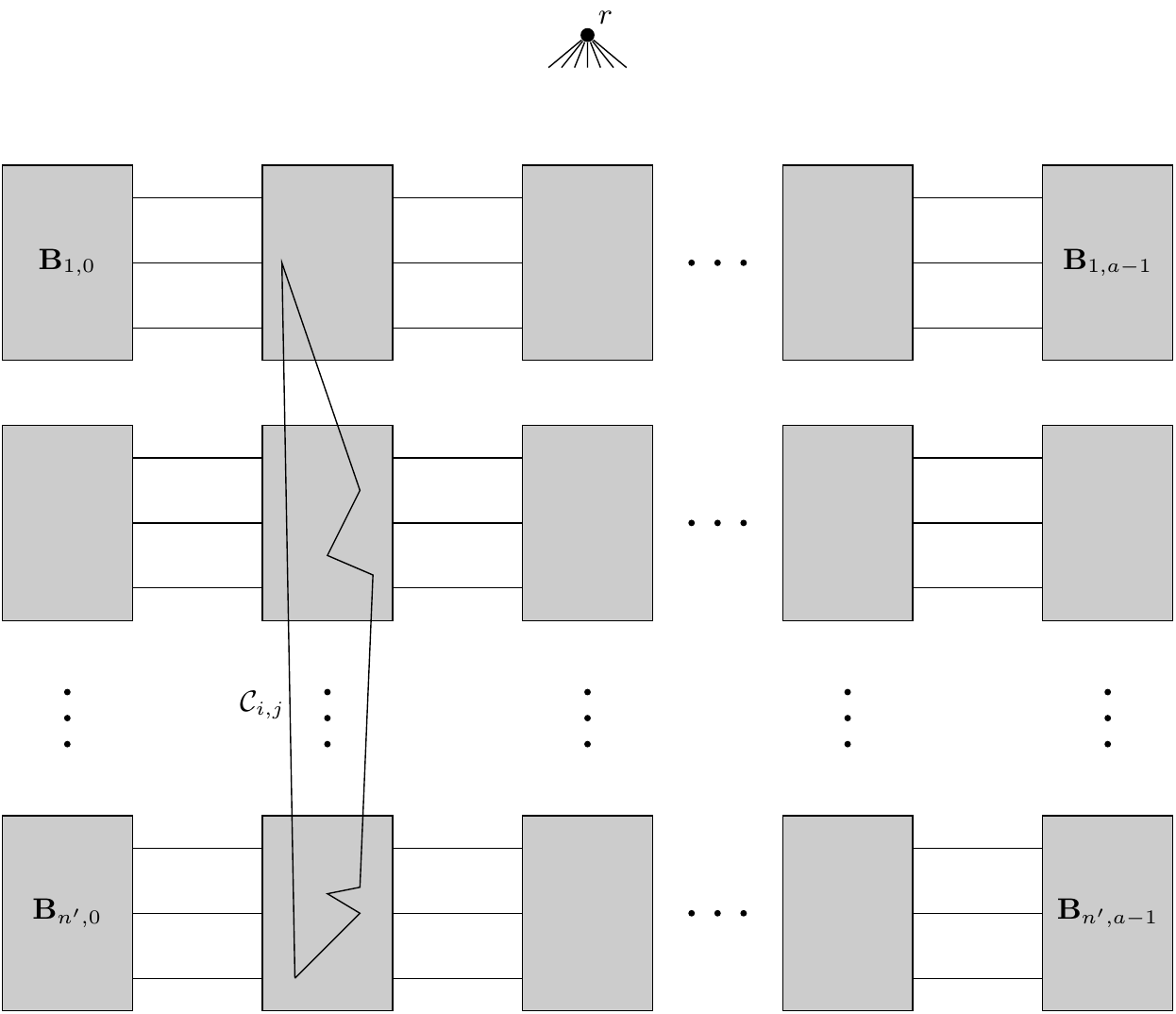}
\caption{The arrangement of the group gadgets in $G$.}
\label{fig:seth-fvsmatrix}
\end{center}
\end{figure}

Finally, we define the size of the feedback vertex set we are looking for
as $\goalsize = \goalsize_h + \goalsize_v + \goalsize_p + \goalsize_q + \goalsize_x + \goalsize_y$,
   where
   \begin{align*}
   \goalsize_h &= \blocksize n'(3a-1) & \goalsize_v &= \blocksize n' a & \goalsize_p &= \blocksize n' \cdot 2a \\
   \goalsize_q &= \blocksize n' \cdot 2a3^\blocksize & \goalsize_x &= n'a & \goalsize_y &= n'a3^\blocksize.
   \end{align*}

\paragraph{Correctness}

\begin{lemma}\label{lem:fvs-seth:corr1}
If $\Phi$ has a satisfying assignment, then there exists a feedback vertex set in $G$ of size $\goalsize$.
\end{lemma}

\begin{proof}
Given a satisfying assignment $\phi$ of the formula $\Phi$ we construct a feedback vertex set $X \subseteq V$ as follows.

For each group of variables $F_\iF$ we consider the sequence $S_\iF \in \{1,2,3\}^\blocksize$
which corresponds to the restriction of the assignment $\phi$ to the variables of $F_\iF$.
Let
\begin{align*}
X_v &= \{v_{\iF,\iv,\iB}^{S_\iF(\iv)}: 1 \leq \iF \leq n', 1 \leq \iv \leq \blocksize, 0 \leq \iB < a\} \\
X_p &= \{p_{\iF,\iv,\iB}^\ivup: 1 \leq \iF \leq n', 1 \leq \iv \leq \blocksize, 0 \leq \iB < a, 1 \leq \ivup \leq 3, \ivup \neq S_\iF(\iv)\} \\
X_q &= \{q_{\iF,\iv,\iB}^{S,\ivup}: 1 \leq \iF \leq n', 1 \leq \iv \leq \blocksize, 0 \leq \iB < a, S \in \{1,2,3\}^\blocksize, 1 \leq \ivup \leq 3, \ivup \neq S_\iF(\iv)\} \\
X_x &= \{x_{\iF,\iB}^{S_\iF}: 1 \leq \iF \leq n', 0 \leq \iB < a\} \\
X_y &= \{y_{\iF,\iB}^S: 1 \leq \iF \leq n', 0 \leq \iB < a, S \in \{1,2,3\}^\blocksize, S \neq S_\iF\} \\
X_z &= \{z_{\iF,\iB}^{S_\iF}: 1 \leq \iF \leq n', 0 \leq \iB < a\}
\end{align*}
Moreover, we define the set $X_h \subseteq \mathcal{H}$ to contain, for each $1 \leq \iF \leq n'$,
$1 \leq \iv \leq \blocksize$, $0 \leq \iB < a$, $1 \leq \ivup \leq 3$ and
$(\iB,\ivup) \neq (a-1,3)$, the vertex $h_{\iF,\iv,\iB}^{\ivup,1}$ if $v_{\iF,\iv,\iB}^\ivup \notin X_v$ and the vertex $h_{\iF,\iv,\iB}^{\ivup,2}$ otherwise.
Note that $|X_v| = \goalsize_v$, $|X_p| = \goalsize_p$, $|X_q| = \goalsize_q$, $|X_x| = \goalsize_x$,
$|X_y| + |X_z| = \goalsize_y$ and $|X_h| = \goalsize_h$. Thus a set
$X = X_v \cup X_p \cup X_q \cup X_x \cup X_y \cup X_z \cup X_h$ is of size $\goalsize$.

To finish the proof we need to verify that $X$ is a feedback vertex set of $G$.
First, we ensure that $X$ includes at least one endpoint of every triangle edge in $G$.
\begin{enumerate}
\item $X$ includes one vertex from each pair $h_{\iF,\iv,\iB}^{\ivup,1}$ and $h_{\iF,\iv,\iB}^{\ivup,2}$.
\item $X$ includes two out of three vertices in each triple $p_{\iF,\iv,\iB}^\ivup$, $1 \leq \ivup \leq 3$, and in each triple $q_{\iF,\iv,\iB}^{S,\ivup}$, $1 \leq \ivup \leq 3$.
\item If $p_{\iF,\iv,\iB}^\ivup \notin X$ or $q_{\iF,\iv,\iB}^{S,\ivup} \notin X$,
  then $v_{\iF,\iv,\iB}^\ivup \in X$.
\item If $y_{\iF,\iB}^S \notin X$ then both $x_{\iF,\iB}^S$ and $z_{\iF,\iB}^S$ are in $X$.
\end{enumerate}
Let $G_0$ be the graph $G$ with deleted guard vertices.
We have just shown that each guard vertex in $G \setminus X$ is of degree zero or one.
Thus if $G_0 \setminus X$ is a forest, then $G \setminus X$ is a forest too.

Now note that $X_p \cup X_q \cup X_v$ separates $(\{r\} \cup \mathcal{V} \cup \mathcal{H}) \setminus X$ from
the rest of the graph $G_0$. To see this recall that
if $p_{\iF,\iv,\iB}^\ivup \notin X$ or $q_{\iF,\iv,\iB}^{S,\ivup} \notin X$,
then $v_{\iF,\iv,\iB}^\ivup \in X$.
Let $G_1 = G_0[\{r\} \cup \mathcal{V} \cup \mathcal{H}]$.
We now show that $G_1 \setminus X$ is a forest.
Note that in $G_1 \setminus X_h$ each cycle contains at least one vertex from $\mathcal{V}$.
Recall that the set $X_v$ contains every third vertex on each path $\mathcal{V}_{\iF,\iv}$. Let
$v_{\iF,\iv,\iB}^\ivup$ and $v_{\iF,\iv,\iB'}^{\ivup'}$ be any two consecutive vertices on
$\mathcal{V}_{\iF,\iv}$ that are not in $X_v$. It is straightforward to check that that in the set of at most $4$ neighbours of these two
vertices in $\mathcal{H}$, at most one is not in $X_h$. Thus,
the vertices in $\mathcal{V} \setminus X_v$ do not take part in any cycle in $G_1 \setminus X$,
    and $G_1 \setminus X$ is a forest.

Now note that $X_v \cup X_q \cup X_x \cup X_y$
separates $(\mathcal{Q} \cup \mathcal{X}) \setminus X$ from the rest of the graph $G_0$.
To see this recall that 
if $p_{\iF,\iv,\iB}^\ivup \notin X$ or $q_{\iF,\iv,\iB}^{S,\ivup} \notin X$,
then $v_{\iF,\iv,\iB} \in X$ and if $x_{\iF,\iB}^S \notin X$ then $y_{\iF,\iB}^S \in X$.
Let $G_2 = G_0[\mathcal{Q} \cup \mathcal{X}]$.
We now show that $G_2 \setminus X$ is a forest.
The vertices in $\mathcal{Q} \cup \mathcal{X}$ from different group gadgets
are not adjacent, thus we focus on a single group gadget $\groupg_{\iF,\iB}$.
Let $S \in \{1,2,3\}^\blocksize$.
Note that in $G_2\setminus X$ the vertices from $\mathcal{Q}_{\iF,\iB}^S\setminus X$ are of
degree at most two, except for the vertex $x_{\iF,\iB}^S$. If $S \neq S_\iF$,
then for any $1 \leq \iv \leq \blocksize$ such that $S(\iv) \neq S_\iF(\iv)$
we have $q_{\iF,\iv,\iB}^{S,S(\iv)} \in X_q$ and the cycle $\mathcal{Q}_{\iF,\iB}^S$
is intersected by $X$. On the other hand, if $S = S_\iF$, then $x_{\iF,\iB}^S \in X_x$.
Thus, the vertices $\mathcal{Q} \setminus X$ are not contained in any cycle in $G_2 \setminus X$.
Moreover, on each cycle $\mathcal{X}_{\iF,\iB}$ we have $x_{\iF,\iB}^{S_\iF} \in X_x$
and we infer that $G_2 \setminus X$ is a forest.

As for $\mathcal{Y}$, note that if $y_{\iF,\iB}^S \notin X$ then $S = S_\iF$
and $x_{\iF,\iB}^S, z_{\iF,\iB}^S \in X$, and $y_{\iF,\iB}^S$ is isolated in $G_0 \setminus X$.

We are left with $\mathcal{Z}$. The graph $G_3 = G_0[\mathcal{Z}]$ consists of $a$ cycles
$\mathcal{C}_{\iC,\ib}$, $0 \leq \iC < m$, $0 \leq \ib < \numcopiessmall$.
Consider a clause $C_\iC$ and an index $0 \leq \ib < \numcopiessmall$.
As $\phi$ satisfies $C_\iC$, there exists a block $F_\iF$, such that a variable from this block
satisfies $C_\iC$. Then $z_{\iF,m\iC+\ib}^{S_\iF}$ is both on the cycle $\mathcal{C}_{\iC,\ib}$
and in $X$, and $G_3 \setminus X$ is a forest, too.
\end{proof}

\begin{lemma}\label{lem:fvs-seth:corr2}
If there exists a feedback vertex set $X$ of size at most $\goalsize$ in the graph $G$,
then $\Phi$ has a satisfying assignment.
\end{lemma}

\begin{proof}
As it is discussed at the begining of the construction process, we may assume that
no guard vertex is in $X$.
Let $X_h = X \cap \mathcal{H}$, similarly we define $X_v$, $X_p$, $X_q$, $X_x$, $X_y$ and $X_z$.
Furthermore, we assume that, among all feedback vertex sets of size at most $\goalsize$ in $G$
that do not contain any guard vertex, the set $X$ is such a one that $|X_p|$
is the smallest possible.

We now lower bound the sizes of the sets $X_h$, $X_v$, $X_p$, $X_q$, $X_x$, $X_y$ and $X_z$.
\begin{enumerate}
\item For each $1 \leq \iF \leq n'$, $0 \leq \iB < a$, $S \in \{1,2,3\}^\blocksize$,
  at least one of the vertices $y_{\iF,\iB}^S$ and $z_{\iF,\iB}^S$ is in $X$
  (as they are connected by a triangle edge), thus $|X_y| + |X_z| \geq n'a3^\blocksize = \goalsize_y$.
\item For each $1 \leq \iF \leq n'$ and $0 \leq \iB < a$ at least
one vertex of the set $X$ needs to hit the cycle $\mathcal{X}_{\iF,\iB}$, to $|X_x| \geq n'a = \goalsize_x$.
\item For each $1 \leq \iF \leq n'$, $0 \leq \iB < a$, $1 \leq \iv \leq \blocksize$
and $S \in \{1,2,3\}^\blocksize$ at least two vertices out of the triple
$\{q_{\iF,\iv,\iB}^{S,\ivup}:1 \leq \ivup \leq 3\}$ need to be included in $X$,
  thus $|X_q| \geq \blocksize n' \cdot 2a3^\blocksize = \goalsize_q$.
\item For each $1 \leq \iF \leq n'$, $0 \leq \iB < a$, $1 \leq \iv \leq \blocksize$
at least two vertices out of the triple $\{p_{\iF,\iv,\iB}^\ivup:1 \leq \ivup \leq 3\}$
need to be included in $X$. Moreover, if $p_{\iF,\iv,\iB}^\ivup \notin X$ for some
$1 \leq \ivup \leq 3$, then $v_{\iF,\iv,\iB}^\ivup \in X$, as otherwise the triangle
edge $v_{\iF,\iv,\iB}^\ivup p_{\iF,\iv,\iB}^\ivup$ is not covered by $X$.
Thus $|X_p| + |X_v| \geq \blocksize n' \cdot 3a = \goalsize_p + \goalsize_v$.
\item For each $1 \leq \iF \leq n'$, $0 \leq \iB < a$, $1 \leq \iv \leq \blocksize$,
  $1 \leq \ivup \leq 3$, such that $(\iB,\ivup) \neq (a-1,3)$,
  at least one endpoint of the triangle edge $h_{\iF,\iv,\iB}^{\ivup,1}h_{\iF,\iv,\iB}^{\ivup,2}$
  needs to be included in $X$. Thus $|X_h| \geq \blocksize n' (3a-1) = \goalsize_h$.
\end{enumerate}
As $|X| \leq \goalsize = \goalsize_h + \goalsize_v + \goalsize_p + \goalsize_q + \goalsize_x + \goalsize_y$, we infer that in all aforementioned inequalities we have equalities, and $r \notin X$.

Recall that we have assumed that $|X_p|$ is the smallest possible. Let
$1 \leq \iF \leq n'$, $0 \leq \iB < a$, $1 \leq \iv \leq \blocksize$ and focus on the triple
$\{p_{\iF,\iv,\iB}^\ivup:1 \leq \ivup \leq 3\}$. If it is wholy contained in $X$,
then  $X \setminus \{p_{\iF,\iv,\iB}^1\} \cup \{v_{\iF,\iv,\iB}^1\}$ is also a feedback vertex
set of $G$, of not greater size, not containing any guard vertex, and with smaller size of $|X_p|$.
Thus $X$ contains exactly two vertices out of each such triple,
$|X_p| = \goalsize_p$, $|X_v| = \goalsize_v$ and $X_v$ contains exactly one vertex
out of each triple $\{v_{\iF,\iv,\iB}^\ivup:1 \leq \ivup \leq 3\}$.

We strengthen the above observation by showing the following claim: for any $1 \leq \iF \leq n'$,
$1 \leq \iv \leq \blocksize$ and any three consecutive vertices $v_A$, $v_B$, $v_C$
on the path $\bigcup_{\iB=0}^{a-1} \mathcal{V}_{\iF,\iv,\iB}$, at least one of these vertices is in $X_v$.
By contradiction, assume that $v_A, v_B, v_C \notin X$. Recall that the graph $G$
contains vertices $h_A^1$, $h_A^2$, $h_B^1$, $h_B^2$, edges $rh_A^1$, $rh_A^2$, $rh_B^1$,
$rh_B^2$, $h_A^1v_A$, $h_A^2v_B$, $h_B^1v_B$, $h_B^2v_C$ and triangle edges
$h_A^1h_A^2$, $h_B^1h_B^2$. 
Note that $|\{h_A^1,h_A^2,h_B^1,h_B^2\} \setminus X| = 2$,
as $|X_h| = \goalsize_h$ and $X_h$ contains exactly one vertex from each pair connected by
a triangle edge. These two vertices in $\mathcal{H} \setminus X$, together with $v_A$, $v_B$,
$v_C$ and the root $r$ induce a subgraph of $G\setminus X$ with $6$ vertices and $6$ edges, a contradiction.

For $1 \leq \iF \leq n'$ and $1 \leq \iv \leq \blocksize$
let us define a sequence $s_{\iF,\iv}(\iB)$, $0 \leq \iB < a$,
such that $v_{\iF,\iv,\iB}^\ivup \in X_v$ iff $\ivup = s_{\iF,\iv}(\iB)$.
By the observation made in the previous paragraph we infer that the sequence $s_{\iF,\iv}$
cannot increase, thus its value can change at most twice.
As $a = \numcopies$, we infer that there exists an index $0 \leq \ib < \numcopiessmall$
such that for all $1 \leq \iF \leq n'$, $1 \leq \iv \leq \blocksize$
we have $s_{\iF,\iv}(m\ib)=s_{\iF,\iv}(m\ib+\iC)$ for all $0 \leq \iC < m$.
For each block $F_\iF$, let $S_\iF = (s_{\iF,\iv}(m\ib))_{\iv=1}^\blocksize \in \{1,2,3\}^\blocksize$
and let $\phi$ be an assignment that corresponds to the sequence $S_\iF$ for each block $F_\iF$.
We claim that $\phi$ satisfies $\Phi$.

Take any clause $C_\iC$, $0 \leq \iC < m$.
Take any block $F_\iF$. As $|X_q| = \goalsize_q$, the set $X_q$
includes exactly two vertices out of each triple
$\{q_{\iF,\iv,m\ib+\iC}^{S,\ivup}:1 \leq \ivup \leq 3\}$
 for $1 \leq \iv \leq \blocksize$, $S \in \{1,2,3\}^\blocksize$.
As $q_{\iF,\iv,m\ib+\iC}^{S,\ivup}$ is connected to $v_{\iF,\iv,m\ib+\iC}^\ivup$
by a triangle edge, we infer that 
$q_{\iF,\iv,m\ib+\iC}^{S,\ivup} \notin X$ iff $v_{\iF,\iv,m\ib+\iC}^\ivup \in X$,
which is equivalent to $S_\iF(\iv) = \ivup$.
Thus $x_{\iF,m\ib+\iC}^{S_\iF} \in X$, as otherwise the cycle $\mathcal{Q}_{\iF,m\ib+\iC}^{S_\iF}$
is disjoint with $X$. As $|X_x| = \goalsize_x$, the set $X_x$ contains exactly one vertex
out of each cycle $\mathcal{X}_{\iF,\iB}$, and we infer that $x_{\iF,m\ib+\iC}^S \notin X$
for $S \neq S_\iF$.
Recall that $x_{\iF,m\ib+\iC}^S$ and $y_{\iF,m\ib+\iC}^S$ are connected by a triangle edge,
thus $y_{\iF,m\ib+\iC}^S \in X$ for $S \neq S_\iF$.
As $|X_y|+|X_z| = \goalsize_y$, we know that the set $X$ contains exactly one endpoint
out of each triangle edge $y_{\iF,m\ib+\iC}^S z_{\iF,m\ib+\iC}^S$, and we infer
that if $z_{\iF,m\ib+\iC}^S \in X$ then $S = S_\iF$.
Finally, if $X$ is a feedback vertex set in $G$, $X$ hits the cycle $\mathcal{C}_{\iC,\ib}$,
thus there exists a block $F_\iF$ and a sequence $S \in \{1,2,3,\}^\blocksize$
such that $z_{\iF,m\ib+\iC}^S \in X$ and the assignment of the variables of the block $F_\iF$
that corresponds to $S$ satisfies $C_\iC$. However, we have proven that $z_{\iF,m\ib+\iC}^S \in X$
implies $S=S_\iF$, thus $\phi$ satisfies $C_\iC$ and the proof is finished.
\end{proof}

\paragraph{Pathwidth bound}

\begin{lemma}\label{lem:fvs-seth:pathwidth}
Pathwidth of the graph $G$ is at most $\blocksize n'+O(\blocksize 3^\blocksize)$.
Moreover a path decomposition of such width can found in polynomial time.
\end{lemma}

\begin{proof}
We give a mixed search strategy to clean the graph with $\blocksize n'+O(\blocksize 3^\blocksize)$
searchers.
First we put a searcher in the root $r$.
This searcher remains there till the end of the cleaning process.

For a gadget $\groupg_{\iF,\iB}$ we call the vertices
$v_{\iF,\iv,\iB}^1$ for $1 \le \iv \le \blocksize$,
  as {\em entry vertices}.
We search the graph in $a = \numcopies$ rounds.
In the beginning of round $\iB$ ($0 \leq \iB < a$)
there are searchers on the entry vertices of the gadget $\groupg_{\iF,\iB}$ for every
$1 \leq \iF \le n'$.
Let $0 \le \iC < m$ and $0 \le \ib < \numcopiessmall$ be integers such that $\iB = \iC + m\ib$.
We place a searcher on the last vertex of the cycle $\mathcal{C}_{\iC,\ib}$
(recall that the vertices on $\mathcal{C}_{\iC,\ib}$ are sorted by the block number).
Then, for each $1 \le \iF \le n'$ in turn we:
\begin{itemize}
  \item put $O(\blocksize 3^\blocksize)$
    searchers on all vertices of the group gadget $\groupg_{\iF,\iB}$,
  \item put $\blocksize$ searchers on entry vertices of the group gadget $\groupg_{\iF,\iB+1}$
  (except for the last round),
  \item put a searcher in the first vertex after the vertices of $\groupg_{\iF,\iB} \cap \mathcal{Z}$ on
  the cycle $\mathcal{C}_{\iC,\ib}$,
  \item remove searchers from all vertices of the group gadget $\groupg_{\iF,\iB}$.
\end{itemize}
The last step of the round is removing the remaining searcher on the cycle $\mathcal{C}_{\iC,\ib}$.
After the last round the whole graph $G$ is cleaned.
Since we reuse $O(\blocksize 3^\blocksize)$ searchers for cleaning group gadgets,
$\blocksize n'+O(\blocksize 3^\blocksize)$ searchers suffice to clean the graph.

Using the above graph cleaning process a path decomposition of width
$\blocksize n'+O(\blocksize 3^\blocksize)$ can be constructed in polynomial time.
\end{proof}

\begin{proof}[Proof of Theorem \ref{thm:fvs-seth}]
Suppose \fvs can be solved in $(3-\eps)^p |V|^{O(1)}$ time
provided that we are given a path decomposition of $G$ of width $p$.
Let $\lambda = \log_3(3-\eps) < 1$.
We choose $\blocksize$ large enough such that
$\frac{\log 3^\blocksize}{\lfloor \log 3^\blocksize \rfloor} < \frac{1}{\lambda}$.
Given an instance of SAT we construct an instance of \fvs
using the above construction and the chosen value of $\blocksize$.
Next we solve \fvs using the $3^{\lambda p} |V|^{O(1)}$ time algorithm.
Lemmata~\ref{lem:fvs-seth:corr1}, \ref{lem:fvs-seth:corr2} ensure correctness,
whereas Lemma~\ref{lem:fvs-seth:pathwidth}
implies that running time of our algorithm is $3^{\lambda \blocksize n'} |V|^{O(1)}$,
however we have 
$$3^{\lambda \blocksize n'}=2^{\lambda \blocksize n'\log 3}=2^{\lambda n'\log 3^\blocksize}\leq 2^C \cdot2^{\lambda n \log 3^\blocksize / \lfloor \log 3^\blocksize \rfloor } = 2^C \cdot 2^{\lambda' n}$$ 
for some $\lambda' < 1$ and $C=\lambda \log 3^\blocksize$. This concludes the proof.
\end{proof}

\end{document}